\documentclass[rmp,aps,twocolumn]{revtex4-1}
\usepackage{graphicx}
\usepackage{amsmath}
\usepackage{amsthm}
\usepackage{amssymb}
\usepackage[usenames]{color}
\usepackage{hyperref}
\RequirePackage[normalem]{ulem}
\hypersetup{
    colorlinks=true,       
    linkcolor=cyan,          
    citecolor=magenta,        
    filecolor=magenta,      
    urlcolor=cyan,           
    runcolor=cyan
}

\usepackage{cancel}
\usepackage{bm}
\usepackage{threeparttable}
\usepackage{subfigure}
\usepackage{upgreek }
\usepackage{marginnote}

\newcommand{\beq}{\begin{equation}}
\newcommand{\eeq}{\end{equation}}
\newcommand{\beqnn}{\begin{equation*}}
\newcommand{\eeqnn}{\end{equation*}}
\newcommand{\bea}{\begin{eqnarray}}
\newcommand{\eea}{\end{eqnarray}}
\newcommand{\beann}{\begin{eqnarray*}}
\newcommand{\eeann}{\end{eqnarray*}}
\newcommand{\bes} {\begin{subequations}}
\newcommand{\ees} {\end{subequations}}

\newcommand{\av}[1]{\langle #1\rangle}
\newcommand{\braket}[2]{\langle #1 | #2\rangle}
\newcommand{\ket}[1]{ | #1\rangle}
\newcommand{\bra}[1]{\langle #1 | }
\newcommand{\ketbra}[2]{|#1\rangle\langle #2|}
\newcommand{\mA}{\mathcal{A}}

\newcommand{\mM}{\mathcal{M}}
\newcommand{\mB}{\mathcal{B}}
\newcommand{\mE}{\mathcal{E}}
\newcommand{\mP}{\mathcal{P}}

\newcommand{\mU}{\mathcal{U}}
\newcommand{\mL}{\mathcal{L}}
\newcommand{\mH}{\mathcal{H}}
\newcommand{\mS}{\mathcal{S}}

\newcommand{\ident}{\openone}
\newcommand{\Tr}{\mathrm{Tr}}
\newcommand{\veps}{\varepsilon}
\newcommand{\eps}{\epsilon}
\newcommand{\1}{\openone}

\newcommand{\abs}[1]{\ensuremath{\left| #1 \right|}}

\newcommand{\norm}[1]{\left\| #1 \right\|}
\newcommand{\ignore}[1]{}
\newcommand{\Texp}{\mathrm{Texp}}
\newcommand{\poly}{\mathrm{poly}}
\newtheorem{definition}{Definition}
\newtheorem{theorem}{Theorem}
\newtheorem{lemma}{Lemma}
\newtheorem{claim}{Claim}

\newcommand{\vp}{\vec{p}}
\newcommand{\Or}{\mathcal{O}}

\begin{document}
\title{Adiabatic Quantum Computing}
\author{Tameem Albash}
\affiliation{Information Sciences Institute, University of Southern California, Marina del Rey, CA 90292}
\affiliation{Department of Physics and Astronomy, University of Southern California, Los Angeles, California
90089, USA}
\affiliation{Center for Quantum
  Information Science \& Technology, University of Southern California, Los Angeles, California
90089, USA}

\author{Daniel A. Lidar}
\affiliation{Department of Physics and Astronomy, University of Southern California, Los Angeles, California
90089, USA} 
\affiliation{Center for Quantum
  Information Science \& Technology, University of Southern California, Los Angeles, California 90089, USA}
\affiliation{Department of Electrical Engineering, University of Southern California, Los Angeles, California
90089, USA}
\affiliation{Department of Chemistry, University of Southern California, Los Angeles, California
90089, USA}

\begin{abstract}
Adiabatic quantum computing (AQC) started as an approach to solving optimization problems, and has evolved into an important universal alternative to the  standard circuit model of quantum computing, with deep connections to both classical and quantum complexity theory and condensed matter physics.  In this review we give an account of most of the major theoretical developments in the field, while focusing on the closed-system setting.  The review is organized around a series of topics that are essential to an understanding of the underlying principles of AQC, its algorithmic accomplishments and limitations, and its scope in the more general setting of computational complexity theory.  We present several variants of the adiabatic theorem, the cornerstone of AQC, and we give examples of explicit AQC algorithms that exhibit a quantum speedup.  We give an overview of several proofs of the universality of AQC and related Hamiltonian quantum complexity theory.  We finally devote considerable space to Stoquastic AQC, the setting of most AQC work to date, where we discuss obstructions to success and their possible resolutions.
\end{abstract}
\maketitle

\tableofcontents

\section{Introduction}

Quantum computation (QC) originated with Benioff's proposals for quantum Turing machines  \cite{Benioff:80,Benioff:82} and Feynman's ideas for circumventing the difficulty of simulating quantum mechanics by classical computers \cite{Feynman1}. This led to Deutsch's proposal for universal QC in terms of what has become the ``standard" model: the circuit, or gate model of QC \cite{Deutsch:89}. 
Adiabatic quantum computation (AQC) is based on an idea that is quite distinct from the circuit model. Whereas in the latter a computation may in principle evolve in the entire Hilbert space and is encoded into a series of unitary quantum logic gates, in AQC the computation proceeds from an initial Hamiltonian whose ground state is easy to prepare, to a final Hamiltonian whose ground state encodes the solution to the computational problem. The adiabatic theorem guarantees that the system will track the instantaneous ground state provided the Hamiltonian varies sufficiently slowly. It turns out that this approach to QC has deep connections to condensed matter physics, computational complexity theory, and heuristic algorithms. 

In its first incarnation, the idea of encoding the solution to a computational problem in the ground state of a quantum Hamiltonian appeared as early as 1988, in the context of solving classical combinatorial optimization problems, where it was called quantum stochastic optimization \cite{Apolloni:1989fj}.\footnote{Even though \cite{Apolloni:1989fj} was published in 1989, it was submitted in 1988, before \cite{Apolloni:88}, which referenced it.} It was renamed quantum annealing (QA) in \cite{Apolloni:88} and reinvented several times \cite{Somorjai:1991oq,Amara:1993rt,finnila_quantum_1994,kadowaki_quantum_1998}.\footnote{It was called ``quasi-quantal method" in  \cite{Somorjai:1991oq}, ``imaginary-time algorithm" in \cite{Amara:1993rt}, and quantum annealing in \cite{finnila_quantum_1994,kadowaki_quantum_1998}. The latter term has become widely accepted.} These early papers emphasized that QA was to be understood as an \emph{algorithm} that exploits simulated quantum (rather than thermal) fluctuations and tunneling, thus providing a quantum-inspired version of simulated annealing (SA) \cite{kirkpatrick_optimization_1983}. The first direct comparison between QA and SA \cite{kadowaki_quantum_1998} suggested that QA can be more powerful.

A very different approach was taken via an \emph{experimental} implementation of QA in a disordered quantum ferromagnet \cite{Brooke1999,brooke_tunable_2001}. This provided the impetus to reconsider QA from the perspective of quantum computing, i.e., to consider a dedicated device that solves optimization problems by exploiting quantum evolution. Thus was born the idea of the quantum adiabatic algorithm (QAA) \cite{farhi_quantum_2000,Farhi:01} [also referred to as adiabatic quantum optimization (AQO) \cite{Smelyanskiy:01,Reichardt:2004}], wherein a physical quantum computer solves a combinatorial optimization problem by evolving adiabatically in its ground state. The term adiabatic quantum computation we shall use here was introduced in \cite{Dam:2001fk}, though the context was still optimization.\footnote{The first documented use of the term ``adiabatic quantum computation" was in \cite{Averin:1998gb}, but the context was an adiabatic implementation of a quantum logic gate in the circuit model.} 

Adiabatic quantum algorithms for optimization problems typically use ``stoquastic" Hamiltonians, characterized by having only non-positive off-diagonal elements in the computational basis. Adiabatic quantum computation with non-stoquastic Hamiltonians is as powerful as the circuit model of quantum computation \cite{aharonov_adiabatic_2007}. In other words, non-stoquastic AQC and all other models for universal quantum computation can simulate one another with at most polynomial resource overhead. For this reason the contemporary use of the term AQC typically refers to the general, non-stoquastic setting, thus extending beyond optimization to any computational. When discussing the case of stoquastic Hamiltonians we will use the term ``stoquastic AQC" (StoqAQC).\footnote{Quantum annealing, like StoqAQC, currently usually involves stoquastic Hamiltonians; we differentiate between them in the following sense: we restrict StoqAQC purely to the case of closed system evolutions, whereas QA refers to a (not necessarily adiabatic) evolution in an open system.}

For most of this review we essentially adopt the definition of AQC from \cite{aharonov_adiabatic_2007}, as this definition allows for the proof of the equivalence with the circuit model, and is thus used to establish the universality of AQC. Interestingly, this proof builds on one of the first QC ideas due to Feynman, which was later shown to allow a general purpose quantum computation to be embedded in the ground state of a quantum system \cite{Feynman:1985ul,Kitaev:book}. A related ground state embedding approach was independently pursued in \cite{Mizel:01,Mizel:02}, around the same time as the original development of the QAA. To define AQC, we first need the concept of a $k$-local Hamiltonian, which is a Hermitian matrix $H$ acting on the space of $p$-state particles that can be written as $H = \sum_{i=1}^r H_i$ where each $H_i$ acts non-trivially on at most $k$ particles, i.e., $H_i = h \otimes \ident$ where $h$ is a Hamiltonian on at most $k$ particles, and $\ident$ denotes the identity operator. 

\begin{definition}[Adiabatic Quantum Computation]
\label{def:AQC}
A $k$-local adiabatic quantum computation is specified by two $k$-local Hamiltonians, $H_0$ and $H_1$, acting on $n$ $p$-state particles, $p\geq 2$. The ground state of $H_0$ is unique and is a product state. The output is a state that is $\veps$-close in $\ell_2$-norm to the ground state of $H_1$. Let $s(t):[0,t_f]\mapsto [0,1]$ (the ``schedule") and let $t_f$ be the smallest time such that the final state of an adiabatic evolution generated by $H(s)=(1-s)H_0+sH_1$ for time $t_f$ is $\veps$-close in $\ell_2$-norm to the ground state of $H_1$. 
\end{definition}

Several comments are in order. (1) A uniqueness requirement was imposed on the ground state of $H_1$ in \cite{aharonov_adiabatic_2007}, but this is not necessary. E.g., in the setting where $H_1$ represents a classical optimization problem, multiple final ground states do not pose a problem as any of the final states represents a solution to the optimization problem. 
(2) Sometimes it is beneficial to consider adiabatic quantum computation in an excited state (see, e.g., Sec.~\ref{sec:stoq-QMA-comp}). (3) As already noted in \cite{aharonov_adiabatic_2007}, it is useful to allow for more general ``paths" between $H_0$ and $H_1$, e.g., by introducing an intermediate ``catalyst'' Hamiltonian that vanishes at $s=0,1$ (see, e.g., Sec.~\ref{sec:catalyst}). 

A crucial question that will occupy us throughout this review is the cost of running an algorithm in AQC. In the circuit model the cost is equated with the number of gates, so one cost definition would be to the count the number of gates needed to simulate the equivalent adiabatic process.  This cost definition presupposes that the circuit model is fundamental, which may be unsatisfactory. In AQC one might be tempted to just use the run time $t_f$, but in order for this quantity to be meaningful it is necessary to  define an appropriate energy scale for the Hamiltonian. In \cite{aharonov_adiabatic_2007} the cost of the adiabatic algorithm was defined to be the dimensionless quantity 
\beq \label{eq:cost}
\mathrm{cost} = t_f \max_s\norm{H(s)}  \ , 
\eeq
in order to prevent the cost from being made arbitrarily small by changing the time units, or distorting the scaling of the algorithm by multiplying the Hamiltonians by some size-dependent factor.\footnote{Unless stated otherwise, we shall always use $\norm{\cdot}$ to denote the operator norm for operators: 
\[\Vert A \Vert = \sup \left\{ \Vert A \ket{\psi} \Vert : \ket{\psi} \in \mH \ \mathrm{with} \ \braket{\psi}{\psi} = 1 \right\} \]
(i.e., the largest singular value of the operator $A$),  and the Euclidean vector norm for vectors. Which one is used will be clear by context.}  
From hereon we will focus on the run time $t_f$, which should be compared to the circuit depth of analogous circuit model algorithms, whereas the full cost in Eq.~\eqref{eq:cost} should be compared to the circuit gate count.

The run time $t_f$ of an adiabatic algorithm scales at worst as $1/\Delta^3$, where $\Delta$ is the minimum eigenvalue gap between the ground state and the first excited state of the Hamiltonian of the adiabatic algorithm \cite{Jansen:07}. If the Hamiltonian is varied sufficiently smoothly, one can improve this to $O(1/\Delta^2)$ up to a polylogarithmic factor in $\Delta$ \cite{Elgart:2012fk}. While these are useful sufficient conditions, they involve bounding the minimum eigenvalue gap of a complicated many-body Hamiltonian, a notoriously difficult problem. This is one reason that AQC has generated so much interest among physicists: it has a rich connection to well studied problems in condensed matter physics. For example, because of the dependence of the run time on the gap, the performance of quantum adiabatic algorithms is strongly influenced by the type of quantum phase transition the same system would undergo in the thermodynamic limit \cite{Latorre:04}.

Nevertheless, a number of examples are known where the gap analysis can be carried out. For example, adiabatic quantum computers can perform a process analogous to Grover search \cite{Grover:97a}, and thus provide a quadratic speedup over the best possible classical algorithm for the Grover search problem \cite{Roland:2002ul}. Other examples are known where the gap analysis can be used to demonstrate that AQC provides a speedup over classical computation, including adiabatic versions of some of the keystone algorithms of the circuit model. However, much more common is the scenario where either the gap analysis reveals no speedup over classical computation, or where a clear answer to the speedup question is unavailable. In fact, the least is known about adiabatic quantum speedups in the original setting of solving classical combinatorial optimization problems. This remains an area of very active research, partly due to the original (still unmaterialized) hope that the QAA would deliver quantum speedups for NP-complete problems \cite{Farhi:01}, and partly due the availability of commercial quantum annealing devices such as those manufactured by D-Wave Systems Inc. \cite{Dwave}, designed to solve optimization problems using stoquastic Hamiltonians. 

The goal of this article is to review the field of AQC from its inception, with a focus on the \emph{closed system} case. That is, we omit the fascinating topic of AQC in open systems coupled to an environment. This includes all experimental work on AQC, and all work on quantum error correction and suppression methods for AQC, as these topics deserve a separate review \cite{Albash-Lidar:RMP-colloq} and including them here would limit our ability to do justice to the many years of work on AQC in closed systems, an extremely rich topic with many elegant results. For the same reasons we also omit the blossoming and closely related fields of holonomic QC \cite{HQC}, topological QC \cite{Nayak:2008aa}, and adiabatic state preparation for quantum simulation \cite{Babbush:2014}. To achieve our goal we organized this review around a series of
topics that are essential to an understanding of the underlying principles of AQC, its algorithmic accomplishments and limitations, and its scope in the more general setting of computational complexity theory. 

We begin by reviewing the adiabatic theorem in Sec.~\ref{sec:AdiabaticTheorem}. The adiabatic theorem forms the backbone of AQC: it provides a sufficient condition for the success of the computation, and in doing so provides the run time of a computation in terms of the eigenvalue gap $\Delta$ of the Hamiltonian and the Hamiltonian's time-derivative. In fact there is not one single adiabatic theorem, and we review a number of different variants that provide different run time requirements, under different smoothness and differentiability assumptions about the Hamiltonian.

Next, we review in Sec.~\ref{sec:algorithms} the handful of \emph{explicit} algorithms for which AQC is known to give a speedup over classical computation. The emphasis is on ``explicit", since Sec.~\ref{sec:universality} provides several proofs for the universality of AQC in terms of its ability to efficiently simulate the circuit model, and \textit{vice versa}. This means that every quantum algorithm that provides a speedup in the circuit model [many of which are known \cite{Jordan-algorithm_zoo}] can in principle be implemented with up to polynomial overhead in AQC. That the number of explicit AQC algorithms is still small is therefore likely to be a reflection of the relatively modest amount of effort that has gone into establishing such results compared to the circuit model. However, there is also a real difficulty, in that performing the gap analysis in order to establish the actual scaling (beyond the polynomial-time equivalence) is, as already mentioned above, in many cases highly non-trivial. A second non-trivial aspect of establishing a speedup by AQC is that when such a speedup is polynomial, relying on universality is insufficient, since the polynomial overhead involved in implementing the transformation from the circuit model to AQC can then swamp the speedup. A good example is the case of Grover's algorithm, where a direct use of the equivalence to the circuit model would not suffice; instead, what is required is a careful analysis and choice of the adiabatic schedule $s(t)$ in order to realize the quantum speedup. 

In Sec.~\ref{sec:QMAC} we go beyond universality into Hamiltonian quantum complexity theory. This is an active contemporary research area, that started with the introduction of the complexity class QMA (``quantum Merlin-Arthur") as the natural quantum generalization of the classical complexity classes "nondeterministic polynomial time" (NP) and MA \cite{Kitaev:book}. The theory of QMA-completeness deals with decision problems that are efficiently checkable using quantum computers. 
It turns out that these decision problems can be formulated naturally in terms of $k$-local Hamiltonians, of the same type that appear in the proofs of the universality of AQC. Thus universality and Hamiltonian quantum complexity studies are often pursued hand-in-hand, and a reduction of $k$ as well as the dimensionality $p$ of the particles appearing in these constructions is one of the main goals. For example, already $k=2$ and $p=2$ leads to both universal AQC and QMA-complete Hamiltonians in 2D, while in 1D $p>2$ is needed for both.\footnote{We
say that $H$ is a $d$D ($d$-dimensional) Hamiltonian if the particles are arranged on a $d$-dimensional grid
and the summands of $H$ couple only pairs of nearest neighbor particles. Note that being $d$D implies that the Hamiltonian is $2$-local.}

We turn our attention to StoqAQC in Sec.~\ref{sec:QA}. This is the setting of the vast majority of AQC work to date. The final Hamiltonian $H_1$ is assumed to be a classical Ising model Hamiltonian, typically (but not always) representing a hard optimization problem such as a spin glass. The initial Hamiltonian $H_0$ is typically assumed to be proportional to a transverse field, i.e., $\sum_i \sigma_i^x$, whose ground state is the uniform superposition state in the computational basis. AQC with stoquastic Hamiltonians is probably less powerful than universal quantum computation, but examples can be constructed which show that it may nevertheless be more powerful than classical computation. Moreover, if we relax the definition of AQC to allow for computation using excited states, it turns out that stoquastic Hamiltonians can even be QMA-complete and support universal AQC. To do justice to this mixed and complicated picture, we first review examples where it is known that StoqAQC does not outperform classical computation (essentially because the eigenvalue gap $\Delta$ decreases rapidly with problem size but classical algorithms do not suffer a slowdown), then discuss examples where StoqAQC offers a quantum scaling advantage over simulated annealing in the sense that it outperforms classical simulated annealing but not necessarily other classical algorithms, and finally point out examples where it is currently not known whether StoqAQC offers a quantum speedup, but one might hope that it does. We also discuss the role of potential quantum speedup mechanisms, in particular tunneling and entanglement.

The somewhat bleak picture regarding StoqAQC should not necessarily be a cause for pessimism. Some of the obstacles in the way of a quantum speedup can be overcome or circumvented, as we discuss in Sec.~\ref{sec:fix-AQC}. In all cases this involves modifying some aspect of the Hamiltonian, either by optimizing the schedule $s(t)$, or by adding certain terms to the Hamiltonian such that small gaps are avoided. This can result in a non-stoquastic Hamiltonian whose final ground state is the same as that of the original Hamiltonian, with an exponentially small gap (often corresponding to a first order quantum phase transition) changing into a polynomially small gap (often corresponding to a second order phase transition). Another type of modification is to give up adiabatic evolution itself, and allow for diabatic transitions. While this results in giving up the guarantee of convergence to the ground state provided by the adiabatic theorem, it can be a strategy that results in better run time scaling for the same Hamiltonian than an adiabatic one.

We conclude with an outlook and discussion of future directions in Sec.~\ref{sec:outlook}. Various technical details are provided in the Appendix.


\section{Adiabatic Theorems} 
\label{sec:AdiabaticTheorem}
The origins of the celebrated quantum adiabatic approximation date back
to Einstein's ``Adiabatenhypothese'': ``If a system be affected
in a reversible adiabatic way, allowed motions are transformed into
allowed motions'' \cite{Einstein:adiabatic}. Ehrenfest was the first
to appreciate the importance of adiabatic invariance, guessing---before
the advent of a complete quantum theory--- that quantum laws would
only allow motions which are invariant under adiabatic perturbations
\cite{Ehrenfest:adiabatic}. The more familiar, modern version of the
adiabatic approximation was put forth by Born and Fock already in 1928 for the case of
discrete spectra~\cite{born_beweis_1928}, after the development of the Born-Oppenheimer approximation for the separation of electronic and nuclear degrees of freedom a year earlier \cite{Born:1927xc}. Kato put the approximation on a firm mathematical foundation in 1950 \cite{kato_adiabatic_1950} and arguably proved the first quantum adiabatic theorem. 

The adiabatic approximation states, roughly, that for a system initially prepared in an eigenstate (e.g., the ground state) $|\veps_{0}(0)\rangle $ of a time-dependent
Hamiltonian $H(t)$, the time evolution governed by the Schr\"{o}dinger equation 
\beq
i\frac{\partial |\psi (t)\rangle }{\partial t}=H(t)|\psi
(t)\rangle 
\eeq 
(we set $\hbar\equiv 1$ from now on) will approximately keep the actual state $|\psi (t)\rangle $ of
the system in the corresponding instantaneous ground state (or other eigenstate) $|\veps_{0}(t)\rangle $ of $H(t)$, provided that $H(t)$ varies ``sufficiently slowly". Quantifying the exact nature of this slow variation is the subject of the Adiabatic Theorem (AT), which exists in many variants. In this section we provide an overview of these variants of the AT, emphasizing aspects that are pertinent to AQC. We discuss the ``folklore" adiabatic condition, that the total evolution time $t_f$ should be large on the timescale set by the square of the inverse gap, and the question of how to ensure a high fidelity between the actual state and the ground state. We then discuss a variety of rigorous versions of the AT, emphasizing different assumptions and consequently different performance guarantees. Throughout this discussion, it is important to keep in mind that ultimately the AT provides only an \emph{upper bound} on the evolution time required to achieve a certain fidelity between the actual state and the target eigenstate of $H(t)$.
 
\subsection{Approximate versions}

Let $|\veps_{j}(t)\rangle $
($j\in \{0,1,2,\ldots \}$) denote the instantaneous eigenstate of $H(t)$
with energy $\veps_{j}(t)$ such that $\veps_{j}(t)\leq \veps_{j+1}(t)$ $\forall j,t$, i.e., $H(t)|\veps_{j}(t)\rangle =\veps_{j}(t)|\veps_{j}(t)\rangle $ and $j=0$ denotes the (possibly degenerate) ground state. Assume that the initial state is prepared in one of the eigenstates $\ket{\veps_j(0)}$.

The simplest as well as one of the oldest traditional versions of
the adiabatic approximation states that a system initialized in an eigenstate $\ket{\veps_j(0)}$ will remain in the same instantaneous eigenstate $\ket{\veps_j(t)}$ (up to a global phase) for all $t\in [0,t_f]$, where $t_f$ denotes the final time, provided \cite{messiah}:
\beq
\max_{t\in [0,t_f]} \frac{|\langle \veps_i\ket{\partial_t \veps_j}|}{|\veps_i-\veps_j|} = \max_{t\in [0,t_f]} \frac{|\bra{\veps_i}\partial_t H\ket{\veps_j}|}{|\veps_i-\veps_j|^2} \ll 1 \,\,\, \forall i\neq j\ .
\label{eq:AT-folklore}
\eeq
This version has been critiqued \cite{PhysRevLett.93.160408,tong_quantitative_2005,du:060403,Wu:2008aa} on the basis of arguments and examples involving a separate, independent timescale. Indeed, if the Hamiltonian includes an oscillatory driving term then the eigenstate population will oscillate with a timescale determined by this term, that is independent of $t_f$, even if the adiabatic criterion \eqref{eq:AT-folklore} is satisfied.\footnote{For example, it is easily checked that when $H(t) = a \sigma^z + b \sin(\omega t) \sigma^x$, the adiabatic condition \eqref{eq:AT-folklore} reduces to $|b\omega| \ll a^2$. However, even if this condition is satisfied the population can oscillate between the two eigenstates: at resonance (when $\omega \approx 2a$) the system undergoes Rabi oscillations with period $\pi/|b|$, a timescale that is independent of $t_f$.} 

A more careful statement of the adiabatic condition that excludes such additional timescales is thus required. The first step is to assume that the Hamiltonian $H_{t_f}(t)$ in the Schr\"{o}dinger equation $\partial \ket{\psi_{t_f}(t)}/\partial t = -i H_{t_f}(t)\ket{\psi_{t_f}(t)}$ can be written as $H_{t_f}(s t_f) = H(s)$, where $s\equiv t/t_f\in [0,1]$ is the dimensionless time, and $H(s)$ is $t_f$-independent. This includes the ``interpolating" Hamiltonians of the type often considered in AQC, i.e., $H(s) = A(s)H_0 + B(s)H_1$ [where $A(s)$ and $B(s)$ are monotonically decreasing and increasing, respectively] and excludes cases with multiple timescales.\footnote{For example, a case such as $H(t) = a \sigma^z + b \sin(\omega t) \sigma^x$ is now excluded since after a change of variables we have $H(s) = a \sigma^z + b \sin(\omega t_f s) \sigma^x$ and evidently $H(s)$ still depends on $t_f$.} The Schr\"{o}dinger equation then becomes 
\beq
\frac{1}{t_f}\frac{\partial \ket{\psi_{t_f}(s)}}{\partial s} = -i H(s)\ket{\psi_{t_f}(s)}\ ,
\eeq
which is the starting point for all rigorous adiabatic theorems.

A more careful adiabatic condition subject to this formulation is given by \cite{Amin:09}:
\beq
\frac{1}{t_f} \max_{s \in [0,1]} \frac{|\langle \veps_i(s)|\partial_s{H}(s)|\veps_j(s)\rangle |}{|\veps_i(s)-\veps_j(s)|^2} \ll 1\,\, \forall j\neq i \ .
\label{eq:AT-folklore2}
\eeq
The conditions \eqref{eq:AT-folklore} and \eqref{eq:AT-folklore2} give rise to the widely used criterion that the total adiabatic evolution time should be large on the timescale set by the minimum of the square of the inverse spectral gap  $
\Delta_{ij}(s) = \veps_{i}(s)-\veps_{j}(s)$. In most cases one is interested in the ground state, so that $\Delta_{ij}(s)$ is replaced by 
\beq
\Delta \equiv \min_{s\in [0,1]} \Delta(s) = \min_{s\in [0,1]} \veps_{1}(s)-\veps_{0}(s)\ .
\label{eq:Delta}
\eeq 
However, arguments such as those leading to Eqs.~\eqref{eq:AT-folklore} and \eqref{eq:AT-folklore2} are approximate, in the sense that they do not result in strict inequalities and do not result in bounds on the closeness between the actual time-evolved state and the desired eigenstate. We discuss this next.

\subsection{Rigorous versions}

The first rigorous adiabatic condition is due to Kato \cite{kato_adiabatic_1950}, and was followed by numerous alternative derivations and improvements giving tighter bounds under various assumptions, e.g., \cite{Teufel:book,Nenciu:93,Avron:99,Hagedorn:2002kx,Reichardt:2004,Ambainis:04,Jansen:07,lidar:102106,PhysRevA.77.042319,Cheung:2011aa,Elgart:2012fk,Ge:2015wo}. All these rigorous results are more
severe in the gap condition than the traditional criterion, and they involve
a power of the norm of time derivatives of the Hamiltonian, rather
than a transition matrix element.

We summarize a few of these results here, and refer the reader to the original literature for their proofs. For simplicity we always assume that the system is initialized in its ground state and that the gap is the ground state gap \eqref{eq:Delta}. We also assume that for all $s \in [0, 1]$ the Hamiltonian $H(s)$ has an eigenprojector $P(s)$ with eigenenergy $\veps_0(s)$, and that the gap never vanishes, i.e., $\Delta > 0$.\footnote{There is a weaker form of the AT, where one does not require a non-vanishing gap \cite{Avron:99}. In this case, as in Theorem~\ref{th:AT1}, the estimate on the error term is $o(1)$ as $t_f\to\infty$.} The ground state, and hence the projector $P(s)$, is allowed to be (even infinitely) degenerate. $P(s)$ represents the ``ideal" adiabatic evolution.

Let $P_{t_f}(s) = \ketbra{\psi_{t_f}(s)}{\psi_{t_f}(s)}$. This is the projector onto the time-evolved solution of the Schr\"{o}dinger equation, i.e., the ``actual" state. Adiabatic theorems are usually statements about the ``instantaneous adiabatic distance" $\| P_{t_f}(s) - P(s)\|$ between the projectors associated with the actual and ideal evolutions, or the ``final-time adiabatic distance" $\| P_{t_f}(1) - P(1)\|$. 
Typically, adiabatic theorems give a bound of the form $O(1/t_f)$ for the instantaneous case, and a bound of the form $O(1/t_f^n)$ for any $n\in\mathbb{N}$ for the final-time case. After squaring, these projector-distance bounds immediately become bounds on the transition probability, defined as $|\braket{\psi^\perp_{t_f}(s)}{\psi_{t_f}(s)}|^2$, where $\ket{\psi^\perp_{t_f}(s)} = Q_{t_f}(s)\ket{\psi_{t_f}(s)}$, with $Q=I-P$.

\subsubsection{Inverse cubic gap dependence with generic $H(s)$}

Kato's work on the perturbation theory of linear operators \cite{kato_adiabatic_1950} introduced techniques based on resolvents and complex analysis that have been widely used in subsequent work. Jansen, Ruskai, and Seiler (JRS) proved several versions of the AT that build upon these techniques \cite{Jansen:07}, and that rigorously establish the gap dependence of $t_f$, without any strong assumptions on the smoothness of $H(s)$. 
Their essential assumption is that the spectrum of $H(s)$ has a band associated with the spectral projection $P(s)$ which is separated by a non-vanishing gap $\Delta(s)$  from the rest. 
Here we present one their theorems:

\begin{theorem}
\label{th:AT0}
Suppose that the spectrum of $H(s)$ restricted to $P(s)$ consists of $m(s)$ eigenvalues separated by a gap $\Delta(s)=\veps_1(s)-\veps_0(s)>0$  from the rest of the
spectrum of $H(s)$, and that $H(s)$ is twice continuously differentiable. Assume that $H$, $H^{(1)}$, and $H^{(2)}$ are bounded operators, an assumption that is always fulfilled in finite-dimensional spaces.\footnote{We use the notation $H^{(k)}(s) \equiv \left(\frac{\partial}{\partial x}\right)^k H(x)|_s$ throughout.} Then for any $s\in [0,1]$,
\begin{align}
\label{eq:JRS}
&\norm{P_{t_f}(s) - P(s)} \leq \frac{m(0) \norm{H^{(1)}(0)}}{t_f \Delta^2(0)}+\frac{m(s) \norm{H^{(1)}(s)}}{t_f \Delta^2(s)}  \notag \\
&\,\, + \frac{1}{t_f}\int_0^s \left( \frac{m\norm{H^{(2)}}}{\Delta^2} + \frac{7m\sqrt{m}\norm{H^{(1)}}^2}{\Delta^3} \right)dx
\end{align}
\end{theorem}
The numerator depends on the norm of the first or second time derivative of $H(s)$, rather than the matrix element that appears in the traditional versions of the adiabatic condition.

Ignoring the $m$-dependence for simplicity, this result shows that the adiabatic limit can be approached arbitrarily closely if (but not only if)
\begin{eqnarray}
t_f &\gg& \max \left\{ \max_{s\in[0,1]} \frac{\norm{H^{(2)}(s)}}{\Delta^2(s)} , \max_{s\in[0,1]} \frac{\norm{H^{(1)}(s)}^2}{\Delta^3(s)} , \right. \nonumber \\
&& \left. \max_{s \in [0,1]} \frac{\|H^{(1)}(s)\|}{\Delta^2(s)} \right\}  \ .
\label{eq:tf-JRS}
\end{eqnarray}

Similar techniques based on Kato's approach can be used to prove a rigorous adiabatic theorem for open quantum systems, where the evolution is generated by a non-Hermitian Liouvillian instead of a Hamiltonian \cite{PhysRevA.93.032118}.

\subsubsection{Rigorous inverse gap squared}
A version of the AT that yields a scaling of $t_f$ with the inverse of the gap squared (up to a logarithmic correction) was given in \cite{Elgart:2012fk}. All other rigorous AT versions to date have a worse gap dependence (cubic or higher). The proof introduces assumptions on $H(s)$ that go beyond those of Theorem~\ref{th:AT0}.  
Namely, it is assumed that $H(s)$ is bounded and infinitely differentiable, and the higher derivatives cannot have a magnitude that is too large, or more specifically, that $H(s)$ belongs to the Gevrey class $G^\alpha$: 

\begin{definition}[Gevrey class]
$H(s)\in G^\alpha$ if $d{H}(s)/ds\neq 0$ $\forall s\in [0,1]$ and there 
exist constants $C,R>0$, 
such that for all $k\geq 1$,
\beq
\max_{s\in[0,1]} \norm{H^{(k)}(s)} \leq C R^k k^{\alpha k}\ .
\eeq
\end{definition}

An example is $H(s) = [1-A(s)]H_0 + A(s) H_1$, where $A(s) = c \int_{-\infty}^s \exp[-1/(x-x^2)] dx$ if $s\in (0,1)$, and $A(s)=0$ if $s\notin [0,1]$. The constant $c$ is chosen so that $A(1)=1$. For this family $\norm{H^{(k)}(s)} = \left| A^{(k)}(s) \right| \|H_1-H_0\| \leq C k^{2k}$, so that $H(s) \in G^2$.

The AT due to \cite{Elgart:2012fk} can now be stated as follows: 

\begin{theorem}
\label{th:AT1}
Assume that $H(s)$ is bounded and belongs to the Gevrey class $G^\alpha$ with $\alpha > 1$, and that $\Delta \ll h$, where $h\equiv \| H(0)\|=\|H(1)\|$. 
If 
\beq
t_f \geq \frac{K}{\Delta^2} |\ln(\Delta/h)|^{6\alpha} 
\eeq
for some $\Delta$-independent constant $K > 0$ (with units of energy),
then the distance 
$\| P_{t_f}(s) - P(s)\|$ is $o(1)$ 
$\forall s\in[0,1]$.
\end{theorem}

This result is remarkable in that it rigorously gives an inverse gap squared dependence, which is essentially tight due to existence of a lower bound of the form $t_f = O(\Delta^{-2}/|\ln \Delta|)$ for Hamiltonians satisfying ${\rm rank}\,H(1) \ll \dim(\mathcal{H})$ \cite{Cao:2012si}. However, the error bound is not tight, and we address this next.

\subsubsection{Arbitrarily small error}
\label{sec:arb-small-err}

Building on work originating with \cite{Nenciu:93} [see also \cite{Hagedorn:2002kx}],
\cite{Ge:2015wo} proved a version of the AT that results in an exponentially small error bound in $t_f$. The inverse gap dependence is cubic. 

Assume for simplicity that $\veps_0(s)=0$ and choose the phase of $\ket{\veps_0(s)}$ so that $\braket{\dot{\varepsilon}_0(s)}{\veps_0(s)}=0$, where the dot denotes $\partial_s$.

\begin{theorem}
\label{th:AT2}
Assume that all derivatives of the Hamiltonian $H(s)$ vanish at $s=0,1$, and moreover that it satisfies the following Gevrey condition: there exist constants $C,R,\alpha>0$ such that for all $k\geq 1$,
\beq
\max_{s\in[0,1]} \norm{H^{(k)}(s)} \leq C R^k \frac{(k!)^{1+\alpha}}{(k+1)^2}\ .
\eeq
Then the adiabatic error is bounded as
\begin{align}
\min_{\theta} \norm{\ket{\psi_{t_f}(1)}-e^{i\theta}\ket{\veps_0(1)}} \leq c_1 \frac{C}{\Delta} e^{-\left(c_2 \frac{\Delta^3}{C^2}t_f\right)^{\frac{1}{1+\alpha}}}
\end{align}
where $c_1 = e R \left(\frac{8\pi^2}{3}\right)^3$ and $c_2 = \frac{1}{4eR^2} \left(\frac{3}{4\pi^2}\right)^5$.
\end{theorem}
Thus, as long as $t_f \gg \frac{C^2}{\Delta^3}$, the adiabatic error is exponentially small in $t_f$. 

The idea of using vanishing boundary derivatives dates back at least to \cite{Garrido:62}. It was also used in \cite{lidar:102106} for a different class of functions than the Gevrey class: functions that are analytic in a strip of width $2\gamma$ in the complex time plane and have a finite number $V$ of vanishing boundary derivatives, i.e., $H^{(v)}(0)=H^{(v)}(1)=0$ $\forall v\in[1,V]$. The adiabatic error is then upper-bounded by $(V+1)^{\gamma+1}q^{-V}$ as along as $t_f\geq \frac{q}{\gamma}V \max_s \norm{H^{(1)}_V(s)}^2/\Delta^3$, where $q>1$ is a parameter that can be optimized given knowledge of $\norm{H^{(1)}_V}$. Thus, the adiabatic error can be made arbitrarily small in the number of vanishing derivatives, while the scaling of $t_f$ with $V$ is encoded into $\norm{H^{(1)}_V}$.\footnote{This corrects an omission in \cite{lidar:102106}, where the dependence of $\norm{H^{(1)}}$ on $V$ was ignored since the supremum of $\norm{H^{(1)}(s)}$ was taken over $s\in[0,1]$ instead of over the region of analyticity of $H(s)$, as noted in \cite{Ge:2015wo}.} An example of a function whose first $V$ derivatives vanish at the boundaries $s=0,1$ is the regularized $\beta$ function $A(s)=\frac{\int_0^s x^V (1-x)^V dx}{\int_0^1 x^V (1-x)^V dx}$ \cite{RPL:10}.
It is possible to further reduce the error quadratically in $t_f$ using an interference effect that arises from imposing an additional boundary symmetry condition \cite{Wiebe:12}.

Note that an important difference between Theorems~\ref{th:AT1} and \ref{th:AT2} is that the former applies for all times $s\in[0,1]$ (``instantaneous AT"), while the latter applies only at the final time $s=1$ (``final-time AT"), which typically gives rise to tighter error bounds.

Also note that Landau and Zener already showed that the transition probability out of the ground state is $O(e^{-C \Delta^2 t_f})$ \cite{LZ1-1,Zener:32} [see \cite{Joye:LZ} for a rigorous proof for analytic Hamiltonians], thus combining an inverse gap square dependence with an exponentially small error bound. However, this result only holds for two-level systems. 

\subsubsection{Lower bound}
\label{sec:ATlowerbound}

Let $H(s)$, with $s\in[0,1]$, be a given continuous Hamiltonian path and $\ket{\veps(s)}$ the corresponding non-degenerate eigenstate path (eigenpath). In the so-called black-box model the only assumption is to be
able to evolve with $H[s(t)]$ for some schedule $s(t)$ (here $s$ is allowed to be a general function of $t$), without exploiting  the unknown structure of $H(s)$. Define the path length $L$ as:
\beq
L = \int_0^1 \norm{\ket{\dot{\veps}(s)}} ds \ ,
\label{eq:L}
\eeq
where dot denotes $\partial_s$. Assuming, without loss of generality, that the phase of $\ket{\veps(s)}$ is chosen so that $\braket{\veps(s)}{\dot{\veps}(s)}=0$, then $L$ is the only natural length in projective Hilbert space (up to irrelevant
normalization factors). 

It was shown in \cite{Boixo:2010fk} that there is a lower bound on the time required to prepare $\ket{\veps(1)}$ from $\ket{\veps(0)}$ with bounded precision:
\beq
t_f > O(L/\Delta)\ .
\label{eq:14}
\eeq
Since an upper bound on $L$ is $\max_s\|\dot{H}(s)\|/\Delta$,%
\footnote{See Appendix~\ref{app:boundproof} as well as Appendix G of \cite{boixo_eigenpath_2009}.}
one obtains the estimate $t_f \sim O(\max_s\|\dot{H}(s)\|/\Delta^2)$, reminiscent of the approximate versions of the adiabatic condition [e.g., Eq.~\eqref{eq:AT-folklore2}].
The proof of the lower bound is essentially based on the optimality of the Grover search algorithm.

The lower bound is nearly achievable using a ``digital", non-adiabatic method proposed in \cite{boixo_fast_2010}, that does not require path continuity or differentiability. The time required scales as $O[(L/\Delta) \log(L/\epsilon)]$, where  $\epsilon$ is a
specified bound on the error of the output state $\ket{\veps(1)}$. 
$L$ is the angular length of the path and is suitably defined to generalize Eq.~\eqref{eq:L} to the non-differentiable case.

Armed with an arsenal of adiabatic theorems we are now well equipped to start surveying AQC algorithms.


\section{Algorithms}
\label{sec:algorithms}

In this section we review the algorithms which are known to provide quantum speedups over classical algorithms. However, to make the idea of a quantum speedup precise we need to draw distinctions among different types of speedups, as several such types will arise in the course of this review. Toward this end we adopt a classification of quantum speedup types proposed in \cite{speedup}. The classification is the following, in decreasing order of strength. 
\begin{itemize}
\item A ``provable" quantum speedup is the case where there exists a proof that no classical algorithm can outperform a given quantum algorithm. The best known example is Grover's search algorithm \cite{Grover:97a}, which, in the query complexity setting, exhibits a provable quadratic speedup over the best possible classical algorithm \cite{Bennett:1997lh}. 
\item A ``strong" quantum speedup was originally defined in \cite{Papageorgiou:2013} by comparing a quantum algorithm against the performance of the best classical algorithm, whether such a classical algorithm is explicitly known or not. This aims to capture computational complexity considerations allowing for the existence of yet-to-be discovered classical algorithms. Unfortunately, the performance of the best possible classical algorithm is unknown for many interesting problems (e.g., for factoring). 
\item A ``quantum speedup" (unqualified, without adjectives) is a speedup against the best available classical algorithm [for example Shor's polynomial time factoring algorithm \cite{Shor:94}]. Such a speedup may be tentative, in the sense that a better classical algorithm may eventually be found. 
\item
Finally, a ``limited quantum speedup" is a speedup obtained when compared specifically with classical algorithms that `correspond" to the quantum algorithm in the sense that they implement the same algorithmic approach, but on classical hardware. This definition allows for the existence of other classical algorithms that are already better than the quantum algorithm. The notion of a limited quantum speedup will turn out to be particularly useful in the context of StoqAQC. 
\end{itemize}
A refinement of this classification geared at experimental quantum annealing was given in \cite{2016arXiv160401746M}.

Using this classification, this section collects most of the adiabatic quantum algorithms known to give a provable quantum speedup (Grover, Deutsch-Jozsa,  Bernstein-Vazirani, and glued trees), or just a quantum speedup (PageRank).\footnote{The glued trees case is, strictly, not an adiabatic quantum algorithm, since it explicitly makes use of excited states. Also, in the PageRank case the evidence for a quantum speedup is numerical.}

Many other adiabatic algorithms have been proposed, and we review a large subset of these in Sec.~\ref{sec:QA}. In a few of these cases there is a scaling advantage over classical simulated annealing, while in some cases there are definitely faster classical algorithms.

\subsection{Adiabatic Grover}
\label{sec:Grover}
%
The adiabatic Grover algorithm~\cite{Roland:2002ul} is perhaps the hallmark example of a provable quantum speedup using AQC, so we review it in detail.  As in the circuit model Grover algorithm \cite{Grover:97a}, informally the objective is to find the marked item (or possibly multiple marked items) in an unsorted database of $N$ items by accessing the database as few times as possible. More formally, one is allowed to call a function $f:\{0,1\}^n \mapsto \{0,1\}$ (where $N=2^n$ is the number of bit strings) with the promise that $f(m)=1$ and $f(x)=0$ $\forall x\neq m$, and the goal is to find the unknown index $m$ in the smallest number of calls. This is an oracular problem \cite{nielsen2000quantum}, in that the algorithm can make queries to an oracle that recognizes the marked items.  The oracle remains a black box, i.e., the details of its implementation and its complexity are ignored.  This allows for an uncontroversial determination of the complexity of the algorithm in terms of the number of queries to the oracle.

For a classical algorithm, the only strategy is to query the oracle until the marked item is found.  Whether the classical algorithm uses no memory, i.e., the algorithm does not keep track of items that have already been checked, or uses an exponential amount of memory (in $n$) to store all the items that have been checked, the classical algorithm will have an average number of queries that scales linearly in $N$.

In the AQC algorithm we denote the marked item by the binary representation of $m$.  The oracle is defined in terms of the final Hamiltonian $H_1 = \ident - \ketbra{m}{m}$, where $\ket{m}$ is the marked state associated with the marked item.  In this representation, the binary representations give the eigenvalues under $\sigma^z$, i.e., $\sigma^z \ket{0} = + \ket{0}$ and $\sigma^z \ket{1} = -\ket{1}$.  The marked state is the ground state of this Hamiltonian with energy 0, and all other computational basis states have energy 1.

\subsubsection{Setup for the adiabatic quantum Grover algorithm}
We use the initial Hamiltonian $H_0 = \ident - \ketbra{\phi}{\phi}$, where $\ket{\phi}$ is the uniform superposition state, 
\beq
\ket{\phi} = \frac{1}{\sqrt{N}} \sum_{i=0}^{N-1} \ket{i} = \ket{+}^{\otimes n}\ ,
\label{eq:usup}
\eeq 
where $\ket{\pm} = \frac{1}{\sqrt{2}}(\ket{0}\pm \ket{1})$.  We take the time-dependent Hamiltonian to be an interpolation:
\begin{align} 
\label{eqt:GroverH}
H(s) &= \left[ 1- A(s) \right] H_0 + A(s) H_1  \\
&= \left[ 1- A(s) \right] (\ident - \ketbra{\phi}{\phi}) + A(s)(\ident - \ketbra{m}{m}) \notag \ ,
\end{align}
where $s=t/t_f \in [0,1]$ is the dimensionless time, $t_f$ is the total computation time, and $A(s)$ is a ``schedule" that can be optimized.  For simplicity, we first consider a linear schedule: $A(s)=s$. Note that $H_1$ is $n$-local.  

If the initial state is initialized in the ground state of $H(0)$, i.e., $\ket{\psi(0)} = \ket{\phi}$, then the evolution of the system is restricted to a two-dimensional subspace, defined by the span of $\ket{m}$ and $\ket{m^{\perp}} = \frac{1}{\sqrt{N-1}} \sum_{i \neq m}^{N-1} \ket{i}$.  In this two-dimensional subspace $H(s)$ can be written as:
\beq
\label{eq:Grover-H}
\left[H(s) \right]_{\ket{m},\ket{m^{\perp}}} 
= \frac{1}{2} \ident_{2 \times 2} - \frac{\Delta(s)}{2} \left( \begin{array}{cc}
 \cos \theta(s) & \sin \theta(s) \\
\sin \theta(s) & - \cos \theta(s)
\end{array} \right) \ ,
\eeq
where:
\bes
\begin{align}
\label{eq:gap-Grover}
\Delta(s) & =  \sqrt{ (1-2 s)^2 + \frac{4}{N} s (1-s) } \ , \\
\cos \theta(s) & =  \frac{1}{\Delta(s)} \left[ 1 - 2 (1-s) \left( 1 - \frac{1}{N} \right) \right] \ , \\
\sin \theta(s) & =  \frac{2}{\Delta(s)} \left(1 - s \right)  \frac{1}{\sqrt{N}}  \sqrt{1 - \frac{1}{N}} \ .
\end{align}
\ees
The eigenvalues and eigenvectors in this subspace are then given by:
\bes
\begin{align}
\veps_0(s) &= \frac{1}{2} \left( 1 - \Delta(s) \right) \ , \quad \veps_1(s) =  \frac{1}{2} \left( 1 + \Delta(s) \right) \ , \label{eqt:M=1Eigenvalues}\\
\ket{\veps_0(s)} & =  \cos \frac{\theta(s)}{2} \ket{m} + \sin \frac{\theta(s)}{2} \ket{m^{\perp}} \ , \\
\ket{\veps_1(s)}  &=  -\sin \frac{\theta(s)}{2} \ket{m} + \cos \frac{\theta(s)}{2} \ket{m^{\perp}} \ .
\end{align}
\ees
The remaining $N-2$ eigenstates of $H(s)$ have eigenvalue $1$ throughout the evolution.  The minimum gap occurs at $s = 1/2$ and scales exponentially with $n$:
\beq
\Delta_{\mathrm{min}} = \Delta(s = 1/2) = \frac{1}{\sqrt{N}} = 2^{-n/2}  \ .
\eeq
(This can be viewed as a special case of Lemma~\ref{lemma-projector-gap} below.)

In our discussion of the adiabatic theorem we saw that without special assumptions on $s(t)$ except that it is twice differentiable, the adiabatic condition is inferred from Eq.~\eqref{eq:JRS}, which requires setting $t_f \gg 2\max_s \norm{\partial_s H(s)}/\Delta^{2}(s) + \int_0^1 \norm{\partial_s H(s)}^2/\Delta^{3}(s) ds $, where we have accounted for the boundary conditions and used the positivity of the integrand to extend the upper limit to $1$.\footnote{Whenever we use the $\gg$ symbol we mean that the larger quantity should be larger by some large multiplicative constant, such as $100$.} 
Differentiating Eq.~\eqref{eq:Grover-H} yields
\beq
\label{eq:dH}
\partial_s H(s) = \left( \begin{array}{cc}
- \left( 1- \frac{1}{N} \right) & \frac{1}{\sqrt{N}} \sqrt{ 1 - \frac{1}{N}} \\
 \frac{1}{\sqrt{N}} \sqrt{ 1 - \frac{1}{N}}  & 1 - \frac{1}{N}
\end{array} \right)\ ,
\eeq
which has eigenvalues $\pm \sqrt{1 - \frac{1}{N}}$, so that $\norm{\partial_s H}  \leq 1$. The other integrand in Eq.~\eqref{eq:JRS}, involving $\norm{\partial^2_s H(s)}/\Delta^{2}(s)$, vanishes after differentiating Eq.~\eqref{eq:dH}. The ground state degeneracy $m(s)=1$ throughout. 
Since
$\int_0^s 1/\Delta^{3}(x) dx = \frac{N}{2}-\frac{N^{3/2} (1-2 s)}{2 \sqrt{{N (1-2 s)^2+4 (1-s) s}}}$, which is a monotonically increasing function of $s$ that approaches $N = \Delta_{\min}^{-2}$ as $s\to 1$, the adiabatic condition becomes 
\beq
t_f \gg 2\max_s \frac{1}{\Delta^{2}(s)} +  \int_0^1 ds \frac{1}{\Delta^3(s)} = \frac{3}{\Delta_{\min}^2}\ .
\label{eq:Grover-gap-dependence}
\eeq
This suggests the disappointing conclusion that the quantum adiabatic algorithm scales in the same way as the classical algorithm.

However, by imposing the adiabatic condition globally, i.e., to the entire time interval $t_f$, the evolution rate is constrained throughout the whole computation, while the gap only becomes small around $s=1/2$. Thus, it makes sense to use a schedule $A(s)$ that adapts and slows down near the minimum gap, but speeds up away from it \cite{Dam:2001fk,Roland:2002ul} [this is related to the idea of rapid adiabatic passage, which has a long history in nuclear magnetic resonance \cite{Powles:1958cq}]. By doing so the quadratic quantum speedup can be recovered, as we address next.

\subsubsection{Quadratic quantum speedup}
\label{sec:quadspeedup}
Consider again the adiabatic condition~\eqref{eq:JRS}, which we can rewrite as:
\beq
t_f \gg 2\max_s \frac{\norm{\partial_s H(s)}}{\Delta^{2}(s)} + \int_0^1 \left( \frac{\norm{\partial^2_s H}}{\Delta^2} + \frac{\norm{\partial_s H}^2}{\Delta^{3}} \right) ds\ ,
\label{eq:24}
\eeq 
where now $H$ and $\Delta$ depend on a schedule $A(s)$. 
Let us now use the ansatz \cite{Jansen:07,Roland:2002ul}
\beq \label{eqt:GroverRate}
\partial_s{A} = c\, \Delta^p[A(s)] \ , \quad A(0)=0 \ , \quad p,c>0 \ .
\eeq
This schedule slows down as the gap becomes smaller, as desired. The normalization constant $c =  \int_0^1  \Delta^{-p}[A(s)] \partial_s{A} ds = \int_{A(0)}^{A(1)} \Delta^{-p}(u) du$ [using $u=A(s)$] is chosen to ensure that $A(1)=1$. 

It follows that:
\begin{align}
&\int_0^1 \left(\frac{\norm{\partial^2_s H[A(s)]}}{\Delta^{2}[A(s)]} + \frac{\norm{\partial_s H[A(s)]}^2}{\Delta^{3}[A(s)]} \right)ds \notag \\
&\leq 4c\int_0^1 \Delta^{p-3}(u) du 
\label{eq:10c}
\end{align}
(the proof is given in Appendix~\ref{app:eq:10c}). 
 Finally, the boundary term in Eq.~\eqref{eq:24} yields $2\max_s \frac{\norm{\partial_s H(s)}}{\Delta^{2}(s)} \leq 4c\Delta_{\min}^{p-2}$.

The case $p=2$ serves to illustrate the main point.
In this case the boundary term is $4c$ and evaluating the integrals yields 
\begin{align*}
&c=\int_0^1 \Delta^{-2}(u) du=
\frac{N}{\sqrt{N-1}}\tan^{-1}\sqrt{N-1} \to \frac{\pi}{2}\sqrt{N}\\
&\int_0^1 \Delta^{-1}(u) du = \frac{\log \left[\frac{\sqrt{N-1} \sqrt{N}+N-1}{\sqrt{N-1} \sqrt{N}-(N-1)}\right]}{2 \sqrt{\frac{N-1}{N}}} \to \log(2N)/2\ ,
\end{align*} 
where the asymptotic expressions are for $N\gg 1$. Substituting this into Eq.~\eqref{eq:10c} yields the adiabatic condition
\beq
t_f \gg 2\pi \sqrt{N} [1+\log(2N)]\ ,
\label{eq:t_f-Grover-rigorous}
\eeq
which is a sufficient condition for the smallness of the adiabatic error, and nearly recovers the quadratic speedup expected from Grover's algorithm. 

The appearance of the logarithmic factor latter is actually an artifact of using bounds that are not tight.\footnote{A detailed analysis of the adiabatic Grover algorithm along with tighter error bounds than we have given here was presented in \cite{RPL:10}.} The quadratic speedup, i.e., the scaling of $t_f$ with $\sqrt{N}$, can be fully recovered by solving for the schedule from Eq.~\eqref{eqt:GroverRate} in the $p=2$ case \cite{Roland:2002ul}.  We first rewrite Eq.~\eqref{eqt:GroverRate} in dimensional time units as $\partial_t{A} = c'\, \Delta^2[A(t)]$, with the boundary conditions $A(0)=0$ and $A(t_f)=1$;.
To solve this differential equation we rewrite it as $t = \int_0^{t} dt = \int_{A(0)}^{A(t)} dA/[c'\Delta^2(A)]$. After integration we obtain

\begin{align} 
\label{eqt:GroverOptimized}
t &= \frac{N}{2c'\sqrt{N-1}}  \left[ \tan^{-1} \left( \sqrt{N-1} \left(2 A(t) - 1 \right) \right)  \right. \notag \\
&\quad \left. + \tan^{-1} \sqrt{N-1} \right] \ .
\end{align}
Evaluating Eq.~\eqref{eqt:GroverOptimized} at $t_f$ gives:
\beq
\label{eq:t_f-opt-Grover}
t_f = \frac{N}{c' \sqrt{N-1}} \tan^{-1} \sqrt{N-1} \to \frac{\pi}{2 c'} {\sqrt{N}} \ ,
\eeq

which is the  expected quadratic quantum speedup. 

One may be tempted to conclude that $t_f$ can be made arbitrarily small since so far $c'$ is arbitrary and can be chosen to be large. However, the adiabatic error bound \eqref{eq:t_f-Grover-rigorous} shows that this is not the case: while it is not tight, it suggests that if $t_f$ scales as $\sqrt{N}$ then $c'$ must scale as $1/\log(2N)$ in order to keep the adiabatic error small. Thus, the general conclusion is that increasing $c'$ results in a larger adiabatic error.%
\footnote{Note that the scaling conclusion $t_f \sim \sqrt{N}$ reported in \cite{Roland:2002ul} is based on the interpretation of Eq.~\eqref{eqt:GroverRate} as a heuristic ``local" adiabatic condition and does not constitute a proof that the adiabatic error is small. The evidence that the rigorous bound \eqref{eq:t_f-Grover-rigorous} is not tight and that $t_f \sim \sqrt{N}$ suffices to achieve a small adiabatic error for the schedule \eqref{eqt:GroverOptimalSchedule} is numerical.}

Inverting Eq.~\eqref{eqt:GroverOptimized} for $A(t)$ [or, equivalently, solving Eq.~\eqref{eqt:GroverRate} for $p=2$] gives the locally optimized schedule 
\beq 
\label{eqt:GroverOptimalSchedule}
A(s) = \frac{1}{2} + \frac{1}{2 \sqrt{N-1}} \tan \left[ (2s-1) \tan^{-1} \sqrt{N-1}  \right] \ ,
\eeq
where we replaced $t/t_f$ [with $t_f$ given by Eq.~\eqref{eq:t_f-opt-Grover}] with $s$. As expected, this schedule rises rapidly near $s=0,1$ and is nearly flat around $s=1/2$, i.e., it slows down near the minimum gap.


Since the choice in Eq.~\eqref{eqt:GroverRate} is not unique, we may wonder if there exists a schedule that gives an even better scaling. Given that Grover's algorithm is known to be optimal in the circuit model setting \cite{Bennett:1997lh,zalka1999grover}, this is, unsurprisingly, not the case, and a general argument to that effect which applies to any Hamiltonian quantum computation was given by \cite{FarhiAnalog}.
We review this argument in the AQC setting, in Appendix~\ref{app:lower-bound}.

\subsubsection{Multiple marked states}

The present results generalize easily to the case where we have $M\geq 1$ marked states, for which Grover's algorithm is known to also give a quadratic speedup in the circuit model \cite{Boyer:96,Biham:1999ye}. The final Hamiltonian can be written as:
\beq
H_1 = \ident - \sum_{m \in \mM} \ketbra{m}{m} \ ,
\eeq
where $\mM$ is the index set of the marked states. 
Let:
\beq
\ket{m_{\perp}} = \frac{1}{\sqrt{N-M}} \sum_{i \notin \mM} \ketbra{i}{i} \ .
\eeq
Instead of evolving in a two-dimensional subspace, the system evolves in an $M+1$ dimensional subspace spanned by $\left(\{ \ket{m}\}_{m \in \mM}, \ket{m_{\perp}} \right)$, and instead of Eq.~\eqref{eq:Grover-H}, the Hamiltonian can be written in this basis as:
\begin{widetext}
\beq
H(s) = \left( \begin{array}{ccccc}
(1-s)\left( 1 - \frac{1}{N} \right) & -\frac{1-s}{N} & \dots && -(1-s) \frac{\sqrt{N-M}}{N} \\
-\frac{1-s}{N}  & (1-s)\left( 1 - \frac{1}{N} \right) & -\frac{1-s}{N}  & \dots  & -(1-s) \frac{\sqrt{N-M}}{N}\\
\vdots & & \ddots & & \vdots \\
-(1-s) \frac{\sqrt{N-M}}{N}  & -(1-s) \frac{\sqrt{N-M}}{N}  & \dots && s + (1-s)\left( 1 - \frac{N-M}{N} \right) 
\end{array} \right) \ .
\eeq
\end{widetext}
This Hamiltonian can be easily diagonalized, and one finds that there are $M-1$ eigenvalues equal to $1-s$, and two eigenvalues 
\beq
\lambda_{\pm} = \frac{1}{2} \pm \frac{1}{2}\sqrt{ (1-2 s)^2 + \frac{4 M}{N} s (1-s) } \ ,
\eeq
that determine the relevant minimum gap: $\Delta(s)  =  \sqrt{ (1-2 s)^2 + \frac{4M}{N} s (1-s)}$. The remaining $N-M-1$ eigenvalues of the unrestricted Hamiltonian are equal to $1$, Comparing the $M=1$ case Eq.~\eqref{eq:gap-Grover} to the present case, the only difference is the change from $1/N$ to $M/N$. Therefore our discussion from earlier goes through with only this modification.
 
In closing, we note that an experimentally realizable version of the adiabatic Grover search algorithm using a single bosonic particle placed in an optical lattice was recently proposed in \cite{Hen:2016aa}.

\subsection{Adiabatic Deutsch-Jozsa algorithm}

Given a function $f:\{0,1\}^n \mapsto \{0,1\}$ which is promised to be either
constant or balanced [i.e., $f(x)=0$ on half the inputs and $f(x)=1$ on the other half], the Deutsch-Jozsa problem is to determine which type the
function is. There exists a quantum circuit algorithm that solves the problem in a single $f$-query~\cite{Deutsch:92}. Classically, the problem requires $2^{n-1}+1$ $f$-queries in the worst case, since it is possible that the first $2^{n-1}$ queries return a constant answer, while the function is actually balanced. It is important to note that the quantum advantage requires a deterministic setting, since the classical error probability is exponentially small in the number of queries.

An adiabatic implementation of the Deutsch-Jozsa algorithm using unitary interpolation was given in \cite{PhysRevLett.95.250503} and an implementation using a linear interpolation was given in \cite{Wei:2006dg}. These algorithms match the speedup obtained in the circuit model [for an earlier example where this is not the case see \cite{Das:2002}], and we proceed to review both. We note that, just like the adiabatic Grover's algorithm, the adiabatic Deutsch-Jozsa algorithm requires $n$-local Hamiltonians. We also note that both the unitary interpolation and linear interpolation strategies we describe here are not unique to the Deutsch-Jozsa problem, and apply equally well to any depth-one quantum
circuit. Thus they should be viewed in this more general context, and are used here with a specific algorithm for illustrative purposes.

\subsubsection{Unitary interpolation} 
\label{sec:DJ-unitary}

The initial Hamiltonian is chosen such that its ground state is the uniform superposition state $\ket{\phi}$ [Eq.~\eqref{eq:usup}] and $N=2^n$,
i.e., $H(0)=\omega\sum_{i=1}^{n}\ket{-}_{i}\bra{-}$, where $\omega $ is the energy
scale. The Deutsch-Jozsa problem can be solved by a single computation of the function $f$
through the unitary transformation $U\ket{x} =(-1)^{f(x)}|x\rangle $ ($x\in \{0,1\}^{n}$)~\cite{Collins:98}, so that in the $\{\ket{x} \}$
(computational) basis $U$ is represented by the diagonal matrix 
$U=\mathrm{diag}[(-1)^{f(0)},...,(-1)^{f(2^{n}-1)}]$.
An adiabatic implementation requires a final Hamiltonian $H(1)$ such that
its ground state is $|\psi (1)\rangle =U\ket{\psi (0)}$. This can be accomplished via a unitary transformation of $H(0)$, i.e., $H(1)=UH(0)U^{\dagger }$.
Then the final Hamiltonian encodes the solution of the Deutsch problem in its ground state, which can be extracted via a measurement of the qubits in the
 $\{\ket{+},\ket{-} \}$ basis (note that this is compatible with the definition of AQC, Def.~\ref{def:AQC}, which does not restrict the measurement basis). 
A suitable unitary interpolation between $H(0)$ 
and $H(1)$ can be defined by $H(s)={\tilde{U}}(s)H(0){\tilde{U}}^{\dagger }(s)$, where ${\tilde{U}}(s)=\exp
\left( i\frac{\pi }{2}sU\right)$, for which ${\tilde{U}}(1) = iU$. Since a unitary transformation of $H(0)$ preserves its spectrum, it does not change the ground state gap, which remains $\omega$. 
The run time of the algorithm can be determined
from the adiabatic condition~\eqref{eq:tf-JRS}, and what remains is the numerator: $\norm{H^{(1)}(s)} = \norm{i\frac{\pi}{2}[U,H(s)]} \leq \pi \norm{H(0)} = \pi$ [and similarly for $\norm{H^{(2)}(s)}$]. This yields
$t_f \gg 1 /\omega $. This result is independent of $n$ so the adiabatic run time is $O(1)$.  

\subsubsection{Linear interpolation} 
\label{sec:DJ-lin}
Unitary interpolations, introduced in \cite{Siu:05}, are somewhat less standard. Therefore we also present the standard linear interpolation method as an alternative. Consider the usual initial Hamiltonian $H_0 = \ident - \ketbra{\phi}{\phi}$ over $n$ qubits, where once again $\ket{\phi}$ is the uniform superposition state. Let the final Hamiltonian be $H_1 = \ident - \ket{\psi}_f\bra{\psi}$, where
\beq
\ket{\psi}_f = \frac{\mu_f}{\sqrt{N/2}}  \sum_{i=0}^{N/2-1}\ket{2i} + \frac{1-\mu_f}{\sqrt{N/2}}  \sum_{i=0}^{N/2-1}\ket{2i+1} \ ,
\eeq
and where
\beq
\mu_f = \abs{\frac{1}{N}\sum_{x\in\{0,1\}^n}(-1)^{f(x)}} \ .
\eeq
Clearly, $\mu_f = 1$ or $0$ if $f$ is constant or balanced, respectively. Therefore $\ket{\psi}_f$ is a uniform superposition over all even (odd) index states if $f$ is constant or balanced, respectively, and a measurement of the ground state of $H_1$ in the computational basis reveals whether $f$ is constant or balanced, depending on whether the observed state belongs to the even or odd sector, respectively. However, we note that one may object to the reasonableness of the final Hamiltonian $H_1$. 
Namely, preparing the state $\ket{\psi}_f$ involves precomputing the quantity $\mu_f$, which directly encodes whether $f$ is constant or balanced and so may be thought to represent an oracle that is too powerful.\footnote{The only efficient way known to compute a quantity similar to $\mu_f$ involves running the Deutsch-Jozsa algorithm in the gate model, where $\frac{1}{N}\sum_{x\in\{0,1\}^n}(-1)^{f(x)}$, without the absolute value, appears as the amplitude of the state $\ket{0}^{\otimes n}$ \cite{nielsen2000quantum}.} Indeed, $H_1$ 
in this construction is not of the standard form $\sum_x f(x) \ketbra{x}{x}$, wherein each oracle call $f(x)$ corresponds to a query about a single basis state $\ket{x}$. Therefore, there is no classical analogue to this oracle in the computational basis.   

Setting this concern in the present version of the algorithm due to \cite{Wei:2006dg} aside,
it remains to determine the adiabatic run time for the adiabatic Hamiltonian $H(s) = (1-s)H_0+sH_1$. The following Lemma \cite{AharonovTa-Shma} comes in handy:
\begin{lemma}
\label{lemma-projector-gap}
Let $\ket{\alpha}$ and $\ket{\beta}$ be two states in some subspace of an $N$-dimensional Hilbert space $\mathcal{H}$, and let $H_\alpha = \ident - \ketbra{\alpha}{\alpha}$, $H_\beta = \ident - \ketbra{\beta}{\beta}$. For any convex combination $H_\eta = (1-\eta)H_\alpha+\eta H_\beta$, where $\eta\in[0,1]$, the ground state gap $\Delta(H_\eta) \geq |\braket{\alpha}{\beta}|$.
\end{lemma}
\begin{proof}
Expand $\ket{\beta} = a\ket{\alpha} + b\ket{\alpha^\perp}$ where $\braket{\alpha}{\alpha^\perp}=0$, and complete $\{\ket{\alpha},\ket{\alpha^\perp}\}$ to an orthonormal basis for $\mathcal{H}$. Writing $H_\eta$ in this basis yields
\beq
H_\eta = \left( \begin{array}{cc}
\eta |b|^2 & \eta a b^* \\
\eta a^* b  & \eta|a|^2+1-\eta
\end{array} \right) \oplus \ident_{(N-2)\times(N-2)} \ .
\eeq
The eigenvalues of this matrix are all $1$ for the identity matrix block and the difference between the eigenvalues in the $2\times 2$ block is $\Delta(H_\eta) = \sqrt{1-4\eta(1-\eta)|b|^2}$. This is minimized for $\eta=1/2$, where it equals $|a|$.
\end{proof}
Applying this lemma, we see that $\Delta[H(s)] \geq |\braket{\phi}{\psi_f}| = 1/\sqrt{2}$. Since $\norm{H^{(1)}(s)} = \norm{H_1-H_0} \leq 2$ and $\norm{H^{(2)}(s)}=0$, it follows from the adiabatic condition~\eqref{eq:tf-JRS} that $t_f$ is independent of $n$, i.e., the adiabatic run time is $O(1)$ as in the circuit model depth.

\subsubsection{Interpretation}
As mentioned above, a classical probabilistic algorithm that simply submits random queries to the oracle will fail with a probability that is exponentially small in the number of queries. One might thus be concerned that the adiabatic algorithms above are no better~\cite{Hen2014}, since they are probabilistic, in the sense that there is a non-zero probability of ending in an excited state. However, for the linear interpolation adiabatic algorithm reviewed, measuring the energy of the final state returns $0$ in the ground state or $1$ in an excited state. In the latter  "inconclusive" case, the algorithm needs to be repeated until an energy of $0$ is found and only then is a computational basis measurement  performed. If an even (or odd) index state is measured, the corresponding constant (or balanced) result is guaranteed to be correct.  Moreover, an excited state outcome can (and should) be made exponentially unlikely using a smooth schedule as per Theorem~\ref{th:AT2}. Thus, the adiabatic algorithms above improve upon a classical probabilistic algorithm in the following sense:
In the adiabatic case, to know with certainty that the function is constant or balanced (an even or odd index measurement result) happens with probability $p=1-q$ where $q\sim e^{-t_f}$, where the run time $t_f$ is independent of $n$. Therefore, the expected number of runs $r$ to certainty in the adiabatic case is 
\begin{align}
\av{r} &= \sum_{r=1}^{\infty} pq^{r-1} r = \frac{1}{1-q} \approx 1+q \sim 1+e^{-t_f}\ .
\end{align}
On the other hand, classically, to know with certainty that the function is constant requires $N=2^n/2+1$ runs or queries (all yielding identical outcomes). 

Finally, we note that there exists a non-adiabatic Hamiltonian quantum algorithm that solves the Deutsch-Jozsa problem in constant time with a deterministic guarantee of ending up with the right answer (i.e., in the ground state)~\cite{Hen2014}. This algorithm is based on finding a fine-tuned schedule $s(t)$.

\subsection{Adiabatic Bernstein-Vazirani algorithm} 
\label{sec:BernsteinVazirani}

The Bernstein-Vazirani problem \cite{Bernstein:93} is to find an unknown binary string $a\in\{0,1\}^n$ with as few queries as possible of the function (or oracle)
\beq 
\label{eqt:BVFunction}
f_a(w) =  w \odot a \in \{0,1\}\ ,
\eeq
where $\odot$ denotes the bitwise inner product modulo $2$, and $w\in\{0,1\}^n$ as well. In the quantum circuit model, it can be shown that $a$ can be determined with $O(1)$ queries \cite{Bernstein:93} whereas classical algorithms require $n$ queries (the classical algorithm tries all $n$ $w$'s with a single $1$ entry to identify each bit of $a$). This is a polynomial quantum speedup. 

Before presenting the adiabatic algorithm we point out the following useful observation.  For an initial state:
\beq
\ket{\Psi(0)} = \sum_{w \in \{0,1\}^n} c_w \ket{w}_A \otimes \ket{\psi_w(0)}_B \ , \quad \sum_{w \in \{0,1\}^n} |c_w|^2 = 1
\label{eq:40}
\eeq
that undergoes an evolution according to the time-dependent Hamiltonian of the form
\beq
H(s) = \sum_{w \in \{0,1\}^n} \ket{w}_A\bra{w} \otimes H_w(s)\ ,
\label{eq:H(s)40}
\eeq 
we have:
\beq
\ket{\Psi(t)} = \sum_{w \in \{0,1\}^n} c_w \ket{w}_A \otimes \ket{\psi_w(t)}_B\ ,
\eeq
where
\begin{align}
\label{eq:43}
\ket{\psi_w(t)}_B &= \ket{\psi_{t_f,w}(s)}_B\\
& = \Texp \left[-i t_f \int_0^s d \sigma H_w(\sigma) \right] \ket{\psi_w(0)}_B \ , \notag
\end{align}
where $\Texp$ denotes the time-ordered exponential. To see this, simply expand the formal solution:
\begin{align}
&\ket{\Psi(t)} = \Texp \left[-i t_f \int_0^s d \sigma H(\sigma) \right] \ket{\Psi(0)} = \\
& \sum_{w \in \{0,1\}^n} c_w \ket{w}_A \otimes \Texp \left[-i t_f \int_0^s d \sigma H_w(\sigma) \right] \ket{\psi_w(0)}_B\ . \nonumber
\end{align}
Thus for each state $\ket{w}$ in subsystem $A$, there is an independently evolving state in subsystem $B$.  In particular, note that adiabaticity in subsystem $B$ does not depend on the size of system $A$.

The adiabatic algorithm \cite{Hen:2013hl} encodes the action of $f_a(w)$ in a Hamiltonian acting on two subsystems $A$ and $B$ comprising $n$ qubits and $1$ qubit respectively:
\bes
\begin{align}
H_1 &= \sum_{w \in \{0,1\}^n} h_w \ , \\
h_w &\equiv -\frac{1}{2}  \ket{w}_A\bra{w} \otimes \left( \ident_B + (-1)^{f_a(w)} \sigma_B^z \right) \ .
\end{align}
\ees
%

The initial Hamiltonian is chosen to be
\begin{eqnarray}
H_0 &=& \frac{1}{2} \left( \ident_A \otimes \left( \ident_B - \sigma_B^x \right) \right) \nonumber \\
&=& \frac{1}{2}  \sum_{w \in \{0,1\}^n} \ket{w}_A \bra{w} \otimes \left( \ident_B - \sigma_B^x \right) \ .
\end{eqnarray}
Any state of the form 
\beq
\ket{\Psi(0)} = \sum_{w \in \{0,1\}^n} c_w \ket{w}_A \otimes \ket{+}_B 
\eeq
is a ground state of $H_0$, with eigenvalue $0$. We assume that the initial state is prepared as the uniform superposition state, i.e., $c_w ={2^{-n/2}}$ $\forall w$.

The total Hamiltonian is thus given by:
\begin{align}
\label{eq:H(s)-BV}
H(s) &= (1-s) H_0 + s H_1 = \sum_{w \in \{0,1\}^n} \ket{w}_A \bra{w} \otimes H_w(s) \ ,
\end{align}
where
\begin{align}
H_w(s) = \frac{1-s}{2} \left( \ident_B - \sigma_B^x \right) - \frac{s}{2} 
\left( \ident_B + (-1)^{f_a(w)} \sigma_B^z \right)\ .
\label{eq:H_w-BV}
\end{align}
The adiabatic algorithm proceeds, after preparation of the initial state, by adiabatic evolution to the final state:
\begin{align}
& \ket{\Psi(t_f)} =  \\
& \frac{1}{2^{n/2}} \sum_{w \in \{0,1\}^n} \ket{w}_A \otimes  \exp \left[ -i t_f \int_0^1  \veps_{0,w}(s) ds \right] \ket{f_a(w)}_B \ , \nonumber
\end{align}
where $\veps_{0,w}(s)$ is the instantaneous ground-state energy of $H_w(s)$, and we have used the general argument from Eqs.~\eqref{eq:40}-\eqref{eq:43}.
Finally an $x$ measurement on subsystem $B$ is performed.  Since we can write 
\beq
\ket{f_a(w)} = \frac{1}{\sqrt{2}} \left( \ket{+} + (-1)^{f_a(w)} \ket{-} \right) \ ,
\eeq
the state collapses to either of the following states with equal probability:
\bes
\begin{align}
\ket{\Psi_+} & = \frac{1}{2^{n/2}}\sum_{w \in \{0,1\}^n} \ket{w}_A \otimes \ket{+}_B =  \ket{\Psi(0)}  \  , \label{eqt:BV+} \\
 \ket{\Psi_-} & = \frac{1}{2^{n/2}} \sum_{w \in \{0,1\}^n} \ket{w}_A \otimes (-1)^{f_a(w)} \ket{-}_B  \ .
\end{align}
\ees
Note that since $f_a(w)$ counts the number of $1$-agreements between $a$ and $w$, we can write:
\beq
\sum_{w \in \{0,1\}^n} (-1)^{f_a(w)} \ket{w}_A = \bigotimes_{k=0}^{n-1}\left( \ket{0}_k + (-1)^{a_k} \ket{1}_k \right)_A \ ,
\eeq
so that 
\beq
\ket{\Psi_-} = \frac{1}{2^{n/2}} \bigotimes_{k=0}^{n-1} \left( \ket{0}_k + (-1)^{a_k} \ket{1}_k \right)_A \otimes \ket{-}_B \ .
\eeq
If the measurement gives $+1$ [i.e., Eq.~\eqref{eqt:BV+}], then the measured state is the initial state and no information is gained and the process must be repeated.  If the measurement gives $-1$, the resulting state in the $A$ subspace encodes all the bits of $a$, since if the $k$-th qubit is in the $\ket{+}_k$ state, then $a_k = 0$ and if it is in the $\ket{-}_k$ state, then $a_k = 1$.  The probability of failure after $m$ tries is $2^{-m}$ so it is exponentially small and $n$-independent.

The run time of the algorithm is also $n$-independent, since only a single qubit (system B) is effectively evolving..

In conclusion, the adiabatic Bernstein Vazirani algorithm finds the unknown binary string $a$ in $O(1)$ time, matching the circuit model depth. Using a similar technique, \cite{Hen:2013hl} presented a  quantum adiabatic version of Simon's exponential-speedup period finding algorithm \cite{Simon:97} [a precursor to Shor's factoring algorithm \cite{Shor:97}], again matching the circuit model depth scaling. An important aspect of these quantum adiabatic constructions is that they go beyond the general-purpose (and hence suboptimal) polynomial-equivalence prescription of universality proofs that map circuit-based algorithms into quantum-adiabatic ones (see Section~\ref{sec:universality}). That equivalence does not necessarily preserve a polynomial quantum speedup, whereas the construction in \cite{Hen:2013hl} discussed here does.

\subsection{The glued trees problem} \label{sec:gluedtrees}

Consider two binary trees, each of depth $n$.  Each tree has $\sum_{j=0}^n {2^j}  = 2^{n+1}-1$ vertices, for a total of $N = 2^{n+2} -2$ vertices, each labelled by a randomly chosen $2n$-bit string.  The two trees are randomly glued as shown in Fig.~\ref{fig:gluedtrees}. More specifically, choose a leaf on the left end at random and connect it to a leaf on the right end chosen at random. Then connect the latter to a leaf on the left chosen randomly among the remaining ones, and so on, until every leaf on the left is connected to two leaves on the right (and vice versa). This creates a random cycle that alternates between the leaves of the two trees. The problem is, starting from the left root, to find a path to the right root in the smallest possible number of steps, while traversing the tree as in a maze. I.e., keeping a record of one's moves is allowed, but at any given vertex one can only see the adjacent vertices. More formally,  an oracle outputs the adjacent vertices of a given input vertex (note that the roots of the trees are the only vertices with adjacency two, so it is easy to check if the right root was found). The problem is, given the name of the left root and access to the oracle, to find the name of the right root in the smallest number of queries. Classical algorithms require at least a sub-exponential in $n$ number of oracle calls, but there exists a polynomial-time quantum algorithm based on quantum walks for solving this problem \cite{Childs:2003}. A polynomial-time quantum almost-adiabatic algorithm was given in \cite{Somma:2012kx}. The qualifier ``almost" is important: the algorithm is not adiabatic during the entire evolution, since it explicitly requires a transition from the ground state to the first excited state and back. We now review the algorithm, which (so far) provides the only example of a (sub-)exponential almost-adiabatic quantum speedup.

Let us denote the bit-string corresponding to the first root by $a_0$ and the second root by $a_{N-1}$.  Define the diagonal (in the computational basis) Hamiltonians
\beq
H_0 = - \ketbra{a_0}{a_{0}} \ , \quad H_1 = - \ketbra{a_{N-1}}{a_{N-1}} \ ,
\eeq
and the states
\beq \label{eqt:col}
\ket{c_j} = \frac{1}{\sqrt{N_j}} \sum_{i \in j\mathrm{-th \ column}} \ket{a_i}\ ,
\eeq
which are a uniform superposition over the vertices in the $j$-th column with $N_j = 2^j$ for $0 \leq j \leq n$ and $N_j = 2^{2n +1 -j}$ for $n+1 \leq j \leq 2n+1$. Note that $\ket{c_0} = \ket{a_0}$ and $\ket{c_{2n+1}} = \ket{a_{N-1}}$.
Let us define the Hamiltonian $A$ associated with the oracle as having the following non-zero matrix elements:
\beq
\bra{c_j} A \ket{c_{j+1}} = \left\{ \begin{array}{cc}
\sqrt{2} & j = n \\
1 & \mathrm{otherwise}
\end{array}\right.
\eeq
\begin{figure}[t] 
   \centering
   \includegraphics[width=0.95\columnwidth]{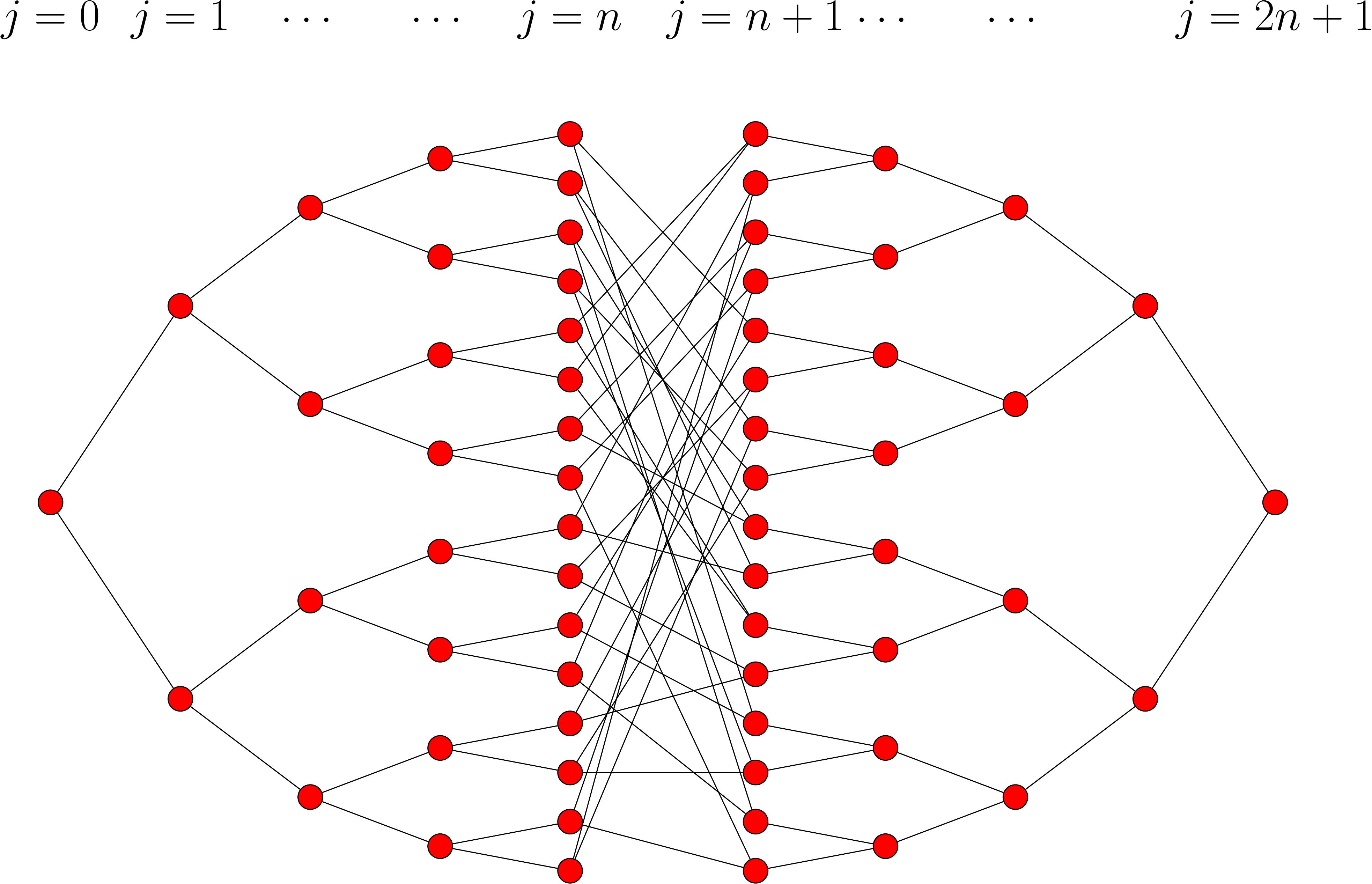} 
   \caption{A glued tree with $n=4$.  The labeling $j$ from Eq.~\eqref{eqt:col} is depicted on top of the tree.}
   \label{fig:gluedtrees}
\end{figure}
We then pick as our interpolating Hamiltonian:
\beq
H(s) = (1-s)\alpha H_0 -s (1-s) A + s \alpha H_1
\label{eq:H-glued}
\eeq
where $\alpha\in (0, 1/2)$ is a constant (independent of $n$) and $s(t)$ is the schedule.  Note that a unitary evolution according to this Hamiltonian will keep a state within the subspace spanned by $\left\{ \ket{c_j} \right\}$ if the state is initially within that subspace.  Since the instantaneous ground state at $s = 0$ ($\ket{a_0}$) is in this subspace, it suffices to only consider this subspace.  Because of the form of the Hamiltonian, the eigenvalue spectrum is symmetric about $s = 1/2$.

In this subspace, at $s_{\times} = \alpha/\sqrt{2}$ (and by symmetry at $1-s_\times$), the energy gap between the ground state and the first excited state closes exponentially in $n$. This is depicted in Fig.~\ref{fig:gluedtrees2}, where $s_1, s_2$ represent the region around $s_\times$ and $s_3, s_4$ represent the region around $1- s_\times$.  In the regions $s\in [0,s_1)$, $s\in [s_2,s_3)$, and $s\in [s_4,1]$, the energy gap between the ground state and first excited state is lower-bounded by $c/n^3$.
The gap between the first and second excited states is lower-bounded by $c'/n^3$ throughout the evolution. Both $c,c'>0$.

The proposed evolution exploits the symmetry and gap structure of the spectrum as follows. 
A schedule is chosen that guarantees adiabaticity only if the energy gap scales as $1/n^3$.
Then, during $s \in [0,s_1)$, the desired evolution is sufficiently adiabatic that it follows the instantaneous ground state. During $s \in [s_1, s_2)$, the evolution is non-adiabatic (since the gap scales as $1/e^n$) and a transition to the first excited state occurs with high probability.  During $s \in [s_2,s_3)$, the  evolution is again sufficiently adiabatic that it follows the instantaneous first excited state.  During $s \in [s_3, s_4)$, the evolution is again non-adiabatic and a transition from the first excited state back to the ground state occurs with high probability.  During $s \in [s_4,1]$, the evolution is again adiabatic and follows the instantaneous ground state. 

Since $t_f = \int_0^1 {ds}({d s}/{dt})^{-1}\sim n^6$, we conclude that $\ket{a_{N-1}}$ can be found in polynomial time.

\begin{figure}[t] 
   \centering
   \includegraphics[width=0.95\columnwidth]{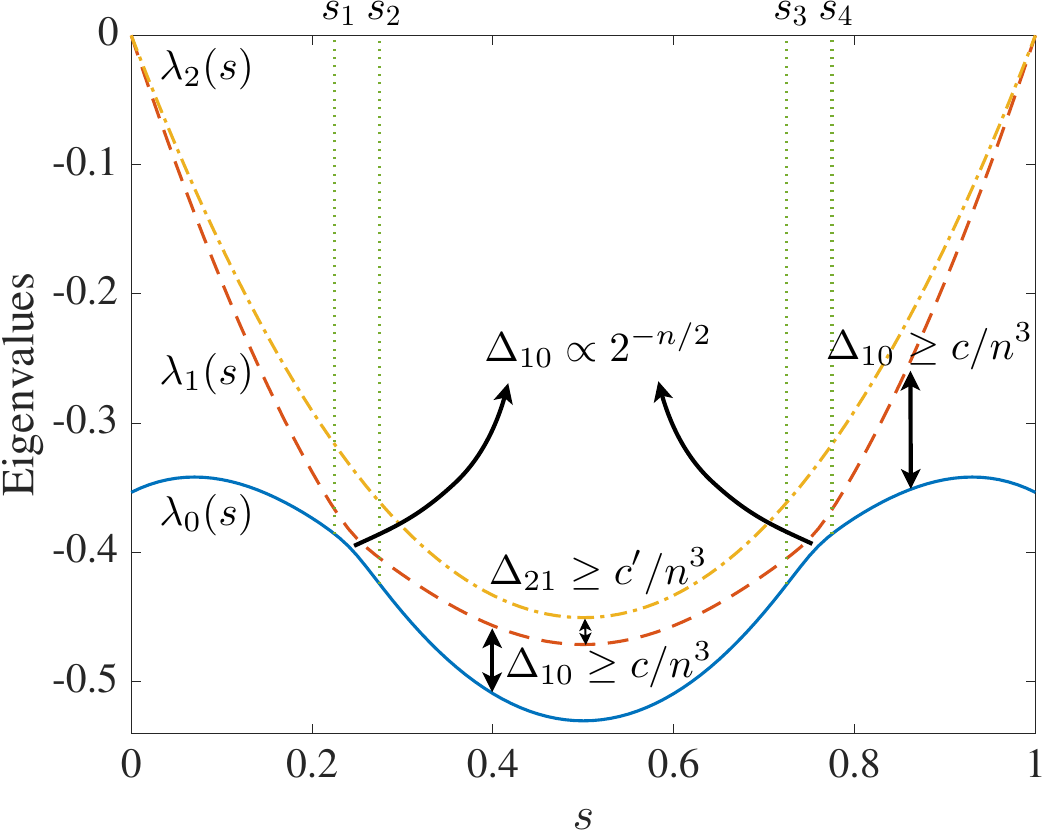} 
   \caption{The ground state ($\lambda_0(s)$, blue solid curve), first excited state ($\lambda_1(s)$, red dashed curve) and second excited state ($\lambda_2(s)$ yellow dot-dashed curve) of the glued-trees Hamiltonian~\eqref{eq:H-glued} for $\alpha = 1/\sqrt{8}$ and $n = 6$. Inside the region $[s_1,s_2]$ and $[s_3,s_4]$, the gap between the ground state and first excited state $\Delta_{10}$ closes exponentially with $n$.  In the region $[s_2,s_3]$, the gap between the ground state and first excited state $\Delta_{10}$ and the gap between the first excited state and second excited state $\Delta_{21}$ are bounded by $n^{-3}$.  Similarly, in the region $[s_4,1]$, the gap between the ground state and first excited state $\Delta_{10}$ is bounded by $n^{-3}$.}
   \label{fig:gluedtrees2}
\end{figure}

\subsection{Adiabatic PageRank algorithm}

We review the adiabatic quantum algorithm from \cite{Garnerone:2012xr} that prepares a state containing the same ranking information as the PageRank vector. The latter is a central tool in data mining and information retrieval, at the heart of the success of the Google search engine \cite{Pagerank}.
Using the adiabatic algorithm, the extraction of the full PageRank vector cannot, in general, be done more efficiently than when using the best classical algorithms known. However, there are particular graph-topologies and specific tasks of relevance in the use of search engines (such as finding just the top-ranked entries) for which the quantum algorithm, combined with other known quantum protocols, may provide a polynomial, or even exponential quantum speedup. 
Note that unlike the previous algorithms we reviewed in this section, which all provided a provable quantum speedup, the current algorithm provides a ``regular" quantum speedup, in the sense that it outperforms all currently known classical algorithms, but better future classical algorithms have not been ruled out.

\subsubsection{Google matrix and PageRank}
\label{sec:Googlematrix}
PageRank can be seen as the stationary distribution 
of a random walker on the web-graph, which spends its time on each page in proportion to the relative importance of that page \cite{citeulike:796239}. 

To model this define the transition matrix $ P_1 $ associated with the (directed) adjacency matrix $ A $ of the graph
\begin{equation}
P_1(i,j)=
\left\{ \begin{array}{ll}
 1/d(i) & \mbox{if $(i,j)$ is an edge of $A$};\\
 0 & \mbox{else},\end{array} \right.
\end{equation}
where $ d(i) $ is the out-degree of the $i$th node. 

The rows having zero matrix elements, corresponding to dangling nodes,
are replaced by the vector $ \vec{e}/n $ whose entries are all $ 1/n $, where $n$ is the number of pages or nodes, i.e., the size of the web-graph.
Call the resulting (right) stochastic matrix $ P_2 $. 
However, there could still be subgraphs with in-links but no out-links. Thus one defines the ``Google matrix" $ G $ as
\begin{equation}
G:=\alpha 
P_2^T + (1-\alpha)E,
\end{equation}
where  $ E \equiv \ket{\vec{v}} \bra{\vec{e}}$.
The ``personalization vector" $ \vec{v} $ is a probability distribution; 
the typical choice is $\vec{v}=\vec{e}/n $.
The parameter $\alpha\in(0,1)$ is the probability that the 
walker follows the link structure of the  web-graph at each step, rather than hop randomly 
between graph nodes according to $\vec{v}$ (Google reportedly uses $ \alpha=0.85 $). 
By construction, $G$ is irreducible and aperiodic,
and hence the Perron-Frobenius theorem \cite{horn2012matrix},\footnote{This theorem states that if all elements of a real symmetric square matrix $A$ are non-negative, then the largest eigenvalue of $A$ is real; furthermore, the components of the corresponding eigenvector can be chosen to be all non-negative.}
ensures the existence of a unique eigenvector with all positive entries
associated to the maximal eigenvalue $1$. This 
eigenvector is precisely the PageRank $ \vp $. Moreover, 
the modulus of the second eigenvalue of $G$ is upper-bounded by $\alpha$ \cite{Nusbaum:03}. 
This is important for the convergence of the power method, the standard computational 
technique employed to evaluate $ \vp $. It uses the fact that for any 
probability vector $ \vp_0 $ 
\begin{equation}
\vec{p}=\lim_{k\rightarrow \infty} G^k \vec{p}_0.
\end{equation}
The power method computes $\vp$ with accuracy  $\nu$ in a time that scales as $O[sn\log(1/\nu)/\log(1/\alpha)]$,
where $s$ is the sparsity of the graph (maximum number of non-zero entries per row of the adjacency matrix). 
The rate of convergence is determined by $\alpha$.

\subsubsection{Hamiltonian and gap}
\label{sec:pagerank-gap}
Consider the following non-local final Hamiltonian associated with
a generic Google matrix $ G $ 
(in this subsection we use $H$ and $h$ for
local and non-local Hamiltonians, respectively):
\begin{equation}
h_1=h(G)\equiv\left(\ident-G \right)^{\dagger} \left( \ident-G \right).
\label{Eq:hp}
\end{equation}
Since $ h(G) $ is positive semi-definite, and $1$ is the maximal
eigenvalue of $ G $ associated with $\vp$, it follows 
that the ground state of $ h(G) $ is given by $ |\pi\rangle \equiv \vp/\norm{\vp}$. 
The initial Hamiltonian has a similar 
form, but it is associated with the Google matrix $ G_c $ of the 
complete graph
\begin{equation}
h_0=h(G_c)\equiv\left(\ident-G_c \right)^{\dagger} \left( \ident-G_c \right).
\label{Eq:hi}
\end{equation}
The ground state of $ h_0$ is the uniform superposition state $ |\psi(0)\rangle = \sum_{j=1}^n|j\rangle/\sqrt{n} $. The basis vectors $|j\rangle $ span the $n$-dimensional Hilbert space of $\log_2 n$ qubits 
The interpolating adiabatic Hamiltonian is
\begin{equation}
h(s)=(1-s)h_0 + s h_1.
\label{Eq:hs}
\end{equation}
Equations~\eqref{Eq:hp}-\eqref{Eq:hs} completely 
characterize the adiabatic quantum PageRank algorithm, apart from the schedule $s(t)$. 

By numerically simulating the dynamics generated by $h(s)$, 
\cite{Garnerone:2012xr} showed that for typical random graph instances generated using the ``preferential attachment model" \cite{BaAl,BoRiSp} and ``copying model" \cite{KlKuRa} (both of which yield sparse random graphs with small-world and scale-free features), 
the typical run time of 
the adiabatic quantum PageRank algorithm scales  as
\begin{equation}
t_f \sim (\log\log n)^{b-1} (\log n)^{b},
\label{Eq:compcomp}
\end{equation}
where $b>0$ is some small integer that depends on the details of the graph parameters. The numerically computed gap scales as $(\log n)^{-b}$, which \cite{Garnerone:2012xr} found to be due to the power law distribution of the out-degree nodes $d(i)$.\footnote{The gap becomes too small for a quantum advantage, i.e., scales as $1/{\rm poly}(n)$, for graphs with only in-degree power-law distribution or when the out-degrees are equal to the in-degrees. This was studied in more detail in \cite{Frees:2013fd}.}

\subsubsection{Speedup}
We next discuss two tasks for which this adiabatic quantum ranking algorithm offers a speedup.

The best currently known classical Markov Chain Monte Carlo {(MCMC)} technique used to evaluate the full PageRank vector 
requires a time (in the bulk synchronous parallel computational model \cite{Valiant:1990}) which scales as $O[\log (n)]$ \cite{Sarma2013}.  The algorithm launches $\log n$ random walks from each node of the graph in parallel (for a total of $n \log (n)$ walkers), with each node communicating $O[\log(n)]$ bits of data to each of its connected neighbors after each step.  After $O[\log (n)]$ steps, the total number of walkers that have visited a node is used to estimate the PageRank of that node.   In the absence of synchronization costs [synchronization and communication are known to be important issues for networks with a large number of processors \cite{Awerbuch1985,Bisseling2004,Rauber2010,Kumar2003}], the classical cost can be taken to be $O[n \log(n)^2]$, i.e., the number of parallel processes multiplied by the duration of each process.\footnote{This analysis improves upon the estimates of the classical cost presented in \cite{Garnerone:2012xr}, and accounts for the critique presented in \cite{Moussa:2013}.}

At the conclusion of the adiabatic evolution generated by the Hamiltonian in Eq.~\eqref{Eq:hs}, the 
PageRank {vector $\vec{p}=\{p_i\}$} is encoded into the {quantum PageRank} state $\ket{\pi}=\sum_{i=1}^n \sqrt{\pi_i}
\ket{i}$ of a $(\log_2 n)$-qubit system, where 
{$ |i\rangle $ denotes the $i$-th node in the graph $G$.}
The probability of measuring node $i$ is $\pi_i = p_i^2/\norm{\vec{p}}^2$. One can estimate ${\pi_i}$ by repeatedly 
sampling the expectation value of the operator $ \sigma^z_i $ in the final state.  
The number of measurements $M$ needed to estimate $\pi_i$ is {given} by the Chernoff-Hoeffding 
bound \cite{citeulike:3392582}, allowing one to approximate {$\pi_i$} with an additive error $e_i$ 
and with $M={\rm poly}(e_i^{-1})$. A nontrivial approximation requires $e_i \leq p_i$ and, these are typically $O(1/n)$.

The fact that the amplitudes of the quantum PageRank state are $\{\sqrt{\pi_i}=p_i/\norm{\vec{p}}\}$, 
rather than $\{\sqrt{p_i}\}$, is a virtue: the number of samples needed to estimate the rank $\pi_i$
with additive error $e_i\sim\pi_i$ scales as $O[n^{2\gamma_i-1}]$, so the total quantum cost is $O[n^{2 \gamma_i -1} \textrm{polylog}(n)]$.\footnote{
It was observed numerically in \cite{Garnerone:2012xr} that ${p_i}
  \propto 1/n^{\gamma_i}$, where $\gamma_i\in (0.6,1]$, and that $\Vert \vec{p}
  \Vert _2^2\propto 1/n$. Let $e_i$ denote the additive error corresponding to
  $\pi_i = p_i^2/\Vert \vec{p} \Vert_2^2 \sim n p_i^2$. It follows from the Chernoff-Hoeffding
  inequality that the number of samples $M(x)$ from the distribution $x$, where
  $x=\pi=\{\pi_i\}$ (output of the quantum algorithm), required for a given, fixed additive estimation error, is proportional to the inverse of the additive error: $M(\pi)\sim 1/e_i$. Assuming $e_i \sim \pi_i$, it
  follows that $M(\pi) \sim 1/\pi_i \sim 1/(n p_i^2) \sim n^{2\gamma_i-1}$. The total cost required to
  prepare the sample in the quantum case is 
  $\Or[\mathrm{polylog}(n)]$.}
 Thus, for the combined task of state preparation and rank estimation, there 
is a polynomial quantum speedup whenever $\gamma_i<1$, namely $O[n^{2\gamma_i-1}\textrm{polylog}(n)]$ \textit{vs}. $O[n \ \textrm{polylog} (n)]$; 
simulations reported in \cite{Garnerone:2012xr} show that this is indeed the case for the top-ranked $\log (n)$ entries, and in applications one is most often interested in the top entries. We emphasize that this holds in the average (not worst) case, and is not a provable speedup; the evidence for the scaling is numerical, and it is unknown whether a classical algorithm for the preparation of $\pi$ rather than $\vec{p}$ may give a similar scaling to the quantum scaling, though if that is the case one could consider quantum preparation of $\{\pi_i^2/\|\pi\|^2\}$, etc.

Another context for useful applications is comparing successive PageRanks, or more generally ``q-sampling" \cite{AharonovTa-Shma}.
Suppose one perturbs the web-graph. 
The adiabatic quantum algorithm 
can provide, in time $O[{\textrm{polylog}}(n)]$, the pre- and post-perturbation
states $\ket{\pi}$ and $\ket{\tilde{\pi}}$ as input to a quantum circuit 
implementing the SWAP-test \cite{PhysRevLett.87.167902}.
To obtain an estimate of the fidelity $|\bra{\pi}\tilde{\pi}\rangle|^2$ one needs to measure an ancilla $O(1)$ times, the number depending only on the desired precision.  In contrast, deciding whether two probability distributions are close classically requires $O[n^{2/3} \log n]$ samples from each \cite{892113}. 
Whenever some relevant perturbation 
of the previous 
quantum PageRank state is observed, one can decide to run 
the classical algorithm again to update the 
classical {PageR\nolinebreak[4]ank}. 


\section{Universality of AQC}
\label{sec:universality}

What is the relation between the computational power of the circuit model and the adiabatic model of quantum computing? It turns out that they are equivalent, up to polynomial overhead. It is well known that the circuit model is universal for quantum computing, i.e., that there exist sets of gates acting on a constant number of qubits each that can efficiently simulate a quantum Turing machine \cite{Deutsch:85,Yao:93}. A set of gates is said to be universal for QC if any unitary operation may be approximated to arbitrary accuracy by a quantum circuit involving only those gates \cite{nielsen2000quantum}. The analog of such a set of gates in AQC is a Hamiltonian. An operational definition of universal AQC is thus to efficiently map any circuit to an adiabatic computation using a sufficiently powerful Hamiltonian. Formally:
\begin{definition}[Universal Adiabatic Quantum Computation]
\label{def:UAQC}
A time-dependent Hamiltonian $H(t)$, $t\in[0,t_f]$, is universal for AQC if, given an arbitrary quantum circuit $U$ operating on an arbitrary initial state $\ket{\psi}$ of $n$ $p$-state particles and having depth $L$, the ground state of $H(t_f)$ is equal to $U\ket{\psi}$ with probability greater than $\epsilon>0$, the number of particles $H(t)$ operates on is $\text{poly}(n)$ $\forall t$, and $t_f=\text{poly}(n,L)$.
\end{definition}
The stipulation that the ground state of $H(t_f)$ is equal to the final state at the end of the circuit ensures that the circuit and the adiabatic computation have the same output. We note that it is possible and useful to relax the ground state requirement and replace it with another eigenstate of $H(t)$ (see, e.g., Sec.~\ref{sec:stoq-QMA-comp}). 
The requirement that the number of particles and time taken by the adiabatic computation are polynomial in $n$ and $L$ ensures that the resources used do not blow up. 
 
We begin, in Sec.~\ref{sec:circuit-to-AQC}, by showing that the circuit model can efficiently simulate AQC. The real challenge is to show the other direction, i.e., that AQC can efficiently simulate the circuit model, which is what we devote the rest of this section to. Along the way, this establishes the universality of AQC.
We present several proofs, starting in Sec.~\ref{sec:history-state} with a detailed review of the history state construction of  \cite{aharonov_adiabatic_2007}, who showed in addition that six-state particles in two dimensions suffice for universal adiabatic quantum
computation. This was improved in \cite{KempeGadget}, using perturbation-theory gadgets, who showed that qubits can be used instead of six-state
particles, and that adiabatic evolution with $2$-local Hamiltonians is quantum universal. A $2$-local model of universal AQC in 2D, which we review in Sec.~\ref{sec:GSE}, was proposed in \cite{MLM:06}, using fermions. Universal AQC using qubits on a two-dimensional grid was accomplished in \cite{OliveiraGadget}. Further simplifications of universal AQC in 2D were presented in \cite{Breuckmann:2014,Gosset:2014rp,Lloyd:2016}, using the space-time circuit model, which we review  in Sec.~\ref{sec:space-time}. The ultimate reduction in spatial dimensionality was accomplished in \cite{Aharonov:2009fi}, who showed that universal AQC is possible with 1D $9$-state particles, as we review in Sec.~\ref{sec:UAQC-1D}. Finally, in Sec.~\ref{sec:gap-amp} we review a construction that allows one to quadratically amplify the gap of any Hamiltonian used in AQC (satisfying a frustration-freeness property), though this requires the computation to take place in an excited state.

\subsection{The circuit model can efficiently simulate AQC}
\label{sec:circuit-to-AQC}

That the circuit model can efficiently simulate the adiabatic model is relatively straightforward and was first shown in \cite{farhi_quantum_2000}.  Assume for simplicity a linear schedule, i.e., an AQC Hamiltonian of the form $H(t) = (1 - \frac{t}{t_f}) H_0 + \frac{t}{t_f} H_1$. The evolution of a quantum system generated by the time-dependent Hamiltonian $H(t)$ is governed by the unitary operator:
\beq
U(t_f,0) = \Texp \left[- i \int_0^{t_f} dt H(t)\right]\ .
\eeq
If $t_f$ satisfies the condition for adiabaticity, $U(t_f,0)$ will map the ground state at $t = 0$ to the ground state at $t_f$. Therefore it suffices to show that the circuit model can simulate $U(t_f,0)$. To do so, we approximate the evolution by a product of unitaries involving time-independent Hamiltonians $H_m' \equiv H( m \Delta t)$:
\beq \label{eqt:approx1}
U(t_f,0) \mapsto U'(t_f,0 ) =  \prod_{m=1}^M U_m' = \prod_{m=1}^M e^{-i \Delta t H_m'}\ ,
\eeq
where $\Delta t = t_f / M$.  The error incurred by this approximation is  \cite{Dam:2001fk}:
\beq
\Vert U(t_f,0) - U'(t_f,0) \Vert \in O \left(\sqrt{t_f \text{poly}(n) / M} \right)\ ,
\label{eq:subdomerr}
\eeq
where we used
\begin{align}
\Vert H(t) - H'_{\lceil m t /t_f\rceil} \Vert &\leq \frac{1}{M} \Vert H_1 - H_0 \Vert \\ \notag
&\in O(\text{poly}(n) / M)\ .
\end{align}
We now wish to approximate each individual term in the product in Eq.~\eqref{eqt:approx1} using the Baker-Campbell-Hausdorff formula \cite{Klarsfeld:89} by:
\beq \label{eqt:CBHDecomp}
U'_m \mapsto U_m'' = e^{-i \Delta t \left( 1 - \frac{m \Delta t}{t_f} \right) H_0} e^{- i \Delta t \frac{m \Delta t}{t_f} H_1} \ ,
\eeq
which incurs an error $ \Vert e^{A + B} - e^A e^B \Vert \in O \left( \Vert A B \Vert \right)$ due to the neglected leading order commutator term $[A,B]/2$, i.e., 
\beq
\Vert U_m' - U_m'' \Vert \in O \left( \frac{t_f^2}{M^2}  \Vert H_0 H_1 \Vert \right)\ .
\eeq
Therefore, accounting for the $M$ terms in the product and observing that the error in Eq.~\eqref{eq:subdomerr} is subdominant, the total error is \cite{Dam:2001fk}:
\beq
\Vert U(t_f,0) - \prod_{m=1}^M U_m'' \Vert \in O \left(\text{poly}(n) t_f^2 / M \right)\ .
\eeq
This means we can approximate $U(t_f,0)$ with a product of $2M$ unitaries provided that $M$ scales as $t_f^2 \text{poly}(n)$.

Depending on the form of $H_0$ and $H_1$, they may need further decomposition in order to write the  terms in Eq.~\eqref{eqt:CBHDecomp} in terms of few-qubit unitaries.  E.g., for the standard initial Hamiltonian $H_0 = -\sum_i \sigma^x_i$, which is a sum of commuting single qubit operators, we can write $e^{-i \Delta t (1 - {t}/{t_f} ) H_0 /K}$ as a product of $n$ one-qubit unitaries. Likewise, assuming that $H_1$ is $2$-local, we can write $e^{-i \frac{m \Delta t}{t_f} \Delta t  H_1/M}$ as a product of up to $n^2$ two-qubit unitaries within the same order of approximation as Eq.~\eqref{eqt:CBHDecomp}.
Thus, $U(t_f, 0)$ can be approximated as a product of unitary operators each of which acts on a few qubits. 
The scaling of $t_f$ required for adiabatic evolution is inherited by the number of few-qubit unitary operators in the associated circuit version of the algorithm.

A more efficient method was proposed in \cite{Boixo:2009aa}, building upon the ideas explained in Sec.~\ref{sec:ATlowerbound}. This ``eigenpath traversal by phase randomization" method applies the Hamiltonian $H(t_j)$ in piecewise continuous manner at random times $t_j$. Each interval $[t_j,t_{j+1}]$ corresponds to a unitary $e^{-i H(t_j)}$, which then needs to be decomposed into one- and two-qubit gates, as above. The randomization introduces an effective eigenstate decoupling in the Hamiltonian eigenbasis (similarly to the effect achieved by projections in the Zeno effect), so that if the initial state is the ground state, the evolution will follow the ground state throughout as required for AQC. The algorithmic cost of this randomization method is defined as the average 
number of times the unitaries are applied, and it can be shown that the cost is $O[L^2/(\varepsilon\Delta)]$, where $\varepsilon$ is the desired maximum error of the final state compared
to the target eigenstate, and $L$ is the path length [Eq.~\eqref{eq:L}]. Since $L\leq\max_s\|\dot{H}(s)\|/\Delta$ as we saw in Sec.~\ref{sec:ATlowerbound}, the worst-case bound on the cost is $\max_s\|\dot{H}(s)\|^2/(\varepsilon\Delta^3)$, up to logarithmic factors.

\subsection{AQC can efficiently simulate the circuit model: history state proof}
\label{sec:history-state}
The goal is, given an arbitrary $n$-qubit quantum circuit, to design an adiabatic computation whose final ground state is the output of the quantum circuit described by a sequence of $L$ one or two-qubit unitary gates, $U_1, U_2, \dots U_L$.  This adiabatic simulation of the circuit should be efficient, i.e., it may incur at most polynomial overhead in the circuit depth $L$. 
In this subsection we review the proof presented in \cite{aharonov_adiabatic_2007}. This was the first complete proof of the universality of AQC, and many of the ideas and techniques introduced therein inspired subsequent proofs, remaining relevant today. 

Let us assume that the $n$-qubit input to the circuit is the $\ket{0\cdots 0}$ state.  After the $\ell$-th gate, the state of the quantum circuit is given by $\ket{\alpha(\ell)}$.  
To proceed, we use the ``circuit-to-Hamiltonian'' construction~\cite{Kitaev:book}, where the final Hamiltonian will have as its ground state the entire history of the quantum computation.  This ``history state" is given by:
\bes
\label{eqt:HistoryState}
\begin{align} 
\ket{\eta}  &= \frac{1}{\sqrt{L+1}}  \sum_{\ell=0}^L \ket{\gamma(\ell)} \\
\ket{\gamma(\ell)}  &\equiv \ket{\alpha(\ell)} \otimes \ket{1^\ell 0^{L-\ell}}_\mathrm{c}  
\end{align}
\ees
where $\ket{1^\ell 0^{L-\ell}}_\mathrm{c}$ denotes the ``Feynman clock" \cite{Feynman:1985ul} register composed of $L+1$ qubits.  The notation means that we have $\ell$ ones followed by $L-\ell$ zeros to denote the time after the $\ell$-th gate.  We wish to construct a Hamiltonian $H_{\mathrm{init}}$ with ground state $\ket{\gamma(0)}$ and a Hamiltonian $H_{\mathrm{final}}$ with ground state $\ket{\eta}$.  Let:
\bes
\begin{align}
H_{\mathrm{init}} & = H_{\mathrm{c-init}} + H_{\mathrm{input}}  + H_{\mathrm{c}} \\
H_{\mathrm{final}} & = \frac{1}{2}H_{\mathrm{circuit}}  + H_{\mathrm{input}}  + H_{\mathrm{c}} \\
H_{\mathrm{circuit}} &= \sum_{\ell=1}^L H_{\ell} \ .
\end{align}
\ees
The full time independent Hamiltonian $H(s)$ is given by~\cite{aharonov_adiabatic_2007}:
\begin{eqnarray}
H(s) &=& (1-s) H_{\mathrm{init}} + s H_{\mathrm{final}}    \\ 
&=& H_{\mathrm{input}}  + H_{\mathrm{c}}+ (1-s) H_{\mathrm{c-init}}  +  \frac{s}{2} H_{\mathrm{circuit}} \nonumber \ .
\end{eqnarray}
The various terms are chosen so that the ground state always has energy $0$:
\begin{itemize}
\item $H_{\mathrm{c}}$: This term should ensure that the clock's state is always of the form $\ket{1^{\ell} 0^{L-\ell}}_{\mathrm{c}}$.  Therefore, we energetically penalize any clock-basis state that has the sequence $01$:
\beq \label{eqt:Hc}
H_{\mathrm{c}} = \sum_{\ell = 1}^{L-1} \ket{0_{\ell} 1_{\ell + 1}}_{\mathrm{c}} \bra{0_{\ell} 1_{\ell + 1}}
\eeq
where $ \ket{0_{\ell} 1_{\ell + 1}}_{\mathrm{c}}$ denotes a $0$ on the $\ell$-th clock qubit and $1$ on the $(\ell+1)$-th clock qubit.  Any illegal clock state will have an energy $\geq 1$.  Any legal clock state will have energy $0$.
\item $H_{\mathrm{c-init}}$: Ensures that the initial clock state is $\ket{0^L}_{\mathrm{c}}$.
\beq \label{eqt:Hcinit}
H_{\mathrm{c-init}} = \ket{1_1}_\mathrm{c} \bra{1_1}
\eeq
Note that we only need to specify the first clock qubit to be in the zero state.  For a legal clock state, Eqs.~\eqref{eqt:Hc} and \eqref{eqt:Hcinit} imply that the rest are in the zero state as well.
\item $H_{\mathrm{input}}$: Ensures that if the clock state is $\ket{0^L}$, then the computation qubits are in the $\ket{0^n}$ state.
\beq
H_{\mathrm{input}} = \sum_{i=1}^n \ket{1_i}\bra{1_i} \otimes \ket{0_1}_\mathrm{c} \bra{0_1}
\eeq
\item $H_{\ell}$: Ensures that the propagation from $\ell -1$ to $\ell$ corresponds to the application of $U_\ell$.  
\bes
\begin{align}
H_1 & = \ident \otimes \ket{0_1 0_2}_{\mathrm{c}} \bra{0_1 0_2} - U_1 \ket{1_1 0_2}_{\mathrm{c}} \bra{0_1 0_2} \nonumber \\
&- U_1^{\dagger} \ket{0_1 0_2}_{\mathrm{c}} \bra{1_1 0_2} + \ident \otimes \ket{1_1 0_2}_{\mathrm{c}} \bra{1_1 0_2} \\
H_{2 \leq \ell\leq L-1} & = \ident \otimes \ket{1_{\ell-1} 0_\ell 0_{\ell+1} }_{\mathrm{c}} \bra{1_{\ell-1} 0_\ell 0_{\ell+1} } \nonumber \\
& - U_\ell \ket{1_{\ell-1} 1_\ell 0_{\ell+1} }_{\mathrm{c}} \bra{1_{\ell-1} 0_\ell 0_{\ell+1} } \nonumber \\
& - U_\ell^{\dagger} \ket{1_{\ell-1} 0_\ell 0_{\ell+1} }_{\mathrm{c}} \bra{1_{\ell-1} 1_\ell 0_{\ell+1}} \nonumber \\
& + \ident \otimes \ket{1_{\ell-1} 1_\ell 0_{\ell+1}}_{\mathrm{c}} \bra{1_{\ell-1} 1_\ell 0_{\ell+1}} \\
H_L & = \ident \otimes \ket{1_{L-1} 0_L}_{\mathrm{c}} \bra{1_{L-1} 0_{L}} \nonumber \\
&- U_L \ket{1_{L-1} 1_{L}}_{\mathrm{c}} \bra{1_{L-1} 0_L} \nonumber \\
& - U_1^{\dagger} \ket{1_{L-1} 0_L} \bra{1_{L-1} 1_L} \nonumber \\
& + \ident \otimes \ket{1_{L-1} 1_L}_{\mathrm{c}} \bra{1_{L-1} 1_L} 
\end{align}
\ees 
Note that the first and last terms leave the state unchanged.  The second term propagates the computational state and clock register forward, while the third term propagates the computational state and clock register backward. 
\end{itemize}

It turns out that the state $\ket{\gamma(0)} = \ket{\alpha(0)} \otimes \ket{0^L}$ is the ground state of $H_{\mathrm{init}}$ with eigenvalue $0$, and $\ket{\eta}$ is the ground state of $H_{\mathrm{final}}$ with eigenvalue $0$.
Let $\mS_0$ be the subspace spanned by $\left\{ \ket{\gamma(\ell)} \right\}_{\ell=0}^L$.  The state $\ket{\alpha(0)}$ is the input to the circuit, so it can be taken to be the $\ket{0\cdots 0}$ state, i.e., the initial ground state is an easily prepared state.  Since the initial state $\ket{\gamma(0)} \in \mS_0$, the dynamics generated by $H(s)$ keep the state in $\mS_0$. It turns out that the ground state is unique for $s \in [0,1]$. By mapping the Hamiltonian within $\mS_0$ to a stochastic matrix, it is possible to find a polynomial lower bound on the gap from the ground state within $\mS_0$:
\beq
 \Delta(H_{\mS_0}) \geq \frac{1}{4} \left( \frac{1}{6L} \right)^2\ .
 \label{eq:gapHS0}
\eeq

It is also possible to bound the global gap (i.e., not restricted to the $\mathcal{S}_0$ subspace) as
\beq
\Delta(H) \geq \Omega(1/L^3)\ .
 \label{eq:gapHL3}
\eeq

A measurement of the final state will find the final outcome of the quantum circuit $\ket{\gamma(L)}$ with probability $\frac{1}{L+1}$.  This can be amplified by inserting identity operators at the end of the circuit, hence causing the history state to include a greater superposition of the final outcome of the circuit.  Together, these results show that there is an efficient 
implementation of any given quantum circuit using the adiabatic algorithm with $H(s)$. Here and elsewhere ``efficient" means up to polynomial overhead, i.e., where $t_f$ scales as a polynomial in $L$.

The proof techniques used to obtain these results are instructive and of independent interest, so we review additional technical details in Appendix~\ref{app:AQC-universality-proof}.

To conclude this section, we briefly mention additional results supporting the equivalence of the circuit and adiabatic approach, in terms of state preparation.  In \cite{AharonovTa-Shma} it was shown (Theorem 2) that any quantum state that can be efficiently generated in the circuit model can also be efficiently generated by an adiabatic approach, and vice versa, for the same initial state.  The proof relies on two important lemmas, the ``sparse Hamiltonian lemma" and the ``jagged adiabatic path" lemma.  The former gives conditions under which a Hamiltonian is efficiently simulatable in the circuit model, and the latter provides conditions under which a sequence of Hamiltonians, defining a path, can have a non-negligible spectral gap. 

\subsection{Fermionic ground state quantum computation}
\label{sec:GSE}

A model of ground state quantum computation (GSQC) using fermions was independently proposed in \cite{Mizel:01} [see also \cite{Mizel:02,Mizel:04} and \cite{Mao:05,Mao:05a}] around the same time as AQC.
In GSQC, one executes a quantum circuit by producing a ground state that spatially encodes the entire temporal
trajectory of the circuit, from input to output. It was shown in \cite{MLM:06} how to adiabatically reach the desired ground state, thus providing an alternative to history state type constructions for universal AQC. One of the differences between the GSQC and history states constructions is that instead of relying on Feynman ``global clock particle'' idea, particles are synchronized \emph{locally} (via CNOT gates), an idea that traces back to \cite{Margolus:90} and was later adopted in some of the space-time circuit-to-Hamiltonian constructions \cite{Breuckmann:2014}.

Consider a quantum circuit with $n$ qubits and depth $L$. We associate $2(L+1)$ fermionic modes with every qubit $q$, via creation operators $a_{q,\ell}^\dagger$ and $b_{q,\ell}^\dagger$, where $\ell=0,\dots,L$. One can view these $2n(L+1)$ modes as the state-space of $n$ spin-$1/2$ fermions, where each
fermion can be localized at sites on a 1D (time)-line of length $L+1$. To illustrate this with a concrete physical system, imagine a two-dimensional array of quantum dots with $L+1$ columns and two rows per qubit, corresponding to the $\ket{0}$ and $\ket{1}$ basis states of that qubit. A total of $n$ electrons are placed in the array. The state of each qubit determines the spin state of the corresponding electron, which in term determines which of the two rows it is in, while the clock of each qubit is represented
by which column the electron is in.

It is convenient to group creation operators into row
vectors $C^\dagger_{q,\ell} = (a_{q,\ell}^\dagger \ b_{q,\ell}^\dagger)$. Then for each single-qubit gate $U^{(1)}_{q,\ell}$ we introduce a term 
\beq
H^{(1)}_{q,\ell}(s) = \left(C^\dagger_{q,\ell} - s C^\dagger_{q,\ell-1} (U^{(1)}_{q,\ell})^\dagger\right)\left(C_{q,\ell} - s C_{q,\ell-1} U^{(1)}_{q,\ell}\right)
\eeq
into the circuit Hamiltonian $H_{\mathrm{circuit}}$. The off-diagonal terms represent hopping or tunneling of the $q$-th electron from site $\ell-1$ to $\ell$ (and v.v.), while $U^{(1)}_{q,\ell}$ acts on the electron's spin. The diagonal terms $C^\dagger_{q,\ell}C_{q,\ell}$ and $C^\dagger_{q,\ell-1}C_{q,\ell-1}$ ensure that $H^{(1)}_{q,\ell}\geq 0$. The parameter $s\in[0,1]$ controls the interpolation from the initial, simple to prepare ground state at $s=0$ when there is no tunneling and every electron is frozen in place, to the full realization of all the gates $U^{(1)}_{q,\ell}$ when $s=1$. 

One can similarly define CNOT Hamiltonian terms between electrons or fermions, whose form can be found in \cite{MLM:06} [see also \cite{Breuckmann:2014}]. These $2$-local terms can be understood as a sum of an identity and NOT term. For such two-qubit gates, the fermions corresponding to the control and target qubits both tunnel forward or backward and the internal spin-state of the target fermion changes depending on the internal state of the control fermion. An important additional ingredient is the addition of a penalty term that imposes an energy penalty on states in which one qubit has gone through the CNOT gate without the other. Instead of the Feynman clock used in the history state construction, there are many local clocks, one per qubit. The synchronization mechanism takes place via the CNOT Hamiltonian. Moreover, the entire construction naturally involves only $2$-local interactions between fermions in 2D. 

While the fermionic GSQC model proposed in \cite{MLM:06} was shown there to be universal for AQC, its gap analysis was incomplete.\footnote{On p.4 of \cite{MLM:06} it was claimed 
that ``$\bra{Z}H\ket{Z} \geq \mathcal{E} O(1/N^2)$" ($N$ is $L$ in our notation), implying a lower bound on the spectral gap of the total Hamiltonian $H$. This claim was based on \cite{Mizel:02}, but was in fact not proven there. Here $\mathcal{E}$ is the energy scale of the CNOT terms, the total Hamiltonian is $H=H_{0}+H_{1}$, where $H_{0}$ contains all of the single qubit terms and $H_{1}$ includes all of the \textsc{CNOT} terms, and $\ket{Z}$ denotes the known ground state of $H_{0}$. As pointed out in \cite{Breuckmann:2014}, the missing step is essentially to exclude zero-energy, invalid time-configurations.}
This was fixed in \cite{Childs:2013ef}, which proved  the ``Nullspace Projection Lemma" that was implicitly assumed in \cite{MLM:06}. This lemma is interesting in its own right, so we reproduce it here:
\begin{lemma}[Nullspace Projection Lemma \cite{Childs:2013ef}]
\label{lem:nullspace-proj}
Let $\Delta(A)$ denote the smallest nonzero eigenvalue of the positive semidefinite operator $A$. Let $H_0$ and $H_1$ be positive semidefinite and assume the nullspace $S$ of $H_0$ is non-empty. Also assume that $\Delta(\left.H_1\right|_S) \geq c >0$ and $\Delta(H_0)\geq d>0$. Then
\beq
\Delta(H_0+H_1) \geq \frac{cd}{c+d+\norm{H_1}}\ .
\eeq
\end{lemma}

As shown in \cite{Breuckmann:2014}, the fermionic GSQC model can be unitarily mapped onto the space-time circuit-to-Hamiltonian model for qubits in 2D, where the gap analysis is more convenient. 
Using the same mapping, \cite{Breuckmann:2014} also showed that the fermionic model of \cite{MLM:06} is in fact QMA-complete. We thus proceed to discuss the space-time model next.

\subsection{Space-time Circuit-to-Hamiltonian Construction}
\label{sec:space-time}

Here we briefly review another construction that realizes universal adiabatic quantum computation \cite{Gosset:2014rp,Lloyd:2016}. This builds on the so-called space-time circuit-to-Hamiltonian construction \cite{Breuckmann:2014}, which in turn is based on the Hamiltonian computation construction of \cite{Janzing:2007}.  
We consider the $2n$-qubit quantum circuit with $n^2$ two-qubit gates, arranged as shown in Fig.~\ref{fig:QCJanzing}.  This form is sufficient for universal quantum computation \cite{Janzing:2007}.  An equivalent representation of the circuit is given in Fig.~\ref{fig:QCJanzing2}, where the $n^2$ gates are arranged in a rotated $n \times n$ grid.  Each plaquette $p$ is associated with a gate $U_p$, of which the majority are identity gates.  Only a $k \times k$ subgrid of the $n\times n$ grid with $k = \sqrt{n}/16$ has non-identity gates, with the subgrid located as shown in Fig.~\ref{fig:QCJanzing3}.  This region is referred to as the interaction region.
\begin{figure}[htbp] 
   \centering
   \subfigure[]{\hspace*{0cm} \includegraphics[width=0.25\columnwidth,trim=0 -10cm 0cm 0]{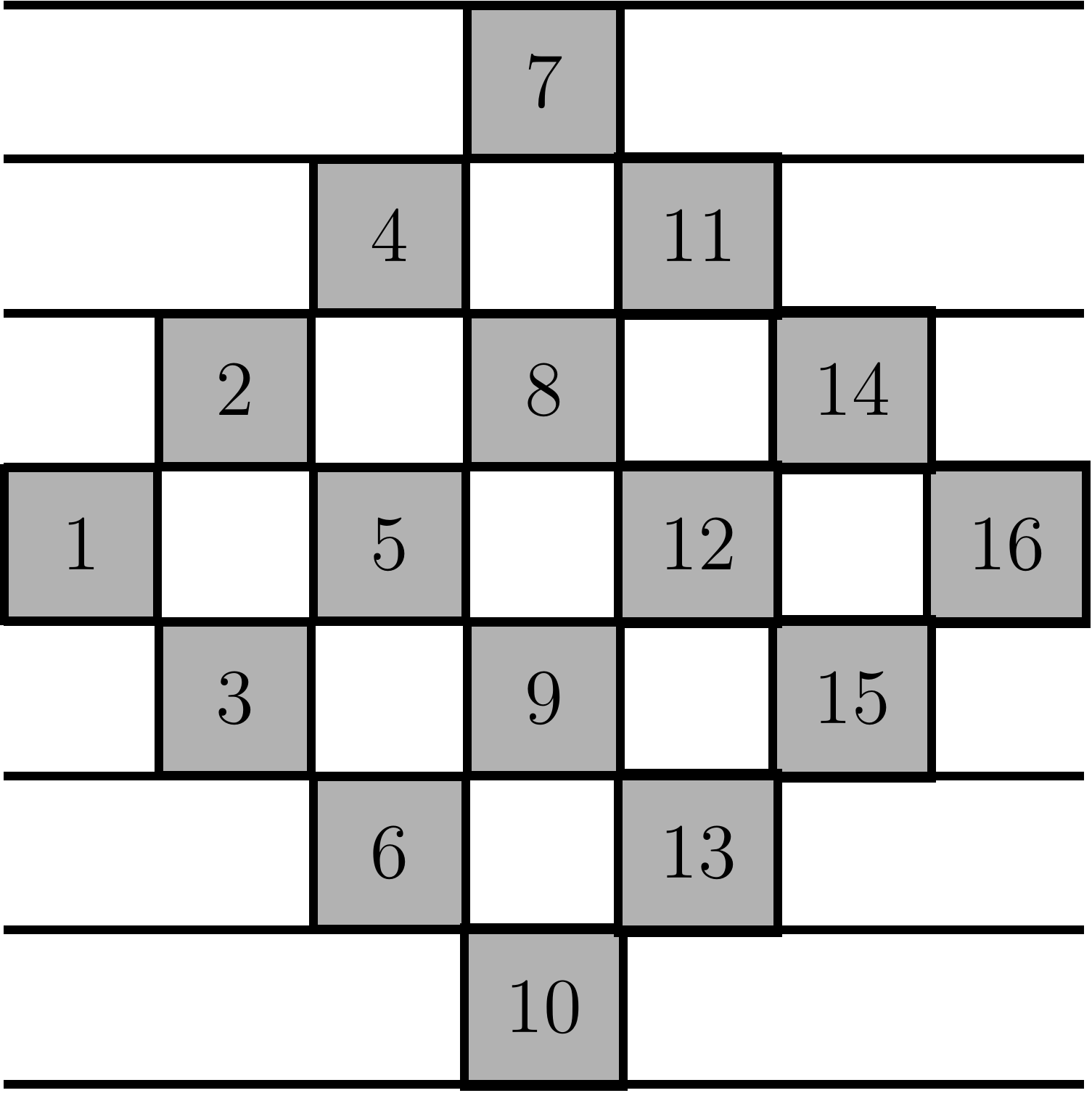} \label{fig:QCJanzing} }  
   \subfigure[]{\hspace*{0cm}\includegraphics[width=0.35\columnwidth]{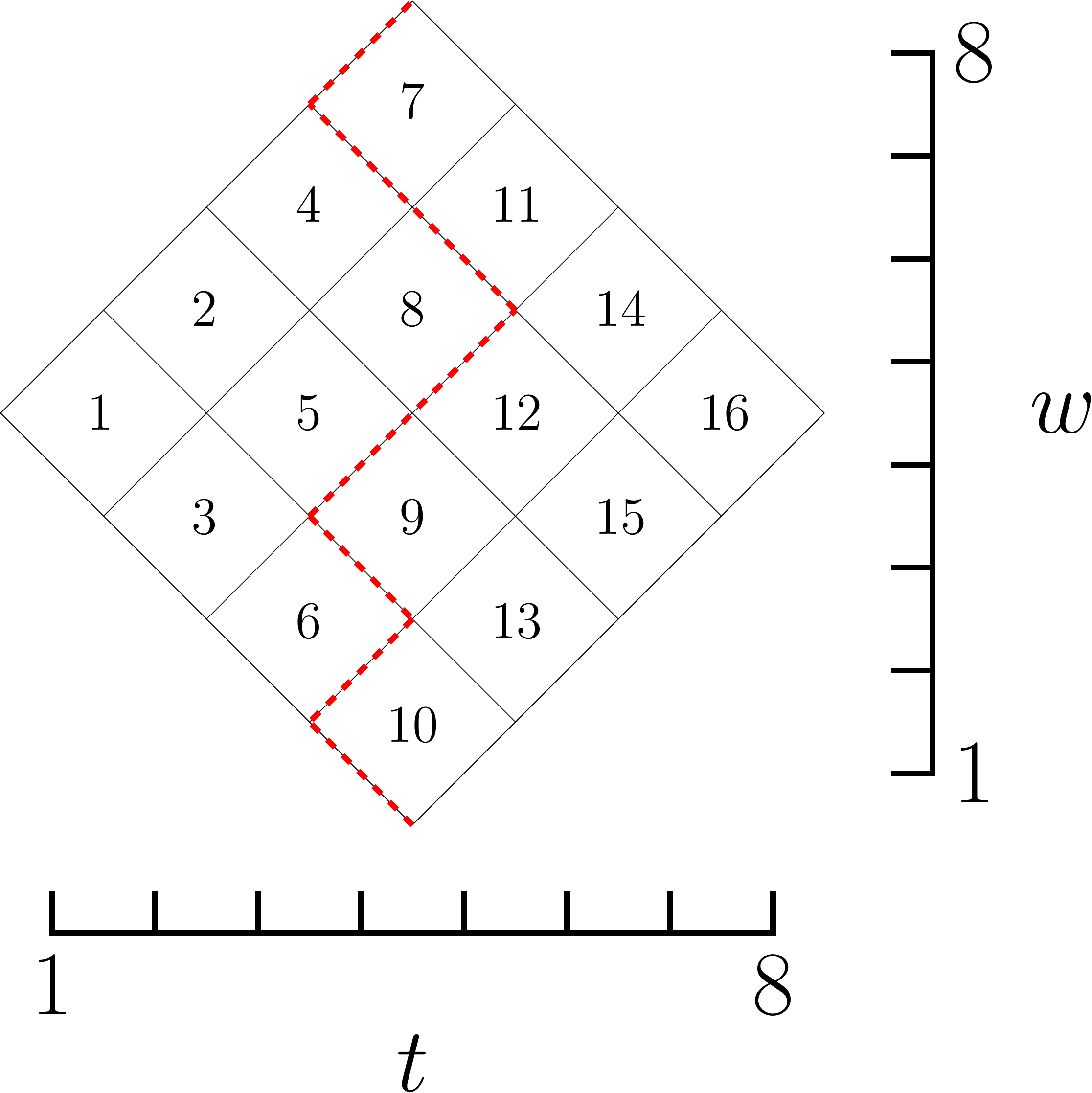} \label{fig:QCJanzing2} } 
   \subfigure[]{\hspace*{0cm}\includegraphics[width=0.29\columnwidth,trim=0 -5cm 0cm 0]{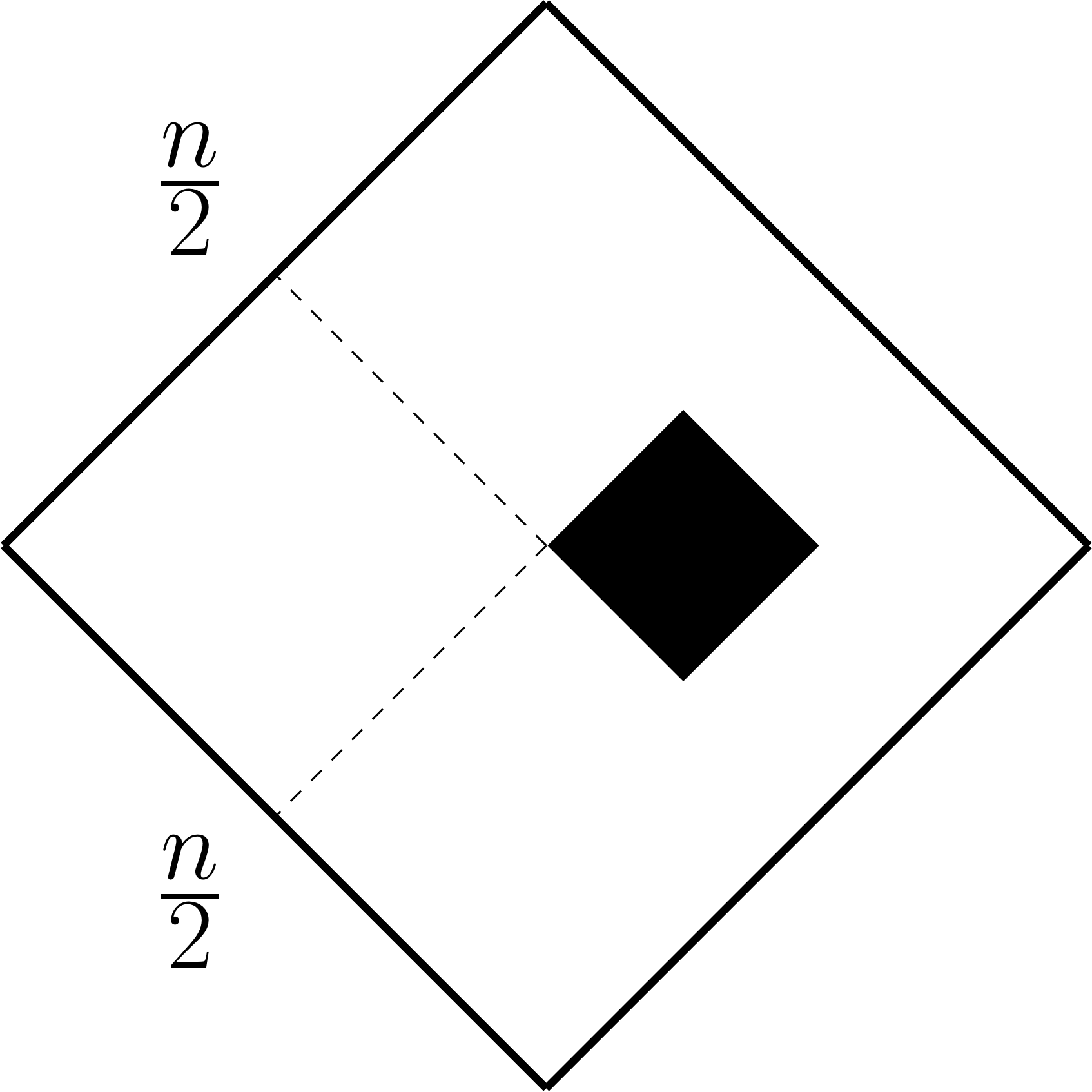} \label{fig:QCJanzing3} }
      \caption{(a) A $2n = 8$ qubit quantum circuit, where each grey square ($n^2 = 16$ in total) corresponds to a 2-qubit gate. (b) An equivalent representation of the quantum circuit in (a) in terms of a rotated grid.  The red dashed line corresponds to an allowed string configuration for the particles.  (c) The circuit is constrained such that the majority of the gates are identity except in a $k \times k$ subgrid (shown in black), located such that its left vertex is at the center of the rotated grid. A successful computation requires the $t$ position of the $2k$ particles with $w$ positions that cross the interaction region, to lie to the right of the interaction region. See also Fig.~1 in \cite{Gosset:2014rp}.}
   \label{fig:example}
\end{figure}

The circuit is mapped to a Hamiltonian $H(\lambda)$, with $\lambda \in [0,1]$.  The Hamiltonian describes the evolution of particles that live on the edges of the rotated $n \times n$ grid.  The positions of the particles are given in terms of the coordinates $(t,w) \in \left\{1, \dots, 2 n \right\}^2$ as shown in Fig.~\ref{fig:QCJanzing2}.  Each particle has two internal degrees of freedom in order to encode the qubits of the circuit.  Let $a_{t,s}[w]$ denote the annihilation operator which annihilates a particle with internal state $s \in \left\{0,1 \right\}$ on the edge $(t,w)$.  The number operator is defined as $n_{t,s}[w] = a^{\dagger}_{t,s}[w] a_{t,s}[w]$, which counts the number of particles (which will be either 0 or 1) at position $(t,w)$ with state $s$.  Let $\mathbf{n}_t [w] = n_{t,0}[w] + n_{t,1}[w]$.

We focus on configurations of particles that form connected segments starting at the top and ending at the bottom [an example is shown in Fig.~\ref{fig:QCJanzing2}], referred to as consistent connected string configurations.  For a fixed $w$ (i.e., a horizontal line on the rotated grid), there is only one occupied edge.  We can describe such configurations in terms of $2n$ bits, denoted by $z$.  Specifically, let the bit value $0$ correspond to an edge going down and left and $1$ correspond to an edge going down and right.  The Hamming weight of such configurations must be $n$, since they start and end in the middle of the grid and so must go left and right an equal number of times.

We are now ready to describe the Hamiltonian:
\beq
H(\lambda) = H_{\mathrm{string}} + H_{\mathrm{circuit}}(\lambda) + H_{\mathrm{input}}\ .
\label{eq:H-spacetime}
\eeq

$\bullet$ $H_{\mathrm{input}}$:  This term ensures that the ground state has the internal state of all particles set to $s = 0$ when the string lies on the left-hand side of the grid by energetically penalizing all states (on the left-hand side) with $s = 1$.  It is given by:
\beq
H_{\mathrm{input}} = \sum_{w = 1}^{2n } \sum_{t \leq n} n_{t,1}[w]\ .
\eeq

$\bullet$ $H_{\mathrm{string}}$: This term ensures that the ground state is in the subspace of connected strings.  Consider a single vertex $v$ in the grid with incident edges labeled by $(t,w), (t+1,w), (t,w+1), (t+1,w+1)$.  We can associate a Hamiltonian $H^v_{\mathrm{string}}$ to each vertex,
\begin{eqnarray} \label{eqt:Hstring}
H^v_{\mathrm{string}} &=& \mathbf n_t[w] + \mathbf n_{t+1}[w] + \mathbf n_t [w+1] + \mathbf n_{t+1}[w+1] \nonumber \\
&& -2 \left( \mathbf n_t[w] + \mathbf n_{t+1}[w] \right) \left( \mathbf n_t[w+1]+\mathbf n_{t+1}[w+1] \right) \nonumber \\
\end{eqnarray}
(for vertices at the boundary of the grid with two or three incident edges, the definition of $H^v_{\mathrm{string}}$ needs to be modified accordingly) such that $H_{\mathrm{string}} = \sum_v H^v_{\mathrm{string}}$.   For connected string configurations, the energy due to this Hamiltonian is zero, while disconnected strings with $L$ string segments have a higher energy $2 L - 2$.

$\bullet$ $H_{\mathrm{circuit}}(\lambda) = \sum_p H_{\mathrm{gate}}^p(\lambda)+ \sqrt{1- \lambda^2} H_{\mathrm{init}}$: 
Define for each plaquette $p$ with borders given by the edges $\{(t,w)$, $(t+1,w)$, $(t,w+1)$, $(t+1,w+1)\}$
\beq
H^p_{\mathrm{gate}}(\lambda) = \mathbf n_t[w] \mathbf n_t[w+1] + \mathbf n_{t+1}[w]\mathbf n_{t+1} [w+1] + \lambda H_{\mathrm{prop}}^p
\eeq
where 
\begin{eqnarray}
H_{\mathrm{prop}}^p &=& - \sum_{\alpha, \beta, \gamma, \delta} \left( \bra{\beta, \delta} U_p \ket{\alpha, \gamma} a_{t+1,\beta}^\dagger a_{t, \alpha} [w] \right. \nonumber \\
&& \left. \times a_{t+1,\delta}^\dagger [w+1] a_{t, \gamma} [w+1] \right) + \mathrm{h.c.}
\end{eqnarray}
The term $H_{\mathrm{prop}}^p$ allows for a pair of particles located on the left (right) edges of a plaquette to hop together such that they both are located on the right (left) edges of a plaquette, with their internal states changed according to $U_p$ ($U_p^{\dagger})$.  Note that this move preserves the connectedness of the string.  Furthermore, the term $\sum_p H_{\mathrm{gate}}^p(0)$ is minimized by a configuration lying either entirely on the left border (corresponding to the bit string $z = 0^n 1^n$) or entirely on the right border ($z = 1^n 0^n$), which in conjunction with  $H_{\mathrm{init}}$, given by
\beq
H_{\mathrm{init}} = \mathbf n_{n+1}[w=1] + \mathbf n_{n+1}[w=2n] \ ,
 \eeq
ensures that the ground state of $H(0)$ is such that all particles lie along the left boundary of the grid.  Including the effect of  $H_{\mathrm{circuit}}(0)$ and $H_{\mathrm{input}}$, the ground state of $H(0)$ is given by $\ket{0^{2n}}\ket{0^n1^n}$ with eigenvalue $1$.  This is an easily prepared ground state.

It can be shown that the ground state of $H(\lambda)$ [Eq.~\eqref{eq:H-spacetime}] along $\lambda \in [0,1]$ is unique and the energy gap above the ground state is lower bounded by $1/\poly(n)$ for all $\lambda \in [0,1]$ [Theorem 1 in Ref.~\cite{Gosset:2014rp}].  To measure the output of the quantum circuit, we measure the $t$ positions of the $2k$ particles for the $w$ values that cross the interaction region (recall that there will always be one particle per horizontal $w$ line).  If we find that all $2k$ particles lie to the right of the interaction region, then their internal states must encode the output of the quantum circuit.  For the choice $k = \sqrt{n}/16$, this occurs with a probability lower bounded by a positive constant.  Together, these properties allow for an efficient (up to polynomial overhead) simulation of the quantum circuit using the adiabatic algorithm generated by $H(\lambda)$.  

However, this implementation requires 4-body interactions [see for example the product term in Eq.~\eqref{eqt:Hstring}].  In~\cite{Lloyd:2016}, improvements to the above construction were presented with only $2$-local interactions using a first order perturbation gadget and a quadratic increase in the number of qubits from the original quantum circuit. The use of only first order perturbation theory is particularly significant, since effective interactions obtained in $k$-th order degenerate perturbation theory with perturbative coupling $g$ and gap $\Delta$ of the unperturbed Hamiltonian scale in strength as $g(g/\Delta)^{k-1}$, leading to a correspondingly small gap of the effective Hamiltonian. In addition, multiple uses of higher-order perturbation theory can increase qubit overhead and complexity.

\subsection{Universal AQC in 1D with $9$-state particles}
\label{sec:UAQC-1D}

The constructions of universal AQC we have reviewed so far are all spatially two-dimensional (2D). It was unclear for some time whether universal AQC is possible in 1D, with some suggestive evidence to the contrary, such as the impressive success of density matrix renormalization group (DMRG) techniques in calculating ground state energies and other properties of a variety of 1D quantum systems \cite{Schollwock:2005qd}. Moreover, classical 1D systems are generally ``easy"; e.g., a 1D restriction of MAX-2-SAT with $p$-state variables can be solved by dynamic programming and hence is in the complexity class P. In addition, the area law implies that 1D systems with a constant spectral gap can be efficiently simulated classically \cite{PhysRevLett.103.050502}. All this implies that adiabatic evolution with 1D Hamiltonians is not useful for universal QC unless certain conditions are met, in particular  a spectral gap that tends to zero. 

This was accomplished in \cite{Aharonov:2009fi}, who proved that it is possible to perform universal AQC using a 1D quantum system of $9$-state particles. 
The striking qualitative difference between the quantum and the
classical 1D versions of the same problem seems surprising. However, 
the $k$-local Hamiltonian problem allows for the encoding of an extra dimension (time), by making the ground state a superposition of states corresponding to different times. This means that the correct analogue of the quantum 1D local Hamiltonian problem is 2D classical MAX-k-SAT, which is NP-complete. 

The proof presented in \cite{Aharonov:2009fi} builds heavily on the history state construction reviewed in Sec.~\ref{sec:history-state}. However, there are a couple of important differences. 
As in the history state construction, the starting point is a quantum circuit $U_x$ acting on $n$ qubits (where 
$x$ is the classical input to the function implemented by the circuit in the universal AQC case). A 1D
$p$-state Hamiltonian is designed which will verify correct propagation according to this circuit. Then  this is used as the final Hamiltonian for the adiabatic evolution.
The problem with directly realizing this in the 1D case is that only the particles nearest to the clock would be able to take advantage of it in order to check correct propagation in time. To overcome this, the circuit $U_x$ is first modified into a new circuit $\tilde{U}_x$ with a distributed clock. The history state construction 
relies on the ability to copy qubits from one column to the next in 
order to move to the next block of gates in the computation, so a new strategy is needed
in 1D. For the modified circuit $\tilde{U}_x$, the qubits are instead placed in a block of
$n$ adjacent particles. One set of gates is performed, and then all of the qubits are moved over
$n$ places to advance time in the original circuit ${U}_x$. More states per particle are needed to accomplish this than in the 2D case. The second main new idea that is needed is related to ensuring that the state of the system had a valid structure. In the 2D case local constraints were used to check that there are no two qubit states in adjacent columns. However, using only local constraints, there is no way to check that there
are exactly $n$ qubit data states in an unknown location in a 1D system, since there are only a constant number of local rules available, which are therefore unable to
count to an arbitrarily large $n$. Instead, it is ensured that, under
the transition rules of the system, any invalid configurations will evolve in polynomial
time into a configuration which can be detected as illegal by local rules. Thus, for every state which is not a valid history state, either the propagation is wrong, which implies
an energy penalty due to the propagation Hamiltonian, or the state evolves to an illegal configuration which is locally detectable, which implies an energy penalty due to the
local check of illegal configurations. For additional technical details required to complete the proof see \cite{Aharonov:2009fi}. A $20$-state translation-invariant modification of the construction from \cite{Aharonov:2009fi} for universal AQC in 1D was given in \cite{Nagaj:2008fv}, improving on a $56$-state construction by \cite{Janzing:2008dz}.

\subsection{Adiabatic gap amplification}
\label{sec:gap-amp}

In all universality constructions the run time of the adiabatic simulation of a quantum circuit depends on the inverse minimum gap of the simulating Hamiltonian. Therefore it is of interest to develop a general technique for amplifying this gap, as was done in \cite{Somma:2013to}. 

Consider a Hamiltonian $H$ with ground state $\ket{\phi}$. The goal is to construct a new Hamiltonian $H'$ that has $\ket{\phi}$ as an eigenstate (not necessarily the ground state) but with a larger spectral gap. 
A quadratic spectral gap amplification is possible when $H$
is frustration-free [see also \cite{Bravyi:2009sp}]:
\begin{definition}
\label{def:FF}
(frustration freeness) A Hamiltonian $H \in \mathbb{C}^N\times\mathbb{C}^N$ is frustration free if it can be written as a sum over positive semi-definite operators: $H = \sum_{k=1}^L a_k \Pi_k$, with $a_k\in[0,1]$ and $L = {\rm polylog}(N)$. Further, if $\ket{\phi}$ is the ground state of $H$ then it is a ground state (i.e., zero eigenvector) of every term in the decomposition of $H$, i.e., $\Pi_k\ket{\phi}=0$ $\forall k$.
\end{definition}
\cite{Somma:2013to} took $\Pi_k$ as projectors.
The quadratic amplification is optimal for frustration-free Hamiltonians in a suitable black-box model, and no spectral gap amplification is
possible, in general, if the frustration-free property is removed. 
An important caveat is that the construction replaces ground state evolution by evolution of a state that lies in the middle of the spectrum; thus it does not fit the strict definition of AQC (Def.~\ref{def:AQC}). We will have another occasion to relax the definition in the same sense, in Sec.~\ref{sec:stoq-QMA-comp}.

We now review the construction in \cite{Somma:2013to} in some detail. To place it in context, note that the universality results we reviewed thus far can be summarized as follows:
Any quantum circuit specified by unitary gates $U_1, \dots, U_Q$ can be simulated by an adiabatic quantum evolution involving frustration-free Hamiltonians: $H(s) =  \sum^L_{k=1} a_k(s) \Pi_k(s)$. 
The ground state of the final Hamiltonian $H(1)$ has large overlap with the output state of the quantum circuit. Moreover, $L$ is polynomial in $Q$, and $\Pi_k$ denotes nearest-neighbor, two-body interactions between spins of corresponding many-body systems in one- or two-dimensional lattices. The inverse minimum gap of $H(s)$ is polynomial in $Q$, and hence so is the duration of the adiabatic simulation.

Now, consider a frustration-free Hamiltonian 
\beq
H(s)=\sum_{k=1}^L a_k \Pi_k(s)\ ,
\eeq 
where each $\Pi_k(s)$ is a projector for all $s\in[0,1]$, and is a local operator. Denote the eigenvalues of this Hamiltonian by $\{\lambda_j\}$, where $\lambda_1=0$ is the ground state energy.
Then, take the Hamiltonian 
\beq
\bar{H}(s)=\sum_{k=1}^L \sqrt{a_k} \Pi_k(s) \otimes (\ketbra{k}{0} +  \ketbra{0}{k})\ ,
\eeq 
where $\ket{0}$ and $\ket{k}$ are ancilla registers defined over one and $\log_2(L)$ qubits, respectively. It can be shown that $\bar{H}(s)$ 
has the desired properties, i.e., if $\ket{\psi(s)}$ was the ground state of $H$, then $\ket{\psi(s)}\ket{10\cdots 0}$ is a (degenerate) zero-eigenvalue eigenstate of $\bar{H}$ and the eigenvalues of $\bar{H}(s)$ are $\{\pm \sqrt{\lambda_j}\}$ [the proof is given in Appendix B of \cite{Somma:2013to}]. Thus, the gap has been  quadratically amplified, and one can evolve with $\bar{H}$ to transform eigenstates at $s=0$ to eigenstates at $s = 1$ and simulate the original quantum circuit with a quadratic speedup over the simulation involving $H$. 

In general $\bar{H}(s)$ will be $\log_2(L)$-local due to the appearance of $\ket{k}$. To avoid these many-body interactions one can represent $\ket{k}$ using a unary encoding, i.e., $\ket{k} \mapsto \ket{0\dots 010\dots 0}$ (with $1$ at the $k$-th position). In this single-particle subspace the new Hamiltonian becomes
\beq
\bar{H}(s)=\sum_{k=1}^L \sqrt{a_k} \Pi_k(s) \otimes (\sigma_{k}^+\sigma_{0}^- +  \sigma_{k}^-\sigma_{0}^+)\ ,
\eeq 
where $\sigma^\pm = (\sigma^x\pm i\sigma^y)/2$ are Pauli raising and lowering operators. Note that since each $\Pi_k(s)$ interacts with the same qubit $0$ of the new register, if the original $H$ was geometrically local, then $\bar{H}$ is not, i.e., it has a central spin geometry. 

One more issue that needs to be dealt with is the degeneracy of the zero eigenvalue. To remove this degeneracy from contributions within the single-particle subspace one can add a penalty term $\frac{1}{4}\sqrt{\Delta}(\ident+\sigma^z_0)$ to $\bar{H}(s)$, which penalizes all states with qubit $0$ in $\ket{0}$; the relevant spectral gap in the single-particle subspace is then still of order $\sqrt{\Delta}$. To remove additional degeneracy from the many-particle subspaces one can add penalties for states that belong to such subspaces. Adding $Z= (L-2)\ident-\sum_{k=0}^L \sigma_k^z$ achieves this since it acts as a penalty that grows with the Hamming weight $a$ of states in the $a$-particle subspace. Thus
\begin{align}
H'(s)=\frac{1}{L^{1/d}}&\left[ \sum_{k=1}^L \sqrt{a_k} \Pi_k(s) \otimes (\ketbra{k}{0} +  \ketbra{0}{k})  \right. \notag \\
&\left. + \frac{1}{4}\sqrt{\Delta}(\ident+\sigma^z_0)\right] + Z \ ,
\label{eq:H'-GAP}
\end{align}
has $\ket{\psi_0}\ket{10\cdots 0}$ as a unique eigenstate of eigenvalue $0$, and all other eigenvalues are at distance at least $\sqrt{\Delta}/L^{1/d}$ if $d\geq 2$.\footnote{The factor $L^{1/d}$ in Eq.~\eqref{eq:H'-GAP} is introduced so that 
the eigenvalues coming from the many-particle subspaces will not mix with the eigenvalues of the single-particle subspace.} This is the desired quadratic gap amplification result.

How far can gap amplification methods go? It was shown in \cite{Schaller:08} that for the one-dimensional transverse-field quantum Ising model, and for the preparation of cluster states \cite{Raussendorf:01}, it is possible to use a series of straight-line interpolations in order to generate a schedule along which the gap is always greater than a constant independent of the system size, thus avoiding the quantum phase transition. However, there exists an efficient method to compute
the ground state expectation values of local operators 
of 2D lattice Hamiltonians undergoing exact adiabatic evolution, and this implies that adiabatic quantum algorithms based on such local Hamiltonians, with unique ground states, can be simulated efficiently if the spectral gap does not scale with the system size \cite{osborne2007simulating}.


\section{Hamiltonian quantum complexity theory and universal AQC}
\label{sec:QMAC}

In this section we review Hamiltonian quantum complexity theory from the perspective of QMA completeness. This theory naturally incorporates decision problems of the type that motivate AQC. Essentially, it concerns a problem involving the ground state of a local Hamiltonian, whose ground state energy is promised to either be below a threshold $a$ or above another threshold $b>a$,
and where $b-a$ is polynomially small in the system size. In some cases this problem is easy, and in other cases it turns out to be so hard that we do not hope to solve it efficiently even on a quantum computer. Characterizing which types of local Hamiltonians fall into the latter category is the subject of QMA-completeness.
  
Hamiltonian quantum complexity theory is an extremely rich subject that is rapidly advancing and has already been reviewed a number of times, so we will only touch upon it and highlight some aspects that are relevant to AQC. Perhaps the most direct connection is the fact that $2$-local Hamiltonians of a form that naturally appears in AQC, are QMA-complete. Additionally, some of the technical tools that played an important role in QMA-completeness locality reductions, such as perturbative gadgets, have also found great use in proofs of the universality of AQC with different Hamiltonians.

The reviews \cite{Aharonov:2002nr,Gharibian:thesis,Gharibian:2015hl}, are excellent resources for additional perspectives and details on Hamiltonian quantum complexity theory.

\subsection{Background}
\label{sec:SAT-background}

\subsubsection{Boolean Satisfiability Problem: $k$-SAT}
%
Consider a Boolean formula $\Phi$ that depends on $n$ literals $x_i\in\{0,1\}$ (with $0$ and $1$ representing False and True, respectively) or their negations.  The problem is to decide whether there exists an assignment of values to the literals that satisfies the Boolean formula, i.e., such that $\Phi = 1$.  If there exists such an assignment then the formula is satisfiable, otherwise it is unsatisfiable. 

The Boolean formula is typically written in conjunctive normal form: it is written in terms of a conjunction (AND - $\land$) of $r$ clauses, where each clause contains the disjunction (OR - $\lor$) of $k$ literals (variables) or their negation (NOT - $\neg$).  A literal and its negation are often referred to as positive and negative literals.  The Boolean formula is written as:
\beq
\Phi = C_1 \land C_2 \land \dots \land C_r
\eeq   
where $C_i = x_{i_1} \lor x_{i_2} \dots \lor x_{i_k}$ and $x_{i_j}$ is the $j$-th positive or negative literal in the $i$-th clause.  The question of Boolean satisfiability, or $k$-SAT, is whether there exists a choice $X = (x_1, \dots, x_n)$ such that $\Phi(X) = 1$.  Note that it only requires $O(k r)$ steps
 to check whether $X$ is a satisfying assignment, yet there are $2^n$ possible choices for $X$.

For $k=3$, the Boolean satisfiability problem, called 3-SAT, is NP-complete.  Let us explain what this means.
%
\subsubsection{NP, NP-complete, and NP-hard}
%
Informally, problems in NP are those whose verification can be done efficiently (e.g., checking whether a Boolean formula is satisfied). An important conjecture, called the Exponential Time Hypothesis \cite{Impagliazzo2001367}, states that there are problems in NP that take exponentially long to solve.

Formally, a decision problem $Q$ is in NP if and only if there is an efficient algorithm $V$, called the verifier, such that for all inputs $\eta$ (e.g., in the case of SAT, this would be the clauses) of the problem:
\begin{itemize}
\item if $Q(\eta) = 1$, then there exists a witness $X$ such that $V(\eta,X) = 1$.
\item if $Q(\eta) = 0$, then for all witnesses $X$ we have $V(\eta, X) = 0$.
\end{itemize}
In both cases, we typically take $|X| = \poly(|\eta|)$, where $|\eta|$ is the number of bits in the binary string associated with the input $\eta$.  The verifier is efficient in the sense that its cost scales as $\poly(|X|)$. In SAT the witness $X$ would be our test assignment.

A decision problem $Q$ is NP-complete if:
\begin{itemize}
\item $Q$ is in NP
\item Every problem in NP is reducible to $Q$ in polynomial time.
\end{itemize}
Here reducibility means that given a problem $A$ in NP and a problem $B$ that is NP-complete, $A$ can be solved using a hypothetical polynomial-time algorithm that solves for $B$.
 A commonly used reduction is the polynomial-time many-to-one reduction \cite{Karp:21-problems}, whereby the inputs of $A$ are mapped into the inputs to $B$ such that the output of $B$ matches the output of $A$.  The hypothetical algorithm then solves $B$ to get the answer to $A$.

A decision problem $Q$ is NP-hard if every problem in NP is reducible to $Q$ in polynomial time. (Note that unlike the NP-complete case, $Q$ does not need to be in NP). 
Clearly, $\text{NP-complete} \subseteq \text{NP-hard}$.

\subsubsection{The $k$-local Hamiltonian Problem}

The history state construction of Sec.~\ref{sec:history-state} relies on a $5$-local Hamiltonian. Such a Hamiltonian belongs to an important class of decision problems known as the $k$-local Hamiltonian Problem, of which a complete complexity classification was given in~\cite{Cubitt:2016vl} subject to restrictions on the set of local terms from which the Hamiltonian can be composed [see also \cite{Bravyi:2014bf}]. Recall that a $k$-local Hamiltonian is a Hermitian matrix that acts non-trivially on at most $k$  $p$-state particles.

The $k$-local Hamiltonian Problem is defined on $n$ qubits, with the following input:
\begin{itemize}
\item A $k$-local Hamiltonian $H = \sum_{i=1}^r H_i$ with $r = \poly(n)$. Each $H_i$ is $k$-local and satisfies $\norm{H_i} = \text{poly}(n)$ and its non-zero entries are specified by $\text{poly}(n)$ bits.
\item Two real numbers $a$ and $b$ specified with $\poly(n)$ bits of precision, such that
\beq
b-a > \frac{1}{\poly(n)}\ .
\label{eq:k-local-gap}
\eeq
\end{itemize}
The output ($0$ or $1$) answers the question: Is the smallest eigenvalue of $H $ smaller than $a$ (output is $1$), or are all eigenvalues larger than $b$ (output is $0$)?  We are promised that the ground state eigenvalue cannot be between $a$ and $b$.\footnote{The quantity $b-a$ is sometimes called the ``promise gap" and is distinct from the spectral gap.} 

We may map 3-SAT to the 3-local Hamiltonian Problem as follows.
For every clause $C_i$ (which involves three literals), we can define a $3$-local projector ${H}_i$ onto all the unsatisfying assignments of $C_i$.  
Because ${H}_i$ is a projector, it has eigenvalues $0$ and $1$, where the $0$ eigenvalue is associated with satisfying assignments and the $1$ eigenvalue with unsatisfying assignments. Therefore:
\beq
H \ket{X} = \sum_{i=1}^r H_i \ket{X} = q \ket{X}
\eeq
where $q$ is the number of unsatisfied assignments by $X$.  Thus 3-SAT is equivalent to the following 3-local Hamiltonian problem: is the smallest eigenvalue of $H$ zero (the 3-SAT problem is satisfiable) or is it at least $1$ (the 3-SAT problem is unsatisfiable)?
%

\subsubsection{Motivation for Adiabatic Quantum Computing}
%
Adiabatic evolution seems well-suited to tackling the $k$-local Hamiltonian problem.  By initializing an $n$-qubit system in an easily prepared ground state, we can (in principle) evolve the system with a time-dependent Hamiltonian whose end point is the $k$-local Hamiltonian.  If the evolution is adiabatic, then we are guaranteed to be in the ground state of the $k$-local Hamiltonian with high probability.  By measuring the state of the system, we can determine the energy eigenvalue of the state (which hopefully is the ground state energy) and hence determine the answer to an NP-complete problem such as 3-SAT. This motivated early work on the quantum adiabatic algorithm \cite{Farhi:01}. 

Another possibility is to try to use AQC as the verifier. However, the quantum algorithm only gives us the answer probabilistically, so we must first define a probabilistic analog of NP and then a quantum version. These new complexity classes are MA and QMA \cite{Kitaev:book}.
\subsection{MA and QMA}
\label{sec:MAandQMA}

Informally, MA can be thought of as a probabilistic analog of NP, allowing for two-sided
errors. Formally,  a decision problem $Q$ is in MA iff there is an efficient probabilistic verifier $V$ such that for all inputs $\eta$ of the problem:
\begin{itemize}
\item if $Q(\eta) = 1$, then there exists a witness $X$ such that $\Pr(V(\eta,X) = 1) \geq \frac{2}{3}$  (completeness).
\item if $Q(\eta) = 0$, then for all witnesses $X$ we have $\Pr(V(\eta,X) = 1) \leq \frac{1}{3}$  (soundness).
\end{itemize}
Again we take $|X| = \poly(|\eta|)$.  MA is typically viewed as an interaction between two parties, Merlin and Arthur.  Merlin provides Arthur with a witness $X$, on which Arthur runs $V$.  If $Q(\eta) = 0$, Merlin should never be able to fool Arthur with a witness $X$ into believing that $Q(\eta) = 1$ with probability $> 1/3$.

Note that there is nothing special about the probabilities $(2/3,1/3)$.  We can generalize our description to MA($c,s$):
\begin{claim}
\label{claim-ampl}
MA($c, c-1/|\eta|^g) \subseteq$ MA$(2/3,1/3) = $MA($1-e^{-|\eta|^g},e^{-|\eta|^g})$, where $g$ is a constant and $c>0$ and $c-1/|\eta|^g < 1$.  
\end{claim}
The proof of this ``amplification lemma" [see, e.g., \cite{Marriott:2005vn,goldreich2008computational,Nagaj:09}] is interesting since it invokes the Chernoff bound, a widely used tool. We thus present it in Appendix~\ref{app:amplification-lemma} for pedagogical interest. 

The complexity class QMA can be viewed as the quantum analogue of MA. Thus, QMA
is informally the class of problems that can be efficiently checked on a quantum computer given a ``witness" quantum state related to the answer to the problem. 
Formally, define a quantum verifier $V$ (a quantum circuit) that takes $\eta$ and a quantum witness state $\ket{X}\in (\mathbb{C}^2)^{\otimes \text{poly}(|\eta|)}$ as inputs and probabilistically outputs a binary number. 
The decision problem $Q$ is said to be in QMA if and only if there exists an efficient (polynomial time) $V$ for all inputs $\eta$ of the problem that satisfies:
\begin{itemize}
\item if $Q({\eta}) = 1$, then there exists a witness $\ket{X}$ such that $\Pr(V(\eta,\ket{X}) = 1) \geq \frac{2}{3}$ (completeness).
\item if $Q({\eta}) = 0$, then for all witnesses $\ket{X}$ we have $\Pr(V(\eta,\ket{X}) = 1) \leq \frac{1}{3}$ (soundness).
\end{itemize}
The amplification lemma applies here as well. The definition of QMA also allows $V$ to have $\text{poly}(|\eta|)$ ancilla qubits each  initialized in the $\ket{0}$ state  \cite{Gharibian:2015hl}.

One can also define the class QCMA, which is similar to QMA except that $\ket{X}$ is a classical state \cite{Aharonov:2002nr}. Since the quantum verifier can force Merlin to send him a classical witness by measuring the witness before applying the quantum algorithm, we have: 
MA $\subseteq$ QCMA $\subseteq$ QMA.

\subsection{The general relation between QMA completeness and universal AQC}

The class of efficiently solvable problems on a quantum computer is bounded error quantum polynomial time (BQP) \cite{Bernstein:93}, which consists of the class of decision problems solvable by a uniform family of polynomial-size quantum circuits with error probability bounded
below $1/2$. Because of the polynomial equivalence between AQC and the circuit model, BQP is also the class of efficiently solvable problems on a universal adiabatic quantum computer. Its classical analog is the class bounded-error probabilistic polynomial time (BPP), and as expected BPP $\subseteq$ BQP \cite{Bernstein:1997}. In addition, BQP $\subseteq$ QCMA \cite{Aharonov:2002nr}. Another interesting characterization is that  BQP$=$QMA$_{\log}$, where QMA$_{\log}$ is the same as QMA except that the quantum proof has $O(\log |\eta|)$ qubits instead of $\text{poly}(|\eta|)$ \cite{Marriott:2005vn}.

This motivates the study of QMA, and in particular QMA completeness, as a tool for understanding universality. Indeed, it is often the case that whenever adiabatic universality can be proven for some class of Hamiltonians, then the local Hamiltonian problem with (roughly) the same class can be shown to be QMA-complete and vice versa.
Note, however, that there is no formal implication from either of those problems to the other \cite{Aharonov:2009fi}. On the one hand, proving QMA-completeness is in general substantially harder than achieving universal AQC, where we can choose the initial state to be any easily prepared state that will help us solve the problem, so we can choose to work in any convenient subspace that is invariant under the Hamiltonian. Indeed, in the history-state construction, we introduce penalty terms to guard against illegal clock states [recall Eq.~\eqref{eqt:Hc}]. For QMA, the states we work with are chosen adversarially from the \emph{full} Hilbert space, and we must be able to check, using only local Hamiltonian terms, that they are of the correct (clock-state) form. On the other hand, proving adiabatic universality involves analyzing the spectral gap of the continuous sequence of Hamiltonians over the \emph{entire} duration of the computation, whereas QMA-completeness proofs are only concerned with one Hamiltonian.

\subsection{QMA-completeness of the $k$-local Hamiltonian problem and universal AQC} 
\label{sec:BiamonteLove}
To prove that a promise problem is QMA-complete, one needs to prove
that it is contained in QMA and that it is QMA-hard. 
The $k$-local Hamiltonian Problem belongs to QMA for any constant $k$, and in fact even for $k = O \left( \log n \right)$ \cite{Kitaev:book}. For pedagogical proofs see \cite{Aharonov:2002nr,Gharibian:thesis,Gharibian:2015hl}. 

The first example of a QMA-hard problem was the $k$-local Hamiltonian problem for $k\geq 5$ \cite{Kitaev:book}, so that in particular the $5$-local Hamiltonian problem is QMA-complete. This was reduced to $3$-local \cite{Kempe:2003,Nagaj:06} and then to $2$-local \cite{KempeGadget}. Note that the $1$-local Hamiltonian problem is in the complexity class P, since one can simply optimize for each $1$-local term independently. Various simplifications of QMA-completeness for the $2$-local case followed. In order to describe these, we first need to define a class of Hamiltonians:
\begin{align}
\label{eq:H-all-models}
H_1 &= \sum_{(i,j)\in \mathcal{E}} J_{ij}^x X_i X_j +  J_{ij}^y Y_i Y_j + J_{ij}^z Z_i Z_j \notag 
\\
& \qquad + \sum_{i \in \mathcal{V}} h^x_i X_i + h^y_i Y_i + h^z_i Z_i \ ,
\end{align}
where $\mathcal{V}$ and $\mathcal{E}$ are the vertex and edge sets of a graph $\mathcal{G}=(\mathcal{V},\mathcal{E})$, and all local fields $\{h^\alpha_i\}$ and couplings $\{J_{ij}^\alpha\}$ ($\alpha\in\{x,y,z\}$) are real. The Heisenberg model corresponds to $J_{ij}^x=J_{ij}^y=J_{ij}^z$, the $XY$ model to $J_{ij}^x=J_{ij}^y$ and $J_{ij}^z=0$, and the Ising model to $J_{ij}^x=J_{ij}^y=0$. When $J_{ij}^\alpha<0$ ($>0$) the interaction between qubits $i$ and $j$ is ferromagnetic (antiferromagnetic). When we write ``fully" below we mean that all interactions have the same sign. Unless explicitly mentioned otherwise we assume that the local fields are all zero.

Most of the simplifications of QMA-completeness are special cases of Eq.~\eqref{eq:H-all-models}:
\begin{itemize}
\item Geometrical locality: nearest-neighbor interactions with $\mathcal{G}$ being a 2D square lattice \cite{OliveiraGadget} or a triangular lattice \cite{Piddock:2015aa}.
\item Simple interactions in 2D: $ZZXX$ and $ZX$ model \cite{Biamonte:07} [defined in Eqs.~\eqref{eqt:ZZXX} and \eqref{eqt:ZX} below], fully ferromagnetic and fully antiferromagnetic Heisenberg model with local fields \cite{Schuch:2009qq}, antiferromagnetic Heisenberg and $XY$ models without local fields \cite{Piddock:2015aa}.
\end{itemize}
Some of the simplifications of QMA completeness use other types of Hamiltonians:
\begin{itemize}
\item Interacting fermions in 2D and the space-time construction \cite{Breuckmann:2014}.
\item Multi-state particles in 1D \cite{Aharonov:2009fi,Nagaj:2008rc,Hallgren:2013dn}. 
\item Non-translationally invariant 1D systems (all two-particle terms identical but position-dependent one-particle terms) \cite{Kay:2008bs}.
\item Translationally invariant 1D systems for which finding the ground state energy is complete for QMA$_{\text{EXP}}$\footnote{QMA$_{\text{EXP}}$ is the same as QMA but with exponential size (in the input)
witness and verification circuit, whereas both are polynomial for QMA. }
 \cite{Gottesman:13}.  
\end{itemize}

The 1D case is interesting since, as remarked in Sec.~\ref{sec:UAQC-1D}, the 1D restriction of MAX-2-SAT with $p$-state variables is in P, yet \cite{Aharonov:2009fi} showed that for $12$-state particles, the problem of approximating the ground state energy of a 1D system is QMA-complete. This result was improved to $11$-state particles in \cite{Nagaj:2008rc}, and then to $8$-state particles in \cite{Hallgren:2013dn}, who also pointed out a small error in \cite{Aharonov:2009fi} that could be fixed by using $13$-state particles. Whether and at which point a further Hilbert space dimensionality reduction becomes impossible remains an interesting open problem. 

The reduction from $5$-local to $2$-local  is done using perturbative gadgets \cite{KempeGadget,OliveiraGadget,Biamonte:07,Bravyi:2008fk,Jordan:08,PhysRevA.91.012315}.  The goal of the gadget is to approximate some target Hamiltonian $H^{\mathrm{T}}$ of $n$ qubits (e.g., the $5$-local Hamiltonian from the history state construction at any time $s$) by a gadget Hamiltonian $H^{\mathrm{G}}$ acting on the same $n$ qubits as well as an additional $\poly(n)$ ancilla qubits.  The gadget Hamiltonian is typically written as
\beq
H^{\mathrm{G}} = H^{\mathrm{A}}+\lambda V\ ,
\eeq
where $H^{\mathrm{A}}$ is an unperturbed Hamiltonian (also called the penalty Hamiltonian), acting only on the ancilla space,
and where $\lambda V$ is a perturbation that acts between the qubits of $H^{\mathrm{T}}$ and the ancilla qubits.  Using perturbation theory, which we review in Appendix~\ref{sec:pert-gadgets}, one can show that the lowest $2^n$ eigenvalues of $H^{\mathrm{G}}$ differ from those of $H^{\mathrm{T}}$ by at most $\epsilon$ and the corresponding eigenstates have an overlap of at least $1 - \epsilon$.  

Completeness of the $2$-local Hamiltonian problem means that every problem in QMA is reducible to the $2$-local Hamiltonian decision problem in polynomial time.  Since this reduction involves perturbative gadgets that preserve the spectrum of the original $5$-local Hamiltonian, this means that the $2$-local Hamiltonian derived from the $5$-local Hamiltonian appearing in the universality proof of Sec.~\ref{sec:history-state} will also have an energy gap that is an inverse polynomial in the circuit length, and that the computation remains in the ground subspace with illegal clock states gapped away by the (now $2$-local) penalty Hamiltonian.  
In the remainder of this subsection we briefly discuss a particularly simple form of $2$-local Hamiltonians that is universal for AQC.

The QMA completeness of general $2$-local Hamiltonians can be extended to show that a more restricted set of $2$-local Hamiltonians composed of real-valued sums of the following pairwise products of Pauli matrices are QMA-complete \cite{Biamonte:07}:
\begin{eqnarray}
\left\{IX,XI,IZ,ZI,ZX,XZ,ZZ,XX
\right\}\ .
\end{eqnarray}
The basic two steps to do this are: (1) Using the result of \cite{Bernstein:1997} that any quantum circuit can be represented using real-valued unitary gates operating on real-valued wavefunctions in the proof of the QMA-completeness of the $5$-local Hamiltonian of the previous subsection, the Hamiltonian terms are all real-valued.  This therefore extends QMA-completeness to $5$-local real Hamiltonians.  (2) The same gadgets used in \cite{KempeGadget,OliveiraGadget} can be used to reduce the locality from five to two.

This can be further simplified to show that ``$ZZXX$" Hamiltonians" that are linear combinations with real coefficients of only
\beq 
\label{eqt:ZZXX}
\left\{IX,XI,IZ,ZI,ZZ,XX
\right\}
\eeq
are QMA-complete. This is done by showing, using perturbation theory, that such Hamiltonians can be used to approximate the $ \sigma^z \otimes \sigma^x$ and $ \sigma^x \otimes \sigma^z$ terms.   Similarly, perturbation theory can be used to show that ``$ZX$" Hamiltonians" that are linear combinations with real coefficients of only
\beq \label{eqt:ZX}
\left\{IX,XI,IZ,ZI,ZX,XZ
\right\}
\eeq
are QMA-complete \cite{Biamonte:07,Cubitt:2016vl,Bravyi:2014bf}.


\section{Stoquastic Adiabatic Quantum Computation}
\label{sec:QA}

In this section we focus on the special class of ``stoquastic Hamiltonians" [originally introduced in \cite{Bravyi:QIC08}], that often arise in the context of quantum optimization. 

\begin{definition}[stoquastic Hamiltonian \cite{Bravyi:2009sp}]
A  Hamiltonian $H$ is called stoquastic with respect
to a basis $\mathcal{B}$ iff $H$ has real nonpositive off-diagonal matrix elements in the basis $\mathcal{B}$.
\label{def:stoq}
\end{definition}
For example, a Hamiltonian is stoquastic in the computational basis iff
\beq
\bra{x}H\ket{x'} \leq 0 \quad \forall x, x'\in\{0,1\}^n\quad x\neq x' \ .
\eeq
The computational basis is often singled out since it plays the role of the basis in which the final Hamiltonian is measured, which sometimes coincides with the basis in which that Hamiltonian is diagonal. The term ``stoquastic" was introduced due to the similarity to stochastic matrices, such as arise in the theory of classical Markov chains. 

Restricting to any basis still leaves some freedom in the definition.  For example, a Hamiltonian $H = - \sum_i \sigma_i^x + H_Z$, where $H_Z$ is diagonal in the computational basis, is clearly stoquastic.  However, applying a unitary transformation $U = \prod_i \sigma_i^z$ to the Hamiltonian gives $H' = \sum \sigma_i^x + H_Z$, which according to Def.~\ref{def:stoq} is not stoquastic in the computational basis.  Applying a local unitary basis transformation should not change the complexity of the problem.  Therefore, for clarity we fix the basis such that the standard initial Hamiltonian always carries a minus sign, i.e., $-\sum_i \sigma_i^x$. From this point forward, we restrict our discussion of stoquasticity to the standard computational basis. With this in mind, the class of stoquastic Hamiltonians includes the fully ferromagnetic Heisenberg and $XY$ models, and the quantum transverse field Ising model [recall Eq.~\eqref{eq:H-all-models}]. 

Given the restriction of the Hamiltonian, one may ask whether there is a complexity class for which the $k$-local stoquastic Hamiltonian problem is complete.  This led to the introduction of the class StoqMA, for which the $k\geq2$-local stoquastic Hamiltonian is StoqMA-complete \cite{Bravyi:2006}. This can be further refined to the result that the transverse Ising model on degree-3 graphs is StoqMA-complete \cite{Bravyi2017}. Rather than give the formal (and rather involved) definition of StoqMA, we note that the only difference between StoqMA and MA is that a stoquastic verifier in StoqMA is allowed to do the final measurement in the $\left\{ \ket{+}, \ket{-} \right\}$ basis, whereas a classical coherent verifier in MA can only do a measurement in the standard $\left\{ \ket{0}, \ket{1} \right\}$ basis.\footnote{MA has an alternative quantum definition as a restricted version of QMA in which the verifier is a coherent classical computer \cite{Bravyi:QIC08}.}
Unlike MA and QMA, the threshold probabilities in StoqMA have an inverse polynomial rather than constant separation; this prevents amplification of the gap between the threshold probabilities based on repeated measurements with majority voting.  Finally, it is known that MA $\subseteq$ StoqMA $\subseteq$ QMA \cite{Bravyi:2006}. 

To capture the important class of problems that are characterized by stoquastic evolution with the constraint of adiabatic evolution, we first introduce the following definition of a model of computation:

\begin{definition}[StoqAQC]
\label{def:StoqAQC}
Stoquastic adiabatic quantum computation (StoqAQC) is the special case of AQC (Definition~\ref{def:AQC}) restricted to $k$-local ($k$ fixed) stoquastic Hamiltonians. 
\end{definition}

Because we defined StoqAQC as a special case of AQC, the computation must proceed in the ground state. However, recall that the algorithm for the glued trees problem (Sec.~\ref{sec:gluedtrees}) is not subject to this ground state restriction and hence is not in StoqAQC. In Sec.~\ref{sec:stoq-QMA-comp} we consider another model of stoquastic computation that is not subject to the ground state restriction. 

StoqAQC has generated considerable interest since experimental implementations of stoquastic Hamiltonians are quite advanced \cite{Bunyk:2014hb,Weber:2017aa}. To characterize its computational power, we introduce a natural promise problem based on StoqAQC, and modeled after the $k$-local Hamiltonian problem:\footnote{We are indebted to Elizabeth Crosson for her help in formulating the StoqAQCEval problem, the BStoqP class, and working out the relations of BStoqP to other complexity classes.}

\begin{definition}[StoqAQCEval]
\label{def:StoqAQCEval}
The StoqAQCEval problem is defined on n qubits, with the following input:
\begin{itemize}
\item a continuous family of $(k\geq 2)$-local ($k$ fixed) stoquastic Hamiltonians $H(s) = \sum_{i=1}^r H_i(s)$ with $r = \text{poly}(n)$ and parameterized by $s \in [0,1]$.  For all $i$ and all $s$, the non-zero entries of $H_i(s)$ are specified by $\text{poly}(n)$ bits of precision, and $\|H_i(s)\| = \text{poly}(n)$.  The ground state energy gap $\Delta[H(s)]$ satisfies $\Delta[H(s)] \geq 1/\text{poly}(n)$ for all $s$.
\item two real numbers $a$ and $b$ specified with $\text{poly}(n)$ bits of precision, and $b- a > 1/\text{poly}(n)$.
\end{itemize}
The output (0 or 1) answers the question: Is the smallest eigenvalue of $H(s = 1)$ smaller than $a$ (output is $1$), or are all  eigenvalues larger than $b$ (output is $0$)?  Just as in the local Hamiltonian problem, we are promised  that the outcome that the ground state energy is between $a$ and $b$ is not possible.
\end{definition}
This allows us to (informally) define the complexity class that captures StoqAQC:
\begin{definition}[BStoqP]
\label{def:BStoqP}
BStoqP is the set of problems that are polynomial-time reducible to StoqAQCEval.
\end{definition}
The StoqAQCEval problem is clearly in StoqMA, because the $k\geq 2$-local stoquastic Hamiltonian problem is StoqMA-complete \cite{Bravyi:2006}. Hence BStoqP$\subseteq$StoqMA, as depicted in Fig.~\ref{fig:complexityclasses}, which summarizes the relations between many of the complexity classes we have discussed.\footnote{As far as we know the related term StoqP was informally introduced by Stephen Jordan in a talk presented at the AQC 2016 conference \cite{Jordan:AQC2016}, showing that StoqP is not equal to BQP unless BQP is in the third level of the polynomial hierarchy.} NP and MA are unlikely to be subsets of BStoqP, since StoqAQC would not be expected to solve NP-complete problems in polynomial time. The tightest inclusion in a classical complexity class we know of is in AM,\footnote{Like MA, the class AM (Arthur Merlin) is a probabilistic generalization of NP. See \url{https://complexityzoo.uwaterloo.ca/Complexity_Zoo} for definitions of the complexity classes mentioned here.}  since the latter includes StoqMA \cite{Bravyi:QIC08}. It is clear that BPP$\subseteq$BStoqP, since $5$-local StoqAQCEval is BPP-hard (using a a classical reversible circuit for universal AQC, with stoquastic gate terms and a $5$-local stoquastic clock Hamiltonian). Finally, we know that BStoqP$\subseteq$BQP, since StoqEvalAQC is in BQP by using the same proof that the adiabatic model in general can be simulated by the circuit model.  

\begin{figure}[t] 
   \centering
   \includegraphics[width=0.95\columnwidth]{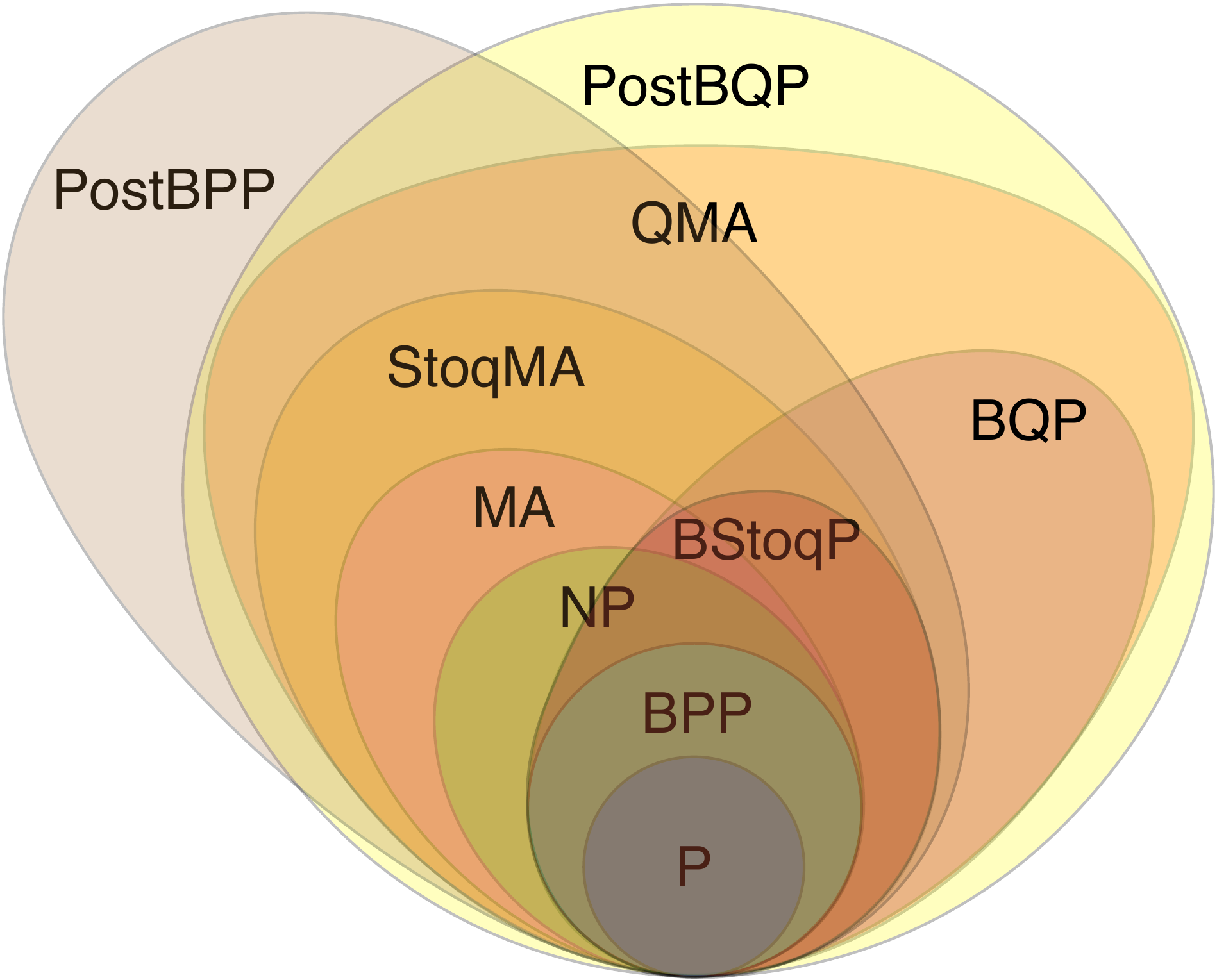} 
   \caption{Known relations between complexity classes relevant for AQC. The BStoqP class defined here (Def.~\ref{def:BStoqP}) lies in the intersection StoqMA and BQP, and includes BPP.}
   \label{fig:complexityclasses}
\end{figure}

\subsection{Why it might be easy to simulate stoquastic Hamiltonians}
\label{sec:stoq-is-easy}

In this subsection we briefly summarize the complexity-theoretic evidence obtained so far that suggests that the StoqAQC setting is less powerful than universal quantum computation. Let us start with a lemma that characterizes the ``classicality" of ground states of stoquastic Hamiltonians.
\begin{lemma}
The ground state $\ket{\psi}$ of a stoquastic Hamiltonian $H$ can always be expressed using only real nonnegative amplitudes: $\ket{\psi} = \sum_{x\in\{0,1\}^n} a_x \ket{x}$, where $a_x \geq 0$ $\forall x$.
\end{lemma} 
\begin{proof}
It follows directly from the stoquastic property that the corresponding Gibbs density matrix $\rho = \exp(-\beta H)/\Tr[\exp(-\beta H)]$ has non-negative matrix elements in the computational basis for any $\beta >0$. In particular, if $H$ is stoquastic then for sufficiently small $\beta$, $\ident - \beta H$ has only non-negative matrix elements. The largest eigenvalue of $\ident - \beta H$ corresponds to the ground state energy of $H$. Thus, by the Perron-Frobenius theorem (see Sec.~\ref{sec:Googlematrix}) the ground state of $H$ can be chosen to have non-negative amplitudes.
\end{proof}
Consequently, if the Hamiltonian is stoquastic, 
a classical probability distribution can be associated with the ground state.
This raises the question as to whether StoqAQC is a model that is capable of quantum speedup over classical  algorithms. Following is the evidence regarding this question.

\begin{enumerate}
\item The ground state energy of the fully ferromagnetic transverse field Ising model can be found to a given additive error in polynomial time with a classical algorithm on any graph, with or without a transverse magnetic field \cite{Bravyi:2016ab}.

\item In \cite{Bravyi:QIC08} it was shown that for any fixed $k$, stoquastic $k$-local Hamiltonian is contained in the complexity class AM. Thus, unless QMA$\subseteq$AM (which is believed to be unlikely), stoquastic $k$-local Hamiltonian is not QMA-complete.
\item It was also shown in \cite{Bravyi:QIC08} that gapped StoqAQC can be simulated in PostBPP, the complexity class described by a polynomial-time classical randomized
computer with the ability to post-select on some subset of the bits after the algorithm is run.\footnote{See also \cite{FarhiHarrow-QAOA}, where gapped StoqAQC was called stoquastic gapped adiabatic evolution, QADI-SG.} I.e., it suffices to call an oracle for problems in PostBPP a polynomial number of times to efficiently sample from the ground state of a gapped stoquastic Hamiltonian.\footnote{PostBPP, also known as BPP$_{\text{path}}$, contains NP. For example,
consider the Grover problem with two registers, a bit-string $x$ for the input and a second register where $f(x)$ is stored. Now if we pick $x$ at random and post-select on the second
register being $1$, we find a marked item. PostBPP is known to be contained in the third level of the polynomial hierarchy \cite{Han:1993fe}.} 

\noindent Suppose that StoqAQC could be used to perform universal quantum computation.
Since gapped StoqAQC can be simulated in PostBPP, this would imply that SampBQP
$\subseteq$ SampPostBPP.%
\footnote{%
In sampling problems we are given an input $x \in \{0, 1\}^n$, and the goal is to sample (exactly or approximately) from some probability distribution over poly$(n)$-bit strings. 
SampBQP and SampPostBPP are the classes of sampling problems solvable on quantum computers and probabilistic classical computers with post-selection, respectively, to within $\eps$ error in total variation (or trace-norm) distance, in time polynomial in $n$ and $1/\eps$ \cite{Aaronson:2010aa}.} In other words, this would imply that
polynomial time quantum algorithms can be simulated classically in polynomial
time using post-selection.
This would then imply that PostBPP=PostBQP
which in turn would collapse the polynomial hierarchy. Thus it is unlikely that StoqAQC is universal for AQC.
\item In \cite{Bravyi:2009sp} it was shown that adiabatic evolution along a path composed entirely of stoquastic frustration-free Hamiltonians (recall Definition~\ref{def:FF}) may be simulated by a sequence of classical random walks, i.e., is contained in BPP.
\end{enumerate}

With this evidence for the potential limitations of stoquastic Hamiltonians, the question arises if they are worthy of pursuit, either theoretically or experimentally. However, it is important to remember that the weakness of stoquastic Hamiltonians arises when one assumes that they generate an evolution that occurs in the ground state. Indeed, we will see in Sec.~\ref{sec:stoq-QMA-comp} that \emph{excited state} stoquastic evolution can be as powerful as universal AQC. 
Moreover, in the next subsection we briefly review counterexamples to the claim that stoquastic Hamiltonians are necessarily easy to simulate using heuristic classical algorithms.

\subsection{Why it might be hard to simulate stoquastic Hamiltonians}
\label{sec:stoq-is-hard}

There does not exist a general theorem that rules out a quantum speedup of StoqAQC over all possible classical algorithms. However, it is often stated that Monte Carlo simulations of StoqAQC do not suffer from the sign problem and will therefore simulate StoqAQC without a slowdown.
Specifically, the conjecture is that if the Monte Carlo simulation starts at $s=0$ in the equilibrium state, and if $s$ changes by a small amount $\epsilon$ from one step to the next, where $\epsilon$ is polynomially small in the system size $n$, the inverse temperature $\beta$ and the spectral gap $\Delta$, then the Monte Carlo simulation stays close to the equilibrium state along the path.  For sufficiently large $\beta$, this would correspond to following the instantaneous ground state.
In this subsection we briefly review theoretical evidence that such a conjecture is not always true.
We focus on two of the most direct classical competitors to StoqAQC: path integral quantum Monte Carlo (PI-QMC), and diffusion quantum Monte Carlo (D-QMC). 

\subsubsection{Topological obstructions}
In \cite{Hastings:2013kk} examples were given of StoqAQC with a polynomially small eigenvalue gap, but where PI-QMC take exponential time to converge. Loosely, the failure of convergence was due to topological obstructions around which the worldlines (trajectories in imaginary time) can get tangled. 

The simplest of the examples can be understood intuitively as follows. A sombrero-like potential is constructed for a single particle with a deep circular minimum of radius $r$. The worldline of the particle in PI-QMC with closed boundary conditions is some closed path that follows this circle in imaginary time.  Because of the depth of the potential at the minimum the distribution of trajectories has very small probability to include any point with radius larger than $r$ and it takes an exponential time in the winding number to transition from one winding number sector to another. Therefore, if an appropriate dimensionless combination of the radius $r$ or the mass of the particle is changed sufficiently fast then PI-QMC fails to equilibrate. At the same time it can be shown that for this example the gap closes polynomially and so one expects that adiabatic evolution requires only polynomial time to find the ground state.

While this example uses winding numbers to construct a protocol for which PI-QMC takes exponential time to equilibrate, PI-QMC can still find the ground state.  To observe a more dramatic effect where not only equilibration is hampered but also the probability of finding the ground state is low, one can introduce stronger topological effects and additionally exploit the discrepancy between $L_1$ and $L_2$-normalized wavefunctions. This was first done in the ``bouquet of circles" example introduced in \cite{Hastings:2013kk}, which shows that PI-QMC can fail to converge even when
using open boundary conditions. The example was designed so that the majority of
the amplitude $\psi$ lies within an expander graph, although the majority of the
probability $|\psi|^2$ does not. Because the endpoints of the wordlines are distributed
according to $\psi$ and not $|\psi|^2$, this effectively ``pins" them to the expander
graph. This pinning means that even though the worldline is in principle open,
the worldline is nevertheless prevented from changing its topological sector
within the bouquet of circles. This then causes failure of convergence. 

A more general method using perturbative gadgets is explained in \cite{Hastings:2013kk}, that allows one to map between continuous variables and spins and applies to all the examples given there.

\subsubsection{Non-topological obstructions}

Diffusion Monte Carlo algorithms should not be affected by topological obstructions that depend on closed boundary conditions, since they do not exhibit periodicity in the imaginary time direction. Rather than use topological obstructions, it is possible to rely entirely on the discrepancy between $L_1$ and $L_2$-normalization to design examples where Monte Carlo methods have differing convergence from AQC. This discrepancy was used in \cite{Jarret:16} to ensure that the walkers in a DQMC algorithm never ``learn" about a potential well that contains the solution, causing DQMC to take exponential time to converge. Since the gap for the adiabatic process is large, QA takes only polynomial time. 

Let $H(s)$ be some stoquastic Hamiltonian acting on a Hilbert space whose basis
states can be equated with the vertices $V$ of some graph. Let $\psi_s(x):V\mapsto \mathbb{C}$ denote the ground state of $H(s)$. Define probability distributions $p_s^{(1)}(x) = \frac{\psi_s(x)}{\sum_{y\in V} \psi_s(y)}$ and $p_s^{(2)}(x) = {\psi^2_s(x)}$. The stoquasticity of $H(s)$ ensures that $\psi_s(x) \geq 0$, so that $p_s^{(1)}(x)$ is a valid probability distribution.

D-QMC algorithms perform random walks designed to ensure that a population of random walkers converges to $p_s^{(1)}(x)$. However, in exponentially large Hilbert spaces there can be vertices such that the distribution associated with the $L_2$-normalized
wavefunction $p_s^{(2)}(x)$ is polynomial, but the distribution associated with the $L_1$-normalized
wavefunction $p_s^{(1)}(x)$ is exponentially small. The idea behind the examples in \cite{Jarret:16} is to exploit this discrepancy to design polynomial-time stoquastic adiabatic processes that the corresponding D-QMC simulations will fail to efficiently simulate.

The main example given in \cite{Jarret:16} is the stoquastic Hamiltonian $H(s) = \frac{1}{n}[L+b(s)W]-c(s)P$, where $L$ is the graph Laplacian of the $n$-bit hypercube, $W$ is the Hamming weight operator (i.e., $W\ket{x} = |x|\ket{x}$ where $|x|$ is the Hamming weight of the bit-string $x$), $P = \ketbra{0\cdots 0}{0\cdots 0}$. In terms of Pauli matrices this Hamiltonian can be written, up to an overall constant, as
\beq
H(s) = -\frac{1}{n}\sum_{j=1}^n \left(X_j + \frac{1}{2}  b(s) Z_j \right)-c(s)P .
\eeq
The schedules $b(s)$ and $c(s)$ are:
\begin{align}
b(s) & = \left\{ \begin{array}{lr}
2sb  &   \\
b  &  
\end{array} \right.
c(s)  = \left\{ \begin{array}{ll}
0  & s\in[0,1/2)  \\
(2s-1)c \  &  s\in[1/2,1]
\end{array} \right. \ . 
\end{align}
For $s\in[0,1/2)$ this is a Hamiltonian of $n$ non-interacting qubits whose gap is easily seen to be $\frac{2}{n}\sqrt{1+(sb/2)^2}$, minimized at $s=0$ where it equals $2/n$. The ground state is given by $\ket{\psi(\theta)}^{\otimes n}$, where $\ket{\psi(\theta)} = \cos(\theta/2)\ket{0}+\sin(\theta/2)\ket{1}$ and $\theta = \tan^{-1}[2/(sb)]$. For $s\in[1/2,1]$ it can be shown that the minimum gap is attained at $s=1/2$, where it equals $1/\sqrt{2n} + O(n^{-3/2})$. Thus the overall minimum gap is polynomial ($2/n$) and the StoqAQC process converges to the ground state $\ket{0\cdots 0}$ in polynomial time. By choosing $b$ so that at $s=1$ we have $\cos(\theta/2) = 1-1/(4n)$,
it is easy to show from the analysis of the non-interacting problem that the probability of ending up in the ground state is $p^{(2)}_{s=1}(0\cdots 0) = \cos^{2n}(\theta_{s=1}/2) 
\to e^{-1/2}$ in the limit $n\to \infty$. 

On the other hand, for the non-interacting problem (when $s\in[0,1/2)$) the D-QMC process\footnote{Here D-QMC refers to the ``Substochastic Monte Carlo (SSMC)" algorithm introduced in \cite{Jarret:16}.} samples from the distribution $p^{(1)}_{s}(x) = \sin(\theta/2)^{|x|} \cos(\theta/2)^{n-|x|}/Z_s$, where $Z_s = \sum_{x\in\{0,1\}^n} \sin(\theta/2)^{|x|} \cos(\theta/2)^{n-|x|} = [\sin(\theta/2)+\cos(\theta/2)]^n$, so that for large $n$ we have $Z_1 \approx (1+1/\sqrt{2n})^n \to e^{\sqrt{n/2}}$ for the same choice of $b$. Thus, for D-QMC the probability of being in the ground state at $s=1/2$ is $p^{(1)}_{s=1/2}(0\cdots 0) = \cos^{n}(\theta_{s=1}/2)/Z_{s=1} \to e^{-1/4}e^{-\sqrt{n/2}}$. Since at $s=1/2$ the random walkers that diffuse in the D-QMC process have a probability to be at the all-zeros string that is of order $e^{-\sqrt{n/2}}$, with high likelihood, no walkers will land on the all-zeros string until the number of time-steps times the number of walkers approaches $e^{\sqrt{n/2}}$. Until this happens it is impossible for the distribution of walkers to be affected by the change in the potential at the all-zeros string that is occurring from $s = 1/2$ to $s = 1$; no walkers have landed there, and the D-QMC algorithm has therefore never queried the value of the potential at that site. Only after allowing for this exponential cost, and by appropriately choosing $c$, does the D-QMC algorithm find the ground state with high probability.\footnote{For $c=2$ one finds that $p^{(1)}_{s=1}(0\cdots 0) = 1/2+O(n^{-1/2})$.}

\subsection{QMA-complete problems and universal AQC using stoquastic Hamiltonians with excited states}
\label{sec:stoq-QMA-comp}

Our definition of StoqAQC (Definition~\ref{def:StoqAQC}) stipulates that the computation must proceed in the ground state. It turns out that if this condition is relaxed, computation with stoquastic Hamiltonians is as powerful as AQC, i.e., it is universal. Here we review a construction by \cite{Jordan:2010fk} of a $3$-local stoquastic Hamiltonian that, by allowing for excited state evolution, is both QMA-complete and universal for AQC.

We start with the QMA-complete Hamiltonian introduced in Sec.~\ref{sec:BiamonteLove}, that can be written as:
\beq
H_{ZZXX} = \sum_i d_i X_i +  h_i Z_i + \sum_{i\leq j} J^x_{ij} X_i X_j + J^z_{ij} Z_i Z_j\ ,
\eeq
where $d_i$, $h_i$, $J^x_{ij}$ and $J^z_{ij}$ are arbitrary real coefficients.  The key idea is to eliminate the negative matrix elements in each term. Toward this end the Hamiltonian is written as:
\beq
H_{ZZXX} = - \sum_k \alpha_k T_k\ ,
\eeq
where $T_k \in \left\{\pm X_i, \pm Z_i, \pm X_i X_j, \pm Z_i Z_j \right\}$ and such that $\alpha_k > 0$.  For an $n$-qubit system, the operators $T_k$ are represented by $2^n \times 2^n$ symmetric matrices with entries taking value $+1,-1,0$.  We use the regular representation of the $Z_2$ group to make the replacement
\beq \label{eqt:Z2substitution}
1 \to \left(\begin{array}{cc}
1 & 0 \\
0 & 1
\end{array} \right)  , \  -1 \to \left(\begin{array}{cc}
0 & 1 \\
1 & 0
\end{array} \right) , \  0 \to \left(\begin{array}{cc}
0 & 0 \\
0 & 0
\end{array} \right)
\eeq
in $T_k$ to define a new operator $\tilde{T}_k$.  The matrix representation of $\tilde{T}_k$ is of size $2^{n+1}\times 2^{n+1}$, and since the original $T_k$ was either $1$-local or $2$-local acting on $n$ qubits, we can interpret $\tilde{T}_k$ as being $2$-local or $3$-local acting on $n+1$ qubits.  Note furthermore that $T_k$ is such that it only has one non-zero entry per row and column, hence with the substitution in Eq.~\eqref{eqt:Z2substitution}, the $\tilde{T}_k$'s are permutation matrices.  We can write the following Hamiltonian acting on $n+1$ qubits:
\beq
\tilde{H}_{ZZXX} = - \sum_k \alpha_k \tilde{T}_k
\eeq
which is a linear combination of permutation matrices with negative coefficients.  This makes $\tilde{H}_{ZZXX}$ a (3-local) stoquastic Hamiltonian.  We can write it as:
\beq
\tilde{H}_{ZZXX} = H_{ZZXX} \otimes \ketbra{-}{-} + \bar{H}_{ZZXX} \otimes \ketbra{+}{+} \ ,
\eeq
where $\bar{H}_{ZZXX} =- \sum_k \alpha_k |T_k|$ and $|T_k|$ is the entrywise absolute value of $T_k$.  To see why this is the case, first consider a positive element $(T_k)_{ij}$.  Then:
\begin{equation*}
- \alpha_k (T_k)_{ij} \otimes \ketbra{-}{-} - \alpha_k (T_k)_{ij} \otimes \ketbra{+}{+} = - \alpha_k (T_k)_{ij} \otimes \ident
\end{equation*}
corresponding to the first replacement in Eq.~\eqref{eqt:Z2substitution}.  For a negative element, we have:
\begin{equation*}
- \alpha_k (T_k)_{ij} \otimes \ketbra{-}{-} + \alpha_k (T_k)_{ij} \otimes \ketbra{+}{+} = - \alpha_k (T_k)_{ij} \otimes \sigma^x
\end{equation*}
corresponding to the second replacement in Eq.~\eqref{eqt:Z2substitution}.  The spectrum of $\tilde{H}_{ZZXX}$ separates into two sectors $\mL_\pm$.  The sector $\mL_-$ is spanned by $\ket{\veps_j} \otimes \ket{-}$, where $\ket{\veps_j}$ are the eigenstates of ${H}_{ZZXX}$, while the sector $\mL_+$ is spanned by $\ket{\bar{\veps}_j} \otimes \ket{+}$, where $\ket{\bar{\veps}_j}$ are the eigenstates of $\bar{H}_{ZZXX}$.  Because the Hamiltonian does not couple the two sectors (there are no interactions that take the ancilla qubit from $\ket{\pm}$ to $\ket{\mp}$), a closed system evolution initialized in the $\mL_-$ sector will remain in that sector.  

Because the spectrum in the $\mL_-$ sector is identical to that of $H_{ZZXX}$, which is capable of universal adiabatic quantum computation, then universal adiabatic quantum computation can be performed in the $\mL_-$ sector.  However, the lowest energy state in $\mL_-$ may not necessarily be the ground state of $\tilde{H}_{ZZXX}$.  Therefore, this establishes universal AQC using a stoquastic Hamiltonian only if we do not restrict ourselves to the ground state of the Hamiltonian.  Attempting to make the lowest energy state in $\mL_-$ be the ground state requires introducing a sufficiently large term proportional to $\ident \otimes \ketbra{+}{+}$ to the Hamiltonian $\tilde{H}_{ZZXX}$, but such a term would make the new Hamiltonian non-stoquastic since it would introduce positive off-diagonal elements.  Therefore this method does not establish universal adiabatic quantum computation using the ground state of a stoquastic Hamiltonian.  

\subsection{Examples of slowdown by StoqAQC}
\label{sec:exp-small-gaps}

It should not come as a surprise that AQC with an arbitrary final Hamiltonian, which is essentially a black box approach, does not guarantee quantum speedups. It can be vulnerable to the same sorts of locality traps confronted by heuristic classical algorithms such as simulated annealing.

A slowdown, or failure of AQC to provide a speedup, is a scenario wherein a more efficient classical algorithm is known. All the known examples that fall into this category arise when the gap closes ``too fast" in the problem size. However, it is important to note that adiabatic theorems provide only upper bounds on run time, not lower bounds. Thus, e.g., an exponentially
small gap does not strictly imply an exponentially long run time. The inverse gap is often treated as a proxy for run time, but a claim such as an equal scaling of the inverse gap and the run time does not hold as a general theorem.

With this caveat in mind, in this subsection we review such ``small-gap" examples in increasing order of generality or difficulty of analysis, which all arise in the StoqAQC context. However, we must first note another important caveat. Namely, in some of the examples we present numerical evidence that is based, in necessity, on finite size calculations. One is then often tempted to extrapolate such evidence to the asymptotic scaling. Of course, any such extrapolations based purely on numerics are conjectures. For example, a claim of exponential scaling can never be proven based on numerics alone, as any finite set of data points can always be perfectly fit by a polynomial of sufficiently high degree. Nevertheless, numerics-driven conjectures about scaling can be quite useful, especially if supported by other, analytical arguments.

Subsequently, we will see in Sec.~\ref{sec:fix-AQC} that there are various methods for circumventing  slowdowns, e.g., via the introduction of non-stoquastic terms. 

\subsubsection{Perturbed Hamming Weight Problems with Exponentially Small Overlaps}
\label{sec:PHW}

The plain Hamming weight problem is described by 
\beq
H_{\mathrm{HW}}(s) = (1-s) \frac{1}{2} \sum_i \left( 1- \sigma_i^x \right)  + s \sum_x |x| \ketbra{x}{x} \ .
\label{eq:118}
\eeq
Its cost function is simply the Hamming weight $|x|$ of the binary bit-string $x$, which is trivially minimized at $x=0^n$. Consider the following perturbation of the plain Hamming weight problem \cite{Dam:2001fk}:
\beq 
h(x) = \left\{ \begin{array}{lr}
|x| & \text{if} \  |x| < n \\
-1 & \text{if} \  |x| = n
\end{array} \right. \ .
\eeq
This is a toy problem that 
is designed to be hard for classical algorithms based on local search: its global optimum lies in a narrow basin, while there is a local optimum with a much larger basin.  An algorithm such as simulated annealing with single spin updates would require exponential time to find the global minimum.

Let us write the corresponding StoqAQC Hamiltonian to make the perturbation explicit:
\beq \label{eqt:vanDamVaziraniPHW}
H(s) = H_{\mathrm{HW}}(s) - s (n+1) \ketbra{1^n}{1^n} \ ,
\eeq
where $\ket{1^n}$ is the all-one state. Denote the instantaneous eigenstates of $H_{\mathrm{HW}}(s)$ by $\left\{ \ket{v_i(s)} \right\}$ ($v_{0}$ denotes the ground state).  Note that the overlap of the all-one-state with the instantaneous ground state of the plain Hamming Weight algorithm is always exponentially small:
\beq
\braket{1^n}{v_0(s)} \leq \frac{1}{\sqrt{2^n}}\ .
\eeq
We will show that this fact causes the adiabatic algorithm as defined in Eq.~\eqref{eqt:vanDamVaziraniPHW} to take exponential time because it leads to an exponentially small gap.

Define a matrix $A(s)$ with elements:
\beq
A_{ij} = \bra{v_i(s)}H(s) \ket{v_j(s)} \ .
\eeq
Note that $A(0)$ is diagonal and $A_{00}(0)=0$, equal to the ground state eigenvalue.  Also $A(1)$ is diagonal, but now $A_{2^n-1,2^n-1}(1)=-1$ is equal to the ground state eigenvalue.  Define a matrix $B$ in the same basis as:
\beq
B_{ij}(s) = \left\{ \begin{array}{ll}
A_{00}(s) & i=j=0 \\
0 & i = 0, j>0 \\
0 & i>0, j=0 \\
A_{ij}(s) & \text{otherwise}
\end{array} \right.
\eeq
The matrix $B$ always has $A_{00}$ as an eigenvalue.  By construction, we know that at $s=1$, the matrix $B$ has $-1$ as its ground state eigenvalue (located in the $2^{n-1} \times 2^{n-1}$ sub-matrix).  Because the matrix transforms continuously between these two extremes, there cannot be a jump in the ground state eigenvalue, so there must be a critical value of $s$, which we denote by $s_c$, where $B$ has a vanishing gap.

The optimal matching distance between $A$ and $B$ expresses how close their eigenvalue spectra are:
\beq
d(A,B) = \min_\pi \max_{1 \leq j \leq 2^n} | \lambda_j - \mu_{\pi(j)} | \ ,
\eeq
where $\pi$ denotes a permutation. Since $A$ and $B$ are Hermitian, this is upper-bounded by $\Vert A - B \Vert_2$ \cite{Bhatia:book}.  The matrix $A-B$ only has non-zero entries $(A-B)_{0,j>0} = A_{0,j>0}$ and $(A-B)_{j>0,0} = A_{j>0,0} = A_{0,j>0}^\ast$, with $A_{0,j>0} = -s(n+1) \braket{v_1(s)}{1^n} \braket{1^n}{v_j(s)}$.  Therefore:
\begin{align}
& \Vert A - B \Vert_2 =  \\
& \quad s (n+1) | \braket{v_1(s)}{1^n}| \sqrt{ \sum_{j=1}^{2^n-1} \braket{1^n}{v_j(s)}\braket{v_j(s)}{1^n}} \notag \\
& \quad =  s (n+1) | \braket{v_1(s)}{1^n}| \sqrt{  1 - |\braket{1^n}{v_0(s)}|^2 } \notag \\
& \quad \leq s(n+1)   | \braket{v_1(s)}{1^n}|  \leq \frac{s(n+1)}{\sqrt{2^n}} \ . \notag
\end{align}
Thus, the gap of $A$ [and hence of $H(s)$] is always upper-bounded by the gap of $B$ plus twice $\Vert A - B \Vert_2$.  Since at $s = s_c$ the gap of $B$ is zero, it follows that the gap of $H(s_c)$ is $\leq s_c (n+1)/\sqrt{2^{n-2}}$  \cite{Dam:2001fk}.  Therefore, the exponentially small overlap between the unperturbed instantaneous ground state and the perturbed final ground state results in the adiabatic algorithm requiring exponential time to reach the final ground state. Informally, this can also be viewed as the inability of local quantum search (fluctuations  induced by the local initial Hamiltonian) to explore the entire (non-local) energy landscape effectively.

\subsubsection{2-SAT on a Ring}
\label{sec:2SAT-ring}
%
In this subsection we review the ``2-SAT on a Ring" problem introduced in the seminal work \cite{farhi_quantum_2000}, which launched the field of AQC. This example is instructive because of its use of the Jordan-Wigner and Fourier transformation techniques, and is also of historical interest. It also serves to illustrate that even a polynomially small gap does not guarantee a quantum speedup. We thus review it in detail.

Consider an $n$-bit SAT problem with $n$ clauses.  Each clause only acts on adjacent bits, i.e., the clause $C_j$ only acts on bits $j$ and $j+1$, where we identify bit $n+1$ with bit 1.  Let each clause be of only two forms: ``agree'' clauses where $00$ and $11$ are satisfying assignments, and ``disagree'' clauses where $01$ and $10$ are satisfying assignments.  Since an odd number of satisfied disagree clauses means that the first bit of the first disagree clause is the opposite of the second bit of the last disagree clause, yet bits $1$ and $n+1$ must agree, there must be an even number of disagree clauses in order for a satisfying assignment to exist. The classical computational cost of finding a satisfying assignment is at most $n$: given the list of clauses, a satisfying assignment is found (assuming an even number of disagree clauses) simply by going around the ring and satisfying each clause one at a time. Note that if $\left\{ w_i \right\}_{i=1}^n$ is a satisfying assignment then so is $\left\{ \neg w_i \right\}_{i=1}^n$.

Let us now define the final Hamiltonian $H_1 = \sum_{i=1}^n H_{C_i}$ associated with the SAT problem, where each clause is represented by:
\begin{align}
H_{C_i} &=
\frac{1}{2} \left( 1 - (-1)^{x_i} \sigma_i^z \sigma_{i+1}^z \right) \\
&  x_i=0\ (1) \text{ if } C_i \text{ is an agree (disagree) clause}  \ . \notag
\end{align}
The ground states of $H_\mathrm{P}$ are then given by $\ket{0}_1\otimes_{i=2}^n \ket{w_i}_i$ and $\otimes_{i=1}^n \ket{ \neg w_i}_i$, where $w_i = \bigoplus_{j=1}^{i-1} x_j$ ($i \geq 2$ and addition modulo $2$). It is possible to gauge away all the disagree clauses. To see this, let $U$ be the unitary transformation defined such that
\beq
U \ket{z_i} = \left\{ \begin{array} {lr} 
\ket{\neg z_i} \ ,& \mathrm{if} \ w_i = 1 \\
\ket{z_i} \ , & \mathrm{if} \  w_j = 0
 \end{array} \right. \ .
\eeq
Under this unitary transformation we have:
\beq
H_1' = U H_1 U^{\dagger} = \sum_{i} \frac{1}{2} \left( \ident_i - \sigma_i^z \sigma_{i+1}^z \right) \ ,
\eeq
i.e., the new final Hamiltonian is a sum of just agree clauses. Note that this unitary transformation requires us to know the ground state, but $H_1'$ and $H_1$ are isospectral, so we can use it for convenience in our gap analysis. The adiabatic computation procedure will be governed by the following time-dependent Hamiltonian:
\beq
H(s) = (1-s) H_0 + s H_1' \ , \quad 0 \leq s \leq 1 \ ,
\eeq
with the initial Hamiltonian $H_0 = \sum_i \ident_i - \sigma_i^x$.
We wish to diagonalize $H(s)$ in order to find its ground state gap. First, define the negation operator $G = \prod_{i=1}^n \sigma_i^x$ such that:
\beq
 G \left( \otimes_{i=1}^n \ket{z_i} \right) = \otimes_{i=1}^n \ket{\neg z_i}\ , 
\eeq
which clearly commutes with $H(s)$. The uniform superposition state, which is the ground state of $H(0)$, is invariant under $G$, i.e., it has eigenvalue $+1$ under $G$.  Therefore, the unitary dynamics will keep the state in the sector with $G=+1$ if it starts in the ground state of $H(0)$.  Let us then write $H(s)$ purely in the $G=+1$ sector. Second, define the Jordan-Wigner transformation
\bes
\begin{align}
b_j & = \sigma_1^x \sigma_2^x \dots \sigma_{j-1}^x \sigma_{j}^-  \\
b_j^{\dagger} & = \sigma_1^x \sigma_2^x \dots \sigma_{j-1}^x \sigma_{j}^+ 
\end{align}
\ees
where $\sigma_j^{\pm} = \frac{1}{2} \left(\sigma_j^z \pm i \sigma_j^y \right)$.  These are fermionic operators that satisfy:
\bes
\begin{align}
\left\{ b_j, b_k \right\} &= 0 \quad \text{(amounts to } \left\{ \sigma_j^- , \sigma_k^- \right\} = 0 \text{)}   \\
\{ b_j, b_k^{\dagger} \} &= \delta_{jk} \quad \text{(amounts to } \left\{ \sigma_j^- , \sigma_k^+ \right\} = \delta_{jk} \text{)}  \ .
\end{align}
\ees
Note that 
\bes
\begin{align}
& b_{j}^{\dagger} b_j  = \frac{1}{2} \left(\ident_j - \sigma_j^x \right) \ , \ j=1,\dots, n \\
& \left(b_j^{\dagger} - b_j \right) \left( b_{j+1}^{\dagger}  + b_{j+1} \right)  = \sigma_j^z \sigma_{j+1}^z \ , \ j = 1, \dots, n-1  \label{eqt:b2} \\
& \left(b_n^{\dagger} - b_n \right) \left( b_{1}^{\dagger}  + b_{1} \right)  = - G \sigma_n^z \sigma_1^z \ .
\label{eqt:b3}
\end{align}
\ees
In order to make Eqs.~\eqref{eqt:b2} and \eqref{eqt:b3} consistent in the $G = +1$ sector, we take $b_{n+1} \equiv - b_1$.  Using this, we have:
\bea
\left. H(s) \right|_{G=+1} &=& \sum_{j=1}^n \left[2 (1-s) b_j^{\dagger} b_j \right. \\
&& \left. + \frac{s}{2} \left( \ident_j - \left(b_j^{\dagger} - b_j \right) \left( b_{j+1}^{\dagger} + b_{j+1} \right) \right) \right] \nonumber
\eea
Third, since this Hamiltonian is invariant under translations $j \mapsto j+1$, define Fourier operators $\beta_p$:
\beq
\beta_p = \frac{1}{\sqrt{n}} \sum_{j=1}^n e^{i \pi p j /n } b_j \ , \quad p = \pm 1, \pm 3, \dots,\pm (n-1) \ ,
\eeq
where for simplicity it is assumed that $n$ is even. Equivalently:
\beq \label{eqt:b's}
b_j = \frac{1}{\sqrt{n}} \sum_{p = \pm 1, \dots} e^{-i \pi p j / n} \beta_p \ ,
\eeq
where we used the fact that $\sum_{p=\pm 1, \dots} e^{i \pi p (j-j')/ n} = n \delta_{j,j'}$. 
Furthermore, note that:
\bes
\begin{align}
\left\{ \beta_{2a-1} , \beta_{2b-1} \right\} & =  \frac{1}{n} \sum_{j,j'} e^{i \pi ((2a-1) j + (2b-1) j')/n} \left\{ b_j, b_{j'} \right\}\nonumber  \\
& = 0 \\
\left\{ \beta_{2a-1} , \beta_{2b-1}^{\dagger} \right\} &= \frac{1}{n} \sum_{j,j'} e^{i \pi ((2a-1) j - (2b-1) j')/n} \left\{ b_j, b_{j'}^{\dagger} \right\} \nonumber \\
& = \frac{1}{n} \sum_{j=1}^n e^{2 \pi i (a - b)/n} = \delta_{a ,b}
\end{align}
\ees
so the set $\left\{ \beta_p \right\}$ comprises valid fermionic operators.  Writing the Hamiltonian in terms of this set, we have:
\bes
\begin{align}
H(s) = & \sum_{p=1,3,\dots} \left[ 2(1-s) \left( \beta_p^{\dagger} \beta_p + \beta_{-p}^{\dagger} \beta_{- p} \right) \right. \nonumber \\  
&  + s \left( \ident - \cos \left( \frac{\pi p}{n} \right) \left( \beta_p^{\dagger} \beta_p  - \beta_{-p} \beta_{-p}^{\dagger} \right)  \right.  \nonumber \\
& \left. \left. + i \sin \left( \frac{ \pi p}{n} \right) \left( \beta_{-p}^{\dagger} \beta_p^{\dagger} - \beta_p \beta_{-p} \right) \right)\right] \\
\equiv & \sum_{p=1, 3, \dots} A_p(s)
\end{align}
\ees
Now that $H(s)$ has finally been written as sum of commuting operators ($ \left[ A_p , A_{p'} \right] = 0$ for $p \neq p'$), we can diagonalize each summand separately.  For a given $p$, let us denote by $\ket{\Omega_p}$ the state that is annihilated by $\beta_p$ and $\beta_{-p}$, i.e., $\beta_p \ket{\Omega_p} =  \beta_{-p}  \ket{\Omega_p} = 0$.  Note that $A_p(s=0) \ket{\Omega_p} = 0$, so $\ket{\Omega_p}$ is the ground state of $A_p$ at $s = 0$ (recall that we already knew that the ground state energy at $s = 0$ was zero).  Let  $\ket{\Sigma_p} = \beta_{-p}^{\dagger} \beta_p^{\dagger} \ket{\Omega_p}$.  $A_p(s)$ keeps states in the subspace spanned by $\ket{\Omega_p}$ and $\ket{\Sigma_p}$ in the same subspace; the initial state is in this subspace, so we can restrict our attention to it.  Let us write $A_p(s)$ in the $\{\ket{\Omega_p},\ket{\Sigma_p}\}$ basis:
\beq
A_p(s) = \left( \begin{array}{cc}
s + s \cos \left( \frac{\pi p}{n} \right) & i s \sin \left( \frac{\pi p}{n} \right) \\
- i s \sin  \left( \frac{\pi p}{n} \right) & 4 - 3 s - s \cos  \left( \frac{\pi p}{n} \right)
\end{array} \right) \ .
\eeq
Diagonalizing this, we find for the energies:
\beq
E_p^{\pm}(s) = 2 - s \pm \left[ (2-3s)^2 + 4 s (1-s)( 1- \cos \left( \frac{\pi p}{n} \right) \right]^{1/2}\ .
\eeq
The instantaneous ground state energy of $H(s)$ is thus given by $\sum_{p=1,3,\dots} E_p^{-}(s)$.  The first excited state energy is given by $E_1^+(s) + \sum_{p=3, \dots} E_p^-(s)$.  The energy gap $\Delta(s)$ is therefore given by:
\bea
\Delta(s) &=& E_1^+(s) - E_1^-(s) \\
&=& 2  \left[ (2-3s)^2 + 4 s (1-s)( 1- \cos \left( \frac{\pi p}{n} \right) \right]^{1/2}\ . \nonumber
\eea
The minimum occurs at $s^{\ast}   = \frac{ 2 \left( 2 + \cos \frac{\pi}{n} \right)}{5 + 4 \cos \frac{\pi}{n} } \to 2/3$ as $n\to\infty$.  Therefore, the minimum gap is given by:
\beq
\Delta(s^{\ast}) = 4 \left| \sin \frac{\pi}{n} \right| \frac{1}{\sqrt{5 + 4 \cos \frac{\pi}{n}}} \to \frac{4 \pi}{3 n} \ ,
\eeq
which implies a polynomial run time for the adiabatic algorithm. As mentioned, the classical computational cost of finding a satisfying assignment is at most $n$. Therefore, despite the polynomially small gap in this example, there is no quantum speedup. This illustrates that a StoqAQC slowdown need not necessarily be associated with an exponentially small gap.

\subsubsection{Weighted 2-SAT on a chain with periodicity}
\label{sec:weighted2SAT}
%
We now discuss another problem, proposed in \cite{Reichardt:2004}, that combines 2-SAT with an exponential slowdown of StoqAQC. It can thus be viewed as exhibiting aspects of the two previous problems we discussed.
 
Consider a weighted 2-SAT problem on a chain with ``agree'' clauses between bits $i, {i+1}$ for $i = 1, \dots, N-1$ with weights:
\beq
J_i = \left\{\begin{array}{rl}
w & \text{if $\lceil \frac{i}{n} \rceil$ is odd} \ , \\
1 & \text{if $\lceil \frac{i}{n} \rceil$ is even}
\end{array} \right.
\eeq
where $n$ is the period and $w > 1$.  As for the previous 2-SAT problem, 
we can map this to a spin-chain with ferromagnetic couplings with strength given by $J_i$.  The adiabatic Hamiltonian is given by:
\beq
H(s) = -(1-s) \sum_{i=1}^{N} \sigma_i^x - s \sum_{i=1}^{N-1} J_i \sigma_i^z \sigma_{i+1}^z \ .
\eeq
This chain has coefficients that alternate between $w$ and $1$ in sectors of size $n$ each, with the $b+1$ odd-numbered sectors being ``heavy" ($J_i=w>1$), $b$ even-numbered sectors being ``light" ($J_i=1$), and where the total number of sectors is $(N-1)/n=2b+1$. Since the chain is ferromagnetic, the ground state of $H(1)$ is trivially the all-$0$ or all-$1$ computational-basis state. The problem is thus classically easy and can be solved by inspection or in time $O(N)$ by a heuristic classical algorithm such as simulated annealing, by simply traversing the chain and updating one spin at a time.

Note that at $s = 0$ there is a unique ground state, while as we just noted, at $s = 1$ the ground state is doubly degenerate.  Therefore, the relevant quantum ground state gap $\Delta$ is not the gap to the first excited state (since at the end of the evolution, this merges with the ground state), but to the second excited state. 

It turns out this gap is exponentially small in the sector size $n$ across a constant range $s\in (1/(1+w),1/2)$. Moreover, there are exponentially many (in $\sqrt{N}$) exponentially small excitations above the ground state for $n\sim \sqrt{N}$. 

More precisely, let $\mu_w = sw/(1-s)$. Theorem 4 in \cite{Reichardt:2004} states that:
\begin{enumerate}
\item For any fixed $s>1/(1+w)$, i.e., $\mu_w<1$, $H(s)$ has one eigenvalue only $O(\mu_w^{n})$ above the ground state energy. This means that the gap $\Delta$ is exponentially decreasing with the sector size $n$. 
\item For $s\in (1/(1+w), 1/(1+\sqrt{w})]$ (i.e., again $\mu_w<1$), $H(s)$ has $2^{b+1}-1$ eigenvalues only $O(b \mu_w^{n})$ above the ground state energy. This means that there are exponentially many (in the number of odd sectors $b+1$) excited states, that likewise have an exponentially small (in $n$) gap from the ground state. Note that $b = [(N-1)/n+1]/2$, so $b\sim\sqrt{N}$ when $n\sim\sqrt{N}$. 
\item For $s\in [1/(1+\sqrt{w}),1/2)$, where $\mu_1=s/(1-s)>1$, $H(s)$ has $2^{b+1}-1$ eigenvalues $O(b \mu_1^{-n})$ above the ground state energy. This again means an exponentially large number (in $b$) of excited states with an exponentially small (in $n$) gap.
\end{enumerate}
The proof uses a Jordan-Wigner transformation to diagonalize $H(s)$, similarly to the technique in Sec.~\ref{sec:2SAT-ring}. The spectral gaps of the $2^N \times 2^N$ matrix $H(s)$ are the square roots of the eigenvalues of an $N \times N$ symmetric, tridiagonal matrix. The Sturm sequence of the principal leading minors of this matrix is then analyzed to bound the eigenvalue gaps of $H(s)$.

Why might we expect this problem to be hard for the adiabatic algorithm?  
Within any given light or heavy sector, the problem (at fixed $s$) is that of a uniform transverse field Ising chain.  Consider the thermodynamic limit $N\gg 1$, and also let $n\gg 1$. In this limit the transverse field Ising chain encounters a phase transition separating the disordered phase and the ordered phase when $1-s=s J_i$ \cite{Sachdev:book}, i.e., at $s=1/(1+J_i)$.  (The boundary of the chain only adds $O(1)$ energy, so it does not impact this intuitive argument in the thermodynamic limit.)  This means that the heavy sectors encounter the phase transition at $s = \frac{1}{1+w}$ whereas the light sectors encounter the phase transition at $s = 1/2$, i.e., the light sectors order after the heavy ones.  At $s=\frac{1}{1+w}$ each heavy sector orders in either the all-$0$ or all-$1$ state, and different heavy sectors are separated by light sectors that have not ordered yet. 
Since the initial Hamiltonian generates only local spin flips, the algorithm is likely to get stuck in a local minimum with a domain wall in one or more disordered sectors, if run for less than exponential time in $n$. This mechanism  in which large local regions order before the whole is well-known in 
disordered, geometrically local optimization problems, giving rise to a Griffiths phase \cite{Fisher:1995zl}.

\subsubsection{Topological slowdown in a dimer model or local Ising ladder}

Another interesting example of a local spin model that leads StoqAQC astray was given in \cite{Laumann:2012hs}. They showed that a translation invariant quasi-1D transverse field Ising model with nearest-neighbor interactions only, the ground state of
which is readily found by inspection, results in exponentially long run times for StoqAQC. The model can be understood as either a dimer model on a two-leg ladder of even length $L$, or, using a duality transformation, a two-leg
frustrated Ising ladder of the same length in a uniform magnetic field,  the ground states of which map onto
the dimer states. The frustrated Ising ladder Hamiltonian is:
\beq
H_1 = -\sum_{\langle i,j \rangle} J_{ij}\sigma_i^z \sigma_j^z - K\sum_{i\in\text{ upper}} \sigma_i^z + \frac{1}{2}U \sum_{i\in\text{ lower}} \sigma_i^z\ ,
\eeq
where upper and lower refer to the legs of the ladder, $J_{ij}=-K$ for the upper-leg couplings and $J_{ij}=K$ for all other (lower-leg and rungs) couplings. 

Quantum dynamics is introduced via a standard initial Hamiltonian $-\Gamma(t) \sum_i \sigma_i^x$, where $\Gamma(0)\gg \|H_1\|$ and $\Gamma(t_f)=0$. The dimer model exhibits a first order quantum phase transition with an exponentially small gap when $K\gg U$, which is inherited by the frustrated Ising ladder model. Namely, the system prefers the sector with exponentially many ground states, while
any degeneracy-lifting interaction favors another containing only $O(1)$ states. StoqAQC selects the wrong sector, tunneling out of which becomes exponentially slow as $\Gamma$ is reduced. 

More specifically, for $K\gg U$, the Hilbert space is spanned by an orthonormal basis of hardcore dimer coverings (``perfect matchings") of the ladder. These fall into three sectors which
are topological in that they are not connected by any local rearrangement of the dimers. The sectors are labeled by a winding number $w$, the difference between the number of dimers on the top and bottom rows (on any fixed plaquette). The model assigns extensive energy $\propto L$ to every state in the $w=0$ sector while leaving the two staggered states $w=\pm 1$ as ground states with energy $0$.   On the other hand, at large $\Gamma$ (strong transverse field) the $w=0$ sector is favored.  Intuitively, slowly turning the transverse field off by reducing $\Gamma$ does not help change the topological sector since any off-diagonal term in the dimer Hilbert space involving only a finite number of rungs in the ladder leaves the winding number $w$ invariant.  This is depicted in Fig.~\ref{fig:ladder}.
Numerical analysis of the Ising ladder confirms this picture by revealing that the gap is exponentially small in $L$ when $K>U$. The critical point is found to be at $\Gamma_c \approx U/b + U^2/(4Kb^3)$, where $b\approx 0.6$ (from exact diagonalization numerics).
\begin{figure}[htbp] 
   \centering
   \includegraphics[width=0.95\columnwidth]{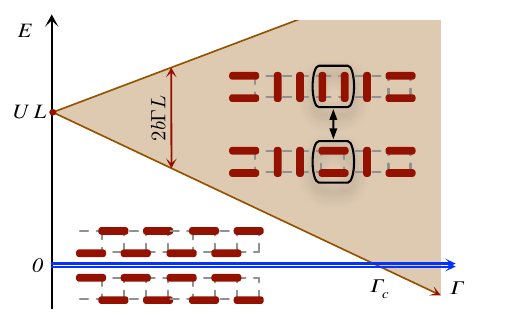} 
   \caption{Energy spectrum of the dimer model on an even length periodic ladder, with the dimer configurations illustrated. The $w = \pm 1$ states are at energy $E = 0$, while the $w = 0$ sector splits into a band for $\Gamma>0$.  For sufficiently large $\Gamma$, the $w = 0$ sector contains the ground state of the system. An unavoided level crossing (first order quantum phase transition) occurs at $\Gamma=\Gamma_c$, which is responsible for the quantum slowdown. From \cite{Laumann:2012hs}.}
   \label{fig:ladder}
\end{figure}

\subsubsection{Ferromagnetic Mean-field Models}
The quantum ferromagnetic $p$-spin model is given by:
\beq
H = -\frac{1}{n^{p-1}}\left( \sum_{i=1}^n \sigma^z_i \right)^p - \Gamma \sum_{i=1}^n \sigma_i^x \ .
\eeq
By inspection it is clear that when $p$ is even, the ground state at $\Gamma = 0$ is either of the two fully-aligned ferromagnetic states, while when $p$ is odd, the unique ground state at $\Gamma = 0$ is the fully-aligned spin-down state.  As $\Gamma$ is tuned from a large value towards zero, the system encounters a first-order phase transition for $p>2$.  This can be readily shown by employing the Suzuki-Trotter decomposition and the static approximation  \cite{Chayes2008,PhysRevB.78.134428,Jorg:2010qa,Suzuki-book} to calculate the partition function $Z$ in the large $n$ limit, where
\beq
Z = \int d m\ e^{-\beta n F(\beta, \Gamma, m)} \ .
\eeq
Here $\beta$ is the inverse temperature, $m$ is the Hubbard-Stratonovich field \cite{Hubbard:1959ix}, and $F$ is the free energy density given by:
\beq
F= (p-1) m^p - \frac{1}{\beta} \log \left[ 2 \cosh \left( \beta \sqrt{\Gamma^2 + p^2 m^{2p-2}} \right) \right] \ .
\eeq
The dominant contribution to $F$ comes from the saddle-point of the partition function $Z$, which provide consistency equations for the field $m$:
\beq
m = p m^{p-1} \frac{\tanh \left( \beta \sqrt{ \Gamma^2 + p^2 m^{2 (p - 1)}} \right)}{ \sqrt{ \Gamma^2 + p^2 m^{2 (p - 1)}}}\ .
\eeq
Solving this equation numerically for $p>2$ reveals a discontinuity in the value of $m$ that minimizes the free energy as $\Gamma$ is tuned through the phase transition point.  At this critical point, the free energy exhibits a degenerate double-well potential, and an instantonic calculation on this potential gives an exponentially small energy gap with system size \cite{Jorg:2010qa}.

\subsubsection{3-Regular 3-XORSAT}
\label{sec:3XORSAT}
All the problems we discussed so far were amenable to a classical solution ``by inspection" (i.e., the solution is obvious from the form of the cost function). Some 
problems were even easy for classical heuristic algorithms performing local search. We now discuss a problem that is non-trivial in this respect, i.e., classically only yields in polynomial time to a tailored approach.

In 3-XORSAT, each clause involves three bits, and there are $M$ clauses and $n$ bits in total.  A clause is satisfied if the sum of the three bits (mod $2$) is a specified value; it can be $0$ or $1$ depending on the clause.  For 3-regular 3-XORSAT, every bit is in exactly three clauses and $M = n$.  This problem is associated with a spin glass phase but is ``glassy without being hard to solve" \cite{Franz:2001,Ricci-Tersenghi1639}: the problem of finding a satisfying assignment can be solved in polynomial time using Gaussian elimination because the problem involves only linear constraints (mod $2$) \cite{haanpaa2006hard}.  

A final Hamiltonian involving $n$ spins can be be written such that each satisfied clause gives energy $0$ and each unsatisfied clause gives energy $1$:
\beq
H_1 = \sum_{c=1}^n \left( \frac{\ident - J_c \sigma^z_{i_1,c} \sigma^z_{i_2,c} \sigma^z_{i_3,c}}{2} \right)\ .
\eeq
Here the index $(i_k, c)$ denotes the three bits associated with clause $c$, and $J_c \in \left\{\pm 1 \right\}$ depending on whether the clause is satisfied if the sum of its bits (mod $2$) is $0$ or $1$.  The StoqAQC Hamiltonian is then given, as usual, by $H(s) = -(1-s) \sum_{i=1}^n \sigma_i^x + s H_1$.
The median minimum gap for random 3-regular 3-XORSAT has numerically been shown to be exponentially small in the system size up to $n=24$ \cite{Jorg:2010a} and $n=40$ \cite{farhi:12} [both using the quantum cavity method \cite{Laumann:2008ys,PhysRevB.78.134428} and QMC simulations], with a first order quantum phase transition at $s = 1/2$. Thus, the numerical evidence suggests that StoqAQC takes exponential time to solve this problem. The same is true for classical heuristic local search algorithms such as WalkSAT \cite{Guidetti:2011}.

We note that since the Hamiltonian gap is not a thermodynamic quantity, one must be careful not to automatically associate a first order quantum phase transition with an exponentially small gap. While the examples presented in this review agree with this rule [for additional examples see \cite{Dusuel:2005,Knysh:06,Jorg:2008,Jorg:2010qa,Bapst:2012}], counterexamples wherein a first order quantum phase transition is associated with a polynomially small gap are known \cite{Cabrera:1987tw,Laumann:2012hs,Tsuda:2013br,Laumann:2015sw}.

\subsubsection{Sherrington-Kirkpatrick and Two-Pattern Gaussian Hopfield Models}

\label{sec:Hopfield}
The Sherrington-Kirkpatrick (SK) model, the prototypical spin glass model, is NP-hard, yet its quantum transverse field Ising model version \cite{Ishii:1985,Usadel:1986,PhysRevB.39.11828,Das:2005aa} exhibits a second order phase transition separating the paramagnetic phase from the spin glass phase \cite{Miller:1993,Ye:1993}. The model is defined via the final Hamiltonian
\beq
H_1 = \sum_{i_1<i_2} J_{i_1i_2} \sigma_{i_1}^z \sigma_{i_2}^z
\eeq 
where the couplings $J_{i_1i_2}$ are zero-mean, independent and identically distributed random variables (e.g., Gaussian, or bimodal, i.e., $J_{i_1i_2}=\pm 1$) and every spin is coupled to every other spin. The adiabatic computation proceeds via 
\beq
H(t) = H_1 - \Gamma(t) \sum_{i=1}^n \sigma_i^x\ ,
\label{eq:H(Gamma)}
\eeq 
where $\Gamma$ is adiabatically reduced to zero.

The polynomial closing of the gap at this phase transition appears promising for AQC. However, a spin glass is dominated by a rough free energy landscape with many local minima forming bottlenecks for classical heuristic local-search algorithms \cite{Parisi:book,Nishimori-book}. 

To gain insight into this phenomenon, and in particular its impact on StoqAQC, \cite{Knysh:2016iq} studied another fully connected model with a vanishing classical gap: the Gaussian Hopfield model, defined generally via 
\beq
J_{i_1i_2\cdots i_p} = \frac{1}{n^{p-1}} \sum_{\mu=1}^r \xi_{i_1}^{(\mu)}\cdots\xi_{i_p}^{(\mu)}
\label{eq:Hopfield}
\eeq 
(Hebb rule), where $\xi_i^{(\mu)}$ are zero-mean i.i.d. random variables of unit variance. By focusing on the analytically more tractable Hopfield model, \cite{Knysh:2016iq}  rigorously analyzed for $r=2$ (the two-pattern case) and $p=2$ (two-local interactions) the properties of local minima away from the global minimum. 

The main insight gained from the theoretical analysis of \cite{Knysh:2016iq} is that the complexity of the model is not determined by the phase transition, but rather by the existence of small-gap bottlenecks in the spin glass phase. Namely, after the occurrence of the polynomially closing gap associated with the second order phase transition separating the paramagnetic and glass phases, there are $O(\log n)$ additional gap minima in the spin glass phase appearing in an approximate geometric progression, a phenomenon that can be attributed to the self-similar properties of the free energy landscape in a $\Gamma$ interval bounded by the appearance of the spin-glass phase.  At these bottlenecks, the gaps scale as a stretched exponential $e^{-c \Gamma_m^{3/4} n^{3/4}}$, where $\Gamma_m$ is the location of the $m$-th minimum.  This is illustrated in Fig.~\ref{fig:Knysh2016}.
Nevertheless this means that StoqAQC suffers a (stretched) exponential slowdown, since the two pattern Gaussian Hopfield model admits an efficient classical solution based on angle sorting and exhaustive search, that scales as $O[n\log(n)]$ \cite{Knysh:2016iq}. Thus, this is another case where a StoqAQC algorithm is too generic to exploit problem structure, and consequently a tailored classical algorithm has exponentially better scaling.

\begin{figure}[htbp] 
   \centering
   \includegraphics[width=0.95\columnwidth]{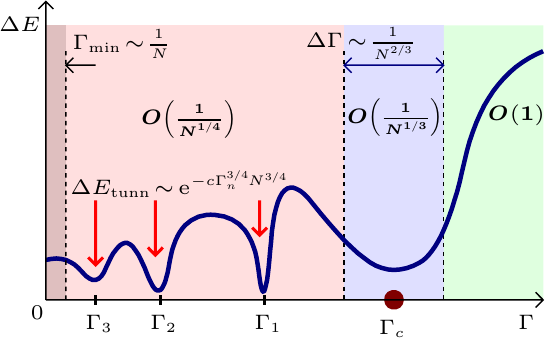} 
   \caption{Illustration showing the gap behavior in the $r =2$, $p=2$ Gaussian Hopfield model.  The paramagetic-spin-glass transition occurs at $\Gamma_c$, with $\Gamma < \Gamma_c$ denoting the spin-glass phase.  The typical gap is denoted using big-O notation.  The spin-glass phase contains $\log n$ additional minima in the gap (indicated by red arrows).  $\Gamma_{\mathrm{min}}$ corresponds to the lowest energy scale of the classical Hamiltonian, which in this case scales as $1/N$, where $N$ was used to represent the variable we denote by $n$.  From \cite{Knysh:2016iq}.}
   \label{fig:Knysh2016}
\end{figure}

\subsection{StoqAQC algorithms with a scaling advantage over simulated annealing}

A substantial effort is underway to develop problems that may exhibit any form of a quantum speedup (recall the classification given in Sec.~\ref{sec:algorithms}).  One approach has been to develop ``tunneling gadgets", i.e., small toy Hamiltonians that exhibit tunneling \cite{Boixo:2016}, and use these gadgets to construct larger problems \cite{PhysRevX.6.031015}.   
An alternative approach has been to develop instances that are believed to exhibit ``small-and-thin" energy barriers in their classical energy landscape  \cite{Katzgraber:2015gf}  in the hope that such barriers persist in the quantum energy landscape where tunneling occurs.
These approaches have been used primarily to assess the performance of the D-Wave devices and are based on numerical analysis, which makes extrapolation and conclusions about asymptotic scaling rather challenging \cite{2016arXiv160401746M,Brady:16}. 

In this subsection we consider several examples of StoqAQC with a demonstrable quantum scaling advantage over simulated annealing (SA).  While none of the examples are demonstrations of an unqualified quantum speedup, these examples are illustrative in that they reveal important qualitative differences between SA, where thermal fluctuations are used to explore the energy landscape, and StoqAQC, where quantum fluctuations are used to explore a different energy landscape. Still, these results are based on a comparison with SA that uses only single-spin updates. SA-like algorithms with cluster-spin updates can be significantly more efficient \cite{Swendsen87,PhysRevLett.62.361,Zintchenko:2015zl,PhysRevLett.115.077201,Houdayer:2001,2016arXiv160401746M},
and their performance relative to StoqAQC is largely an open question. The same is true for parallel tempering (aka exchange Monte Carlo) \cite{swendsen:86,Earl:2005pd,Hukushima:1996,Marinari:1992,katzgraber:06a}.

\subsubsection{Spike-like Perturbed Hamming Weight Problems}
\label{sec:spike}
%
We start with a problem for which there is no (limited) quantum speedup, in order to set up the more interesting problems that follow. Consider a cost function $f(x)$ to be minimized with $x \in \left\{ 0 , 1 \right\}^n$ an $n$-bit string.  The final Hamiltonian can generically be written as:
\beq
H_1 = \sum_{x} f(x) \ketbra{x}{x}\ .
\eeq
We first consider the cost function of the ``plain" Hamming weight problem:
\beq
f(x) = |x|
\eeq
where $|x|$ denotes the Hamming weight of the $n$-bit string $x$ [as in Eq.~\eqref{eq:118}].  This problem is equivalent to a system of $n$ non-interacting spins in a global (longitudinal) field, which is of course a trivial problem that can be solved in time $O(1)$, e.g., by parallelized SA running with a a single thread for each spin. The scaling of the time needed by the quantum algorithm is ${O}(n^{1/2})$, and the full cost of the quantum algorithm is $O(n^{3/2})$ according to Eq.~\eqref{eq:cost}, since it requires $O(n)$ single-qubit terms in the Hamiltonian. 
A fairer comparison is to an SA algorithm that is ignorant of the structure of the problem. In this case one can show that the cost for single-spin update SA with random spin selection is lower bounded by $\mathcal{O}(n \log n)$ \cite{Muthukrishnan:2015ff}. 

Next we consider a more interesting problem, referred to as the ``spike", first studied in~\cite{Farhi-spike-problem}.  The cost function is given by:
\beq
f(x) = \left\{ \begin{array}{ll} 
n & \mathrm{if} \  |x| = n/4 \\
|x| & \mathrm{otherwise}
\end{array} \right.
\eeq
Since the barrier scales with $n$, we can expect that single-spin-update SA will take $\exp(n)$ time to traverse the barrier.  However, it can be shown that the quantum gap scales as $\Omega(n^{-1/2})$ \cite{Farhi-spike-problem,2015arXiv151106991K}, so the adiabatic algorithm only takes polynomial time.  

This type of ``perturbed" Hamming weight problem can be generalized, while still retaining an advantage over single-spin-update SA.  For cost functions of the form
\beq
\label{eq:148}
f(x) = \left\{ \begin{array}{ll} 
|x| + h(n) & \mathrm{if} \  l(n) < |x| < u(n) \\
|x| & \mathrm{otherwise}
\end{array} \right.
\eeq
satisfying $h[(u-l)/\sqrt{l}] = o(1)$, the minimum gap of the quantum algorithm is lower bounded by a constant \cite{Reichardt:2004} [see Appendix~A of Ref.~\cite{Muthukrishnan:2015ff} for a pedagogical review of the proof].  The SA run time, on the other hand, scales exponentially in $\max_n h(n)$.  

Similarly, consider barriers with width proportional to $n^{\alpha}$ and height proportional to $n^{\beta}$, i.e.
\beq
f(x) = \left\{ \begin{array}{ll} 
|x| + n^{\alpha} & \mathrm{if} \  \frac{n}{4} - \frac{1}{2}  n^{\beta} < |x| < \frac{n}{4} + \frac{1}{2} n^{\beta} \\
|x| & \mathrm{otherwise}
\end{array} \right. \ .
\eeq
When $\alpha$ and $\beta$ satisfy $\alpha + \beta \geq 1/2$, $\alpha < 1/2$, and $2 \alpha + \beta \leq 1$, the minimum gap scales polynomially as $n^{1/2 - \alpha - \beta}$ \cite{Brady:2015rc,Brady:2016}, while the SA run time scales exponentially in $n^{\alpha}$.

\subsubsection{Large plateaus}

The above examples have relied on energy barriers in the classical cost that scale with problem size to foil single-spin-update SA. This agrees with the intuition that a StoqAQC advantage over SA is associated with tall and thin barriers \cite{PhysRevB.39.11828,RevModPhys.80.1061}. However, somewhat counterintuitively, it is also possible to foil SA by having very large plateaus in the classical cost function. Specifically, consider:
\beq
f(x) = \left\{ \begin{array}{ll} 
u -1 & \mathrm{if} \  l < |x| < u\\
|x| & \mathrm{otherwise}
\end{array} \right.
\eeq
where $l, u = \mathcal{O}(1)$ [a special case of Eq.~\eqref{eq:148}].  SA with single-spin-updates and random spin selection has run time $O(n^{u-l-1})$, where $u-l-1$ is the plateau width \cite{Muthukrishnan:2015ff}.  This polynomial scaling arises because the energy landscape provides no preferred direction and SA then behaves as a random walker on the plateau.  Numerical diagonalization shows that the quantum minimum gap is constant and the adiabatic run time is only $O(n^{1/2})$, where the scaling with $n$ arises from the numerator of the adiabatic condition  \cite{Muthukrishnan:2015ff}. Thus StoqAQC has a polynomial scaling advantage over SA in this case.

A natural question is whether these potential quantum speedup results for StoqAQC relative to SA survive when other algorithms are considered. The answer is negative.  \cite{Muthukrishnan:2015ff} showed that a diabatic evolution is more efficient than the adiabatic evolution to solve these problems, and a similar efficiency is achieved using classical spin-vector dynamics. There is also a growing body of numerical \cite{2014arXiv1410.8484C,Brady:2015rc,Muthukrishnan:2015ff,PhysRevX.6.031015} and analytical \cite{2015arXiv151008057I,Crosson-Harrow:2016} research that shows that quantum Monte Carlo methods exhibit similar or even identical advantages over SA for many spike-like perturbed Hamming weight problems.

\subsection{StoqAQC algorithms with undetermined speedup}

In this subsection we focus on examples where it is currently unknown whether there is a quantum speedup or slowdown for StoqAQC.

\subsubsection{Number partitioning}
The number partitioning problem is a canonical NP-complete problem \cite{Garey:book} that is defined as follows: given a set of $n$ positive numbers $\left\{ a_i\right\}_{i=1}^n$, the objective is to find a partition $\mP$ of this set that minimizes the partition residue $E$ defined as:
\beq
E = \left| \sum_{j \in \mP}a_j - \sum_{j \notin \mP} a_j\right|\ .
\eeq
The problem exhibits an easy-hard phase transition at the critical value $b/n=1$, where $b$ is the number of bits used to represent the set $\{a_i\}$ \cite{Mertens:1998qd,Borgs:2001rc}. In the hard phase it roughly corresponds to finding the minimum in a set of $2^n$ numbers \cite{Mertens:2001qq}. To translate it into Ising spin variables let $s_j = 1$ when $j \in \mP$ and $s_j = -1$ otherwise, so that
\beq
E = \left| \sum_{j=1}^n a_j s_j \right|\ ,
\eeq
which can then be turned into a Mattis-like Ising Hamiltonian whose ground state is the minimizing partition:
\beq
H_1 = \sum_{i,j=1}^n a_i a_j s_i s_j\ .
\label{eq:H-number-part}
\eeq
The energy landscape of this final Hamiltonian is known to be extremely rugged in the hard phase \cite{Smelyanskiy:2002,Stadler:2003}, and the asymptotic behavior can already be seen for small sizes $n$.  While SA effectively requires the searching of all possible bit configurations with a run time $\propto 2^{0.98n}$ \cite{Stadler:2003}, numerical simulations of StoqAQC exhibit a slightly better run time $\propto 2^{0.8n}$ \cite{PhysRevX.6.031015}.  State-of-the-art classical algorithms have scalings as low as $2^{0.291 n}$ \cite{Becker2011}. 

It should be noted that number partitioning is known as the ``easiest hard problem" \cite{Hayes:02} due to the existence of efficient approximation algorithms that apply in most (though of course not all) cases, e.g., a polynomial time approximation algorithm known as the differencing method \cite{Karmarkar:1982la}. It should further be noted that if all the $a_j$'s are bounded by a polynomial in $n$, then
integer partitioning can be solved in polynomial time by dynamic programming
\cite{Mertens:2003aa}. The NP-hardness of the number partitioning problem
requires input numbers of size exponentially large in $n$ or, after division by the maximal input number, of exponentially high precision. This is problematic since the $\{a_j\}$ are used as coupling coefficients
in the adiabatic Hamiltonian~\eqref{eq:H-number-part}, and suggests that a different encoding will be needed in order to allow AQC to meaningfully address number partitioning.

\subsubsection{Exact Cover and its generalizations}
\label{sec:exact-cover}
We briefly review the adiabatic algorithm for Exact Cover, which initiated and sparked the tremendous interest in the power of AQC when it was first studied in \cite{Farhi:01}. While the optimistic claim made in that paper, that ``the quantum adiabatic algorithm worked well, providing evidence that quantum computers (if large ones can be built) may be able to outperform ordinary computers on hard sets of instances of NP-complete problems" turned out to be premature, the historical impact of this study was large, and it led to the avalanche of work that forms the core of this review.

The Exact Cover $3$ (EC3) problem is an NP-complete problem that is a particular formulation of 3-SAT [recall Sec.~\ref{sec:SAT-background}] whereby each clause $C$ (composed of three bits $x_{C_1}, x_{C_2}, x_{C_3}$ that are taken from the set of variables $\left\{ x_i\in\{0,1\} \right\}_{i=1}^n$) is satisfied if $x_{C_1} + x_{C_2} + x_{C_3} = 1$.  There are only three satisfying assignments: (1,0,0), (0,1,0), and (0,0,1).  A $3$-local Hamiltonian $H_{C}$ can be associated with each clause, that assigns an energy penalty to the unsatisfying assignments \cite{Farhi:01,Latorre:04}:
\begin{eqnarray}
H_C &=& \frac{1}{8} \left[ \left(1 + \sigma_{C_1}^z \right) \left(1 + \sigma_{C_2}^z \right) \left(1 + \sigma_{C_3}^z \right) \right. \nonumber \\
&& + \left(1 - \sigma_{C_1}^z \right) \left(1 - \sigma_{C_2}^z \right) \left(1 - \sigma_{C_3}^z \right) \nonumber \\
&& + \left(1 - \sigma_{C_1}^z \right) \left(1 - \sigma_{C_2}^z \right) \left(1 + \sigma_{C_3}^z \right) \nonumber \\
&& + \left(1 - \sigma_{C_1}^z \right) \left(1 + \sigma_{C_2}^z \right) \left(1 - \sigma_{C_3}^z \right)  \nonumber \\
&& \left. + \left(1 + \sigma_{C_1}^z \right) \left(1 - \sigma_{C_2}^z \right) \left(1 - \sigma_{C_3}^z \right) \right]\ .
\label{eq:EC-1}
\end{eqnarray}
A $2$-local alternative is \cite{Young:2010}:
\beq
H_{C} = \frac{1}{4} \left( \sigma_{C_1}^z + \sigma_{C_2}^z + \sigma_{C_3}^z - 1 \right)^2\ .
\label{eq:EC-2}
\eeq
The final Hamiltonian is then given by $H_1 = \sum_{C} H_{\mathrm{C}}$.  If the ground state energy is $0$, then an assignment exists that satisfies all clauses.  The adiabatic algorithm is given as usual by $H(s) = (1-s) H_0 + s H_1$, with $H_0 = \sum_{i=1}^n \frac{1}{2} \left( 1 - \sigma_i^x \right)$.  

For instances with a unique satisfying assignment, while the initial (small $n$) scaling of the typical minimum gap (median) is consistent with polynomial \cite{Farhi:01,Latorre:04}, the true (large $n$) scaling is exponential and can be associated with a first order phase transition \cite{Young2008,Young:2010} occurring at intermediate $s=s_c>0$.  The fraction of instances with this behavior increases with increasing problem size \cite{Young:2010}. This illustrates the perils of extrapolating the asymptotic scaling from studies based on small problem sizes.

A natural generalization of the Exact Cover problem is to have the sum of $K$ variables sum to 1 for a clause to be satisfied, which defines the problem known as ``1-in-$K$ SAT".  Another is to have the clause satisfied unless all the variables are equal, which defines the problem ``$K$-Not-All-Equal-SAT".  Both of these are NP-complete and have been shown analytically to exhibit a first order phase transition for sufficiently large $K$ \cite{Smelyanskiy:2004}.  Numerical results for locked 1-in-3 SAT and locked 1-in-4 SAT --- where ``locked' is the additional requirement that every variable is in at least two clauses and that one cannot get from one satisfying assignment to another by flipping a single variable \cite{PhysRevLett.101.078702,1742-5468-2008-12-P12004} --- have been shown to exhibit an exponentially small gap at the satisfiability transition \cite{hen:11}. 

Since all these problems are NP-complete, there is no polynomial-time classical algorithm known for their worst-case instances. Using StoqAQC has, in all cases that have been studied to date, resulted in exponentially small gaps. Thus, whether these problems can be sped up (even polynomially) is at this time still an open problem.

\subsubsection{3-Regular MAXCUT}

For 3-regular MAXCUT, the problem is to find the assignment that gives the maximum number of satisfied clauses, where each bit appears in exactly three clauses.  Each clause involves only two bits and is satisfied if and only if the sum of the two bits (mod $2$) is $1$.  The number of clauses is $M = 3n/2$.  The final Hamiltonian can be written as:
\beq
H_1 = \sum_{c=1}^{3n/2} \left( \frac{\ident + \sigma^z_{i_1,c} \sigma^z_{i_2,c} }{2} \right)\ ,
\eeq
where the index $(i_k, c)$ denotes the two bits associated with clause $c$. This model can also be viewed as an antiferromagnet on a 3-regular random graph. Because the random graph in general has loops of odd length, it is not possible to satisfy all of the
clauses. This problem is NP hard.
\begin{figure}[b] 
   \centering
   \includegraphics[width=0.95\columnwidth]{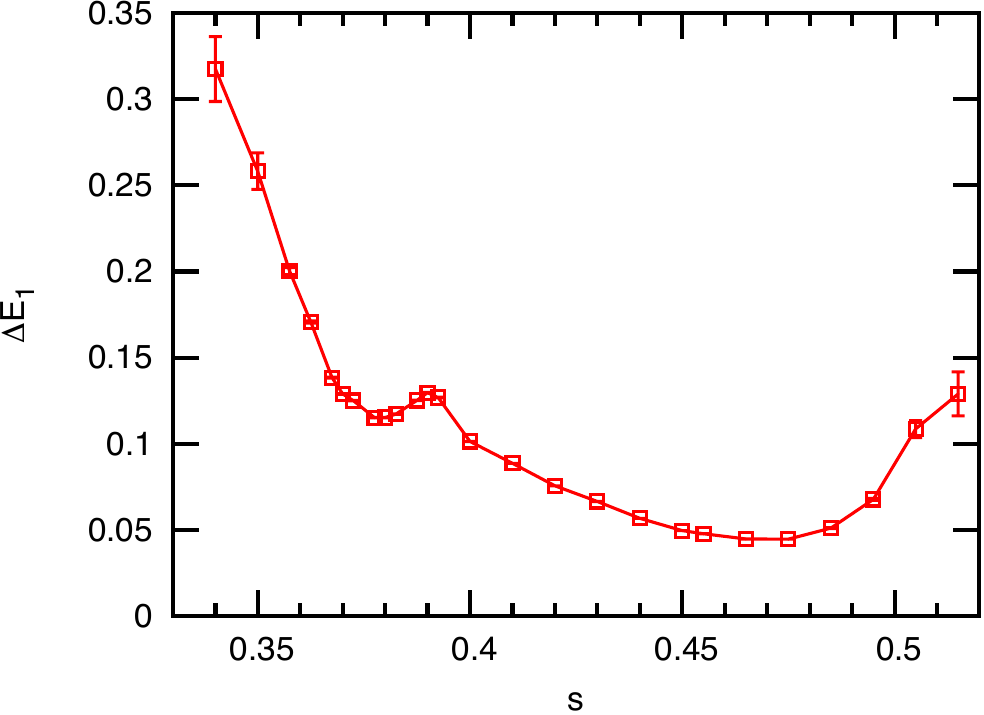} 
   \caption{The gap to the first even excited state for an instance of size $n=128$, exhibiting two minima.  The lower minimum occurs well within the spin-glass phase, while the higher minimum is associated with the second order phase transition.  From \cite{farhi:12}.}
   \label{fig:Farhi2012Fig11}
\end{figure}

For random instances of this problem, where there is a doubly degenerate ground state (the smallest possible because of the $Z_2$ symmetry) and with a specified energy of $n/8$, the standard adiabatic Hamiltonian $H(s) = -(1-s) \sum_{i=1}^n \sigma_i^x + s H_1$ exhibits, for sufficiently large sizes of up to $n=160$, two minima in the energy gap \cite{farhi:12} (see Fig.~\ref{fig:Farhi2012Fig11} for an example).  The first minimum, at $s\approx 0.36$, is associated with a second-order phase transition from paramagnetic to glassy, and the gap closes polynomially with system size. The second minimum occurs inside the spin-glass phase, with a gap that closes exponentially (or possibly a stretched exponential).  Therefore, while the first minimum does not pose a problem for the adiabatic algorithm (although it has been shown that the quantum algorithm with a linear interpolating schedule does not pass through the associated glass phase transition faster than SA \cite{Liu:2015}), the second minimum implies an exponential run time.

\subsubsection{Ramsey numbers}

An adiabatic algorithm for the calculation of the Ramsey numbers $R(k,l)$ was proposed in \cite{Gaitan:2012qz}. $R(k,l)$ is the smallest integer $r$ such that every graph on $r$ or more vertices contains either a $k$-clique or an $l$-independent set.\footnote{A $k$-clique is a subset of $k$ vertices such that every two distinct vertices are adjacent. Equivalently, the subgraph induced by the clique is a complete graph. An independent set is a subset of the vertices no two of which are adjacent. $R(k,l)$ can be phrased as the ``party problem": What is the smallest number of guests one can invite to a party such that there is always either a group of $k$ guests that all know each other, or a group of $l$ guests, none of whom know each other? Such a threshold number always exists \cite{Ramsey:1930yf}.}  Computing them by brute force is doubly exponential in $N=\max\{k,l\}$ [note that $R(k,l)=R(l,k)$] using graph coloring techniques, as follows: Try every one of the $2^{N(N-1)/2}$ colorings of the edges of the complete graph $K_N$ with the colors blue and red. For every coloring, check whether or not there is an induced subgraph on $k$ vertices with only blue edges, or an induced subgraph on $l$ vertices with only red edges. If every coloring contains at least one of the desired subgraphs, we are done. Otherwise, increment $N$ by $1$ and repeat. Except for certain special values of $k$ and $l$, no better algorithm is currently known.

The idea in \cite{Gaitan:2012qz} is to construct a cost function $h(G)$ for a graph $G$ where
\beq
h(G) = \mathcal{C}(G) + \mathcal{I}(G)
\eeq
where $\mathcal{C}(G)$ counts the number of $m$-cliques in the graph $G$ and $\mathcal{I}(G)$ counts the number of $l$-independent sets in the graph $G$.  The cost $h(G)$ equals zero only if there does not exist an $k$-clique or an $l$-independent set.  This will only occur if $R(k,l) > N$.  The algorithm then proceeds as follows.  By mapping $h(G)$ over $K_N$ to a final Hamiltonian $H_1$, the adiabatic algorithm $H(s) = -(1-s) \sum_{i=1}^n \sigma_i^x + s H_1$ is performed and the final energy of the state is measured.  If $h(G) = 0$, then $N$ is incremented by $1$ and the experiment is repeated.  This process continues until the first occurrence of $h(G)>0$, in which case $N = R(k,l)$. Thus the algorithm is essentially an adiabatic version of the graph coloring method described above. It is unknown whether its StoqAQC version improves upon the classical brute force $2^{N(N-1)/2}$ scaling. The adiabatic quantum algorithm was simulated in \cite{Gaitan:2012qz} and shown to correctly determines the Ramsey numbers $R(3,3)$ and $R(2,s)$ for $5\leq s \leq 7$. It was also shown there that Ramsey number computation is in QMA.

An adiabatic algorithm for generalized Ramsey numbers 
(where the induced subgraphs are trees rather than complete graphs) 
was presented in \cite{Ranjbar:2016ty}. Whether this results in a quantum speedup is also unknown. We also remark that Ising formulations for many NP-complete and NP-hard problems, including all of Karp's $21$ NP-complete problems \cite{Karp:21-problems}, are known \cite{2013arXiv1302.5843L}, but it is unknown whether they are amenable to a quantum speedup.

\subsubsection{Finding largest cliques in random graphs}
\label{sec:cliques}

The fastest algorithm known to date for the NP-hard problem of finding a largest clique in a graph runs in time $O(2^{0.249n})$ for a graph with $n$ vertices \cite{Robson:clique}.\footnote{As stated, this is actually an algorithm for the complementary maximum independent set (MIS) problem, but this is sufficient since $\text{MIS}(G) = \text{max-clique}(\bar{G})$ for any graph $G$ and its complement $\bar{G}$, and the algorithm applies for arbitrary $G$.}  For random graphs, a super-polynomial time is required to
find cliques larger than $\log n$ using the Metropolis algorithm, while the maximum clique is likely to be of size very close to $2 \log n$ \cite{Jerrum:1992ve}. One of the earliest papers on the quantum adiabatic algorithm was concerned with the largest clique problem for random
graphs \cite{Childs:clique}, though the algorithm presented there works for general graphs. The results were numerical and showed, by fixing the desired success probability, that the median time required by the
adiabatic algorithm to find the largest clique in a random graph are consistent with quadratic growth for
graphs of up to $18$ vertices. These results on small graphs probably do not capture the asymptotic behavior of the algorithm (the coefficients grow rapidly and have alternating sign), which is likely to be dominated by exponentially small gaps [however, to the extent that these are due to perturbative crossings, they can be avoided by techniques we discuss in Sec.~\ref{sec:PerturbativeCrossings}, in particular as related to the maximum independent set problem \cite{Choi15022011}].

\subsubsection{Graph isomorphism}

In the graph isomorphism problem, two $N$-vertex graphs $G$ and $G'$ are given, and the task is to determine whether there exists a permutation of the vertices of $G$ such that it preserves the adjacency and transforms $G$ to $G'$, in which case the graphs are said to be isomorphic.  If and only if the graphs are isomorphic does there exist a permutation matrix $\sigma$ that satisfies
\beq
A' = \sigma A \sigma^T \ ,
\eeq
where $A$ and $A'$ are the adjacency matrices of $G$ and $G'$ respectively.
An adiabatic algorithm to determine whether a pair of graphs are isomorphic was first proposed in \cite{hen:12b}, mostly for strongly regular graphs, and generalized to arbitrary graphs in \cite{Gaitan:2014ix}, which also showed how to determine
the permutation(s) that connect a pair of isomorphic graphs, and the automorphism group of a given graph. The final Hamiltonian formulated in \cite{Gaitan:2014ix} is such that when the ground state energy vanishes the graphs are isomorphic and the bit-string $s = (s_0, \dots, s_{N-1})$ associated with the ground state gives an $N \times N$ permutation matrix $\sigma(s)$ to perform the transformation:
\beq
\sigma(s)_{ij} = \left\{ \begin{array}{lr}
0  & \mathrm{if} \ s_j > N-1 \\
\delta_{i,s_j} & \mathrm{if} \ 0 \leq s_j \leq N-1
\end{array} \right. \ .
\eeq
The computational complexity of these adiabatic algorithm is currently unknown. However, a recent breakthrough gave a quasipolynomial ($\exp[(\log n)^{O(1)}]$) time classical algorithm for graph isomorphism \cite{Babai:graph-isomorphism}. It seems unlikely that this can be improved upon by using StoqAQC without deeply exploiting problem structure.

\subsubsection{Machine learning}

Quantum machine learning is currently an exciting and rapidly moving frontier in the context of the circuit model \cite{Rebentrost:2014uq,Lloyd:2014fk,Wiebe:2014eu}, though it must be evaluated carefully \cite{Aaronson:2015lh}. One StoqAQC approach is to find a quantum version of the classical method of boosting, wherein multiple weak classifiers (or features) are combined to create a single strong classifier \cite{Meir,Freund}. The task is to find the optimal set of weights of the weak classifiers so as to minimize the training error of the strong classifier on a training data set. After this training step, the strong classifier is then applied to a test data set. This optimization problem can be mapped to a quadratic unconstrained binary optimization (QUBO) problem, which can then be trivially turned into an Ising spin Hamiltonian suitable for adiabatic quantum optimization, where the binary variables represent the weights. 
This idea was implemented in \cite{Neven1,Neven,Neven2,Denchev:2012hc,Pudenz:2013kx,Babbush:2014ij}, where the ground states found by the adiabatic algorithm encode the solution for the weights. 

Another idea is to learn the weights of a Boltzmann machine or, after the introduction of a hidden layer, a reduced Boltzmann machine \cite{Hinton:2006fv}. The latter forms the basis for various modern methods of deep learning. StoqAQC approaches for this problem were developed in \cite{Adachi:2015qe,Amin:2016,Benedetti:2016bs}

Neither the classical nor the quantum computational complexity is known in this case, but scaling of the solution time with problem size is not the only relevant criterion: classification accuracy on the test data set is clearly another crucial metric. It is possible, though at this point entirely speculative, that the quantum method will lead to better classification performance. This can come about in the case of ground state degeneracy, if the weights are reconstructed via ground state solutions and if quantum and classical heuristics for solving the QUBO problem  find different ground states \cite{Matsuda:2009uq,Azinovic:2016kx,2016arXiv160607146M,Zhang:2017}.

\subsection{Speedup mechanisms?}

While the universality of AQC suggests that similar speedup mechanisms are at play as in the circuit model of quantum computing, the situation is less clear regarding StoqAQC. Here we discuss two potential mechanisms, tunneling and entanglement, that might be thought to endow StoqAQC with an advantage over classical algorithms. 

\subsubsection{The role of tunneling}
\label{sec:role-of-tunneling}
It is often stated that an advantage of StoqAQC over classical heuristic local-search algorithms is that the quantum system has the ability to tunnel through energy barriers, which can provide an advantage over classical algorithms such as simulated annealing that only allow probabilistic hopping over the same barriers. Indeed, such a qualitative picture motivated some of the early research on quantum annealing [e.g., \cite{finnila_quantum_1994}]. However, this statement requires a careful interpretation as it has the potential to be misleading. Whereas only the final cost function --- which generates the energy landscape that the classical random walker explores --- matters for the classical algorithm, this energy landscape does not become relevant for the quantum evolution until the end. Therefore, tunneling does not occur on the energy landscape defined by the final cost function alone, if it occurs at all. A different notion of tunneling is at work, which we now explain.

\begin{figure*}[ht] 
   \centering
   \subfigure[]{\includegraphics[width=2in]{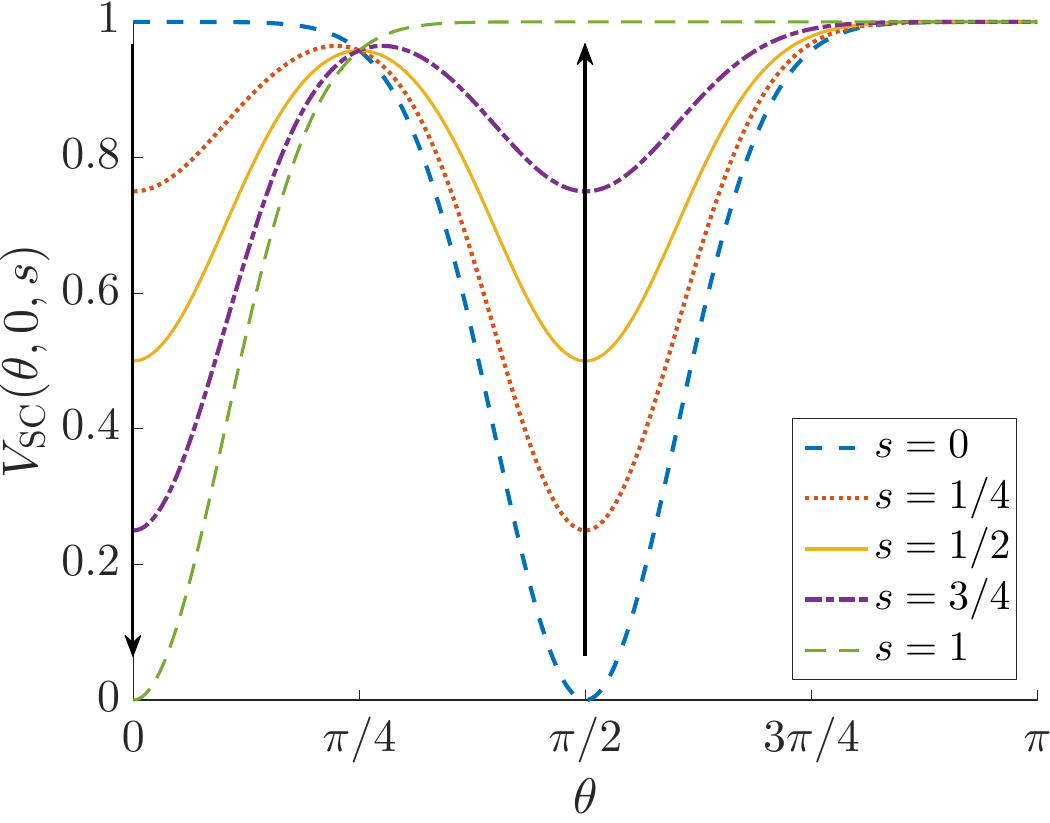} }
   \subfigure[]{\includegraphics[width=2in]{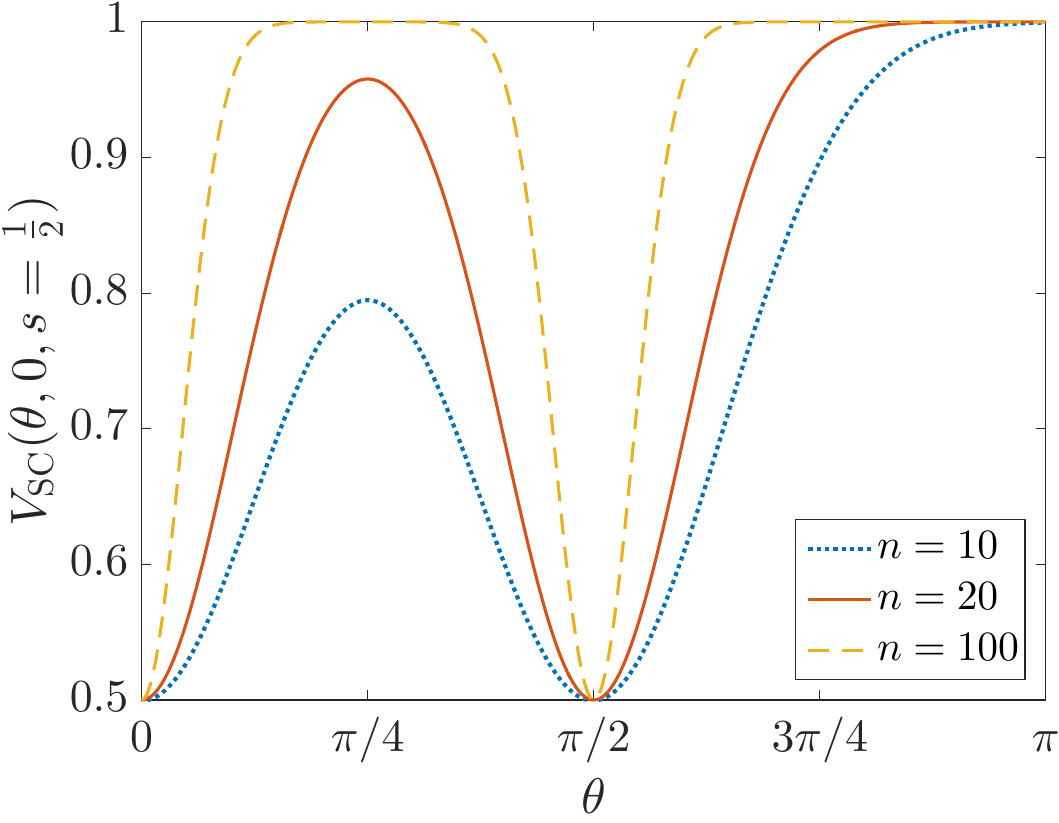} }
   \subfigure[]{\includegraphics[width=2in]{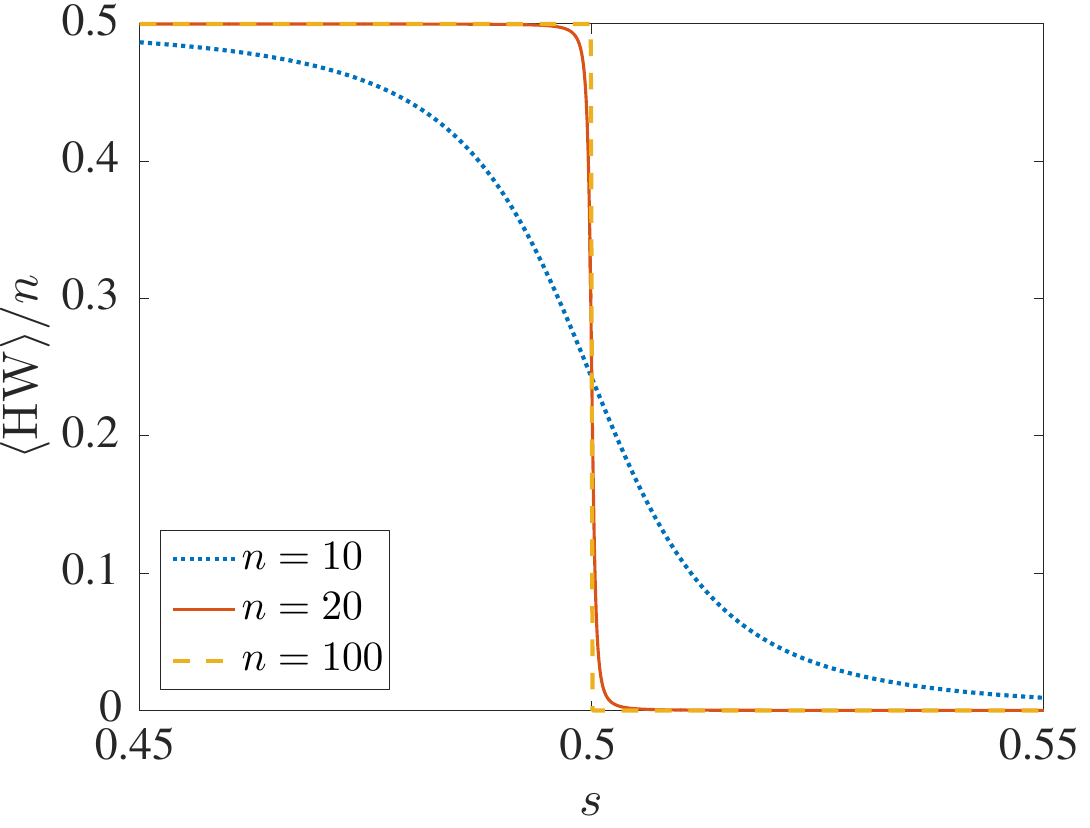} \label{fig:GroverHW}}
   \caption{Analysis of tunneling in the Grover problem. (a) The semiclassical potential for $n=20$ at different dimensionless times $s$.  The arrows indicate the behavior of the local minima as $s$ increases.  There is a discrete jump in the position of the global minimum at $s=1/2$, where it changes from being at $\theta \approx \pi/2$ to $\theta \approx 0$, corresponding to a first order quantum phase transition.  (b) The behavior of the potential when the two minima are degenerate at $s=1/2$.  As $n$ grows, both the barrier height grows (and saturates at $1$) and the curvature of the local minima grows.  (c) The expectation value of the Hamming Weight operator [defined in Eq.~\ref{eqt:HW}] of the instantaneous ground state as $n$ grows.  This is to be interpreted as the system requiring $O(n)$ spins to tunnel in order to follow the instantaneous ground state as the system crosses the minimum gap at $s=1/2$.}
   \label{fig:Grover}
\end{figure*}

The standard notion of tunneling from single-particle quantum mechanics involves a semiclassical potential where classically allowed and classically forbidden regions can be defined.  Starting from a many-body Hamiltonian, there is no unique way to take the semiclassical limit.  Consider one such limit, based on the spin-coherent path integral formalism \cite{klauder1979path}: 
\beq
\!\! \bra{\Omega(t_f)} \Texp[-i \int_0^{t_f} d \tau H(\tau)] \ket{\Omega(0)} = \int \mathcal{D} \Omega(t) e^{\frac{i}{\hbar} S[\Omega(t)]}\ ,
\eeq
where the action $S[\Omega(t)]$ is given by:
\beq
\!\! S[\Omega(t)] = \int_0^{t_f} d t \left( i \hbar \bra{\Omega(t)} \partial_t \ket{\Omega(t)} - \bra{\Omega(t)} H(t) \ket{\Omega(t)} \right)\ ,
\eeq
and 
\begin{align}
\label{eqt:SpinCoherent}
\ket{\Omega} &\equiv \ket{\theta, \varphi}  \\
&\equiv \otimes_{j=1}^n \left[ \cos (\theta_j/2) \ket{0}_j + e^{i \varphi_j} \sin (\theta_j/2) \ket{1}_j \right]\ . \notag
\end{align}
is the spin-coherent state \cite{arecchi1972atomic}.

Despite the absence of a true kinetic term, we can identify the semiclassical potential as:
\beq \label{eqt:VSC}
V_{\mathrm{SC}}(\{\theta_j\}, \{\varphi_j\}, t) =  \bra{\Omega} H(t) \ket{\Omega} 
\eeq
This form for $V_{\mathrm{SC}}$ has been used \cite{Farhi-spike-problem,Schaller:2010,Boixo:2016,Muthukrishnan:2015ff} to capture many of the relevant features of StoqAQC problems endowed with qubit-permutation symmetry; this symmetry often allows for analytical and numerical progress.\footnote{Note that by using a product-state ansatz via the symmetric spin-coherent state, the semiclassical approach  implicitly
takes advantage of the bit-symmetry of the problem. This is
inaccessible to an algorithm that has only black-box access to $f$, thus limiting the generality of this approach.}

We illustrate this approach with the Grover Hamiltonian [Eq.~\eqref{eqt:GroverH}].  Recall that the final Hamiltonian is $H_1 = \ident - \ketbra{m}{m}$, where $\ket{m}$ is the marked state associated with the marked item. As a cost function, this is the antithesis of the ``tall and narrow" potential that is often associated with a classical speedup: $\bra{x}H_1\ket{x} = 1-\delta_{x,m}$, i.e., the potential is flat everywhere, except for a well of constant depth at the marked state. Nevertheless, we now show that following the instantaneous ground state will involve the tunneling of $O(n)$ qubits.

Without loss of generality we may assume that the ``marked" state is the all-zero bit string.
%
%
Setting $\theta_j\equiv\theta$ and $\varphi_j\equiv\varphi$ $\forall j$ in Eq.~\eqref{eqt:SpinCoherent}, the Hamiltonian can be written succinctly as:
\bea
H(s) &=& (1-s) \left( \ident - \ketbra{\Omega(\pi/2, 0)}{\Omega(\pi/2, 0)} \right) \nonumber \\
&&+ s \left(\ident - \ketbra{\Omega(0, 0)}{\Omega(0, 0)} \right)\ .
\eea
The semiclassical potential for the Grover problem is then:
\bea \label{eqt:GroverSC}
V_{\mathrm{SC}}(\theta,0) &=& \left(1-s \right) \left( 1 - \frac{1}{2^n}  \left( 1 + \sin \theta  \right)^n \right) \nonumber \\
&& + s \left( 1- \frac{1}{2^n} \left(1 +  \cos\theta \right)^n \right) \ .
\eea
The locations of the two degenerate minima at $s=1/2$ are given by the pair of transcendental equation:
\bes
\begin{align}
\frac{1-\cos \theta + \sin \theta}{1+\cos \theta - \sin \theta} & = \left( \frac{1+\sin\theta}{1+\cos\theta} \right)^n \ , \\
\frac{1+\cos \theta - \sin \theta}{1-\cos \theta + \sin \theta} & = \left( \frac{1+\cos\theta}{1+\sin\theta} \right)^n \ ,
\end{align}
\ees
which in the limit of $n \to \infty$ have solutions $0$ and $\pi/2$ respectively.  This equation is invariant under $\theta \to \pi/2 -\theta$, which corresponds to the two minima.  Since the semiclassical potential in Eq.~\eqref{eqt:GroverSC} at $s=1/2$ is also invariant under $\theta \to \pi/2 -\theta$, the local minima have identical structure.  Using the Hamming Weight operator defined as:
\beq \label{eqt:HW}
\mathrm{HW} = \frac{1}{2} \sum_{i=1}^n \left( \ident - \sigma_i^z \right)
\eeq
this potential suggests that in the large $n$ limit, we can expect that $n/2$ spins need to be flipped in order to move from the $\theta \approx \pi/2$ minimum to the $\theta \approx 0$ minimum, i.e.,
\beq
\bra{\Omega(\pi/2,0)} \mathrm{HW} \ket{\Omega(\pi/2,0)} - \bra{\Omega(0,0)} \mathrm{HW} \ket{\Omega(0,0)} = n/2\ .
\eeq
The instantaneous ground state, as it passes through the minimum gap at $s=1/2$, indeed exhibits this behavior, as shown in Fig.~\ref{fig:Grover}.  

However, the more general role of tunneling in providing quantum speedups is not by any means evident. This topic was studied in detail in \cite{Muthukrishnan:2015ff}, which showed that tunneling is neither necessary nor sufficient for speedups in the class of perturbed Hamming weight optimization problems with qubit permutation symmetry. 

Our discussion here has been restricted to coherent tunneling, and compelling arguments have been presented in \cite{Boixo:2016,PhysRevX.6.031015,Andriyash:2017} that incoherent, thermally assisted tunneling plays a computational role in quantum annealing. However, this mechanism is in the open-system setting, which is outside the scope of this review.  
Moreover, its role in \cite{Boixo:2016,PhysRevX.6.031015} is limited to a prefactor, and does not translate into a scaling advantage, i.e., it does not qualify as a speedup according to the classification of \cite{speedup}. 

\subsubsection{The role of entanglement}
\label{sec:ent-role}

The role that entanglement plays in quantum computation with pure states in the circuit model depends on the entanglement measure used. On the one hand, it is well known that for any circuit-model quantum algorithm operating on pure states, the presence of multi-partite entanglement quantified via the Schmidt-rank (with a number of parties that increases unboundedly with input size), is necessary if the quantum algorithm is to offer an exponential speedup over classical computation \cite{Jozsa:2003fk}. On the other hand, universal quantum computation can be achieved in the standard pure-state circuit model
while the entanglement entropy (or any other suitably continuous entanglement
measure) of every bipartition is small in each step of the computation \cite{Van-den-Nest:2013aa}.
The corresponding role of entanglement in the computational efficiency of AQC remains an open question. Partly this is because the connection between entanglement and spectral gaps is not yet very well understood, and partly this is because even if entanglement is present, its computational role in AQC is unclear. 

{The area law asserts that for any subset $S$ of particles, the entanglement entropy between $S$ and its complement is bounded by the surface area of $S$ rather than the trivial bound of the volume of $S$. While generic quantum states do not obey an area law \cite{Hayden:2006aa}, and there are 1D systems for which there is exponentially more entanglement than suggested by the area law \cite{Movassagh:2016},} a sweeping conjecture in condensed matter physics is that in a \emph{gapped} system the entanglement spreads only over a finite length, which leads to  area laws for the entanglement entropy \cite{Eisert:2010ao}.\footnote{Here ``gapped" means $O(1)$, whereas in AQC ``gapped" usually means $O[1/\text{poly}(n)]$.} E.g., the area law for {gapped} 1D systems, proved in \cite{Hastings:2007lr}, states that for the ground state, the entanglement of any interval is upper bounded by a constant independent of the size of the interval.  While this leaves open the question of the general dependence of the upper bound on the spectral gap $\Delta$, this means that the ground state of such systems is accurately described by polynomial-size matrix product states (MPSs) \cite{White:1992,White:1992b,Ostlund:1995}. In \cite{GottesmanHastings:09} it was shown that for certain 1D system the entanglement entropy in some regions can be as high as $\text{poly}(1/\Delta)$. This demonstrates that the entanglement entropy can become large as the gap becomes small. Two other important recent results are the existence of a polynomial time algorithm for the ground state of 1D gapped local Hamiltonians with constant ground-state energy \cite{Landau:2015fu,Huang:2014aa}, and the fact that 1D quantum many-body states satisfying exponential decay of correlations always fulfill an area law \cite{Brandao:2013kl}.  

However, the connection between entanglement entropy and gaps is not nearly as clear in higher dimensional systems, even though entanglement close to quantum phase transitions is a well developed subject \cite{Osborne:02,Osterloh:02,Vidal:03,WuSarandyLidar:04,Amico:2008uq}.

It is not surprising that entanglement is necessary for the computation to succeed if the intermediate ground states that the system must follow are entangled.  This was verified explicitly in \cite{Bauer:15}, where the quantum state was represented by an MPS and projected entangled-pair states (PEPS) \cite{Verstraete:2004,Verstraete:2008}. This work showed that the probability of finding the ground state of an Ising spin glass on either a planar or non-planar two-dimensional graph increases with the amount of entanglement in the MPS state or PEPS state.  Furthermore, even a small amount of entanglement gives improved success probability over a mean-field model. However, this does not resolve the role entanglement plays in generating a speedup.

In an attempt to address this, the entanglement entropy for the adiabatic Grover algorithm was studied, and it was found to be bounded ($\leq 1$) throughout the evolution \cite{Orus:2004oz}. This was also observed numerically for systems with $10$ qubits \cite{WEN:2008vl}. In an effort to check whether more entanglement may help the Grover speedup, \cite{WEN:2009fu} considered adding an additional term to the Hamiltonian to make the ground state more entangled, to reach an $O(1)$ scaling in a Grover search task.  However, since it is impossible to achieve a better-than-quadratic speedup in the Grover search problem without introducing an explicit dependence on the marked state \cite{Bennett:1997lh}, this result is not conclusive in linking entanglement with enhanced computational efficiency.  Furthermore, a two-dimensional path for the Grover problem using the quantum adiabatic brachistochrone approach [see Sec.~\ref{sec:QAB}] that gives a higher success probability for the same evolution time relative to the standard one-dimensional path for the Grover problem, in fact has \emph{less} entanglement (negativity) \cite{PhysRevLett.103.080502}. 

The entanglement entropy in the adiabatic algorithm for the Exact Cover problem, where no speedup is known (recall Sec.~\ref{sec:exact-cover}), scales linearly with problem size for $n\leq 20$ \cite{Orus:2004oz,Latorre:04}. 

Further studies have also shown this lack of correlation between performance and the amount of entanglement entropy.  In \cite{Hauke:2015cr} simulations of adiabatic quantum optimization were performed of a trapped ion Hamiltonian with $n=16$ of the form:
\beq
H_1 = J \sum_{i \neq j}^n \frac{\sigma_i^z \sigma_j^z}{|i-j|} + \sum_i h_i^z \sigma_i^z + V \sum_{i \neq j}^n \sigma_i^z \sigma_j^z\ ,
\eeq
with $100$ disorder realizations of $h_i^z$. It was found that a large entanglement entropy has little significance for the success probability of the optimization task.  

Overall, these results indicate that the connection between entanglement and algorithmic efficiency in AQC is currently wide open and deserves further study.


\section{Circumventing slowdown mechanisms for AQC}
\label{sec:fix-AQC}

In this section we collect several insights into mechanisms that explain slowdowns in the performance of adiabatic algorithms. We also discuss mechanisms for circumventing such slowdowns. Several important ideas will be reviewed: avoiding the use certain initial and final Hamiltonians, modifying the adiabatic schedule, avoiding quantum phase transitions, and avoiding perturbative energy level crossings.

\subsection{Avoiding poor choices for the initial and final Hamiltonians}

We first show that if one chooses the initial Hamiltonian to be the one-dimensional projector onto the uniform superposition state $\ket{\phi}$, and uses a linear interpolation, then an improvement beyond a Grover-like quadratic speedup is impossible as long as the final Hamiltonian $H_1$ is diagonal in the computational basis.  Specifically, for an adiabatic algorithm of the form
\beq
H(t) = \left( 1 - \frac{t}{t_f} \right) E \left( \ident - \ketbra{\phi}{\phi} \right) + \frac{t}{t_f} H_1\ ,
\eeq
the run time $t_f$ for measuring the ground state of $H_1$ with probability $p$ is lower bounded by [Theorem 1 of \cite{Farhi:05}; see also \cite{Znidaric:2006ph}]:
\beq \label{eqt:FarhiBoundT}
t_f \geq \frac{2}{E} \left( 1 - \sqrt{1-p} \right) \sqrt{\frac{N}{k}}  - 2 \frac{\sqrt{p}}{E}\ ,
\eeq
where $N = 2^n$ and $k$ is the degeneracy of the ground state of $H_1$. To see this, define an operator $V_x$ for $x = 0, \dots, N-1$ that is diagonal in the computational basis:
\beq
\bra{z} V_x \ket{z} = e^{2 \pi i z x / N}\ ,
\eeq
and let $\ket{x} = V_x \ket{\phi} = \frac{1}{\sqrt{N}} \sum_{z = 0}^{N-1} e^{2 \pi i z x / N} \ket{z}$. Now define the modified adiabatic algorithm:
\beq \label{eqt:Hx}
H_x(t) = \left( 1 - \frac{t}{t_f} \right) E \left( \ident - \ketbra{x}{x} \right) + \frac{t}{t_f} H_1 \ .
\eeq
Note that $\ket{x = 0} = \ket{\phi}$ implies that $H_0(t) = H(t)$.  For each $x$, the final state is given by $\ket{\psi_x} = U_x(t_f,0) \ket{x}$, with success probability $p_x = \bra{\psi_x} P \ket{\psi_x}$, where $P$ is the projector onto the ground subspace of $H_1$.  Using $H_x(t) = V_x H_0 V_x^{\dagger}$, we have $U_x(t,0) = V_x U_0(t,0) V_x^{\dagger}$, and hence $p_x = p, \forall x$ since $V_x$ commutes with $P$.  We should already see a potential problem for having $t_f$ scale better than $\sqrt{N}$, since if we were to run the algorithm backward, we would find the state $\ket{x}$, which would be solving the Grover problem (note that the initial Hamiltonian~\eqref{eqt:Hx} is the Grover Hamiltonian in a rotated basis).

Now define an evolution according to an $x$-independent Hamiltonian:
\beq
H_\mathrm{R}(t)  = \left( 1 - \frac{t}{t_f} \right) E  \ident+ \frac{t}{t_f} H_1 \ ,
\eeq
and let $\ket{g_x} = \frac{1}{\sqrt{p}} P \ket{\psi_x}$.  Consider the difference in the reverse-evolutions associated with $H_{\mathrm{R}}(t)$ and $H_x(t)$ from $\ket{g_x}$:
\beq
S(t) = \sum_x \Vert \left(U_x^{\dagger}(t_f,t) - U_{\mathrm{R}}^{\dagger}(t_f,t) \right)\ket{g_x} \Vert^2 \ .
\eeq
We can write $\ket{g_x} = \sqrt{p} \ket{\psi_x} + \sqrt{ 1 - p } \ket{\psi_x^{\perp}}$, where $\ket{\psi_x^{\perp}}$ is orthogonal to $\ket{\psi_x}$.  Using $U_x^{\dagger}(t_f, 0) \ket{\psi_x} = \ket{x}$ and defining $\ket{\mathrm{R}_x} = U_{\mathrm{R}}^{\dagger}(t_f,0) \ket{g_x}$, we have:
\bes
\begin{align}
S(0) & = \sum_x \Vert\sqrt{p}\ket{x} + \sqrt{1-p} \ket{x^{\perp}} - \ket{\mathrm{R}_x} \Vert^2 \\
& = 2 N - \sum_x \left[\sqrt{p} \braket{x}{\mathrm{R}_x} + \sqrt{1-p} \braket{x^{\perp}}{\mathrm{R}_x} + \mathrm{c.c.} \right] \\
& \geq 2 N - 2\sqrt{p} \sum_x | \braket{x}{\mathrm{R}_x} |  - 2 N  \sqrt{ 1 - p} \ .
\end{align}
\ees
Since $H_{\mathrm{R}}$ commutes with $H_1$, the state $\ket{\mathrm{R}_x}$ is an element of the $k$-dimensional ground subspace of $H_1$.  Choosing a basis $\left\{ \ket{G_i} \right\}_{i=1}^k$ for this subspace, and writing $\ket{\mathrm{R}_x} = \sum_{i=1}^k \alpha_{x,i} \ket{G_i}$, we have:
\begin{align}
\label{eqt:trick}
\sum_x | \braket{x}{\mathrm{R}_x} |  & \leq \sum_{x,i} | \alpha_{x,i} | \cdot  |\braket{x}{G_i}| \\
& \leq \sqrt{\sum_{x,i} | \alpha_{x,i} | \sum_{x',i'}  |\braket{x'}{G_{i'}}|} = \sqrt{N k} \ . \notag
\end{align}
Therefore, we have:
\beq
S(0) \geq 2 N \left( 1 - \sqrt{1-p} \right) - 2 \sqrt{N k p}\ .
\eeq
In order to upper-bound $S(0)$, we use $S(t_f) - S(0) \leq \int_0^{t_f} | \frac{d}{dt} S(t) | dt$ with $S(t_f) = 0$.  The derivative can be computed using the Schr\"odinger equation:
\begin{align}
 \frac{d}{dt}S(t) & =  - i \sum_x \bra{g_x} U_x(t_f,t) \left[ H_x(t) - H_{\mathrm{R}}(t) \right] U_{\mathrm{R}}^{\dagger}(t_f,t) \ket{g_x} \notag \\
 &\qquad + \mathrm{c.c.} \notag  \\
& = -2 \Im \sum_x \left( 1 - \frac{t}{t_f} \right) E \bra{g_x} U_x(t_f,t)\ket{x}  \times \notag \\
&\qquad \bra{x} U_{\mathrm{R}}^{\dagger}(t_f,t) \ket{g_x} \ .
\end{align}
Thus:
\bes
\begin{align}
\left| \frac{d}{dt}S(t)  \right| & \leq 2 E \left( 1 - \frac{t}{t_f} \right) \sum_x \left| \bra{x} U_{\mathrm{R}}^{\dagger}(t_f,t) \ket{g_x} \right| \notag \\
 & \leq 2 E \left( 1 - \frac{t}{t_f} \right) \sqrt{N k}\ , \label{eqt180b}
 \end{align}
 \ees
where in Eq.~\eqref{eqt180b} we used the same trick as in Eq.~\eqref{eqt:trick}.  Therefore, $\int_0^{t_f} | \frac{d}{dt} S(t) | dt \leq E t_f \sqrt{N k}$.  Putting the upper and lower bound for $S(0)$ together, we have:
\beq
E t_f \sqrt{N k} \geq 2 N \left( 1 - \sqrt{1 -  p}\right) - 2 \sqrt{N kp}\ ,
\eeq
which yields Eq.~\eqref{eqt:FarhiBoundT}.

As an example of the relevance of this result, consider the trivial case of $n$ decoupled spins in a global magnetic field. For an initial Hamiltonian that reflects the bit-structure of the problem, e.g., the standard $H_0 = -\sum_i \sigma_i^x $, the run time of the adiabatic algorithm scales as $\sqrt{n}$ \cite{Muthukrishnan:2015ff,PhysRevA.95.032335}.  If, however, we were to choose instead the projector initial Hamiltonian, the result above shows that we would find a dramatically poor scaling despite the simplicity of the final Hamiltonian.

A similar result is found if all structure is removed from the final Hamiltonian.  Namely, if $H_1 = \sum_z h(z) \ketbra{z}{z}$, we can define a permutation $\pi$ over the $N$ computational basis states such that $h^{[\pi]}(z) = h(\pi^{-1}(z))$. Assume that the initial Hamiltonian is $\pi$-independent and that $c(t)$ satisfies $|c(t)| \leq 1$. Then, for the permuted Hamiltonian $H_{1,\pi} = \sum_z h(z) \ketbra{\pi(z)}{\pi(z)}$, one can show that if the adiabatic algorithm
\beq \label{eqt:permutationAlgorithm}
H_\pi(t) = H_0 + c(t) H_{1,\pi}
\eeq
succeeds with probability $p$ for a set of $\epsilon N!$ permutations, then [Theorem 2 of \cite{Farhi:05}]:
\beq
t_f \geq \frac{\epsilon^2 p}{16 h^{\ast}} \sqrt{N-1} - \frac{\epsilon \sqrt{\epsilon/2}}{4 h^{\ast}} \ ,
\eeq
where $h^{\ast} = \sqrt{\sum_z h(z)^2 / N-1}$.  This result means that no algorithm of the form of Eq.~\eqref{eqt:permutationAlgorithm} can find the minimum of $H_{1,\pi}$ with a constant probability for even a fraction of all permutations if $t_f$ is $o(\sqrt{N})$.

The lesson from this analysis is \emph{what not to do} when designing quantum adiabatic algorithms: avoid choosing the initial Hamiltonian to be the one-dimensional projector onto the uniform superposition state if a better-than-quadratic speedup is hoped for, and avoid removing structure from the final Hamiltonian.

\subsection{Quantum Adiabatic Brachistochrone}
\label{sec:QAB}
%
Modifying the adiabatic schedule adaptively so that it slows down as the gap decreases is an approach that is essential for obtaining a quadratic speedup using the adiabatic Grover algorithm [recall Sec.~\ref{sec:Grover}]. 
Here we discuss how such ideas, including the condition for the locally optimized schedule [Eq.~\eqref{eqt:GroverRate}] can be understood as arising from a variational time-optimal strategy for determining the interpolating Hamiltonian between $H_0$ and $H_{1}$ \cite{PhysRevLett.103.080502}.  By time-optimal, we mean a strategy that gives rise to the shortest total evolution time $t_f$ while guaranteeing that the final evolved state $\ket{\psi(t_f)}$ is close to the desired final ground state $\ket{\veps_0(t_f)}$.  The success of the strategy is judged by the trade-off between $t_f$ and the fidelity $F(t_f) = \left| \braket{\psi(t_f)}{\veps_0(t_f)} \right|^2$.
We first discuss this method generally and then show how it applies to the adiabatic Grover case.

The interpolating Hamiltonian's time-dependence comes from a set of control parameters $\vec{x}(t) = \left( x^1(t), \dots, x^M(t) \right)$, i.e., $H(t) = H[\vec{x}(t)]$.  We can parameterize $\vec{x}(t)$ in terms of a dimensionless time parameter $s(t)$ with $s(0) = 0$ and $s(t_f) = 1$, where $v = \frac{d s}{dt}$ characterizes the speed with of motion along the control trajectory $\vec{x}[s(t)]$.  The total evolution time is then given by:
\beq
t_f = \int_0^1 \frac{ds}{v(s)}\ .
\label{eq:t_f-from-s}
\eeq
Motivated by the form of the adiabatic condition, let us define the following Lagrangian 
\bea \label{eqt:adiabaticL}
\mathcal{L}[\vec{x}(s),\dot{\vec{x}}(s)] & \equiv& \frac{\Vert \partial_s H(s) \Vert_{\mathrm{HS}}^2}{\Delta^p(s)}  \\
&=&  \sum_{i,j} \frac{\Tr\left( \partial_{x^i}H(s) \partial_{x^j} H(s)\right)}{\Delta^p(s)}  \partial_s x^i(s) \partial_s x^j(s)  \nonumber
\eea
($p>0$), and adiabatic-time functional
\bea \label{eqt:adiabaticfunctional}
\mathcal{T}[\vec{x}(s)] = \int_0^1 ds \mathcal{L}[\vec{x}(s),\dot{\vec{x}}(s)] \ ,
\eea
where $\Vert A \Vert_{\mathrm{HS}} \equiv \sqrt{\Tr \left( A^{\dagger} A \right)}$ is the Hilbert-Schmidt norm, chosen to ensure analyticity (this choice is not unique, but other choices may not induce a corresponding Riemannian geometry). The time-optimal curve $\vec{x}_{\rm QAB}(s)$ is the quantum adiabatic brachistochrone (QAB), and is the solution of the variational equation $\delta {\mathcal T[\vec{x}(s)]}/\delta \vec{x}(s) = 0$.

Alternatively, the problem can be thought of in geometrical terms. The integral in Eq.~\eqref{eqt:adiabaticfunctional} is of the form $\int ds \sum_{i,j} g_{ij}(\vec{x}) \partial_s x^i \partial_s x^j$, which defines a reparametrization-invariant object.  Therefore, using results from differential geometry, the Euler-Lagrange equations derived from extremizing Eq.~\eqref{eqt:adiabaticfunctional} are simply the geodesic equations associated with the metric $g_{ij}$ appearing in $\mathcal{L}[\vec{x}(s),\dot{\vec{x}}(s)]  = g_{ij}(\vec{x})\dot{x}^i\dot{x}^j$ (Einstein summation convention):
\beq
\label{eq:EL-QAB}
\partial_s^2 x^k + \Gamma_{ij}^k \partial_s x^i \partial_s x^j = 0\ ,
\eeq
where $\Gamma_{ij}^k = \frac{1}{2} g^{kl} \left( \partial_j g_{li} + \partial_i g_{lj} - \partial_l g_{ij} \right)$ are the Christoffel symbols (connection coefficients) and 
\beq
g_{ij}(\vec{x}) = \frac{\Tr[\partial_i H(\vec{x})\partial_j H(\vec{x})]}{\Delta^p(\vec{x})} \ .
\eeq  
To find the variational time-optimal strategy associated with minimizing Eq.~\eqref{eqt:adiabaticfunctional}, the procedure is thus as follows: (a) solve Eq.~\eqref{eq:EL-QAB} to find the optimal path $\vec{x}_{\rm QAB}(s)$; (b) compute the adiabatic error using the Schr\"odinger equation along this optimal path (or multi-parameter schedule). Note that to compute the metric requires knowledge of the gap, or at least an estimate thereof. 

The optimal path is a geodesic in the parameter manifold endowed with the Riemannian metric $g$. This metric gives rise to a curvature tensor $\mathbf{R}$, which can be computed from the metric tensor and the connection using standard methods \cite{Nak}. Namely, $\Gamma \sim g^{-1}\partial g \sim \Delta^{-1} \partial \Delta$, and $\mathbf{R} \sim \partial^2 g + g \Gamma^2 \sim \Delta^{-p-2}$. Thus, the smaller the gap, the higher the curvature.

Let us illustrate with a simple example.  Consider the following Hamiltonian with \emph{a single} control parameter $x^1(s)$:
\beq
H(s) = \left( 1 - x^1(s) \right) P_a^{\perp} + x^1(s) P_b^{\perp} \ ,
\eeq
where we have defined the projector $P_{a}^{\perp} = \ident - \ketbra{a}{a}$ and similarly for $P_{b}^{\perp}$. This includes the Grover problem as the special case where $\ket{a}$ is the uniform superposition and $\ket{b}$ is the marked state.  We can always find a state $\ket{a^{\perp}}$ such that $\ket{b} = \alpha_0 \ket{a} + \alpha_1 \ket{a^{\perp}}$, where $\alpha_0 = \braket{a}{b}$.
Therefore, the evolution according to $H(s)$ occurs in a two dimensional subspace spanned by $\ket{a}$ and $\ket{a^\perp}$, and:
\bes
\begin{align}
\partial_{x^1} H(s) & =  - P_a^{\perp} + P_b^{\perp} \\
 \Tr \left(\partial_{x^1} H(s)  \partial_{x^1} H(s)  \right) & = 2 \left( 1 - |\alpha_0|^2 \right) \\
& \hspace{-3cm} \Delta(s)   =   \sqrt{1 - 4 \left( 1- |\alpha_0|^2 \right) x^1(s) (1-x^1(s))} \\ 
g_{11} & =  \frac{2(1-|\alpha_0|^2)}{\Delta(s)^3} \ .
\end{align}
\ees
The geodesic equation is then given by:
\bea
\frac{d^2}{ds^2} x^1(s) &&  \\
&& \hspace{-1cm} + \frac{p \left( 1- 2 x^1(s) \right) \left( 1- |\alpha_0|^2 \right)}{1 - 4 x^1(s) \left( 1- x^1(s) \right) \left( 1- |\alpha_0|^2 \right)} \left( \frac{d}{ds} x^1(s) \right)^2 = 0 \ . \nonumber
\eea
In the case of $p = 4$, we can solve this equation analytically, and the solution with the boundary conditions $x^1(0) = 1- x^1(1) = 0$ is given by:
\beq
x^1(s) = \frac{1}{2} + \frac{|\alpha_0|}{2 \sqrt{1 - |\alpha_0|^2}} \tan \left[ \cos^{-1} \left(|\alpha_0| \right) \left( 2 s - 1 \right) \right]
\label{eq:QAB1D-Grover}
\eeq
(note that $\cos^{-1} (|\alpha_0|) = \tan^{-1} \left( \frac{\sqrt{ 1 - |\alpha_0|^2}}{|\alpha_0|}   \right)$). Remarkably, this is equivalent to the expression we found for the Grover problem [Eq.~\eqref{eqt:GroverOptimalSchedule}] if we take $\alpha_0 = 1/\sqrt{N}$, despite the different choice of norm and value of $p$.  This shows that the optimal schedule for the Grover problem has a deep differential geometric origin.

We can extend the analysis to \emph{two} control parameters such that the time-dependent Hamiltonian is given by:
\beq
H(s) = x^1(s)  P_a^{\perp} + x^2(s) P_b^{\perp}\ .
\eeq
The associated QAB (or geodesic) path can be found numerically, and it turns out that it is not of the form $x^2(s)=1-x^1(s)$, i.e., it is different from the \cite{Roland:2002ul} path given by Eq.~\eqref{eq:QAB1D-Grover}.
The optimal two-parameter path reduces the adiabatic error relative to the latter [see Fig.~\ref{Rezakhani09a}], but can of course not reduce the (already optimal) $\sqrt{N}$ scaling. The two-parameter QAB also has lower curvature than the Roland-Cerf path [see Fig.~\ref{Rezakhani09b}], which implies that it follows a path with a larger gap and less entanglement than the latter \cite{PhysRevLett.103.080502}, as mentioned in Sec.~\ref{sec:ent-role}. 

The differential geometric approach to AQC was further explored in \cite{PhysRevA.82.012321}, where its connections to quantum phase transitions were elucidated, within a unifying information-geometric framework. See also \cite{Zulkowski:2015ij}.
\begin{figure}[htbp] 
   \centering
   \subfigure[]{\includegraphics[width=2in]{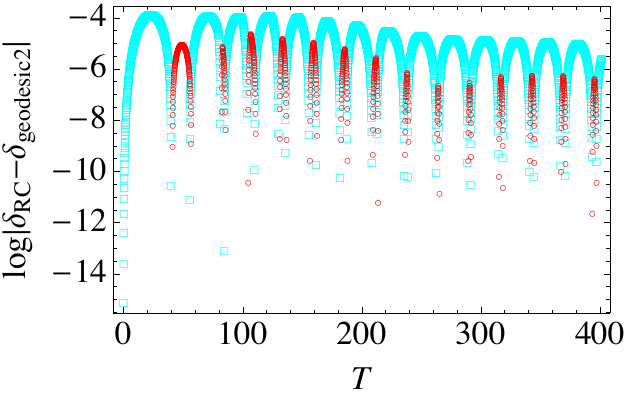} \label{Rezakhani09a}}
    \subfigure[]{\includegraphics[width=2in]{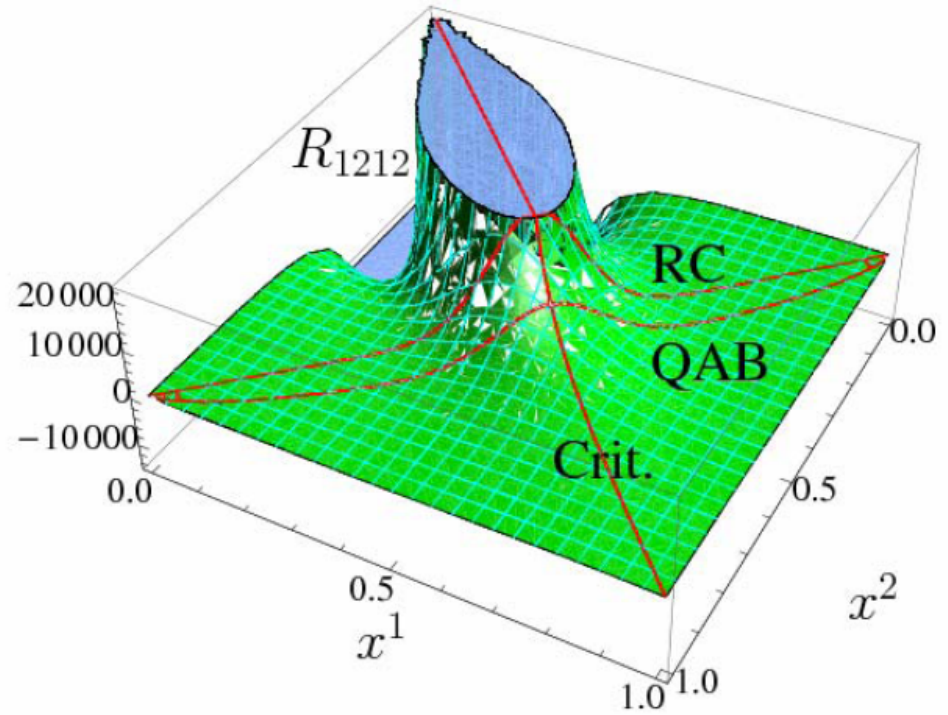} \label{Rezakhani09b}}
   \caption{(a) Final-time error $\delta(T) = \sqrt{1 - F(T)}$ ($T=T_f$ in our notation) for the single-control (denoted RC for Roland-Cerf) and two-parameter control (denoted geodesic2) geodesic paths for the Grover problem with $n = 6$. Squares (cyan) indicate where the two-parameter geodesic path outperforms (i.e. has a lower error than) the single-parameter path; circles (red) correspond to the opposite case. (b) The curvature tensor component $R_{1212}$ for $n=3$.  The curves on the curvature surface show the case of the standard linear interpolation $x_2 = 1-x_1$ (denoted Crit.), the path followed by the one-parameter geodesic (denoted RC), and the path followed by the two-parameter geodesic (denoted QAB).  From \cite{PhysRevLett.103.080502}.}
   \label{Rezakhani09}
\end{figure}

\subsection{Modifying the initial Hamiltonian}
\label{sec:mod-H_0}

Rather than modifying the adiabatic interpolation, one may modify the initial Hamiltonian. Such a strategy was pursued in \cite{Farhi:2011:QAA:2011395.2011396} and tried on a particular set of 3-SAT instances, where the clauses are picked randomly subject to satisfying two disparate planted solutions and then penalizing one of them with a single additional clause. This was done in order to generate instances with an avoided crossing at the final time $s=1$, reproducing the type of obstacle to AQC envisioned in \cite{Altshuler2010}.

It was then shown that in this case, by picking a random initial Hamiltonian of the form
\beq
H_{0} = \frac{1}{2} \sum_{i=1}^n c_i \left( \ident - \sigma_i^x \right)\ ,
\eeq
where $c_i$ is a random variable taking value $1/2$ or $3/2$ with equal probability, it is possible to remove the small gap encountered by the standard adiabatic algorithm with high probability. Since this strategy does not rely on information about the specific instance, it appears to be quite general. Therefore, if the algorithm is to be run on a single instance of some optimization problem, the adiabatic algorithm should be run repeatedly with different initial Hamiltonians  \cite{Farhi:2011:QAA:2011395.2011396}. 

An alternative approach based on modifying the initial Hamiltonian, with a different goal, was proposed in \cite{Perdomo-Ortiz:2011fh}, whereby an initial guess for the solution (a computational basis state) is used as the initial state of the adiabatic algorithm.  An initial Hamiltonian is used with this state as its ground state. Evolution to the final Hamiltonian then proceeds according to a standard schedule. If the final state that is measured is not the ground state of the final Hamiltonian (due to diabatic transitions), the algorithm can be repeated with the measured state as the new initial state.  Such ``warm start" repetitions of the algorithm exhibited improved performance compared to the standard approach for 3-SAT problems, although the results were limited to relatively small system sizes of $6$ and $7$ qubits.

 \subsection{Modifying the final Hamiltonian}
 \label{sec:modified-H_1}

The same problem can be specified by two or more different final Hamiltonians, as we saw, e.g., in the case of Exact Cover $3$ (EC3), in terms of Eqs.~\eqref{eq:EC-1} and \eqref{eq:EC-2}. It was claimed in \cite{Altshuler2010} that adiabatic quantum optimization fails for random instances of EC3 because of
Anderson localization. The claim, which we discuss in more detail in Sec.~\ref{sec:PerturbativeCrossings}, was based on the form given in Eq.~\eqref{eq:EC-2}. However, as argued in \cite{Choi15022011}, it is possible to reformulate the final Hamiltonian for EC3 such that the argument in \cite{Altshuler2010} may not apply. Namely, for any pair of binary variables $x_{C_i},x_{C_j}$ in the same clause $C$, add a term $D_{ij}x_{C_i}x_{C_j}$ with $D_{ij}>0$; this is permissible since in order for a clause to be satisfied, exactly one variable must take value
$1$, whereas the other two are $0$. Numerical evidence for up to $15$ bits suggests that the addition of the new set of arbitrary parameters $D_{ij}$ may avoid the Anderson localization issue \cite{Choi15022011}. This example illustrates a general principle, that it can be incorrect to conclude from the failure
of one specific choice of the final Hamiltonian that all quantum adiabatic algorithms fail for the same problem. 

\subsection{Adding a catalyst Hamiltonian} 
\label{sec:catalyst}

We define a ``catalyst" as a term that (1) vanishes at the initial and final times, but is present at intermediate times, (2) is a sum of local terms with the same qubit-interaction graph as the final Hamiltonian $H_1$, (3) does not use any other information specific to the particular instance.

Consider, e.g.:
\beq
{H}(s) = (1-s) H_{0} + s (1-s) H_{\mathrm{C}} + s H_{1} \ .
\eeq
The specific form of $H_{\mathrm{C}}$ is of course important, but even a randomly chosen catalyst can help \cite{FarhiAQC:02,Farhi:2011:QAA:2011395.2011396,Zeng:2016bs}.  We illustrate how $H_{\mathrm{C}}$ can turn a slowdown (exponential run time) into a success (at worst polynomial run time) for a specific problem with a specific $H_{\mathrm{C}}$ that is analytically tractable.  Consider a final Hamiltonian of the form
\beq
H_{1} = \sum_{z} h(z) \ketbra{z}{z}\ ,
\eeq
where $z$ denotes an $n$-bit string, and $h(z) = \sum_{i<j<k} h_3(z_i, z_j,z_k)$ with
\beq
h_3(z_1,z_2,z_3) = \left\{ \begin{array}{lr}  
0,& z_1 + z_2 + z_3 = 0 \\
3,& z_1 + z_2 + z_3 = 1 \\
1,& z_1 + z_2 + z_3 = 2 \\
1,& z_1 + z_2 + z_3 = 3 \\
\end{array} \right. \ .
\eeq
The all-zero bit string minimizes the final Hamiltonian with energy $0$.

The cost function $h(z)$ is bit-permutation symmetric and only depends on the Hamming weight $|z|$, which facilitates the analysis. Specifically \cite{Farhi-spike-problem}:
\begin{eqnarray}
h(z)  &=& \frac{3}{2} |z| \left( n - |z| \right) \left( n - |z| - 1 \right) + \frac{1}{2} |z| \left( |z| - 1 \right) \left(n - |z| \right) \nonumber \\
&& + \frac{1}{6} |z| \left( |z| - 1 \right) \left( |z| - 2 \right)\ .
\end{eqnarray}
The final Hamiltonian can then be written in terms of the total spin operators $S^{\alpha} = \frac{1}{2} \sum_{i=1}^n \sigma^\alpha_i$ by using the mapping $|z| \mapsto \frac{n}{2} - S^z$.  The initial Hamiltonian is taken to be 
\begin{align}
H_0 = \frac{(n-1)(n-2)}{2} \left( \frac{n}{2} \ident - S^x \right) \ ,
\end{align}
[the unconventional normalization is to ensure that both $H_{1}$ and $H_{0}$ scale similarly with $n$ \cite{farhi_quantum_2000}].  $H_{\mathrm{C}}$ is taken to be identical for all combinations of three bits in order to preserve the permutation symmetry:
\beq
H_{\mathrm{C}} = - 2 n \left(S^x S^z + S^z S^x \right) \ .
\eeq
Note that this catalyst is non-stoquastic. A useful way to characterize the change due to the introduction of $H_{\mathrm{C}}$ is to study the semi-classical potential associated with the Hamiltonian:
\beq
V(s, \theta, \varphi)  = \bra{\theta, \varphi}H(s) \ket{\theta, \varphi}\ ,
\eeq
where $\ket{\theta, \varphi}$ is the spin-coherent state defined in Eq.~\eqref{eqt:SpinCoherent}. 
In the large $n$ limit we have \cite{FarhiAQC:02}:
\begin{eqnarray}
\lim_{n \to \infty} V / (2/n)^3 &=& 2 ( 1- s)( 1- \sin \theta \cos \varphi ) \nonumber \\
&& + \frac{1}{6}s \left( 13 + 3 \cos \theta - 9 \cos ^2 \theta - 7 \cos^3 \theta \right) \nonumber \\
&&- 8 s(1-s) \cos \theta \sin \theta \cos \varphi\ ,
\end{eqnarray}
where the three terms arise from the initial, final and catalyst Hamiltonians, respectively.  We display the behavior of this potential in Fig.~\ref{fig:FarhiPotential}.  In the absence of $H_{\mathrm{C}}$, there is a value of $s$ where the potential has degenerate minima, and the system must tunnel from the right well to the left well in order to follow the global minimum.  This point is associated with an exponentially small gap \cite{Farhi-spike-problem}, i.e., the algorithm requires exponential time to follow the global minimum.  However, in the presence of $H_{\mathrm{C}}$ the potential never exhibits such an obstacle; there is always a single global minimum that the system can follow from $s = 0$ to $s = 1$ with polynomial run time.  

Using this method of introducing a catalyst Hamiltonian, improvements were generally observed on a large number of MAX 2-SAT instances of size $n=20$ (by directly solving the Schr\"{o}dinger equation) \cite{crosson2014different}. Both stoquastic and non-stoquastic $H_{\mathrm{C}}$ were tried and improved the success rate, but the difference between stoquastic and non-stoquastic was not decisive. 

A similar study was performed in \cite{Hormozi:2016} on fully-connected Ising instances, 
$H_{1} = \sum_{i=1}^n h_i \sigma^z_i + \sum_{i< j}^n J_{ij} \sigma_i^z\sigma_j^z$, of size $n \leq 17$, where the $J_{ij}$'s and $h_i$'s were picked from a continuous Gaussian distribution with zero mean and unit variance.  The authors observed that a stoquastic catalyst generally improves the performance of easy instances by boosting the minimum gap and reducing the number of anticrossings.  The fraction of instances affected tends to grow with increasing problem size.  This is in stark contrast to non-stoquastic catalysts that tend to improve the performance of the very hard instances, but the fraction of improved instances remains constant with increasing problem size. Furthermore, the gap does not generically increase with the addition of this catalyst, and the number of anticrossings grows.  This latter feature leads to the increased success probability as population lost from the ground state at one anticrossing can be recovered at a later anticrossing.
\begin{figure}[th] 
   \centering
   \subfigure[]{\includegraphics[width=0.45\columnwidth]{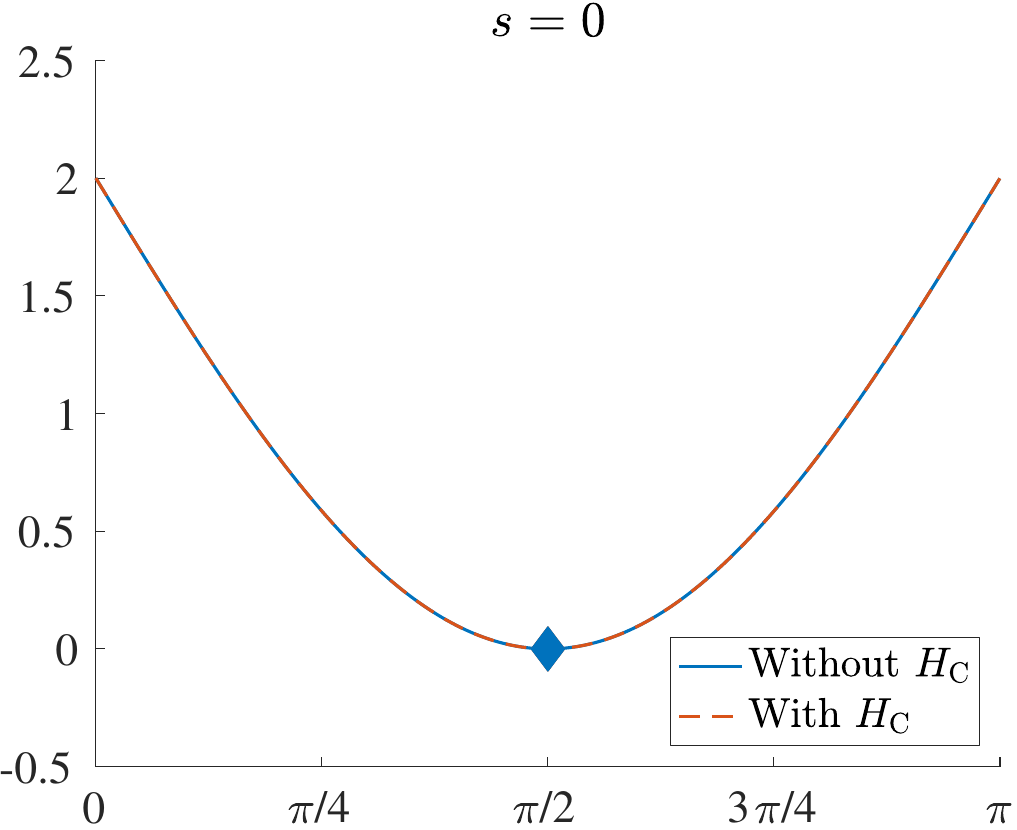}} 
   \subfigure[]{\includegraphics[width=0.45\columnwidth]{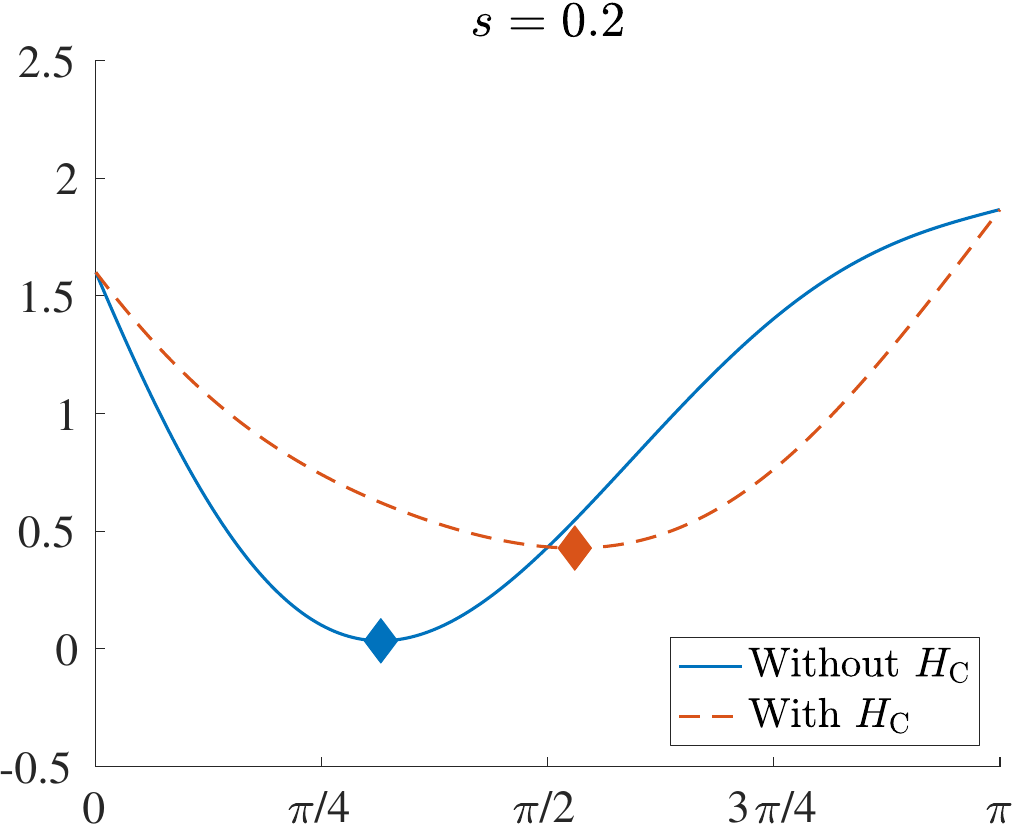}} 
   \subfigure[]{\includegraphics[width=0.45\columnwidth]{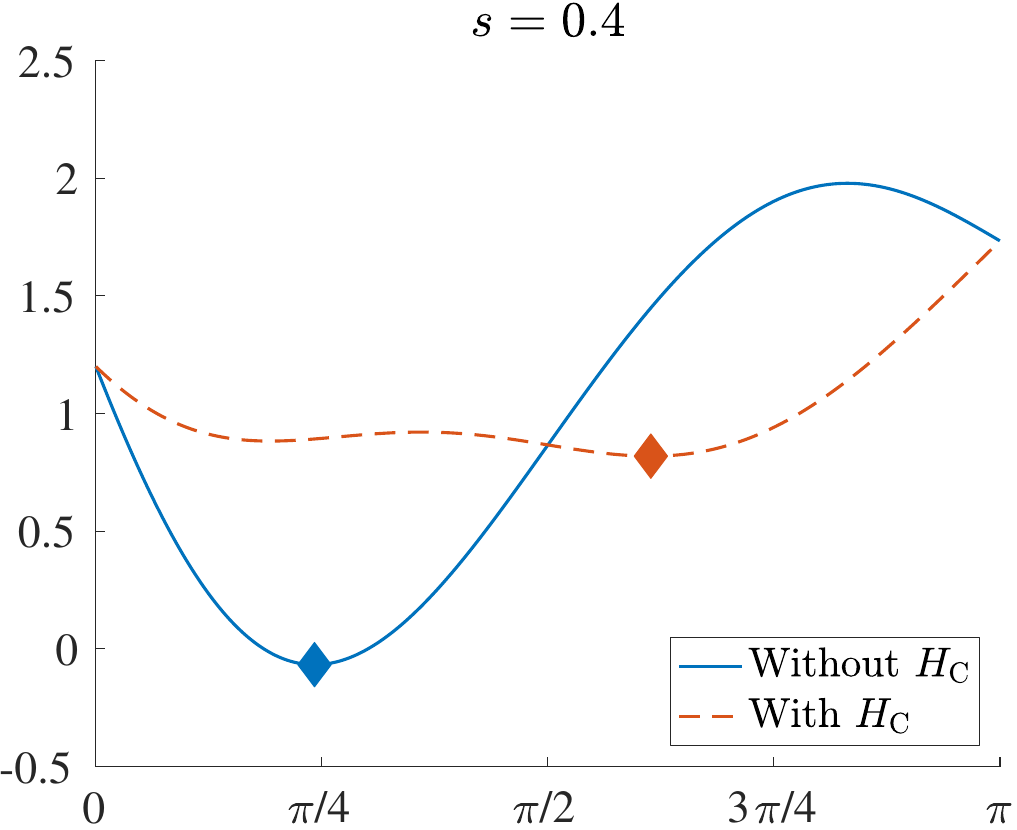}} 
   \subfigure[]{\includegraphics[width=0.45\columnwidth]{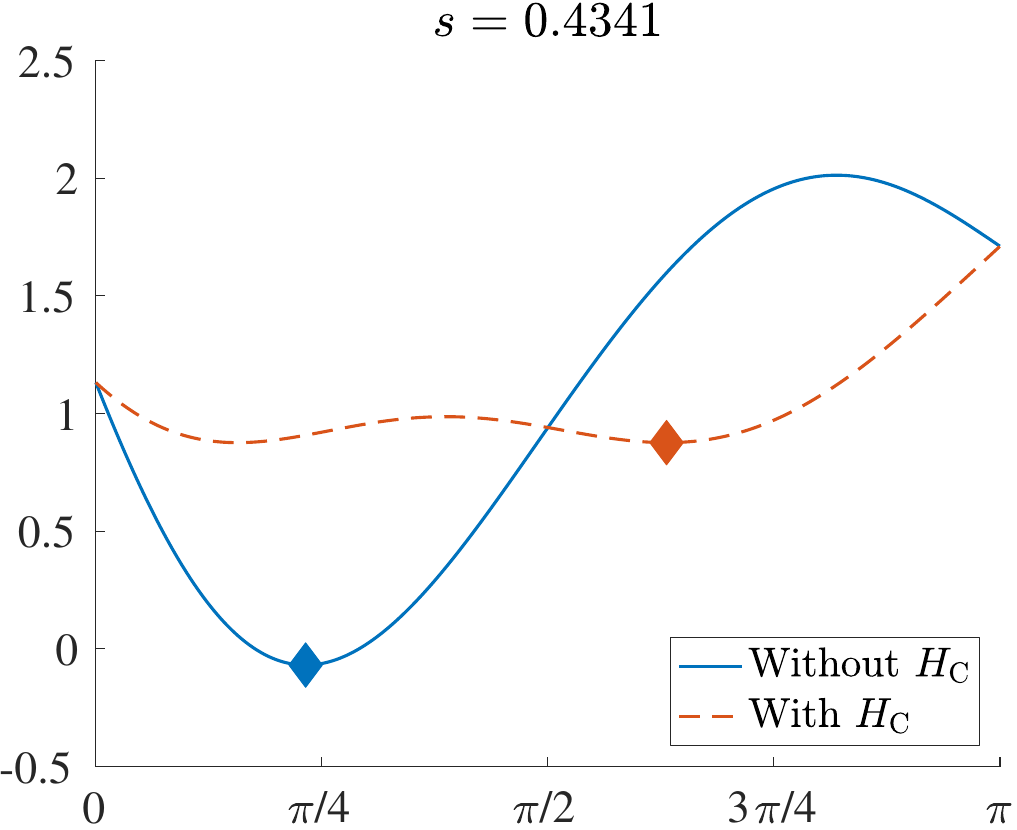}}
    \subfigure[]{\includegraphics[width=0.45\columnwidth]{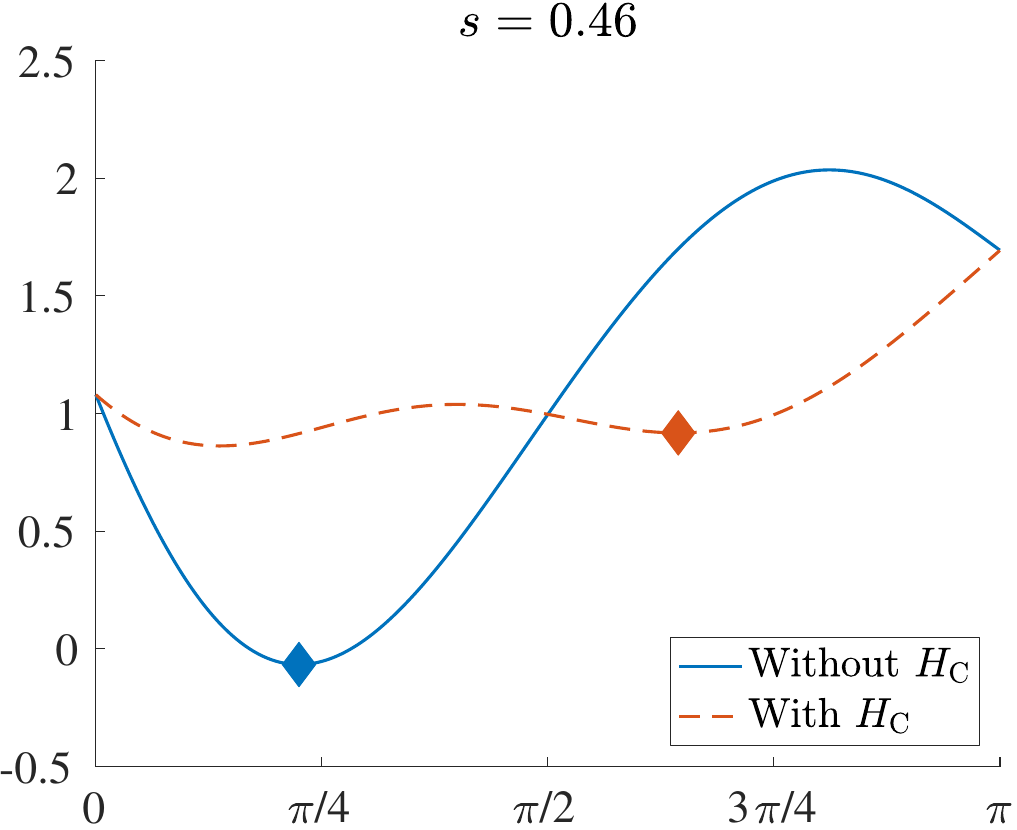}}
    \subfigure[]{\includegraphics[width=0.45\columnwidth]{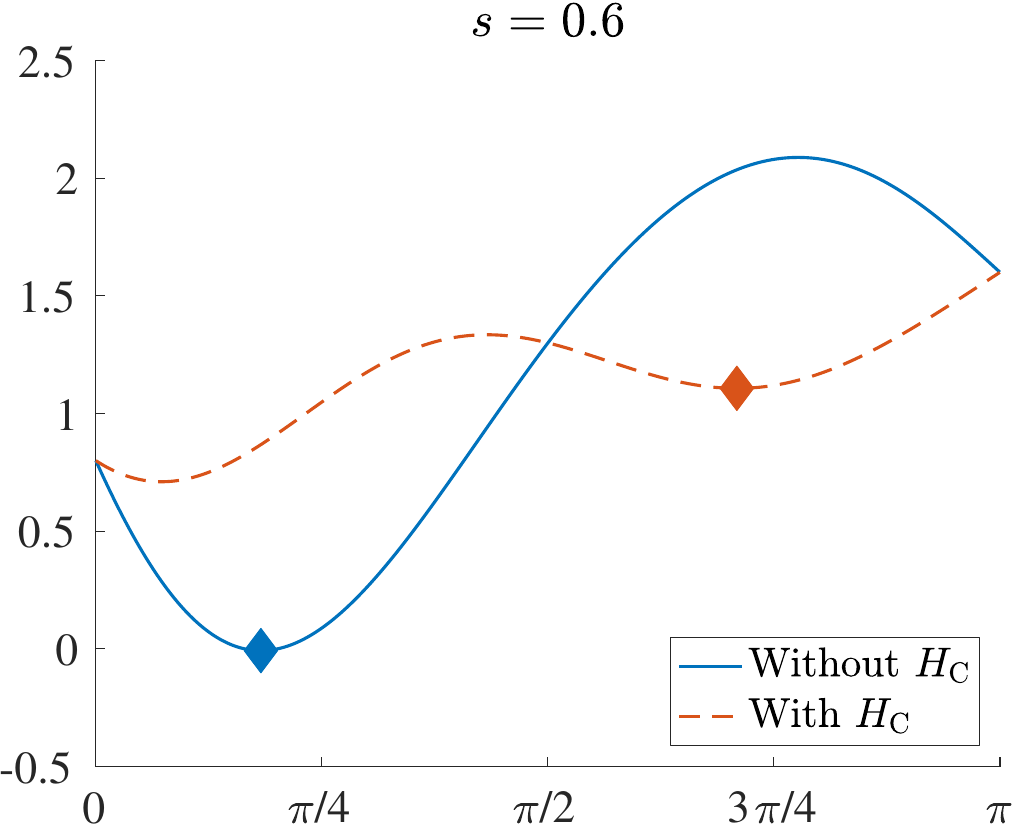}}
    \subfigure[]{\includegraphics[width=0.45\columnwidth]{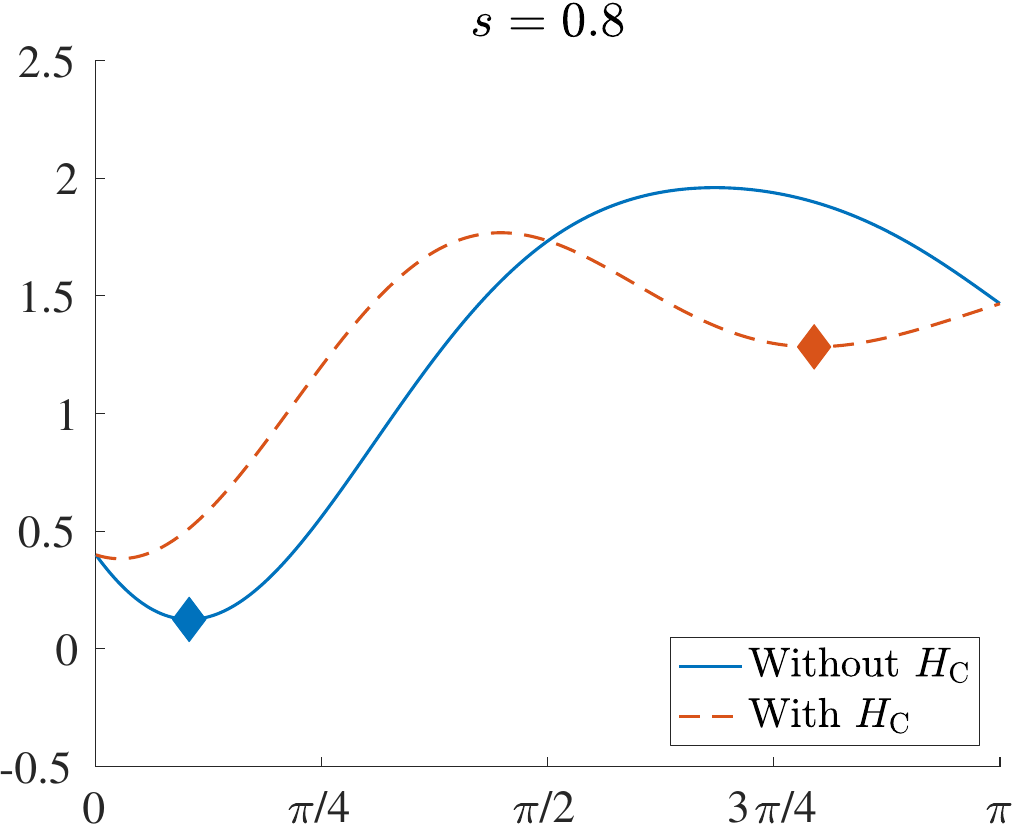}}
    \subfigure[]{\includegraphics[width=0.45\columnwidth]{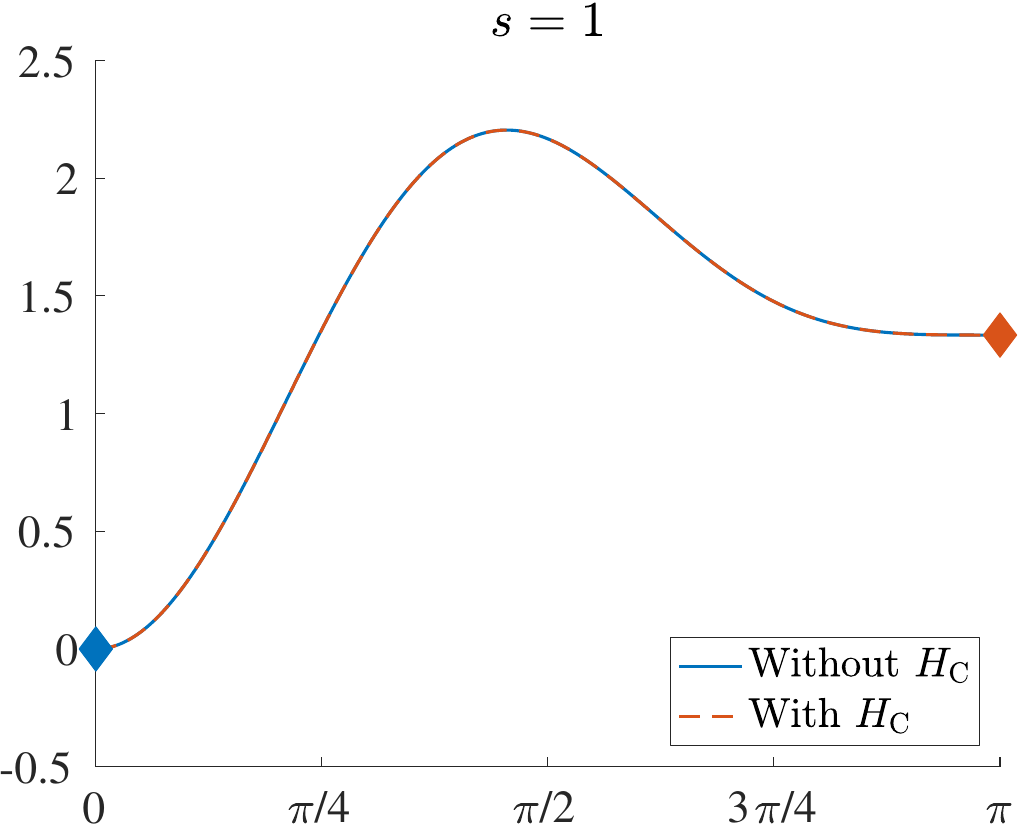}}
   \caption{The diamonds represent the minima followed by a polynomial run time.  In the case with $H_{\mathrm{C}}$, the potential can follow the global minimum polynomial time.  In the case without $H_{\mathrm{C}}$, there is an $s$ value where the potential has a degenerate minimum, and the algorithm cannot tunnel to the new global minimum in polynomial time.}
   \label{fig:FarhiPotential}
\end{figure}

\subsection{Addition of non-stoquastic terms}

The addition of non-stoquastic terms was already considered numerically in the previous subsection; here we focus on analytical results obtained for certain mean field models.

Quantum statistical-mechanical techniques (Trotter-Suzuki decomposition, replica method under the replica-symmetric ansatz, and the static approximation) were used in \cite{Seki:2012,Seoane:2012,Seki:2015,Nishimori:2016yq} to analyze infinite-range Ising models with ferromagnetic as well as random interactions. These studies concluded that non-stoquastic terms can sometimes modify first-order quantum phase transitions (with an exponentially small gap) in the stoquastic Hamiltonian to second-order transitions (with a polynomially small gap) in the modified, non-stoquastic Hamiltonian. 

The Hamiltonian is of the form
\beq
H(s,\lambda) = (1-s)H_0 - s\left(\lambda H_{1,z}^{(p)} + (1-\lambda)H_{1,x}^{(k)} \right)
\eeq
where $H_0 = -\sum_{i=1}^n \sigma_i^x$ is a standard initial Hamiltonian, and 
\beq 
H_{1,\alpha}^{(q)} =  n\left(\frac{1}{n}\sum_{i=1}^n \sigma_i^\alpha \right)^q\ , \quad \alpha\in\{x,z\}\ , \quad q\in\{p,k\}\ ,
\eeq
where $\lambda\in[0,1]$ controls the strength of the non-stoquastic term $H_{1,x}^{(k)}$, and both $p$ and $k$ are integers $\geq 2$ that determine the locality of the model. The parameter $\lambda$ is increased to $1$ along with $s$, so that the final Hamiltonian is the infinite-range $p$-body ferromagnetic Ising model $H(1,1) = -H_{1,z}^{(p)}$. Also the $r$-pattern Hopfield model was studied, where $\lambda H_{1,z}^{(p)}$ is replaced by $-\sum_{1\leq i_1<\cdots<i_p\leq n}J_{{i_1}\cdots {i_p}}\sigma^z_{i_1}\cdots\sigma^z_{i_p}$, where $J_{{i_1}\cdots {i_p}}$ is given in Eq.~\eqref{eq:Hopfield}, with $\xi_{i_p}$ being $\pm 1$ with equal probability.

In the ferromagnetic case \cite{Seki:2012,Seoane:2012} showed that for $p\geq 4$, a two-local non-stoquastic $XX$ term \footnote{This addition does not result in a truly non-stoquastic Hamiltonian; there exists a local unitary transformation that makes the Hamiltonian stoquastic.  Specifically, rotate $\sigma^z$ to $\sigma^x$.  In this new basis, the Hamiltonian is stoquastic.} changes the first order phase transition to a second order one, for an appropriately chosen path in the $(\lambda,s)$ plane, starting from $(\lambda_0,0)$ (with arbitrary $\lambda_0$) and ending at $(1,1)$. The situation in the Hopfield model case is identical to the ferromagnetic case, for an extensive number of patterns $r\propto n$. For a fixed number of patterns $p\geq 5$ is sufficient and $p>3$ is necessary in order to avoid first order phase transitions \cite{Seki:2015}.

\subsection{Avoiding perturbative crossings} 
\label{sec:PerturbativeCrossings}

An important slowdown mechanism we already alluded to in Sec.~\ref{sec:mod-H_0} is due to anti-crossings very close to the end of the evolution, that can result in an extremely small minimum gap.  These crossings are often referred to as perturbative, because a perturbative expansion back in time from the final Hamiltonian [e.g., perturbation theory in $\Gamma$ for Eq.~\eqref{eq:H(Gamma)}] yields perturbed states that cross in energy very close to where the exact eigenstates anti-cross, with a gap that is exponentially small in the Hamming weight of the unperturbed crossing states \cite{AminChoi:2009} [shown there in the context of the weighted maximum independent set problem; see also \cite{Foini:2010,Farhi:2011:QAA:2011395.2011396}]. This problem of perturbative crossings was demonstrated for the NP-complete Exact Cover problem [recall Sec.~\ref{sec:exact-cover}] in \cite{Altshuler2010}, who related the mechanism  of exponentially small spectral gaps to Anderson localization of the eigenfunctions
of $H(s)$ in the space of the solutions. They showed that the Hamming weight between such states can be $\Theta(n)$, which is clearly problematic for the adiabatic algorithm. It was also claimed in \cite{Altshuler2010} that these anti-crossings appear with high probability as the transverse field goes to zero; however the latter claim did not survive a more accurate analysis that took into account the extreme value statistics of the energy levels: the exponential degeneracy of the ground state, which is a distinguishing feature of random NP-complete problems with a discrete spectrum (such as Exact Cover), dooms the proposed mechanism \cite{Knysh:2010bf,Knysh:2016iq}.

Nevertheless, this does not rule out the occurrence of exponentially small gaps close to the end of the evolution. Furthermore, it is plausible that the mechanism for avoided level crossings presented in \cite{Altshuler2010} may not necessarily be restricted to the end of the evolution, but may occur throughout a many-body-localized phase \cite{Laumann:2015sw}. In light of this we now discuss a rather general way to circumvent such perturbative crossings, that differs from the random initial Hamiltonian approach presented in Sec.~\ref{sec:mod-H_0}.

Using the NP-hard maximum independent set problem, it was shown that this problem occurs only for one particular implementation of the adiabatic algorithm, and different choices can avoid the problem \cite{Choi:2010rz}. In fact, \cite{Dickson:2011ad} showed that there is always some choice of the initial and final Hamiltonians that avoids such perturbative crossings (note that this does not include non-perturbative crossings).  Furthermore, this choice can be made efficiently, i.e., in polynomial time, space and energy \cite{Dickson:2011lo}, as we now summarize.

The idea of \cite{Dickson:2011lo} is to cause the ground state to diverge from all other states by changing the degeneracy of the spectrum of the final Hamiltonian, such that the ground state is the most degenerate, the first excited state less degenerate, up to the highest excited state, which will be the least degenerate.  Consider an $n$-qubit Ising Hamiltonian of the form 
\beq
H_{1} = \sum_{i\in M} h_i \sigma^z_i + \sum_{\{i,j\}\in M} J_{ij} \sigma_i^z\sigma_j^z\ ,
\label{eq:IsingH}
\eeq
where $h_i, J_{ij} \in \left\{ 0, \pm 1 \right\}$ and $M$ specifies the non-zero terms, of which there are $m$.  In order to simplify the analysis, assume that there are no single bit-flip degeneracies, meaning that there are no degenerate states that are Hamming distance $1$ from each other.  For each non-zero $h_i$ term that the ground state satisfies, i.e., $h_i \sigma^z_i = -1$, add $a\geq 1$ ancilla qubits with an interaction of the form:
\beq
H_h = \sum_{k=1}^a b \left(h_i \sigma_i^z + 1 \right) \left( \sigma^z_{i_k} +1 \right) / 2\ ,
\eeq
where $\{i_1,\dots,i_a\}\in M_h$. This term vanishes when the term $h_i$ is satisfied, regardless of the orientation of the $a$ ancillas, whereas otherwise it gives an energy $b n_1$, where $n_1$ is the number of ancillas pointing up.  Note that when the term is unsatisfied, when all ancilla spins point down the energy cost is zero.  This is important because we do not want to change the energy of the ground state configuration.

Similarly, for each (non-zero) $J_{ij}$, also add $a$ ancillas with the following interaction term:
\beq
H_J = \sum_{k=1}^a \left( J_{ij} \sigma^z_i \sigma^z_j + 1 \right) \left( \sigma^z_{(ij)_k} + 1  \right) /2\ ,
\eeq
where $\{ij_1,\dots,ij_a\}\in M_J$. This introduces $3$-local terms; it is possible to use $2$-local terms to achieve the same result at the expense of introducing an additional ancilla for each term [see the Appendix of \cite{Dickson:2011lo} for details].  

The spectrum of the new Hamiltonian (with $m a $ additional ancilla qubits) is the original spectrum when all ancilla qubits point down, and all the new energy states correspond to flips of the ancilla qubits, with increased energy. Note that this means that no new local minima were introduced.  Now consider the following adiabatic algorithm Hamiltonian:
\beq
H = \lambda H_0+ H_1'  \ , \quad H_1' = H_1+H_h+H_J\ ,
\eeq
with $\lambda$ decreased from $\infty$ (proportional to $abm$ suffices) to $0$, and where the initial Hamiltonian includes transverse fields on the ancilla qubits:
\beq
 H_0= -  \sum_{i\in M\cup M_h\cup M_J} \sigma^x_i  \ .
\eeq
Consider a non-degenerate classical state ${\alpha}$ with energy $E_{\alpha}$ under the action of $H_1$.  It becomes degenerate under the action of $H_1'$.  Let $\ket{\alpha}$ denote  the uniform superposition over all these degenerate states with energy $E_{\alpha}$.  Introducing $\lambda > 0$ breaks the degeneracy, and from first order degenerate perturbation theory (see Appendix~\ref{sec:perturbationTheory}) the state $\ket{\alpha}$ is the new lowest energy eigenstate within the subspace spanned by the unperturbed degenerate states with energy $E_{\alpha}$. The correction to its energy is $E_{\ket{\alpha}} = E_{\alpha} + \lambda E_{\ket{\alpha}}^{(1)} + \dots$, where 
\begin{align}
E_{\alpha} = &-(\# \text{ of terms in $H_1'$ satisfied by }\alpha) \notag \\
&+ (\# \text{ of terms in $H_1'$ unsatisfied by }\alpha)  \\
=& -2(\# \text{ of terms in $H_1'$ satisfied by }\alpha) + m \notag
\end{align}
(recall that $h_i,J_{ij}\in\{0,\pm 1\}$), and
\begin{align}
E_{\ket{\alpha}}^{(1)} &= \bra{\alpha}  H_0\ket{\alpha} \\
&=  - a \left( \text{\# of terms in $H_1'$ satisfied by ${\alpha}$}\right) \notag \\
&= \frac{a}{2} \left( E_{\alpha} - m \right)\ .
\notag 
\end{align}
Note that in $E_{\alpha}$ the contribution is entirely due to $H_1$, while in $E_{\ket{\alpha}}^{(1)}$ the contribution is entirely due to $H_h+H_J$.

Taking $a = b = n^2$, it can be shown that higher order corrections do not depend on $a$, and hence the first order correction dominates the behavior.  Therefore, it is clear that a state $\ket{\alpha}$ with a lower (final) energy than a state $\ket{\beta}$ has a larger negative slope (first-order perturbation energy correction).  Therefore, the states $\ket{\alpha}$ and $\ket{\beta}$ grow farther apart for $\lambda > 0$ according to first order perturbation theory.  This means that the perturbative crossing is avoided.

This method works in general for the problem of finding the ground
state of an arbitrarily-connected Ising model with local fields, and is fully stoquastic. Thus, all NP-complete problems can be attacked
using StoqAQC without encountering perturbative crossings.
Of course, this does not prove that StoqAQC can solve NP-complete problems
in polynomial time. However, it does mean that proving otherwise requires identifying some effect other than perturbative crossings that unavoidably results in exponentially long adiabatic run times.

\subsection{Evolving non-adiabatically}

Our discussion so far has been restricted to adiabatic evolutions, where the minimum gap controls the efficiency of the quantum algorithm.  
However, as we have seen with the glued-trees problem in Sec.~\ref{sec:gluedtrees}, the quantum evolution can take advantage of the presence of two avoided-level crossings (and their associated exponentially small gaps) to leave and return to the ground state with high probability in polynomial time, whereas an adiabatic evolution would have required exponential time. Setting aside the fascinating and intricate field of open-system AQC where relaxation can play a beneficial role in returning the computation to the ground state [the subject of a separate review \cite{Albash-Lidar:RMP-colloq}], this is one among several cases where non-adiabatic, i.e., diabatic evolution enhances the performance of a quantum algorithm based on Hamiltonian computation. Another example is~\cite{crosson2014different} [see also \cite{Hormozi:2016}] where it was observed that evolving rapidly (as well as starting from excited states) increased the success probability on the hardest instances of randomly generated $n=20$ MAX-2-SAT instances with a unique ground state.  When evolving rapidly, population leaks into the first excited state before the avoided-level crossing and then returns to the ground state after the avoided-level crossing.  An instance of this behavior is shown in Fig.~\ref{fig:Crosson2014Fig}.
\begin{figure}[htbp] 
   \centering
   \includegraphics[width=0.95\columnwidth]{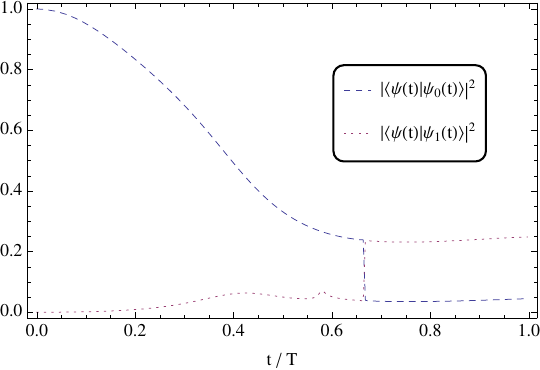}
   \caption{Overlap squared of the evolved wavefunction $\ket{\psi(t)}$ and the instantaneous ground state $\ket{\psi_0(t)}$ and first excited state $\ket{\psi_1(t)}$ for an instance of MAX-2-SAT with $n=20$ and with a total time $T = 10$.  Because of the rapid evolution, population leaks out of the ground state and hence the decrease in the ground state population.  There is an avoided level crossing at approximately $t/T = 0.65$, where the population between the ground state and first excited state are effectively swapped.  Therefore, if more substantial leaking into the first excited state occurs, this will lead to an increase in probability of finding the ground state at the end of the evolution.  From \cite{crosson2014different}.}
   \label{fig:Crosson2014Fig}
\end{figure}

A similar result was observed in~\cite{Muthukrishnan:2015ff} for a large class of Perturbed Hamming Weight problems (recall Sec.~\ref{sec:spike}), but with the difference that the rapid evolution diabatically pushes population to higher excited states and then returns to the ground state through a series of avoided-level crossings, a phenomenon called ``diabatic cascade".

These results raise the question of whether adiabatic evolution is in fact the most efficient choice for running a quantum adiabatic algorithm. After all, the goal is to find the ground state once, with the highest probability and in the shortest amount of time. Therefore, rather than maximizing the probability by increasing the evolution time $t_f$, we can instead use many rapid repetitions of the algorithm to simultaneously shorten $t_f$ and increase the success probability. Let $p_{\mathrm{S}}(t_f)$ denote the single-run success probability of the algorithm with evolution time $t_f$. The probability of failing to find the ground state after $R$ independent repetitions is $(1-p_\mathrm{S})^{R}$, so the probability of succeeding at least once is $1-(1-p_\mathrm{S})^{R}$, which we set equal to the desired probability $p_\mathrm{d}$. The trade-off between success probability and run time is therefore well captured by the time-to-solution (TTS) metric, which measures the time required to find the ground state at least once with probability $p_\mathrm{d}$ (typically taken to be 99\%):
\beq
\mathrm{TTS}(t_f) = t_f \frac{\ln (1 - p_\mathrm{d})}{\ln [1 - p_\mathrm{S}(t_f)]} \ .
\eeq
Other metrics exist that quantify this tradeoff, e.g., without insisting on finding the ground state \cite{King:2015cs}, or that make use of optimal stopping theory and assign a cost to each run \cite{Vinci:2016tg}.

For $p_\mathrm{S}\lesssim 1$ (close to the adiabatic limit), only a single (or few) repetitions of the algorithm are necessary and the TTS  scales linearly with $t_f$.  As $t_f$ is lowered, the success probability typically decreases and more repetitions are necessary, but the TTS may in fact be lower because of the smaller $t_f$ value.  The optimal $t_f$ for the algorithm minimizes the TTS, and is defined as:
\beq
\mathrm{TTS}_{\mathrm{opt}} = \min_{t_f>0} \mathrm{TTS}(t_f) \ .
\eeq
Benchmarking of algorithms then proceeds as follows.  For a specific class of problem instances of varying sizes $n$, TTS$_{\mathrm{opt}}$ is calculated for each size $n$.  The scaling of the algorithm with $n$ is then determined from the scaling with $n$ of TTS$_{\mathrm{opt}}$, as, e.g., in \cite{q108}.  

One benefit of this approach is in obtaining a quantum scaling advantage over specific classical algorithms.. 
For example, the constant gap perturbed Hamming weight oracular problems \cite{Reichardt:2004} [Sec.~\ref{sec:PHW}] and the ``spike" problem of \cite{Farhi-spike-problem} [Sec.~\ref{sec:spike}] with a polynomially closing quantum gap, can be solved in ${O}(1)$ time using a classical algorithm. However, QA exhibits a scaling advantage over SA for these problems in the sense that QA offers a TTS that scales better than SA with single-spin updates \cite{Muthukrishnan:2015ff}.


\section{Outlook and Challenges}
\label{sec:outlook}

Adiabatic quantum computing has blossomed from a speculative alternative approach for solving optimization problems, to a formidable alternative to other universal models of quantum computing, with deep connections to both classical and quantum complexity theory, and condensed matter physics. 

In this review we have given an account of most of the major theoretical developments in the field. Of course, some omissions were inevitable. For example, a potentially promising application of AQC is in quantum chemistry, where the calculation of molecular energies can be formulated in terms of a second-quantized fermionic Hamiltonian that is mapped, via a generalized Jordan-Wigner transformation \cite{Jordan:28,Bravyi:00}, to a non-stoquastic qubit Hamiltonian \cite{Aspuru-Guzik:05,Seeley:2012}. This mapping generates $k$-local interactions, but perturbative gadgets can be used to reduce the problem to only $2$-local interactions \cite{Babbush:2014}. The ground state of the mapped Hamiltonian can then be prepared using adiabatic evolution followed by appropriate measurements to determine the energy spectrum. However, the scaling of the minimum gap for such a preparation procedure is not known, and hence this is an example of AQC with a non-stoquastic Hamiltonian for which it is unknown whether a quantum speedup is possible. A variety of other interesting AQC results with an unknown speedup, and which we did not have the space to review here in detail, can be found in \cite{Rajak:2014qa,OGorman:2015qf,Kurihara:2014fu,Sato:2014dz,Hashizume:2015hc,Rosenberg:2016tg,Goto:2016kl,Inack:2015oq,Santra:2016cr,Durkin:2016dq,Behrman:2016bh,III:2016qf,Karimi:2016ve,Raymond:2016ly,Miyahara:2016zr,Chancellor:2016ys,Karimi:2016vn,Cao:2016fk}.

Moreover, to make the review comprehensive and detailed enough to be self-contained, we focused only on the closed-system setting, thus completely ignoring the important problem of AQC in open systems, with the associated questions of error correction and fault tolerance. We also left out the experimental work on AQC and quantum annealing. These important topics will be the subject of a separate review \cite{Albash-Lidar:RMP-colloq}.

Due to the prominence of stoquastic Hamiltonians in the body of work on AQC, we coined a new term, StoqAQC, which is roughly what was meant when the term ``quantum adiabatic algorithm" was first introduced. Correspondingly, we devoted a substantial part of this  review to StoqAQC, despite the fact that there are indications that this model of computation may not be more powerful than classical computing. Its prominence is explained by the fact that it is easier to analyze than universal AQC, which requires non-stoquastic terms, and by the fact that it is easier to implement experimentally [see, e.g., \cite{Bunyk:2014hb,Weber:2017aa}]. The relatively short history of AQC has witnessed a fascinating battle of sorts between attempts to show that StoqAQC fails to deliver quantum speedups, and corresponding refutations by clever tweaks. To put this and other results we have discussed in the proper perspective, we conclude with a list of $10$ key theoretical challenges for the field of AQC:
\begin{enumerate}
\item Prove or disprove that StoqAQC is classically efficiently simulatable.
\item Find an NP-hard optimization for which AQC gives a quantum speedup in the worst case.
\item Find a class of non-oracular, physically realizable optimization problems for which AQC gives a quantum speedup.  
\item Identify a subset of non-stoquastic Hamiltonians for which ground state preparation can be done efficiently using adiabatic evolution.
\item Formulate every quantum algorithm that gives a speedup in the circuit model natively as an AQC algorithm (i.e., directly, without using perturbative gadgets).
\item Find a problem that can be solved with a quantum speedup using AQC, that was not previously known from other models of quantum computing.
\item Give a way to decide whether adiabatic evolution gives rise to a stronger or weaker speedup than non-adiabatic (diabatic) evolution for a given problem.
\item Predict the optimal adiabatic schedule for a given problem without \textit{a priori} knowledge of the size and/or position of its spectral gap.
\item Prove or disprove that tunneling can generate a (scaling, not prefactor) quantum speedup in AQC. 
\item Establish the relation between entanglement and quantum speedups using AQC.
\end{enumerate}

Solving these problems will likely keep researchers busy for years to come, require interdisciplinary collaborations, and will significantly advance our understanding of AQC. We hope that this review will
catalyze new and productive approaches, enhancing our repertoire of algorithms that give rise to quantum speedups from the unique perspective of AQC.

\acknowledgments
We are grateful to Vicky Choi, Elizabeth Crosson, Itay Hen, Milad Marvian, Siddharth Muthukrishnan, Hidetoshi Nishimori, and Rolando Somma for useful discussions. This work was supported under ARO Grant No. W911NF-12-1-0523, ARO MURI Grants No. W911NF-11-1-0268 and No. W911NF-15-1-0582, and NSF Grant No. INSPIRE-1551064. The research is based upon work (partially) supported by the Office of the Director of National Intelligence (ODNI), Intelligence Advanced Research Projects Activity (IARPA), via the U.S. Army Research Office contract W911NF-17-C-0050. The views and conclusions contained herein are those of the authors and should not be interpreted as necessarily representing the official policies or endorsements, either expressed or implied, of the ODNI, IARPA, or the U.S. Government. The U.S. Government is authorized to reproduce and distribute reprints for Governmental purposes notwithstanding any copyright annotation thereon.
\newpage

\appendix

\section{Technical calculations}

\subsection{Upper bound on the adiabatic path length $L$}
\label{app:boundproof}

Right below Eq.~\eqref{eq:14} we claimed that an upper bound on $L$ is $\max_s\|\dot{H}(s)\|/\Delta$.
To see this differentiate the eigenstate equation $H\ket{\veps_a}=\veps_a\ket{\veps_a}$ for the normalized instantaneous eigenstate $\ket{\veps_a}$ and inner-multiply by $\bra{\veps_b}$, with $b\neq a$, to get $(\veps_a-\veps_b)\braket{\veps_b}{\dot{\veps}_a} = \bra{\veps_b}\dot{H}\ket{\veps_a}$. 
Let $\Delta_{ba}=\veps_b-\veps_a$ and $\Delta_a = \min_{b}\min_s\Delta_{ba}(s)$.
Using our phase choice: 
\begin{align}
| \dot{\varepsilon}_a \rangle &= \sum_b | \varepsilon_b \rangle \langle \varepsilon_b | \dot{\varepsilon}_a \rangle = \sum_{b \neq a} | \varepsilon_b \rangle \langle \varepsilon_b | \dot{\varepsilon}_a \rangle \notag \\
&=   - \sum_{b \neq a}  | \varepsilon_b \rangle \langle \varepsilon_b | \dot{H} | \varepsilon_a \rangle / \Delta_{ba}\ .
\end{align}
Thus 
\begin{align}
\Vert | \dot{\varepsilon}_a \rangle \Vert &\leq  \frac{1}{\Delta_a} \Vert \sum_{b \neq a} | \varepsilon_b \rangle \langle \varepsilon_b | \dot{H} | \varepsilon_a \rangle \Vert  \notag \\
&  \leq   \frac{1}{\Delta_a} \Vert  \sum_{b \neq a} | \varepsilon_b \rangle \langle \varepsilon_b | \Vert \Vert  \dot{H} | \varepsilon_a \rangle  \Vert \leq  \frac{1}{\Delta_a}\Vert  \dot{H} | \Vert \ ,
\end{align}
where in the last equality we used the definition of the operator norm and the fact that $\sum_{b \neq a} | \varepsilon_b \rangle \langle \varepsilon_b|$ is a projector. Integration just multiplies by $1$.

\subsection{Proof of the inequality given in Eq.~\eqref{eq:10c}}
\label{app:eq:10c}
Note that 
\bes
\begin{align}
\norm{\partial_s H[A(s)]} &= |\partial_s{A}|\norm{H_1-H_0} \leq 2 |\partial_s{A}|\\
\norm{\partial^2_s H[A(s)]} &\leq 2 |\partial^2_s{A}|
\end{align}
\ees
for the interpolating Hamiltonian \eqref{eqt:GroverH}. Also note that 
\bes
\begin{align}
\partial^2_s A(s) &= c p \Delta^{p-1} [A(s)] \frac{d\Delta}{d A} \partial_s A(s)  \\
&= c^2 p \frac{d\Delta}{d A} \Delta^{2p-1}[A(s)]\ .
\end{align}
\ees
Thus:
\bes
\begin{align}
&\int_0^1 \left(\frac{\norm{\partial^2_s H[A(s)]}}{\Delta^{2}[A(s)]} + \frac{\norm{\partial_s H[A(s)]}^2}{\Delta^{3}[A(s)]} \right)ds \notag \\
&\quad \leq 2 \int_0^1 \left(\frac{|\partial^2_s A|}{\Delta^{2}[A(s)]} +  \frac{2|\partial_s A|^2}{\Delta^{3}[A(s)]} \right)ds \\
& \quad = \int_0^1 2c^2 \left(p \frac{d\Delta}{d A}+2\right) \Delta^{2p-3}[A(s)] ds \\
\label{eq:appA3}
&\quad = 2c  \int_0^1 \left(p \frac{d\Delta}{d u}+2\right) \Delta^{p-3}(u) du\\
\label{eq:appA4}
&\quad = 4c\int_0^1 \Delta^{p-3}(u) du  \ ,
\end{align}
\ees
where in line~\eqref{eq:appA3} we used the change of variables $u=A(s)$, so that $ds=du/\partial_s A = du/[c \Delta^p(u)]$, and in line~\eqref{eq:appA4} we used $B(p)\equiv 2c  \int_0^1 \Delta^{p-3}d\Delta = 0$, since $B(2)=2c\ln[\Delta(1)/\Delta(0)]=0$ and $B(p\neq 2)= \frac{2c}{p-2}\left(\Delta^{p-2}(1)-\Delta^{p-2}(0)\right)=0$, due to the boundary conditions $\Delta(0)=\Delta(1)=1$ [Eq.~\eqref{eq:gap-Grover}].

\section{A lower bound for the adiabatic Grover search problem}
\label{app:lower-bound}

Here we show that there is no schedule that gives a better scaling for the adiabatic Grover problem than the one discussed in Sec.~\ref{sec:quadspeedup}, resulting in a quadratic quantum speedup. The argument is due to \cite{Roland:2002ul}, which in turn is based on the general Hamiltonian quantum computation argument by \cite{FarhiAnalog}. 

To show this, consider two different searches, one for $m$ and another for $m'$. We do not allow the schedule to depend on $m$, i.e., the same schedule must apply to all marked states. Let us denote the states for each at the end of the algorithm by $\ket{\psi_m(t_f)}$ and $\ket{\psi_{m'}(t_f)}$.  In order to be able to distinguish if the search gave $m$ or $m'$, we must require that $\ket{\psi_m(t_f)}$ and $\ket{\psi_{m'}(t_f)}$ are sufficiently different. Let us define the distance (or infidelity)
\beq
D_{mm'}(t) \equiv 1 - | \braket{\psi_m(t)}{\psi_{m'}(t)} |^2\ ,
\eeq
[note that $D_{mm'}(0)=D_{mm}(t)=0$] and demand that:
\beq
D_{mm'}(t_f)  \geq \eps \ , \quad m \neq m' \ .
\eeq
First, we have a lower bound on the sum:
\begin{align}
\label{eq:lowerboundsum}
\sum_{m,m'}  D_{mm'}(t_f) &=  \sum_{m \neq m'} D_{mm'}(t_f) \notag \\
&\geq  \sum_{m \neq m'} \epsilon = N(N-1) \epsilon \ .
\end{align}

Next, let us find an upper bound on the sum. We write the Hamiltonian \eqref{eqt:GroverH} explicitly as $H(t) = \ident - [1-A(t)] \ketbra{\phi}{\phi} + H_{1m}(t)$ where $H_{1m}(t) = -A(t) \ketbra{m}{m}$.%
\footnote{Optimality applies to arbitrary driving
Hamiltonians. Hence the lower bound
holds more generally, and does in fact not require the initial Hamiltonian to be a projector onto the
uniform superposition as we have done here for simplicity.} 
Then:
%
\begin{align}
\label{eqt:Grover2Terms}
\frac{d}{dt} D_{mm'}(t) & = 2 \Im \left[ \bra{\psi_{m}}(H_{1m} - H_{1m'}) \ket{\psi_{m'
}} \braket{\psi_{m'}}{\psi_{m}} \right] \notag \\
& \leq 2 \left|  \bra{\psi_{m}}(H_{1m} - H_{1m'}) \ket{\psi_{m'
}} \braket{\psi_{m'}}{\psi_{m}} \right| \notag \\
&  \leq 2 |  \bra{\psi_{m}} H_{1m}  \ket{\psi_{m'}} | + 2 |  \bra{\psi_{m}} H_{1m'}  \ket{\psi_{m'}} |  \ .
\end{align}
Let us now sum over all $m$ and $m'$:
\begin{align}
& \sum_{m,m'} \frac{d}{dt} D_{mm'}(t) \leq 4 \sum_{m,m'}  | \bra{\psi_{m}} H_{1m}  \ket{\psi_{m'}} |  \\
& \quad  \leq 4 \sum_{m,m'} \Vert H_{1m} \ket{\psi_{m'}} \Vert \Vert \ket{\psi_{m}} \Vert = 4 \sum_{m,m'} \Vert H_{1m} \ket{\psi_{m'}} \Vert \ , \notag
\end{align}
where we first used the fact that under the sum the two terms in the last line of Eq.~\eqref{eqt:Grover2Terms} are identical, and then we used the Cauchy-Schwartz inequality ($| \langle x|y \rangle | \leq \Vert x \Vert \Vert y \Vert$).  Now we note that:
\begin{eqnarray}
\sum_{m} \norm{ H_{1m} \ket{\psi_{m'}}}^2 &=& \sum_m \bra{\psi_{m'}} H_{1m} H_{1m} \ket{\psi_{m'}}  \\
&=& A^2(t) \sum_m \braket{\psi_{m'}}{m}\braket{m}{\psi_{m'}} = A^2(t) \ ,\notag
\end{eqnarray}
so that
\beq
\left( \sum_m \Vert H_{1m} \ket{\psi_{m'}} \Vert \right)^2 = \left( \vec{x} \cdot \vec{y} \right)^2 \leq (\vec{x} \cdot \vec{x}) (\vec{y} \cdot \vec{y}) = N A^2(t)\ ,
\eeq
where $\vec{x} = \left(\Vert H_{11} \ket{\psi_{m'}} \Vert , \Vert H_{12} \ket{\psi_{m'}} \Vert , \dots,  \Vert H_{1N} \ket{\psi_{m'}} \Vert  \right)$ and $\vec{y} = \left(1, 1, \dots, 1 \right)$.  Therefore, we have:
\begin{align}
 \sum_{m,m'} \frac{d}{dt} D_{mm'}(t) &\leq 4 \sum_{m,m'} \Vert H_{1m} \ket{\psi_{m'}} \Vert \nonumber \\
&  \leq 4 \sum_{m'} \sqrt{N} A(t) = 4 N \sqrt{N} A(t)\ .
\end{align}
If we integrate both sides we have:
\begin{align}
\sum_{m,m'} D_{mm'}(t_f) &\leq 4 N \sqrt{N} \int_0^{t_f} A(t) dt \leq 4 N \sqrt{N} t_f  \ .
\end{align}
Combining this with Eq.~\eqref{eq:lowerboundsum}, we have $N (N-1) \epsilon \leq 4 N \sqrt{N} t_f $, and hence:
\beq
t_f \geq \frac{\epsilon}{4} \frac{N-1}{\sqrt{N}}\ ,
\eeq
so that the computation must last a minimum time of $O(\sqrt{N})$ if the schedule is to be agnostic about the identity of the marked state.  
Therefore, the solution using the locally optimized schedule is asymptotically optimal.

\section{Technical details for the proof of the universality of AQC using the History State construction}
\label{app:AQC-universality-proof}

We add details to the proof sketch given in Sec.~\ref{sec:history-state}.

\subsection{$\ket{\gamma(0)}$ is the ground state of $H_{\mathrm{init}}$}
Let us first check that $\ket{\gamma(0)}$ is the ground state of $H_{\mathrm{init}}$ with eigenvalue $0$.  Note that $H_{\mathrm{init}}$ is a sum of projectors, so it is positive semi-definite.
Therefore if we find a state with energy $0$, then it is definitely a ground state.  The all-zero clock state is annihilated by $H_{\mathrm{c}}$, $H_{\mathrm{input}}$, and $H_{\mathrm{c-init}}$, so we have $H_{\mathrm{init}} \ket{\gamma(0)} = 0$, i.e., $\ket{\gamma(0)}$ is a ground state of $H_{\mathrm{input}}$.  We shall show later that it is a unique ground state.

\subsection{$\ket{\eta}$ is a ground state of $H_{\mathrm{final}}$}
Next we check that $\ket{\eta}$ is a ground state of $H_{\mathrm{final}}$ with eigenvalue $0$.  First, we wish to show that $H_{\mathrm{final}}$ is positive semi-definite.  We already know that $H_{\mathrm{input}}$ and $H_{\mathrm{c}}$ are positive semi-definite, so we only need to show this to be the case for the $H_\ell$'s.  This follows since it is easily checked that 
$H_{\ell} = H_{\ell}^\dagger = \frac{1}{2}H_{\ell}^2$, so that 
$\bra{X} H_{\ell} \ket{X} =  \frac{1}{2} \bra{X} H_{\ell}^{\dagger} H_{\ell} \ket{X} 
= \frac{1}{2} \Vert H_{\ell} \ket{X} \Vert^2 \geq 0$. 
Thus $H(s)$ is positive semi-definite, since it is a sum of positive semi-definite terms.

Let us now check that $H_{\mathrm{final}}$ annihilates $\ket{\eta}$.  First, because $\ket{\eta}$ only involves legal clock states, it is annihilated by $H_{\mathrm{c}}$.  Next,
\beq
H_{\mathrm{input}} \ket{\eta} = H_{\mathrm{input}} \frac{1}{\sqrt{L+1}} \ket{\alpha(0)} \otimes \ket{0^L}_{\mathrm{c}} = 0\ .
\eeq
Finally, note that the only non-zero terms in $\sum_{\ell=0}^L H_{\ell} \ket{\eta}$ are of the form:
\bes
\begin{align}
& H_{\ell} \ket{\alpha(\ell-1)} \otimes \ket{1^{\ell-1} 0^{L-\ell+1}}_{\mathrm{c}} =  \\ 
&\ket{\alpha(\ell-1)} \otimes \ket{1^{\ell-1} 0^{L-\ell+1}}_{\mathrm{c}} 
- \ket{\alpha(\ell)} \otimes \ket{1^{\ell} 0^{L-\ell}}_{\mathrm{c}} \nonumber \\
& H_{\ell} \ket{\alpha(\ell)} \otimes \ket{1^{\ell} 0^{L-\ell}}_{\mathrm{c}}  = \\
&-\ket{\alpha(\ell-1)} \otimes \ket{1^{\ell-1} 0^{L-\ell+1}}_{\mathrm{c}} 
+ \ket{\alpha(\ell)} \otimes \ket{1^{\ell} 0^{L-\ell}}_{\mathrm{c}} \nonumber \ ,
\end{align}
\ees
which cancel.  Therefore, $\ket{\eta}$ has eigenvalue $0$ and is a ground state of $H_{\mathrm{final}}$.

\subsection{Gap bound in the space spanned by $\left\{ \ket{\gamma(\ell)} \right\}_{\ell=0}^L$}
Let $\mS_0$ be the space spanned by $\left\{ \ket{\gamma(\ell)} \right\}_{\ell=0}^L$.  Let us show that $H(s)$ acting on any state in $\mS_0$ keeps it in this subspace:
\bes \label{eqt:HinS0}
\begin{align}
H_{\mathrm{c}} \ket{\gamma(\ell)} & = 0  \ , \\
H_{\mathrm{input}} \ket{\gamma(\ell)} & = 0  \ , \\
H_{\mathrm{c-init}} \ket{\gamma(\ell)} & = \left\{
 \begin{array}{lr}
0 \ ,  & \ell  = 0 \\
 \ket{\gamma(\ell)} \ , & \ell \neq 0  
 \end{array} \right. \\
 H_{\ell} \ket{\gamma(\ell')} & = \delta_{\ell',\ell} \left( \ket{\gamma(\ell-1)} - \ket{\gamma(\ell)} \right) \nonumber \\
 & + \delta_{\ell', \ell} \left( \ket{\gamma(\ell)} - \ket{\gamma(\ell-1)} \right)\ .
\end{align}
\ees
Since the initial state $\ket{\gamma(0)} \in \mS_0$, the dynamics generated by $H(s)$ keep the state in $\mS_0$. Thus, it is sufficient to bound the gap in this subspace.  In the basis given by $\left\{ \ket{\gamma(\ell)} \right\}_{\ell=0}^L$, we can write an $(L+1) \times (L+1)$ matrix representation of the Hamiltonian in the $\mS_0$ subspace, which, using Eq.~\eqref{eqt:HinS0}, is:
\begin{eqnarray}
H_{\mS_0}(s) &=& (1-s)\left( \begin{array}{rrrrrr}
0 & 0 & 0 &\dots &&\\
0 & 1 &  0 \\
& & 1 \\
&&&\ddots\\
&&&&1&0\\
&&&&&1
\end{array} \right)  \\
&&  + s  \left( \begin{array}{rrrrrr}
\frac{1}{2} & - \frac{1}{2} & 0 & 0 & \dots & 0\\
-\frac{1}{2} & 1 &  - \frac{1}{2} & \\
0 & -\frac{1}{2} & 1 & -\frac{1}{2} & \\
\vdots&& \ddots & \ddots & \ddots \\
&&&-\frac{1}{2}&1& -\frac{1}{2} \\
0 &&&& -\frac{1}{2} & \frac{1}{2}
\end{array} \right)\ . \notag
\end{eqnarray}

\subsubsection{Bound for $s < 1/3$}
Let us first bound the gap for $s < 1/3$.  The Gerschgorin circle theorem states \cite{Gershgorin}:\\ \emph{Let $A$ be any matrix with entries $a_{ij}$.  Consider the disk $D_i$ (for $1 \leq i \leq n$) in the complex plane defined as: $D_i = \left\{ z \ \Big|\  |z-a_{ii} | \leq \sum_{j \neq i} |a_{ij}| \right\}$.  Then the eigenvalues of $A$ are contained in $\cup_i D_i$ and any connected component of $\cup_i D_i$ contains as many eigenvalues of $A$ as the number of disks that form the component.} 

Consider the cases $i=1$, $i=L+1$, and $i\neq1,L+1$ separately.  Note that when $s < 1/3$:
\bes
\begin{align}
\left[H_{\mS_0}(s) \right]_{11} & = \frac{1}{2} s < \frac{1}{6} \ , \quad \sum_{j \neq 1} | a_{1j}|  = \frac{1}{2} s < \frac{1}{6}\ , \\
\left[H_{\mS_0}(s) \right]_{L+1,L+1} & = 1 - \frac{1}{2} s > 5/6 \notag \\
&\qquad \sum_{j \neq L+1} | a_{L+1,j}|  = s < \frac{1}{3}\ , \\
\left[H_{\mS_0}(s) \right]_{ii} & = 1 \ , \quad i \neq 1, L+1 \notag \\
& \qquad \sum_{j \neq i} | a_{1j}|  = \frac{1}{2} s < \frac{1}{3} \ .
\end{align}
\ees
Therefore, we can identify a disk $D_1$ centered at $z \leq \frac{1}{6}$ on the real line with radius $\leq \frac{1}{6}$. The closest possible disk to it which does not overlap it, is the disk $D_{L+1}$ centered at $z \geq 5/6$ with a radius $\leq 1/3$.  Therefore, since no disks intersect $D_1$, it covers the smallest values on the real line, and it follows that the ground state for $s < 1/3$ is unique.  Furthermore, we have learned that the minimum gap is a constant of at least $1/6$, since that is the closest distance between $D_1$ and $D_{L+1}$.  (This also proves that $\ket{\gamma(0)}$ is the unique ground state at $s = 0$.)

\subsubsection{Bound for $s \geq 1/3$}
Now let $s \geq 1/3$ and consider the matrix representation of $G(s) \equiv \ident - H_{\mS_0}(s)$ in the same basis:
\beq
G(s) = \left( \begin{array} {rrrrr}
1-\frac{1}{2}s & \frac{1}{2}s \\
\frac{1}{2} s & 0 & \frac{1}{2} s\\
& \ddots & \ddots & \ddots \\
&&&0 & \frac{1}{2}s \\
&&& \frac{1}{2}s & \frac{1}{2}s
\end{array} \right)\ .
\eeq
This matrix is Hermitian and has all non-negative real entries for $0<s\leq 1$.   Note that increasing powers of $G(s)$ fill more of the matrix, and $G(s)^{L+1}$ has all positive entries for $0<s\leq 1$.  We can thus invoke Perron's theorem:\\ 
\emph{Let $G$ be a Hermitian matrix with real non-negative entries.  If there exists a finite $k$ such that all entries of $G^k$ are positive, then $G$'s largest eigenvalue is positive and all other eigenvalues are strictly smaller in absolute value.  The eigenvector corresponding to the largest eigenvalue is unique, and all its entries are positive.}  

Therefore, by Perron's theorem, $G(s)$'s largest eigenvalue $\mu$ must be positive, and the associated unique eigenvector $\vec{\alpha} = \left(\alpha_1, \dots, \alpha_{L+1} \right)$ has $\alpha_i > 0$.  Let us use this to define a matrix $P$ with entries $P_{ij} = \frac{\alpha_j}{\mu \alpha_i} G_{ij} \geq 0$, such that 
\beq
\sum_j P_{ij} = \frac{1}{\mu \alpha_i} \sum_j G_{ij} \alpha_j  = 1 \ , 
\eeq
where we used that $\vec{\alpha}$ is an eigenvector of $G$ with eigenvalue $\mu$.  Thus $P$ is a stochastic matrix (it has only non-negative entries and its rows sum to $1$).  Now note that if $\left(\alpha_1 v_1, \dots \alpha_{L+1} v_{L+1} \right)$ is a left eigenvector of $P$ with eigenvalue $\nu/\mu$, then $\left( v_1, \dots v_{L+1} \right)$ is an eigenvector of $G$ with eigenvalue $\nu$:
\begin{eqnarray}
\frac{\nu}{\mu} \alpha_j v_j &=& \sum_i \alpha_i v_i P_{ij} = \sum_i v_i \frac{\alpha_j}{\mu} G_{ij} = \frac{\alpha_j}{\mu} \sum_i G_{ji} v_i  \nonumber \\
&\Longrightarrow& \nu v_j = \sum_i G_{ji} v_i \ .
\end{eqnarray}
It is straightforward to check that the reverse also holds: if $\left( v_1, \dots v_{L+1} \right)$ is an eigenvector of $G$ with eigenvalue $\nu$, then $\left(\alpha_1 v_1, \dots \alpha_{L+1} v_{L+1} \right)$ is a left eigenvector of $P$ with eigenvalue $\nu/\mu$.
By taking $\vec{v} = \vec{\alpha}$, which corresponds to the largest eigenvalue ($\nu=\mu$) of $G$, it then follows that  $\vec{\alpha^2} = \left( \alpha_1^2, \dots \alpha_{L+1}^2 \right)$ is a left eigenvector of $P$ with the maximal eigenvalue $1$.
If we normalize $\vec{\alpha^2}$, i.e., define 
\beq
\label{eq:Pi}
\vec{\Pi} =  \frac{1}{Z}\left( \alpha_1^2, \dots \alpha_{L+1}^2 \right)\ , \quad Z = \sum_i \alpha_i^2\ ,
\eeq
then $\vec{\Pi}$ is the limiting distribution of $P$, i.e., $P \vec{\Pi} = \vec{\Pi}$.  We can then relate the energy gap between the highest and second highest eigenvalue (let us denote it by $\delta/\mu$) of $P$ to the energy gap between the ground state energy of $H$ (given by $1-\mu = \lambda$) and the first excited state (given by $1-\delta$)
\beq
\Delta_{\mathrm{largest}}(P) = 1 - \frac{\delta}{\mu} = \frac{\mu - \delta}{\mu} = \frac{\Delta(H_{\mS_0})}{1- \lambda} \ .
\eeq
where ``largest" denotes the gap from the largest eigenvalue of $P$ to the next largest eigenvalue.
We wish to bound the gap of $P$ and hence of $H(s)$.  Let us define a non-empty set $\mB \subseteq \left\{1, 2, \dots, L+1 \right\}$ satisfying $\sum_{i \in \mB} \Pi_i \leq \frac{1}{2}$, where $\Pi_i$ are the entries of $\vec{\Pi}$.  Then the conductance of $P$, $\varphi(P)$ is defined as:
\beq \label{eqt:conductance}
\varphi(P) = \min_{\mB } \frac{F(\mB)}{\Pi(\mB)}\ ,
\eeq
where
\bes
\begin{align}
F(\mB) & = \sum_{i \in \mB} \sum_{j \notin \mB} \Pi_i P_{ij}  \label{eqt:F(B)}\ ,\\
\Pi(\mB) & = \sum_{i \in \mB} \Pi_i\ .
\end{align}
\ees
The Conductance bound \cite{Sinclair198993} then states that
\beq \label{eqt:conductancebound}
\Delta_{\mathrm{largest}}(P) \geq \frac{1}{2} \varphi(P)^2\ .
\eeq

To use the result of the Conductance bound, we would like to show that the ground state of $H(s)$ [and hence the eigenstate associated with the largest eigenvalue of $G(s)$] is monotone, i.e., that $\alpha_1 \geq \alpha_2 \geq \dots \geq \alpha_{L+1} \geq 0$.  The case $s=0$ is obvious, so consider $s>0$. First note that $G(s)$ applied to a monotone vector generates a monotone vector, i.e., $G(s)$ preserves monotonicity.  To see, this consider $G(s) \vec{v} = \vec{w}$ with $\vec{v}$ monotone.  The components of $\vec{w}$ are given by:
\bes
\begin{align}
w_1 &= \left( 1- \frac{1}{2}s \right) v_1 + \frac{1}{2} s v_2 \\
w_k & = \frac{1}{2} s v_{k-1} + \frac{1}{2} s v_{k+1} \ , \quad 2 \leq k \leq L \\
w_{L+1} & = \frac{1}{2} s v_L + \frac{1}{2} s v_{L+1}
\end{align}
\ees
Therefore we have:
\bes
\begin{align}
w_1 - w_2 &= (1 -  s) v_1 + \frac{1}{2}s \left( v_2 - v_3 \right) \\
w_k - w_{k+1} & = \frac{1}{2} s \left( v_1 - v_2 + v_3 - v_4 \right) \ , \quad 2 \leq k \leq L-1 \\
w_{L} - w_{L+1} & = \frac{1}{2} s \left( v_{L-1} - v_L \right)\ ,
\end{align}
\ees
which clearly are all $\geq 0$ by the monotonicity of $\vec{v}$ and $s\leq 1$.  Therefore $\vec{w}$ is also monotone.

Recall that $G(s)$ is Hermitian, so it has an orthonormal set of eigenvectors $\left\{ \ket{v_i} \right\}_{i=1}^{L+1}$ with eigenvalues $\mu_i$, where $\ket{v_1} = \vec{\alpha}$ and $\mu_1  = \mu$.\footnote{Note that we abuse notation and mix kets with standard vectors here, and also do not use transpose notation to distinguish column from row vectors.} Because these eigenvectors form a basis we can always find a set of coefficients  $\left\{c_i \right\}_{i=1}^{L+1}$ such that:
\beq
\sum_i c_i \ket{v_i} = \left(1, \dots, 1 \right) = \vec{1}\ .
\eeq 
Then:
\bes
\begin{align}
\label{eq:115a}
&\left( \frac{1}{\mu_1} G(s) \right)^{k} \sum_i c_i \ket{v_i} = \sum_i \left( \frac{\mu_i}{\mu_1} \right)^{k} c_i \ket{v_i}  \\
&\qquad \Longrightarrow \left( \frac{1}{\mu_1} G(s) \right)^{k} \vec{1}^{\mathrm{T}} = \sum_i \left( \frac{\mu_i}{\mu_1} \right)^{k}  \vec{1}^{\mathrm{T}}\ .
\end{align}
\ees
Using $| \mu_i | < \mu_1$ by Perron's theorem, we have from Eq.~\eqref{eq:115a} that:
\beq
\lim_{k \to \infty} \left( \frac{1}{\mu_1} G(s) \right)^{k} \sum_i c_i \ket{v_i} = c_1 \ket{v_1} \ ,
\eeq
Since the quantities $\left( \frac{1}{\mu_1} G(s) \right)^{k}$ (for $k\geq L+1$), $\sum_i c_i \ket{v_i}=\vec{1}$, and $\ket{v_1} = \vec{\alpha}$ are all positive, it follows that also $c_1 > 0$.  Since $\vec{1}$ is monotone and $G(s)$ preserves monotonicity, we have finally established that $\ket{v_1} = \vec{\alpha}$ is monotone.  This then implies that $\vec{\Pi}$ [Eq.~\eqref{eq:Pi}] is monotone.

We are ready to calculate the conductance $\varphi(P)$. 
First, consider the case where the first index (of $\vec{\Pi}$) is in the set $\mB$, i.e., $1 \in \mB$.  Let $k$ be the smallest index such that $k \in \mB$ but $k+1 \notin \mB$.  (Note that from the form of $P$, only $P_{11}$, $P_{j,j+1}$, $P_{L+1,L+1}$ are nonzero.)  Then we have for $F(\mB)$:
\begin{eqnarray}
F(\mB) &=& \sum_{i\in \mB, {i\neq k}} \sum_{j \notin \mB} {\Pi_i P_{ij}}+ \Pi_k P_{k,k+1}  \geq \Pi_k P_{k,k+1} \nonumber \\
&=& \Pi_k \frac{\sqrt{\Pi_{k+1}}}{\mu \sqrt{\Pi_k}} \left[G(s)\right]_{k,k+1}= \frac{\sqrt{\Pi_k \Pi_{k+1}}}{1-\lambda} \left[G(s)\right]_{k,k+1} \nonumber \\
&\geq& \frac{\Pi_{k+1}}{1-\lambda} \left[G(s)\right]_{k,k+1}\ ,
\end{eqnarray}
where the last inequality follows from the monotonicity of $\vec{\Pi}$. Because $0 < 1 - \lambda \leq 1$, and $\left[ G(s) \right]_{k,k+1} = \frac{1}{2} s \geq 1/6$ for $s \geq 1/3$, it follows that:
\beq
F(\mB = \left\{1, \mathrm{others}, k \right\} ) \geq \frac{\Pi_{k+1}}{6}\ .
\eeq
Since by definition $\Pi(\mB) \leq 1/2$, then $\Pi(\bar{\mB}) \geq 1/2$ where $\bar{\mB}$ is the complement of ${\mB}$, but since the largest possible size of $\bar{\mB}$ is $L$ (recall that $1\in\mB$), it follows that $\Pi(\bar{\mB}) \leq L \Pi_{k+1}$, so that $\Pi_{k+1} \geq 1/(2L)$, and hence:
\beq \label{eqt:ratioFPi}
\frac{F(\mB = \left\{1, \mathrm{others}, k \right\} )}{\Pi(\mB)} \geq \frac{1}{6L}\ .
\eeq
Next consider the case where $1 \notin \mB$.  Now let $k$ be the smallest index such that $k \notin \mB$ but $k+1 \in \mB$. Then:
\bes
\begin{align}
F(\mB) & = \sum_{i \in \mB, i \neq k+1} \sum_{j \notin \mB} \Pi_i P_{ij} + \Pi_{k+1} P_{k+1,k} \\
& \geq \Pi_{k+1} P_{k+1,k} \geq \frac{\Pi_{k+1}}{6}\ .
\end{align}
\ees
In this case, since the maximum size of $\mB$ is $L$ but it excludes the index $1$, we have $\Pi(\mB) \leq L \Pi_{k+1}$, so that $F(\mB) \geq \Pi(\mB)/(6L)$.  Therefore, we again find the condition~\eqref{eqt:ratioFPi}.  Thus, by the conductance bound [Eq.~\eqref{eqt:conductancebound}]:
\beq
\Delta(P) = \frac{\Delta(H_{\mS_0})}{1- \lambda} \geq \frac{1}{2} \left( \frac{1}{6L} \right)^2\ .
\eeq
Now since $\lambda$ is the ground state of $H_{\mS_0}$, for any state in $\ket{v}\in\mS_0$, we must have $\bra{v} H_{\mS_0} \ket{v} \geq \lambda$.  In particular:
\beq
\bra{\gamma(0)} H_{\mS_0} \ket{\gamma(0)} = \frac{1}{2} s \geq \lambda \ ,
\eeq
i.e., $\lambda \leq 1/2$. This finally yields Eq.~\eqref{eq:gapHS0}.

\subsection{Gap bound in the entire Hilbert space}
Let us now go a step further and show how the global gap (i.e., not restricted to the $\mS_0$ subspace) scales with $L$.  Let $\mS$ denote the subspace spanned by all legal clock states.  The dimensions of this subspace will be $\dim(\mS) = (L+1) 2^n$, since we have $L+1$ legal clock states and $2^n$ computational states.  $H(s)$ acting on any state in $\mS$ does not generate any illegal clock states, so for any $\ket{v} \in \mS$ we have $H(s) \ket {v} \in \mS$.  Similarly, for any state in the orthogonal subspace $\mS^{\perp}$, i.e., the subspace of illegal clock states, for any state $\ket{v^{\perp}} \in \mS^{\perp}$, we have $H(s) \ket{v^{\perp}} \in \mS^{\perp}$.  Therefore, the eigenstates of $H(s)$ below either to $\mS$ or to $\mS^{\perp}$, and $H(s)$ is block diagonal with blocks $H_{\mS}(s)$ and $H_{\mS^{\perp}}(s)$ that can be diagonalized independently.

Let us first restrict to $H_{\mS^{\perp}}(s)$.  $H_c$ penalizes all illegal clock states by at least one unit of energy and acts as the identity on the computational qubits.  Therefore, it shifts the entire spectrum of $\mS^{\perp}$ by at least one unit of energy.  Since the remaining terms are positive semi-definite, they cannot lower the energy.  Therefore, regardless of the form of the ground state in the subspace, it has an energy of at least one unit.    

Let us now restrict to $H_{\mS}(s)$ and define:
\beq \label{eqt:gammaj}
\ket{\gamma_j(\ell)} = \ket{\alpha_j(\ell)} \otimes \ket{1^{\ell} 0^{L-\ell}}_\mathrm{c}\ ,
\eeq
where $\ket{\alpha_j(\ell)}$ is the state of the circuit at time $\ell$ had the input state been given by the binary representation of $j$ (e.g., if $j = 4$, the input configuration of the circuit would have been $\ket{0^{n-3}1_30_20_1}$).  Note that $\ket{\gamma_0(\ell)} = \ket{\gamma(\ell)}$.  Let $\mS_j$ denote the space spanned by $\left\{ \ket{\gamma_j(\ell)} \right\}_{\ell=0}^L $.  Since $H_\mS(s)$ cannot mix states with different $j$ subindices (it can only propagate forward or backward in $\ell$), $H_{\mS}(s)$ is block diagonal in the subspaces $\mS_j$.  Therefore, we only need to find the minimum ground state energy of $H_{\mS_{j>0}}(s)$ to determine the minimum gap from $H_{\mS_0}(s)$.  

In order to determine the ground state energy of $H_{\mS_j}$, we notice that we can write:
\beq
H_{\mS_j}(s) = H_{0}(s) + H_{\mS_j, \mathrm{input}} \ , \quad j > 0
\eeq
where $H_0(s)$ has exactly the same spectral properties as $H_{\mS_0}$ except in the $\mS_j$ subspace.  The reason for this decomposition is because $H_{\mathrm{input}}$ is zero in $\mS_0$ and hence is absent from $H_{\mS_0}(s)$.  Note that:
\beq
H_{\mS_j, \mathrm{input}} \ket{\gamma_j(\ell)} = \left\{ \begin{array}{lr}
k \ket{\gamma_j(0)} , & \ell = 0 \\
0 \ , & \ell > 0
\end{array} \right.
\eeq
(recall that  $H_{\mathrm{input}}$ projects onto the $0$ clock state, which is why for $\ell > 0$ we have zero.)  Therefore, in the basis $\left\{ \ket{\gamma_j(\ell)} \right\}_{\ell=0}^L$, we can write the matrix representation of $H_{\mS_j, \mathrm{input}}$ as:
\beq \label{eqt:HSj}
H_{\mS_j, \mathrm{input}} = \left( \begin{array}{rrrrr}
k & & & 0 \\
& 0 \\
& & \ddots \\
0&&& 0
\end{array} \right)\ ,
\eeq
where $k \geq 1$ denotes the number of 1's in the binary representation of $j > 0$.    In particular, note that it is diagonal in this basis.  We now use the Geometrical Lemma  
\cite{Kitaev:book,Aharonov:2002nr}:\\
\begin{lemma}[Geometrical Lemma] 
Let $H_1$ and $H_2$ be two Hamiltonians with ground state energies $g_1$ and $g_2$ respectively.  Both Hamiltonians have a ground state energy gap to the first excited state that is larger than $\Lambda$.  Let the angle between the two ground subspaces be $\theta$.\footnote{The angle $\theta$ is defined via $\cos \theta = \max_{v_1,v_2} |\braket{v_1}{v_2}|$, where $\ket{v_i}$ belongs to space $i$.}.  Then the ground state energy ($g_0$) of $H_0 = H_1 + H_2$ is at least $g_1 + g_2 +  \Lambda (1- \cos \theta)$.
\end{lemma}

Let $H_1 = H_{0}$ and $H_2 = H_{\mS_1, \mathrm{input}}$.  We saw that the ground state gap of $H_1$ is $\Omega(1/L^2)$ and that of $H_2$ is $1$, so we can take $\Lambda = \Omega(1/L^2)$.  The ground state energy of $H_2$ is $g_2 = 0$.  Therefore, using the Geometrical Lemma, we have $g_0 - g_1 \geq \Lambda( 1- \cos \theta)$. It remains to bound the angle between the two ground spaces.  From Eq.~\eqref{eqt:HSj}, it is clear that the (degenerate) ground state of $H_2$ can be written as a linear combination of $\left\{ \ket{\gamma_j(\ell)} \right\}_{\ell=1}^L$, whereas the (unique) ground state of $H_1$ can be written as a monotone vector in $\left\{ \ket{\gamma_j(\ell)} \right\}_{\ell=0}^L$.  Therefore:
\begin{align}
\cos \theta  &= \max_{\left\{ c_{\ell'} \right\}} \left| \sum_{\ell = 0}^L \alpha_\ell \bra{\gamma_j (\ell)} \sum_{\ell'=1}^L c_{\ell'} \ket{\gamma_j(\ell')} \right| \notag \\
& =  \max_{\left\{ c_{\ell'} \right\}} \left| \sum_{\ell = 1}^L \alpha_\ell c_\ell \braket{\gamma_j (\ell)} {\gamma_j(\ell)} \right|  \\
&\leq \max_{\left\{ c_{\ell'} \right\}}  \sum_{\ell = 1}^L \alpha_\ell \left| c_\ell \right|  \leq \sqrt{\frac{L}{L+1}}  \leq 1 - \frac{1}{2 L}  \ , \notag
\end{align}
where we have used that $\vec{\alpha}$ is monotone so 
that $c_\ell = \alpha_\ell \sqrt{\frac{L+1}{L}}$ maximizes the sum.  Therefore, the global gap can be bounded from below by $\Omega(1/L^3)$, which is Eq.~\eqref{eq:gapHL3}.
\section{Proof of the Amplification Lemma (Claim~\ref{claim-ampl})}
\label{app:amplification-lemma}

To prove Claim~\ref{claim-ampl}, define a new verifier $V^{\ast}(\eta,X)$ which amounts to repeating $V(\eta,X)$ $K$ times, where $K=\text{poly}(|\eta|)$ to keep the verifier efficient, and taking a majority vote on the output, i.e., $\Pr(V^{\ast}(\eta,X) = 1) = \Pr \left(\sum_{i=1}^K V_i > K/2 \right)$, where $V_i\in\{0,1\}$ is the random number  associated with the $i$-th run of $V(\eta,X)$.

Now recall the multiplicative Chernoff bound:
\begin{align}
\Pr \left( \sum_{i=1}^K Y_i \leq \left( 1- \beta \right) K p \right) & \leq e^{ - \beta^2 K p /2 } \ ,\,\,\, 0 < \beta < 1 \notag \\
\Pr  \left( \sum_{i=1}^K Y_i \geq \left( 1+ \beta \right) K p \right) & \leq e^{  - \beta^2 K p /(2+\beta)  }\ , \,\,\, 0 < \beta \ ,
\end{align}
for $p = \mathbb{E}(Y)$ where $Y\in\{0,1\}$ is a random variable.
Consider first take the case where $Q(\eta) = 1$.  In that case, $p \geq 2/3$.  If we now pick $\beta = 1- 1/(2p)$ (i.e., $1/4\leq \beta \leq 1/2$) in the Chernoff bound, then:
\begin{align}
\Pr \left(\sum_{i=1}^K V_i > \frac{K}{2} \right) &= 1 -  \Pr \left(\sum_{i=1}^K V_i \leq \frac{K}{2} \right)  \\
 &\geq 1 - e^{  - \frac{ (p-1/2)^2 K}{2 p} } \geq  1 - e^{  - \frac{ (2/3-1/2)^2 K}{4/3} } \ .\notag
\end{align}
For the case where $Q(\eta) = 0$, $p \leq 1/3$, take $\beta = 1/(2p) -1 > 0$, so  that:
\begin{align}
 \Pr \left(\sum_{i=1}^K V_i > \frac{K}{2} \right)  & \leq \Pr \left(\sum_{i=1}^K V_i   \geq \frac{K}{2} \right) \\
& \leq e^{  -\frac{(p-1/2)^2 K }{p(p + 1/2)}}  \leq  e^{-\frac{(1/3-1/2)^2 K }{(1/3 + 1/2)/3}} \ . \notag
\end{align}
This shows that MA$(2/3,1/3) =$ MA($1-e^{-|\eta|^g},e^{-|\eta|^g})$, since $K=\text{poly}(|\eta|)$.

To show that MA($c, c-1/|\eta|^g$) $\subseteq$ MA(2/3,1/3), it is sufficient 
to show that when $Q(\eta) = 0$ it is exponentially unlikely that Merlin is able to fool Arthur that $Q(\eta) = 1$.  Therefore consider the probability of fooling Arthur, i.e., $\Pr \left( V^{\ast}(\eta, X) = 1 \right) > c$ when $Q(\eta)=0$.  Take $p = \Pr \left(V(\eta,X) = 1 \right) = c- 1/|\eta|^g$.  Then:
\begin{eqnarray}
\Pr \left( \text{Arthur fooled} \right) &=& \Pr \left( \sum_{i=1}^K V_i \geq K c \right)  \\
&=& \Pr \left( \frac{1}{K} \sum_{i=1}^K V_i \geq  p + \epsilon \right)\ ,\notag
\end{eqnarray}
where we take $\epsilon = 1 / |\eta|^g$.  Recall the additive Chernoff bound:
\beq
\Pr \left(\frac{1}{K} \sum_{i=1}^{K} Y_i \geq p + \epsilon \right)  \leq e^{-K D(p+\epsilon \Vert p )}\ ,
\eeq
where 
\beq
D(x \Vert y) = x \ln \frac{x}{y} + (1-x) \ln \frac{1-x}{1-y}
\eeq 
is the Kullback-Leibler divergence.  Expanding $D(p+\epsilon \Vert p ) = \frac{\epsilon^2}{2 p(1-p)} + O(\epsilon^3)$, we see that if $K = \epsilon^{-2-\veps}$, where $0< \veps \ll 1$, then we can exponentially suppress the probability that Arthur is fooled by Merlin while keeping $K=\text{poly}(|\eta|)$. 

\section{Perturbative Gadgets}

\label{sec:pert-gadgets}

In this section we review the subject of perturbative gadgets, which have played an important role in the reduction of the locality of interactions in the proofs of QMA completeness and the universality of AQC. These tools are generally useful. Our discussion is based primarily on ~\cite{Jordan:08} [see also \cite{Bravyi:2008fk}]. To set up the appropriate tools we first briefly review degenerate perturbation theory.

\subsection{Degenerate Perturbation Theory \`a la  \cite{Bloch1958329}} 
\label{sec:perturbationTheory}

Consider $H  = H_0 + \lambda V$ where $H_0$ has a $d$-dimensional degenerate ground subspace $\mE_0$ with energy 0.  Let $\ket{\psi_1}, \dots \ket{\psi_d}$ be the lowest $d$ energy eigenstates of $H$ with energies $E_1, \dots, E_d$, and let their span define the subspace $\mE$.  The goal is to define a perturbative expansion (in $\lambda$) for the effective Hamiltonian $H_{\mathrm{eff}}$ of $H$, defined as:
\beq
H_{\mathrm{eff}}(H,d) = \sum_{j=1}^d E_j \ket{\psi_j}\bra{\psi_j} \ .
\label{eq:Heff}
\eeq
We will show that this expansion converges provided $\lambda$ satisfies
\beq \label{eqt:perturbationCondition}
\Vert \lambda V \Vert < \gamma /4 \ ,
\eeq
where $\gamma$ is the gap to the first excited state of $H_0$.   We first show how to construct this effective Hamiltonian in terms of other, more convenient operators.

Let $P_0$ be the projection onto $\mE_0$, and define:
\beq
\ket{\alpha_j} = P_0 \ket{\psi_j} \ , \quad j = 1, \dots, d\ .
\eeq
For $\lambda$ sufficiently small [this will amount to satisfying Eq.~\eqref{eqt:perturbationCondition}], the states $\left\{ \ket{\alpha_j} \right\}_{j = 1}^d$ are linearly independent since the states $\left\{ \ket{\psi_j} \right\}_{j = 1}^d$ are only slightly perturbed from the eigenstates of $H_0$.  Note that the states $\left\{\ket{\alpha_j}\right\}_{j = 1}^d$ are not necessarily orthogonal or normalized.  There exists an operator $\mU$ 
such that:
\bes
\begin{align}
\mU \ket{\alpha_j} & = \ket{\psi_j} \ , \quad j = 1, \dots, d \\
\mU \ket{\phi} & = 0 \ , \quad \forall \ket{\phi} \in \mE_0^{\perp}\ .
\end{align}
\ees
This means that:
\begin{eqnarray} \label{eqt:P0U}
P_0 \mU \ket{\alpha_j} &=& P_0 \ket{\psi_j} = \ket{\alpha_j} \nonumber \\
& \Longrightarrow& P_0^2 \mU \ket{\alpha_j} = P_0 \mU \ket{\alpha_j} = P_0 \ket{\alpha_j} \nonumber \\
&\Longrightarrow& P_0 \mU = P_0 \ .
\end{eqnarray}
Let $\tilde{\mU}$ be the operator satisfying:
\bes
\begin{align}
\tilde{\mU} \ket{\psi_j} & = \ket{\alpha_j} \ , \quad j = 1, \dots, d \\
\tilde{\mU} \ket{\phi} & = 0 \ , \quad \forall \ket{\phi} \in \mE^{\perp} \ .\label{eqt:notP0}
\end{align}
\ees
Note that $\tilde{\mU}$ is not the inverse of $\mU$ because $\mU$ is not invertible on the entire Hilbert space.  Also, $\tilde{\mU}$ is not $P_0$ because of Eq.~\eqref{eqt:notP0} (it annihilates all states outside of $\mE$).  Note that:
\beq
\mU P_0 \tilde{\mU} \ket{\psi_j} = \ket{\psi_j} \ .
\label{eq:UP0Utilde}
\eeq

Now define:
\beq \label{eqt:perturbationA}
\mA = \lambda P_0 V \mU\ .
\eeq
Note that the states $\left\{ \ket{\alpha_i} \right\}_{i=1}^d$ are right eigenvectors of $\mA$ with eigenvalues $E_1, \dots, E_d$ respectively:
\begin{align}
\mA \ket{\alpha_j} &= \lambda P_0 V \ket{\psi_j} = P_0 \left( H_0 + \lambda V \right) \ket{\psi_j} = P_0 E_j \ket{\psi_j} \notag \\
&= E_j \ket{\alpha_j}\ ,
\end{align}
where we used that $P_0 H_0 = 0$ because the eigenvalue of the ground subspace of $H_0$ is zero.  The effective Hamiltonian associated with $H$ can now be constructed using $\mU, \tilde{\mU}, \mA$:
\beq \label{eqt:HeffA}
H_{\mathrm{eff}}(H, d) = \mU \mA \tilde{\mU}\ .
\eeq
To see this note that:
\bes
\begin{align}
\mU \mA \tilde{\mU} \ket{\phi} & = 0 \ , \quad \forall \ket{\phi} \in \mE^{\perp} \\
\mU \mA \tilde{\mU} \ket{\psi_j} & = \mU \mA  \ket{\alpha_j}  = E_j \mU \ket{\alpha_j} = E_j \ket{\psi_j} \ ,
\end{align}
\ees
which is identical to the action of $H_{\mathrm{eff}}$ on a complete set of vectors.  The strategy is now to find a perturbative expansion for $\mU$ (we will not need the explicit expansion of $\tilde{\mU}$, so we do not provide it here), construct $\mA$ using Eq.~\eqref{eqt:perturbationA}, find $\left\{\ket{\alpha_j}\right\}_{j=1}^d$ and $\left\{ E_j \right\}_{j=1}^d$ as, respectively, the right eigenvectors and eigenvalues of $\mA$, and apply $\mU$ to $\ket{\alpha_j}$ to get a perturbative expansion for $\ket{\psi_j}$.

It can be shown that the desired perturbative expansion of $\mU$ and $\mA$ is given by:
\bes  \label{eqt:SeriesU}
\begin{align} 
\mU & = P_0 + \sum_{m=1}^\infty \mU_m  \\
\mA & = P_0 V  \sum_{m=1}^{\infty} \mU_m = \sum_{m=1}^{\infty} \mA_m\  ,
 \label{eqt:seriesA}
\end{align}
\ees
where
\bes
\begin{align}
\mU_m &= \!\!\!\!\!\sum_{\substack{\ell_1 \geq 1, \ell_2 \geq 0, \dots, \ell_m \geq 0 \\ \ell_1 + \dots + \ell_m = m \\ \ell_1 + \dots + \ell_p \geq p, \,\,1 \leq p\leq m-1}} \!\!\!\!\!\left( S_{\ell_1} \lambda V \right) \left( S_{\ell_2} \lambda V \right) \dots ( S_{\ell_m} \lambda V ) P_0\ , \\
\label{eqt:Sell}
S_{\ell} &= \left\{ \begin{array}{rl}
 \frac{1}{\left(-H_0\right)^{\ell}}(\ident-P_0) \ , & \ell > 0 \\
- P_0 \ , & \ell = 0
\end{array} \right. \ .
\end{align}
\ees

The series in Eq.~\eqref{eqt:SeriesU} converges for $\Vert \lambda V \Vert < \gamma /4$.  To see this note that:
 \begin{align}
 \Vert \mU \Vert  & = \Vert \mU_0 + \sum_{m=1}^{\infty} \mU_m \Vert \leq \Vert \mU_0 \Vert + \sum_{m=1}^{\infty} \Vert \mU_m \Vert \notag \\
 & \leq 1 + \sum_{m=1}^{\infty} \lambda^m \sum\nolimits ' \Vert S_{\ell_1} V S_{\ell_2} \dots S_{\ell_m} V P_0 \Vert \\
 & \leq 1 + \sum_{m=1}^{\infty} \lambda^m \sum\nolimits'  \Vert S_{\ell_1} \Vert   \dots \Vert S_{\ell_m}\Vert \Vert V \Vert \ , \notag
\end{align}
where the sum $\sum\nolimits'$  involves summing all the different ways to add up to $m$ while satisfying convexity, i.e., $\ell_1 + \ell_2 + \dots \ell_p \geq p$.  
Because of the form of $S_\ell$ [Eq.~\eqref{eqt:Sell}], we have:
\beq
\Vert S_{\ell} \Vert = \left(\frac{1}{E_1^{(0)} } \right)^{\ell}= \frac{1}{\gamma^{\ell}}\ ,
\eeq
where $E_1^{(0)}$ is the energy of the first excited state of $H_0$ (corresponds to the state that minimizes $H_0 Q_0$ to calculate the operator norm).  Therefore, we have:
\beq
\Vert \mU \Vert \leq 1 + \sum_{m=1}^{\infty} \lambda^m \sum\nolimits' \frac{ \Vert V \Vert^m}{\gamma^m} \ .
\eeq
The sum $\sum\nolimits'$ is less than the number of ways to add up to $m$ using $m$ non-negative integers, which is given by $2m-1 \choose m$.  However, since $\sum_{j=0}^{2m-1} {2m-1 \choose j }= 2^{2m-1}$, it is clear that ${2m-1 \choose m} \leq 2^{2m-1}$.  Therefore, we can upper bound the sum with this value:
\beq
\Vert \mU \Vert \leq 1 + \sum_{m=1}^{\infty} 2^{2m-1} \frac{ \Vert \lambda V \Vert^m}{\gamma^m} \ .
\eeq
This series converges if the condition for $\lambda$ in Eq.~\eqref{eqt:perturbationCondition} is satisfied.
\subsection{Perturbative Gadgets}
\label{sec:gadgets}

For a $k$-local target Hamiltonian $H^{\mathrm{T}}$, the goal is to construct a $2$-local ``gadget" Hamiltonian $H^{\mathrm{G}}$, whose low energy spectrum (captured by an effective Hamiltonian $H_{\mathrm{eff}}$) approximates the spectrum of $H^{\mathrm{T}}$.  In order to do so, we shall use the expression in Eq.~\eqref{eqt:HeffA} for the effective Hamiltonian in terms of the operators $\mU$ and $\mA$ and use their perturbative expansion from the previous subsection.  We shall show that for our gadget Hamiltonian, the perturbative expansion of the effective Hamiltonian matches that of the target Hamiltonian.

The perturbative gadget we review here uses a strongly bound set of ancillas, coupled to the target qubits via
weaker interactions, where the latter are treated as a perturbation. $H^{\mathrm{T}}$ is then generated in
 low order perturbation theory of the combined system consisting of both ancilla and target qubits. Such
gadgets first appeared in the proof of QMA-completeness of the $2$-local Hamiltonian problem via a reduction from
$3$-local Hamiltonian, where they were used to construct effective $3$-body interactions from $2$-body ones \cite{KempeGadget}.

Let $H_s$ denote a $k$-local term. For the $i$-th qubit in $H_s$, we associate an arbitrary direction in $\mathbb{R}^3$ denoted by $\hat{n}_{s,i}$.  
A general $k$-local target Hamiltonian acting on $n$ qubits can then be expressed as:
\beq \label{eqt:Hcomp}
H^{\mathrm{T}} = \sum_{s=1}^r c_s H_s\ ,
\eeq
with $H_s = \sigma_{s,1} \sigma_{s,2} \dots \sigma_{s,k}$ where $\sigma_{s,j} = \hat{n}_{s,j} \cdot \vec{\sigma}_{s,j}$.  The goal is to simulate $H^{\mathrm{T}}$ using only 2-local interactions.  Toward this end,  introduce $k$ ancilla qubits for each $H_s$, for a total of $r k$ ancilla qubits.  Define
\bes 
\label{eqt:gadget}
\begin{align} 
H^{\mathrm{G}} & = H^{\mathrm{A}} + \lambda V = \sum_{s=1}^r H^{\mathrm{A}}_s + \lambda \sum_{s=1}^r V_s\\
H^{\mathrm{A}}_s  &  = \sum_{i<j}^k  \frac{1}{2} \left( \ident - Z_{s,i} Z_{s,j} \right) \\
V_s & = \sum_{j=1}^k c_{s,j} \sigma_{s,j} \otimes X_{s,j} \\
c_{s,j} & = \left\{ \begin{array}{lr}
c_s \ , & j = 1 \\
1 \ , & j \neq 1
\end{array} \right.\ ,
\end{align}
\ees
where $X_{s,j}, Z_{s,j}$ are the Pauli-($x,z$) operators on the $j$-th ancilla qubit of $H_s$.  Note that the ground state of $H_s^{\mathrm{A}}$ is given by the span of $\left\{ \ket{0_1 \dots 0_k}_s^{\mathrm{A}}, \ket{1_1 \dots 1_k}_s^{\mathrm{A}} \right\}$.

Consider the $k$-local ancilla operator $X_s \equiv X_{s,1} \otimes X_{s,2} \otimes \dots \otimes X_{s,k}$.  This operator clearly commutes with $H_{\mathrm{G}}$.
Therefore $H_{\mathrm{G}}$ and the set of operators $\left\{ X_{s} \right\}_{s=1}^r$ share a set of eigenstates.  The operator $X_s$ has eigenvalues $\pm 1$, each with degeneracy $2^{k-1}$ (to see this simply write $X_s$ in the basis $\left\{ \ket{\pm}_{s,1} \otimes \ket{\pm}_{s,2} \otimes \dots \otimes \ket{\pm}_{s,k}\right\}$).  Therefore, $H^{\mathrm{G}}$ can be block diagonalized into $2^r$ blocks, where each block corresponds to a fixed $X_{s} = \pm 1$ for $s = 1, \dots, r$ with dimension $2^n 2^{r(k-1)}$.  Let $H^{\mathrm{G}}_{+}$ denote the block with $X_s = 1 \ , \forall s$.  

Note that since $H^{\mathrm{G}}_+$ will be used to approximate $H^{\mathrm{G}}$, the system will need to be initialized to have $X_s = 1 \ , \forall s$.  The eigenstate of $X_s$ with eigenvalue $1$ is given by
\beq \label{eqt:plusState}
\ket{+}_s = \frac{1}{\sqrt{2}} \left( \ket{0_1 \dots 0_k}_s + \ket{1_1 \dots 1_k}_s \right)\ ,
\eeq
so the ancilla qubits must be initialized to be in the state $\bigotimes_{r=1}^s \ket{+}_s$.

We wish to show that the low energy spectrum of $H^{\mathrm{G}}_+$ approximates the spectrum of $H^{\mathrm{T}}$.  Our task is to calculate $H_{\mathrm{eff}}(H^{\mathrm{G}}_+,2^n)$ [in the notation of Eq.~\eqref{eq:Heff}] perturbatively to $k$-th order in $\lambda$.  $\lambda V$ will perturb the ground subspace of $H^{\mathrm{A}}$ in two ways:
\begin{enumerate}
\item It shifts the energy of the entire subspace;
\item It splits the degeneracy of the ground subspace beginning at $k$-th order in perturbation theory.  It is this splitting that will allow us to mimic the spectrum of $H_{\mathrm{T}}$.
\end{enumerate}
We analyze the shift and splitting separately.  To do this, define:
\beq
\tilde{H}_{\mathrm{eff}}(H,d,\Delta) \equiv H_{\mathrm{eff}}(H,d) - \Delta \Pi\ ,
\eeq
where $\Pi$ is the projection onto the space spanned by $\left\{ \ket{E_j} \right\}_{j=1}^d$.  Note that the eigenstates of $\tilde{H}_{\mathrm{eff}}$ and $H_{\mathrm{eff}}$ are identical, and the energy gaps between energy levels are identical too.

Let us start with the case where $r=1$, i.e., $H^{\mathrm{T}} = \sigma_1 \sigma_2 \dots \sigma_k$, so that 
$H^{\mathrm{A}} = \sum_{i=1}^k \sum_{j=i+1}^k \frac{1}{2} \left( 1- Z_i Z_j \right)$,  and $V = \sum_{j=1}^k \sigma_j \otimes X_j$.
We first wish to construct $\mA$ [Eq.~\eqref{eqt:perturbationA}] for $H^{\mathrm{G}}_{+}$.  Note that $H^{\mathrm{A}}$ has a ground state of zero energy (corresponds to all qubits with $Z_i = 1$ or all qubits with $Z_i = -1$), and the first excited state has energy $\gamma = k-1$ (let $Z_1 = -1$, all the rest are $+1$).  Furthermore:
\beq
\Vert V \Vert = \Vert \sum_{j=1}^{k} \sigma_j \otimes X_j \Vert 
\leq \sum_{j=1}^{k} \Vert \sigma_j \otimes X_j \Vert = k \ .
\label{eq:normV}
\eeq
Therefore, by Eq.~\eqref{eqt:perturbationCondition}, the perturbative expansion will converge if $\lambda < \frac{k-1}{4k}$.  Because of the form of $\mA$ [Eq.~\eqref{eqt:seriesA}], all $\mA_m$ terms are sandwiched between $P_0$ operators.  Thus, all non-zero terms in $\mA$ must take states in $\mE_0$ and return them to states in $\mE_0$.  Since we have restricted  to the $X = \otimes_{i=1}^k X_i = + 1$ sector, $\mE_0$ is restricted to have the ancilla qubits in the $\ket{+}$ state [Eq.~\eqref{eqt:plusState}].  Therefore, we can write:
\beq
P_0 = \ident \otimes P_+\ ,
\eeq
where $P_+$ is the projection onto the $\ket{+}$ ancilla state.

Each term in $V = \sum_{j=1}^k \sigma_j \otimes X_j$ only flips a single ancilla qubit.  Therefore, in order for $\mA$ take a state out of $\mE_0$ and return it, the power of $V$ must either flip all ancilla qubits or flip some and return them back to their original value.  The former process (flipping all qubits) first happens at $k$-th order in perturbation theory.  The latter process (flipping and returning) can happen at lower orders than $k$, but $\mA$ is then  proportional to $P_0$ since the product of $V$'s effectively cancel.  To see how this works, consider $\mA$ up to second order for $k>2$.  From the perturbation expansion [Eq.~\eqref{eqt:seriesA}] we have:
\beq
\mA_{\leq 2} = \lambda P_0 V P_0 + \lambda^2 P_0 V S_1 V P_0\ ,
\eeq
but $P_0 V P_0 = 0$ since $V\ket{+}$ is orthogonal to $\ket{+}$.  On the other hand, $V P_0$ takes the system to a state with energy $k-1$ for $H^{\mathrm{A}}$, so $S_1 V P_0 = - VP_0/(k-1)$.  Therefore: 
\beq
\mA_{\leq 2} = - \frac{\lambda^2}{k-1} P_0 V^2 P_0\ .
\eeq 
Now note that:
\beq
V^2 = \sum_i \left( \sigma_i \otimes X_i \right)^2 + \sum_{i\neq j} \left( \sigma_i \otimes X_i \right) \left( \sigma_j \otimes X_j \right)\ .
\eeq
The cross-terms are annihilated by $P_0 \cdot P_0$ since they take the state out of $\mE_0$.  The diagonal term is proportional to the identity on the ancilla qubits, so we have:
\beq
\mA_{\leq 2} = - \frac{\lambda^2}{k-1} \Omega P_0\ ,
\eeq
where $\Omega$ is an operator that depends on the particular orientation of the $\sigma_j$'s.  This argument extends to order $k-1$ so that:
\beq
\mA_{\leq k-1} = \sum_{m \ \mathrm{even}} \lambda^m \Omega_m P_0 \ .
\eeq
At order $k$, something new happens.  There are now cross-term that involve all $X_i$'s once, i.e., 
\[
\lambda^k P_0 \left( \sigma_1 \otimes X_1 \right) S_1 \left( \sigma_2 \otimes X_2 \right) S_1 \dots S_1 \left( \sigma_k \otimes X_k \right) P_0\ .
\]
The $S_1$ operator measures the successive change in energy of the system to give an overall constant of:
\begin{eqnarray}
\left( - \frac{1}{k-1} \right) \left( - \frac{1}{2 (k-2)} \right) \dots ( - \frac{ 1}{(k-1)1} ) \nonumber \\
= \prod_{j=1}^{k-1} \left( -\frac{1}{j(k-j)} \right) = \frac{(-1)^{k-1}}{(k-1)!^2}\ .
\end{eqnarray}
Therefore, the cross-terms (of which there are $k!$, since the operators can be multiplied in any order) then take the form:
\beq
- \frac{(-\lambda)^k k!}{(k-1)!^2}P_0 \left( \sigma_1 \dots \sigma_k \otimes X \right) P_0\ .
\eeq
Thus:
\beq
\mA_{\leq k} = f(\lambda) P_0 - \frac{k (-\lambda)^k}{(k-1)!} P_0 \left( H^\mathrm{T} \otimes X \right) P_0\ ,
\eeq
where $f(\lambda)$ is some $k$-th order polynomial in $\lambda$, with coefficients that depend on $H^\mathrm{T}$.  Using the form of $H_{\mathrm{eff}}$ in Eq.~\eqref{eqt:HeffA}, we have:
\begin{eqnarray}
\label{eq:212}
H_{\mathrm{eff}} (H^\mathrm{G}_-, 2^n )  &=& f(\lambda) \mU P_0 \tilde{\mU}   \\
&& \hspace{-2cm} - \mU \left(  \frac{k (-\lambda)^k}{(k-1)!} P_0 \left( H^\mathrm{T} \otimes X \right) P_0 + O(\lambda^{k+1}) \right) \tilde{\mU} \ .\notag
\end{eqnarray}
Recall that $\mU P_0 \tilde{\mU} \ket{\psi_j} = \ket{\psi_j}$ [Eq.~\ref{eq:UP0Utilde}]
so $\mU P_0 \tilde{\mU}$ acts as the identity in $\mE$ so that the first term in Eq.~\eqref{eq:212} can be dropped.  Furthermore, we can replace $\mU$ and $\tilde{\mU}$ in the second term by their $\lambda^0$ counterpart since we are only keeping terms to order $k$ and the term in the parenthesis is already of order $k$.  Therefore:
\begin{align}
\tilde{H}_{\mathrm{eff}}  (H^\mathrm{G}_+, 2^n , f(\lambda) ) & = - \frac{k (-\lambda)^k}{(k-1)!} P_0 \left( H^\mathrm{T} \otimes X \right) P_0 \notag  \\
&\quad + O(\lambda^{k+1})
\label{eq:Htilde-eff} \\
& = - \frac{k (-\lambda)^k}{(k-1)!}  \left( H^\mathrm{T} \otimes P_+ \right) + O(\lambda^{k+1})\ . \notag
\end{align}
This shows that the target Hamiltonian $H^\mathrm{T}$ appears as the leading order term in the effective Hamiltonian that describes the $2^n$-Hilbert space of the $n$ target qubits, albeit with a diminished magnitude of order $\lambda^k/(k-1)!$.

Let us now consider the general $r$ case, i.e., the Hamiltonian in Eq.~\eqref{eqt:Hcomp}.  We note that just as in the $r=1$ case, $H^{\mathrm{A}}$ again has an energy gap of $k-1$. Generalizing from Eq.~\eqref{eq:normV}, the perturbative expansion then converges for:
\beq
\lambda < \frac{k-1}{4 \Vert V \Vert}\ .
\label{eq:lambda-cond}
\eeq
In the sector where $X_s = +1$, $H^{\mathrm{A}}$ has the state $\otimes_{s=1}^r \ket{+}_s$ as a ground state.  Since $H^{\mathrm{A}}$ acts as the identity on the computational qubits, the ground state is $2^n$-fold degenerate.

In the perturbation expansion for $\mA$, products of $V$ again appear.  Each $V_s$ acts on a different ancilla register.  Therefore, at order $k$, cross-terms of different $V_s$'s cannot flip all $k$ ancilla qubits in a register, so they are annihilated by $P_0 \cdot P_0$.  The only cross-terms that contribute are $k$ products of a given $s$ where each ancilla qubit appears once.  Therefore, the natural generalization of the previous result is recovered, namely, Eq.~\eqref{eq:Htilde-eff} continues to hold with $H^\mathrm{T}$ replaced by the sum over $r$ terms as in Eq.~\eqref{eqt:Hcomp}, 
where again $f(\lambda)$ is some polynomial in $\lambda$ of order $k$ with coefficients that depend on $c_s H_s$, and where $P_+$ is the projector onto $\otimes_s \ket{+}_s$.

Note that the convergence condition~\eqref{eq:lambda-cond} 
requires the interaction term $V$ to be stronger than the effective interaction it generates, which scales as $\lambda^k$ [as can be seen from Eq.~\eqref{eq:Htilde-eff}]. This may pose implementation difficulties, since a practical device is likely to have only a limited range of interaction strengths.  Weaker gadgets can be implemented that circumvent this problem, albeit at the cost of a larger overhead of ancillary qubits \cite{Cao:2014}. The idea is to replace strong interactions by repetition of interactions with ``classical" ancillas.
Additional gadgets simplifications and resource reductions were proposed in \cite{PhysRevA.91.012315}.

%

\end{document}